\documentclass[10pt,journal,compsoc]{IEEEtran}
\usepackage[utf8]{inputenc}\DeclareUnicodeCharacter{2212}{-}
\usepackage[nolist]{acronym}
\usepackage{cite}
\usepackage[hidelinks,bookmarks=false]{hyperref}
\usepackage{amsmath,amssymb,amsfonts, amsthm}
\usepackage{graphicx}
\usepackage{textcomp}
\usepackage[table]{xcolor}
\usepackage{tikz}
\usepackage{pgfplots}
\usepackage{multirow}
\usepackage{makecell}
\usepackage{algorithm}
\usepackage{algpseudocode}
\ifCLASSOPTIONcompsoc
    \usepackage[caption=false,font=normalsize,labelfon
t=sf,textfont=sf]{subfig}
\else
\usepackage[caption=false,font=footnotesize]{subfi
g}
\fi
\usepackage{framed}
\usepackage{soul}
\usepackage{bm}
\usepackage[normalem]{ulem}
\usepackage{pdfcomment}
\usepackage{import}
\usepackage{tkz-euclide}
\usepackage{stmaryrd}
\usepackage{xparse}
\usepackage{mathtools}
\usepackage{cite}
\usepackage[shortlabels]{enumitem}
\usepackage{selectp}
\usepackage{arydshln}
\usepackage{fancyhdr}

\usepgfplotslibrary{external}
\tikzset{every picture/.style={line width=0.75pt}}
\usetikzlibrary{positioning, shapes, trees, automata, calc, fit, automata, decorations.pathmorphing, math}
\usetikzlibrary{arrows.meta}
\def\BibTeX{{\rm B\kern-.05em{\sc i\kern-.025em b}\kern-.08em
		T\kern-.1667em\lower.7ex\hbox{E}\kern-.125emX}}
	     
\newcommand{\tagg}[1]{%
	\colorbox{green}{#1}
}

\theoremstyle{definition}

\newtheorem{theorem}{Theorem}
\newtheorem{lemma}{Lemma}
\newtheorem{proposition}{Proposition}
\newtheorem{corollary}{Corollary}

\tikzset{
    database/.style={
        path picture={
            \draw (0, 1.5*\database@segmentheight) circle [x radius=\database@radius,y radius=\database@aspectratio*\database@radius];
            \draw (-\database@radius, 0.5*\database@segmentheight) arc [start angle=180,end angle=360,x radius=\database@radius, y radius=\database@aspectratio*\database@radius];
            \draw (-\database@radius,-0.5*\database@segmentheight) arc [start angle=180,end angle=360,x radius=\database@radius, y radius=\database@aspectratio*\database@radius];
            \draw (-\database@radius,1.5*\database@segmentheight) -- ++(0,-3*\database@segmentheight) arc [start angle=180,end angle=360,x radius=\database@radius, y radius=\database@aspectratio*\database@radius] -- ++(0,3*\database@segmentheight);
        },
        minimum width=2*\database@radius + \pgflinewidth,
        minimum height=3*\database@segmentheight + 2*\database@aspectratio*\database@radius + \pgflinewidth,
    },
    database segment height/.store in=\database@segmentheight,
    database radius/.store in=\database@radius,
    database aspect ratio/.store in=\database@aspectratio,
    database segment height=0.1cm,
    database radius=0.25cm,
    database aspect ratio=0.35,
}

\tikzset{every picture/.style={line width=0.75pt}}
\tikzset{every picture/.style={line width=0.75pt}}
\tikzset{cross/.style={cross out, draw=black, fill=none, minimum size=2*(#1-\pgflinewidth), inner sep=0pt, outer sep=0pt}, cross/.default={2pt}}
\usetikzlibrary{positioning, shapes, trees, automata, calc, fit, automata, shapes.misc, arrows, patterns}

\newcommand{\ie}{\textit{i.e.~}}

\def\myvspacebeforesec{0cm}
\def\myvspacebeforesubsec{0cm}

\DeclarePairedDelimiterX{\Iintv}[1]{\llbracket}{\rrbracket}{\iintvargs{#1}}
\NewDocumentCommand{\iintvargs}{>{\SplitArgument{1}{,}}m}
{\iintvargsaux#1} %

\def\ben{\begin{equation*}}
\def\een{\end{equation*}}
\def\barr{\begin{array}}
\def\earr{\end{array}}

\DeclareMathAlphabet{\mathpzc}{OT1}{pzc}{m}{it}

\newcommand{\asymCloud}[4]{
\begin{scope}[shift={#1},scale=#3]
\draw (-1.6,-0.7) .. controls (-2.3,-1.1)
and (-2.7,0.3) .. (-1.7,0.3)coordinate(asy1) .. controls (-1.6,0.7)
and (-1.2,0.9) .. (-0.8,0.7) .. controls (-0.5,1.5)
and (0.6,1.3) .. (0.7,0.5) .. controls (1.5,0.4)
and (1.2,-1) .. (0.4,-0.6)coordinate(asy2) .. controls (0.2,-1)
and (-0.2,-1) .. (-0.5,-0.7) .. controls (-0.9,-1)
and (-1.3,-1) .. cycle;
\node at ($(asy1)!0.5!(asy2)$) (#4) {#2};
\end{scope}
}

\fancypagestyle{firststyle}
{
	\fancyhead[c]{This article has been accepted for publication in IEEE/ACM Transactions on Networking. This is the author's version which has not been fully edited and
  content may change prior to final publication. Citation information: DOI \href{https://doi.org/10.1109/TNET.2022.3180763}{10.1109/TNET.2022.3180763}}
	\fancyfoot[c]{© 2022 IEEE. Personal use is permitted, but republication/redistribution requires IEEE permission.See https://www.ieee.org/publications/rights/index.html for more information.}
}

\fancypagestyle{secondstyle}
{
	\fancyfoot[c]{© 2022 IEEE. Personal use is permitted, but republication/redistribution requires IEEE permission.See https://www.ieee.org/publications/rights/index.html for more information.}
}

\begin{document}
	%
	\title{Worst-case Delay Bounds in Time-Sensitive Networks with Packet Replication and Elimination}
	
	%
	%
	
	\author{Ludovic Thomas, Ahlem Mifdaoui, Jean-Yves Le Boudec
		\IEEEcompsocitemizethanks{\IEEEcompsocthanksitem L. Thomas and A. Mifdaoui are with ISAE-Supaéro, Université de Toulouse. Toulouse, France.\protect\\
			E-mail: ludovic.thomas@isae-supaero.fr
			\IEEEcompsocthanksitem J-Y. Le Boudec is with the I\&C Departement, EPFL. Lausanne, Switzerland.}
		\thanks{Manuscript received someday; revised some other day}}

	\IEEEtitleabstractindextext{%
		\begin{abstract}
%
%
%
%
%
%
Packet replication and elimination functions are used by time-sensitive networks (as in the context of \acs{IEEE} \acs{TSN} and \acs{IETF} \acs{DetNet}) to increase the reliability of the network.
Packets are replicated onto redundant paths by a replication function.
Later the paths merge again and an elimination function removes the duplicates.
This redundancy scheme has an effect on the timing behavior of time-sensitive networks and many challenges arise from conducting timing analyses.
The replication can induce a burstiness increase along the paths of replicates, as well as packet mis-ordering that could increase the delays in the crossed bridges or routers.
The induced packet mis-ordering could also negatively affect the interactions between the redundancy and scheduling mechanisms such as traffic regulators (as with \aclp{PFR} and \aclp{IR}, implemented by \acs{TSN} \emph{\acl{ATS}}).
Using the network calculus framework, we provide a method of worst-case timing analysis for time-sensitive networks that implement redundancy mechanisms in the general use case, \emph{i.e.}, at end-devices and/or intermediate nodes.
We first provide a network calculus toolbox for bounding the burstiness increase and the amount of reordering caused by the elimination function of duplicate packets.
We then analyze the interactions with traffic regulators and show that their shaping-for-free property does not hold when placed after a packet elimination function.
We provide a bound for the delay penalty when using per-flow regulators and prove that the penalty is not bounded with interleaved regulators.
Finally, we use an industrial use-case to show the applicability and the benefits of our findings.

		\end{abstract}
		
		}

	\maketitle
	\thispagestyle{firststyle}
	\pagestyle{secondstyle}

	\IEEEdisplaynontitleabstractindextext

	%
	\IEEEpeerreviewmaketitle

	\IEEEraisesectionheading{\section{Introduction}\label{sec:introduction}}
	

%
%
%
%


\IEEEPARstart{T}{ime-sensitive networks} were specified by the \ac{DetNet} working group of the \ac{IETF}, as well as the \ac{TSN} task group of the \ac{IEEE}, for supporting safety-critical applications in several domains, such as aerospace \cite{ieeeP8021DPTSN}, automation \cite{iecIECIEEE608022019} and automotive \cite{ieeeDraftStandardLocal2019b}.

As opposed to the \emph{best-effort} service, safety-critical applications require a \emph{deterministic} service \cite[\S 3.1]{finnDeterministicNetworkingArchitecture2019} \cite{farkasTSNBasicConcepts2018a}
with
zero congestion loss, high levels of reliability, bounded out-of-order delivery and 
guarantees on the
end-to-end latency of each flow.
Time-sensitive networks provide this service by relying on a set of redundancy and scheduling mechanisms.
The former reduce the probability of end-to-end losses whereas the later aim to guarantee latency bounds
\cite{farkasTSNBasicConcepts2018a}.






Verifying these bounds on a network is a known intractable issue for simulators and real-life experiments because worst-case situations are not captured by stochastic metrics \cite[\S 1]{bouillard2018deterministic}.
Therefore, both the \ac{TSN} and the \ac{DetNet} working groups recommend using analytical tools for conducting worst-case timing analyses and for proving the determinism of the network's service \cite[\S L.3, \S N.2]{ieee8021Q} \cite{ietf-detnet-bounded-latency-08}.
Among them, the network-calculus framework \cite{leboudecNetworkCalculusTheory2001} computes latency, jitter, and backlog bounds, assuming that the sources [resp., the servers] respect some contract of maximum traffic generation [resp., of minimum service]. It has been used to prove certification requirements in avionics \cite{AFDX}.
%
%
%
%
%
The worst-case timing performance of time-sensitive networks, when focusing only on scheduling mechanisms, has been widely analyzed in the literature ~\cite{mohammadpourLatencyBacklogBounds2018,leboudecTheoryTrafficRegulators2018, thomasCyclicDependenciesRegulators2019,thomasTimeSynchronizationIssues2020,maile2020network,zhaoQuantitativePerformanceComparison2021}.

\begin{table}[b]
    \caption{\label{tab:introduction:comparison-names} Main Acronyms Used in the Paper and Comparison with the Terms of the Working Groups.}
    \resizebox*{\linewidth}{!}{
        \def\mll{0.35\linewidth}
\def\mlll{0.27\linewidth}
\begin{tabular}{r l|l|l}
    \multicolumn{2}{l|}{In this paper} & \begin{minipage}{\mll} \begin{flushleft}Term used in \acs{DetNet} \cite{finnDeterministicNetworkingArchitecture2019} \end{flushleft} \end{minipage} & \begin{minipage}{\mll} \begin{flushleft} Term used in \acs{TSN} \cite{TimeSensitiveNetworkingTSNa}   \end{flushleft} \end{minipage}    \\
    \hline
    \textbf{\acs{PREF}} & \hspace{-0.3cm} \begin{minipage}{\mlll} \begin{flushleft} Packet replication and elimination functions \end{flushleft} \end{minipage} & \begin{minipage}{\mll} \begin{flushleft} Packet replication and elimination functions \end{flushleft} \end{minipage} & \begin{minipage}{\mll} \begin{flushleft} \textbf{\acs{FRER}:} Frame replication and elimination  for redundancy \cite{IEEEStandardLocal2017a}  \end{flushleft} \end{minipage} \\
    \hline
    \textbf{\acs{PRF}} & \hspace{-0.3cm} \begin{minipage}{\mlll} \begin{flushleft} Packet replication function \end{flushleft} \end{minipage} & \begin{minipage}{\mll} \begin{flushleft} Packet replication function \end{flushleft} \end{minipage} & \begin{minipage}{\mll} \begin{flushleft} Stream splitting function \cite[\S 7.7]{IEEEStandardLocal2017a} \end{flushleft} \end{minipage} \\
    \hline
    \textbf{\acs{PEF}} & \hspace{-0.3cm} \begin{minipage}{\mlll} \begin{flushleft} Packet elimination function \end{flushleft} \end{minipage} & \begin{minipage}{\mll} \begin{flushleft} Packet elimination function  \end{flushleft} \end{minipage} & \begin{minipage}{\mll} \begin{flushleft} Sequence recovery function \cite[\S 7.4.2]{IEEEStandardLocal2017a} \end{flushleft} \end{minipage} \\
    \hline
    \textbf{\acs{POF}} & \hspace{-0.3cm} \begin{minipage}{\mlll} \begin{flushleft} Packet  ordering function \end{flushleft} \end{minipage}  & \begin{minipage}{\mll} \begin{flushleft} Packet ordering function  \end{flushleft} \end{minipage}  &  \cellcolor{black!25} \begin{minipage}{\mll} \begin{flushleft} Does not exist in \acs{TSN} (March 2022) \end{flushleft} \end{minipage} \\
\end{tabular}
    }
\end{table}

The main issue addressed in this paper is the effect of the redundancy mechanisms on the delay guarantees in time-sensitive networks.
These redundancy mechanisms, such as \emph{\acf{FRER}} \cite{IEEEStandardLocal2017a} in \ac{TSN} and \emph{\acf{PREOF}} \cite{finnDeterministicNetworkingArchitecture2019} in \ac{DetNet}, 
%
decrease the end-to-end packet-loss ratio by distributing ``\emph{the contents of [\dots] flows over multiple paths in time and/or space, so that the loss of some of the paths does need not cause the loss of any packets}''\cite{finnDeterministicNetworkingArchitecture2019}.
To do so, the \ac{DetNet} \acf{PRF}
replicates each incoming packet into several outgoing packets that can take different paths (Fig.~\ref{fig:prob-statement:three-sections}).
The paths then merge and multiple copies of the packet (the replicates) reach a \acf{PEF} that forwards only the first replicate and eliminates the subsequent ones (the duplicates).
The \ac{PEF} generally relies on a sequence number in the packet header to identify the replicates \cite{finnDeterministicNetworkingArchitecture2019}.
In Table~\ref{tab:introduction:comparison-names}, we compare the terms used in \ac{IETF} \ac{DetNet} and \ac{IEEE} \ac{TSN}.
The main acronyms used thorough the paper are also listed in Table~\ref{tab:introduction:comparison-names}.

The \ac{PREF} could increase the worst-case end-to-end latency of the flows: \cite[\S 3.1]{finnDeterministicNetworkingArchitecture2019} recalls that their use is
``\emph{constrained by the need to meet the users' latency requirements}''.
Therefore, understanding how \ac{PREF} affect the worst-case latency guarantees is fundamental to: (i) determine the applicability of \ac{PREF} in industrial networks; (ii) perform trade-offs between latency and loss-ratio requirements; and (iii) design networks with stringent requirements on both aspects.

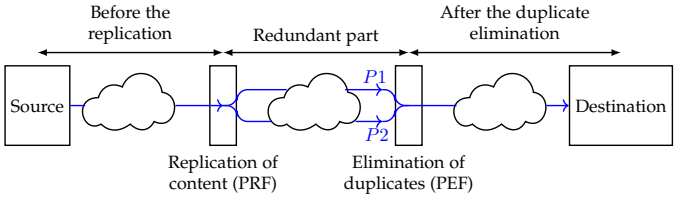
\begin{figure}\centering
    \resizebox*{\linewidth}{!}{\begin{tikzpicture}
    \tikzstyle{n} = [draw, minimum height=1.5cm]
    \tikzstyle{g} = [blue]
    \tikzstyle{r} = [blue]
    \def\mscale{0.5}
    \tikzstyle{ms} = [xshift=1cm]
    \node[n] at (0,0) (source) {Source};
    \asymCloud{(2,0)}{}{\mscale}{c1};
    \node[n, minimum width=0.5cm, label=below:\makecell{Replication of\\content (PRF)}] at (3.5,0) (prf) {};
    \asymCloud{(5.5,0)}{}{\mscale}{c2};
    \node[n, minimum width=0.5cm, label=below:\makecell{Elimination of\\duplicates (PEF)}] at (7,0) (pef) {};
    \asymCloud{(9,0)}{}{\mscale}{c3};
    \node[n] at (11,0) (destination) {Destination};
    \draw[-, blue] (source.east) -- (0.9,0); 
    \draw[->, blue] (2.6,0) -- (prf.center); 
    \draw[-, blue] (pef.center) -- (7.9,0);
    \draw[->, blue] (9.6,0) -- (destination.west);

    \draw[-, g, rounded corners=0.2cm] (prf.center) -- (prf.east) |- (4.7,0.3);
    \draw[-, r, rounded corners=0.2cm] (prf.center) -- (prf.east) |- (4.4,-0.3);

    \draw[->, g] (5.8,0.3) -- (6.5,0.3) node[pos=0.8,above] {$P1$};
    \draw[->, r] (6,-0.3) -- (6.5,-0.3) node[pos=0.8, below] {$P2$};
    \draw[-, g, rounded corners=0.2cm] (6.5,0.3) -| (pef.west) -- (pef.center);
    \draw[-, r, rounded corners=0.2cm] (6.5,-0.3) -| (pef.west) -- (pef.center);

    \draw[latex-latex] (0,1) -- (3.5,1) node[pos=0.5, above] {\makecell{Before the\\replication}};
    \draw[latex-latex] (3.5,1) -- (7,1) node[pos=0.5, above] {Redundant part};
    \draw[latex-latex] (7,1) -- (11,1) node[pos=0.5, above] {\makecell{After the duplicate\\elimination}};
\end{tikzpicture}%
}%
    \caption{\label{fig:prob-statement:three-sections} 
    The three sections of a replicate's path: before, in, and after the redundant part.}
\end{figure}

Three main challenges arise from conducting worst-case timing analyses of \acs{PREF}. 
First, the replication of packets through the network can induce a burstiness increase along the paths of replicate packets, which leads to increasing delay and backlog bounds in the crossed nodes.
Second, the traffic exiting the \ac{PEF} can exhibit both an increased burstiness and a mis-ordering of the packets.
This can lead to increased delay bounds in the nodes placed after the \ac{PEF}.
Third, the coexistence of the packet mis-ordering with the burstiness increase could negatively affect the behavior of the devices that have been designed for tackling each issue individually.
For example, \acfp{POF} have been specified in \ac{IETF} \ac{DetNet} for removing only the packet mis-ordering.
Similarly, traffic regulators (also called \emph{shapers}) are scheduling mechanisms designed for removing only the burstiness increase.
If a traffic regulator (as in \ac{TSN} \acs{ATS}, \emph{\acl{ATS}}) is placed after the \ac{PEF} for removing the burstiness increase caused by the redundancy, then the packet mis-ordering that coexists with this burstiness increase could negatively affect the behavior of the traffic regulator.
The existing worst-case timing analyses of redundancy mechanisms in time-sensitive networks
\cite{heiseRealtimeGuaranteesDependability2018,heiseSAFDXDeterministicHighavailability2014a}
are limited to the assumption of using redundancy mechanisms at the end-systems as with \emph{\ac{AFDX}} \cite{AFDX} and \emph{\ac{PRP}} \cite{iecIEC6243932016}.
%
This assumption discards the main challenges detailed above.
More recent works \cite[\S 4.3.2]{heiseRealtimeGuaranteesDependability2018} \cite{pahlevanRedundancyManagementSafetyCritical2018} based on simulation consider redundancy mechanisms at intermediate nodes.
However, firm conclusions are difficult to draw from simulations as they not cover the worst-case behavior.


Therefore, our primary goal in this paper is to bridge these gaps and to provide a method of worst-case timing analysis for time-sensitive networks that implement redundancy mechanisms in the general use-case, i.e., at end-systems and/or intermediate nodes. Specifically:
\begin{itemize}
\item 
We provide a network-calculus toolbox that enables the computation of upper bounds on the burstiness increase (Theorem~\ref{thm:toolbox:pef-oac}) and on the amount of reordering (Theorem~\ref{thm:toolbox:reordering}) due to the elimination of duplicate packets. 
Theorem~\ref{thm:toolbox:pef-oac} is useful for computing delay bounds in the nodes located after the elimination of duplicates, whereas the bound from Theorem~\ref{thm:toolbox:reordering} can be compared to the application's requirements to decide if the packets should be reordered prior to their delivery.
\item
We analyze the interactions between redundancy mechanisms and traffic regulators. 
We show that the packet mis-ordering due to the elimination of duplicates leads to a bounded increase of the worst-case delay with \acp{PFR} (Theorem~\ref{thm:regulators:pfr-2d}) and to unbounded delays with \acp{IR} (\emph{e.g.}, \ac{TSN} \acs{ATS}) (Theorem~\ref{thm:regulators:ir-instable}).
The problem goes away if the packets are re-ordered after the elimination and before the regulator (Theorem~\ref{thm:regulators:preof-for-free}).
\item 
We conduct  performance analyses for an industrial use-case that highlight the interest of our introduced approach to tighten the delay bounds in comparison to intuitive computation approaches.
\end{itemize}

In Section~\ref{sec:prob-statement}, we illustrate the issues posed by \ac{PREF} using a toy example.
In Section~\ref{sec:related-work}, we relate our proposed approach to the state of the art, 
and we describe the system model in Section~\ref{sec:system-model}.
Our main theoretical contributions are detailed in Sections~\ref{sec:toolbox} and \ref{sec:regulators}, that cover the network-calculus toolbox for redundancy mechanisms and the analysis of the interaction between such mechanisms and traffic regulators.
Finally, we validate our approach on an industrial use-case in Section~\ref{sec:evaluation}.

\vspace{\myvspacebeforesec}
\section{Illustration of the Issues Posed by Packet Replication and Elimination}\label{sec:prob-statement}

In this section, we illustrate the issues posed by \acf{PREF} in time-sensitive networks, as identified in the Introduction.
We first detail the burstiness increase and the mis-ordering introduced by \ac{PREF}.
Afterwards, we focus on the problems arising from the interactions between \ac{PREF} and traffic regulators.

\vspace{\myvspacebeforesubsec}
\subsection{Burstiness and Misordering Introduced by PREFs}
\label{sec:prob-statement:pref-issues}

\begin{figure}[b]\centering
    \resizebox*{\linewidth}{!}{\begin{tikzpicture}
	\node[draw] at (0,0) (s) {\makecell{Packet\\Replication\\Function (PRF)}};
	\node[draw] at (3,0.75) (a) {$[0,1]$};
	\node[draw] at (3,-0.75) (b) {$[6,7]$};
	\node[draw] at (6.75,0) (c) {\makecell{Packet\\Elimination\\Function (PEF)}};
	\node[draw, red, rounded corners=0.5cm, fit={(a)}, inner sep=0.3cm] (cNode) {};
	\node[draw, red, rounded corners=0.5cm, fit={(b)}, inner sep=0.3cm] (dNode) {};
	\node[anchor=west, red] at ([xshift=-0.2cm] dNode.south east) {$D$};
	\node[anchor=west, red] at ([xshift=-0.2cm] cNode.north east) {$C$};

	\draw[->, rounded corners=0.2cm] ([yshift=0.2cm]s.east) -- ([xshift=0.5cm,yshift=0.2cm] s.east) -- ([xshift=-0.5cm] a.west) -- (a.west);
	\draw[->, rounded corners=0.2cm] ([yshift=-0.2cm]s.east) -- ([xshift=0.5cm,yshift=-0.2cm] s.east) -- ([xshift=-0.5cm] b.west) -- (b.west);
	
	\draw[->, rounded corners=0.2cm] (a.east) -- ([xshift=0.5cm] a.east) -- ([xshift=-1.5cm,yshift=0.2cm] c.west) -- ([yshift=0.2cm]c.west) node[pos=0.5,anchor=center] (tA) {};
	\draw[->, rounded corners=0.2cm] (b.east) -- ([xshift=0.5cm] b.east) -- ([xshift=-1.5cm,yshift=-0.2cm] c.west) -- ([yshift=-0.2cm]c.west) node[pos=0.5,anchor=center] (tB) {};
	\draw[-,red] ([yshift=-0.1cm] tA.center) -- ([yshift=0.1cm] tA.center) node[pos=1,above] {\texttt{outC}};
	\draw[-,red] ([yshift=-0.1cm] tB.center) -- ([yshift=0.1cm] tB.center) node[pos=0,below] {\texttt{outD}};

	\draw[red] ([xshift=0.5cm] c.west) ++(90:-0.9) arc (90:-90:-0.9) node[pos=0, anchor=center] (mt) {} node[pos=1, anchor=center] (mtt) {};
	\draw[red] (mt.center) -- ++(1.5cm,0);
	\draw[red, dotted] (mt.center) ++(1.5cm,0) -- ++(0.5cm,0);
	\draw[red] (mtt.center) -- ++(1.5cm,0);
	\draw[red, dotted] (mtt.center) ++(1.5cm,0) -- ++(0.5cm,0);
	\node[red, anchor=south] at (mtt.center) {$F$};

	\draw[red] ([xshift=-1cm] s.west) ++(90:0.5) arc (90:-90:0.5) node[pos=0, anchor=center] (mmt) {} node[pos=1, anchor=center] (mmtt) {};
	\draw[red] (mmt.center) -- ++(-0.1cm,0);
	\draw[red, dotted] (mmt.center) ++(-0.1cm,0) -- ++(-0.5cm,0);
	\draw[red] (mmtt.center) -- ++(-0.1cm,0);
	\draw[red, dotted] (mmtt.center) ++(-0.1cm,0) -- ++(-0.5cm,0);
	\node[red, anchor=south] at (mmt.center) {$B$};


	\draw[->] ([xshift=-1.5cm] s.west) -- (s.west) node[pos=0.8,anchor=center] (tIn) {};
	\draw[-,red] ([yshift=0.1cm] tIn.center) -- ([yshift=-0.1cm] tIn.center) node[pos=0,above] {\texttt{in}};
	\draw[->] (c.east) -- ([xshift=1.5cm] c.east) node[pos=0.5,anchor=center] (tOut) {};
	\draw[-,red] ([yshift=-0.1cm] tOut.center) -- ([yshift=0.1cm] tOut.center) node[pos=1,above] {\texttt{out}};	
\end{tikzpicture}%
}%
    \caption{\label{fig:prob-statement:toy-example} Toy example used thorough the paper. A flow is replicated on two paths, $C$ and $D$, with different delay bounds. The paths then merge into $F$, that removes the duplicates.}
\end{figure}
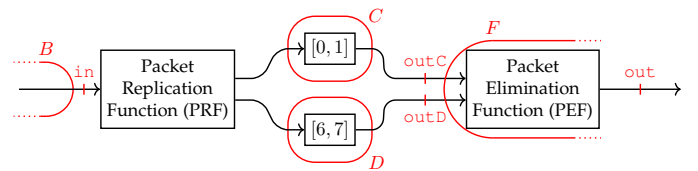


To highlight the effect of \ac{PREF} on the burstiness and packet order, we consider the toy example in Fig.~\ref{fig:prob-statement:toy-example}:
a periodic flow with a rate $r_0$ of one packet per every time unit is replicated at the output of the vertex $B$ and sent over two paths: $C$ [resp., $D$], with a minimum delay of zero time units [resp., six time units] and a maximum delay of one time unit [resp., seven time units].
A possible trace of packets for the toy example is given in Fig.~\ref{fig:prob-statement:double-rate}.
Here, the path through $C$ drops all  Data Units 1 to 6: they are only received through $D$ with a latency of seven time units (7 t.u.).
After 7 t.u., the link through $C$ is available again and the Data Units 7 to 14 are received through both $C$ and $D$, with a latency of 1 [resp., 7] t.u.
The \acs{PEF} receives the sum of ``\texttt{outC}'' and ``\texttt{outD}''.
It drops the duplicates but forwards the packets that contain not-already-seen data units.
Its output is on the Line ``\texttt{out}''.

We observe that the traffic after the \ac{PEF} is much more bursty than before the replication function: between t.u. 8 and t.u. 13, the \ac{PEF} simultaneously outputs the ``older'' packets 1-6 received through the ``long'' path and the ``newer'' packets 7-12 received through the now-active ``short`` path.
This increases the load on the downstream link with a doubled rate, $2r_0$, for a duration of 6 t.u.

The toy example hence suggests that \acf{PREF} can significantly increase the flows' burstiness, which could further worsen the congestion and the worst-case delay in the downstream nodes.
Obtaining a bound on this burstiness increase is important for computing the end-to-end latency of the flow.
Indeed, a delay bound for the flow in the third section of Fig.~\ref{fig:prob-statement:three-sections} (after the \ac{PEF}) can be obtained from such an upper bound on the flow's worst-case traffic at the output of the \ac{PEF} and from a lower bound on the minimum service provided by the nodes located after the \ac{PEF}.


A first approach for bounding the traffic of the flow after the \ac{PEF}, which we denote as \emph{intuitive}, consists in doing as if the \ac{PEF} would never drop a packet (i.e., even the duplicates are forwarded).
This approach requires the network engineer to dimension all the downstream nodes in order to support a sustained double rate.
In Theorem~\ref{thm:toolbox:pef-oac}, we provide a better bound for the traffic at the output of the \ac{PEF}.
It leads to better end-to-end latency bounds, as we show in Sec.~\ref{sec:evaluation}.

We also observe that \ac{PREF} can create a mis-ordering: In the toy example, Data Unit 6 exits the \ac{PEF} five time units after Data Unit 7.
Obtaining an upper bound for this mis-ordering is important for comparing it to the application's requirements. 
We provide such bound in Theorem~\ref{thm:toolbox:reordering}.


\begin{figure}\centering
    \resizebox*{0.95\linewidth}{!}{\begin{tikzpicture}
	\tikzstyle{lw} = [line width=2pt]
	\tikzstyle{every node}=[font=\Large]
	
	\draw [->,lw] (-1,0) -- (15,0);
	
	\draw (0,-0.2) -- (0,0.2);
	\draw (1,-0.2) -- (1,0.2);
	\draw (2,-0.2) -- (2,0.2);
	\draw (3,-0.2) -- (3,0.2);
	\draw (4,-0.2) -- (4,0.2);
	\draw (5,-0.2) -- (5,0.2);
	\draw (6,-0.2) -- (6,0.2);
	\draw (7,-0.2) -- (7,0.2);
	\draw (8,-0.2) -- (8,0.2);
	\draw (9,-0.2) -- (9,0.2);
	\draw (10,-0.2) -- (10,0.2);
	\draw (11,-0.2) -- (11,0.2);
	\draw (12,-0.2) -- (12,0.2);
	\draw (13,-0.2) -- (13,0.2);
	\draw (14,-0.2) -- (14,0.2);
	
	\draw[lw,blue,->] (0,0) -- (0,1) node[pos=1,above, black] {1};
	\draw[lw,blue,->] (1,0) -- (1,1) node[pos=1,above, black] {2};
	\draw[lw,blue,->] (2,0) -- (2,1) node[pos=1,above, black] {3};
	\draw[lw,blue,->] (3,0) -- (3,1) node[pos=1,above, black] {4};
	\draw[lw,blue,->] (4,0) -- (4,1) node[pos=1,above, black] {5};
	\draw[lw,blue,->] (5,0) -- (5,1) node[pos=1,above, black] {6};
	\draw[lw,blue,->] (6,0) -- (6,1) node[pos=1,above, black] {7};
	\draw[lw,blue,->] (7,0) -- (7,1) node[pos=1,above, black] {8};
	\draw[lw,blue,->] (8,0) -- (8,1) node[pos=1,above, black] {9};
	\draw[lw,blue,->] (9,0) -- (9,1) node[pos=1,above, black] {10};
	\draw[lw,blue,->] (10,0) -- (10,1) node[pos=1,above, black] {11};
	\draw[lw,blue,->] (11,0) -- (11,1) node[pos=1,above, black] {12};
	\draw[lw,blue,->] (12,0) -- (12,1) node[pos=1,above, black] {13};
	\draw[lw,blue,->] (13,0) -- (13,1) node[pos=1,above, black] {14};

	\node[anchor=south, red] at (15,0)  {\texttt{in}};
	
	\draw [->,lw] (-1,-2) -- (15,-2);
	\draw (0,-2.2) -- (0,-1.8);
	\draw (1,-2.2) -- (1,-1.8);
	\draw (2,-2.2) -- (2,-1.8);
	\draw (3,-2.2) -- (3,-1.8);
	\draw (4,-2.2) -- (4,-1.8);
	\draw (5,-2.2) -- (5,-1.8);
	\draw (6,-2.2) -- (6,-1.8);
	\draw (7,-2.2) -- (7,-1.8);
	\draw (8,-2.2) -- (8,-1.8);
	\draw (9,-2.2) -- (9,-1.8);
	\draw (10,-2.2) -- (10,-1.8);
	\draw (11,-2.2) -- (11,-1.8);
	\draw (12,-2.2) -- (12,-1.8);
	\draw (13,-2.2) -- (13,-1.8);
	\draw (14,-2.2) -- (14,-1.8);
	\node[anchor=south, red] at (15,-2)  {\texttt{outC}};
	
	\draw[lw,blue,->] (7,-2) -- (7,-1) node[pos=1, anchor=south, black] {7};
	\draw[lw,blue,->] (8,-2) -- (8,-1) node[pos=1, anchor=south, black] {8};
	\draw[lw,blue,->] (9,-2) -- (9,-1) node[pos=1, anchor=south, black] {9};
	\draw[lw,blue,->] (10,-2) -- (10,-1) node[pos=1, anchor=south, black] {10};
	\draw[lw,blue,->] (11,-2) -- (11,-1) node[pos=1, anchor=south, black] {11};
	\draw[lw,blue,->] (12,-2) -- (12,-1) node[pos=1, anchor=south, black] {12};
	\draw[lw,blue,->] (13,-2) -- (13,-1) node[pos=1, anchor=south, black] {13};
    \draw[lw,blue,->] (14,-2) -- (14,-1) node[pos=1, anchor=south, black] {14};
	
	\draw [->,lw] (-1,-4) -- (15,-4);
	\draw (0,-4.2) -- (0,-3.8);
	\draw (1,-4.2) -- (1,-3.8);
	\draw (2,-4.2) -- (2,-3.8);
	\draw (3,-4.2) -- (3,-3.8);
	\draw (4,-4.2) -- (4,-3.8);
	\draw (5,-4.2) -- (5,-3.8);
	\draw (6,-4.2) -- (6,-3.8);
	\draw (7,-4.2) -- (7,-3.8);
	\draw (8,-4.2) -- (8,-3.8);
	\draw (9,-4.2) -- (9,-3.8);
	\draw (10,-4.2) -- (10,-3.8);
	\draw (11,-4.2) -- (11,-3.8);
	\draw (12,-4.2) -- (12,-3.8);
	\draw (13,-4.2) -- (13,-3.8);
	\draw (14,-4.2) -- (14,-3.8);
	\node[anchor=south, red] at (15,-4)  {\texttt{outD}};

	\draw[lw,blue,->] (7,-4) -- (7,-3) node[pos=1, anchor=south, black] {1};
	\draw[lw,blue,->] (8,-4) -- (8,-3) node[pos=1, anchor=south, black] {2};
	\draw[lw,blue,->] (9,-4) -- (9,-3) node[pos=1, anchor=south, black] {3};
	\draw[lw,blue,->] (10,-4) -- (10,-3) node[pos=1, anchor=south, black] {4};
	\draw[lw,blue,->] (11,-4) -- (11,-3) node[pos=1, anchor=south, black] {5};
	\draw[lw,blue,->] (12,-4) -- (12,-3) node[pos=1, anchor=south, black] {6};
	\draw[lw,red,dashed,->] (13,-4) -- (13,-3) node[pos=1, anchor=south, black] {7};
	\draw[lw,red,dashed,->] (14,-4) -- (14,-3) node[pos=1, anchor=south, black] {8};

	\draw [->,lw] (-1,-6) -- (15,-6);
	\draw (0,-6.2) -- (0,-5.8)	node[pos=0,below] {1};
	\draw (1,-6.2) -- (1,-5.8)	node[pos=0,below] {2};
	\draw (2,-6.2) -- (2,-5.8)	node[pos=0,below] {3};
	\draw (3,-6.2) -- (3,-5.8)	node[pos=0,below] {4};
	\draw (4,-6.2) -- (4,-5.8)	node[pos=0,below] {5};
	\draw (5,-6.2) -- (5,-5.8)	node[pos=0,below] {6};
	\draw (6,-6.2) -- (6,-5.8)	node[pos=0,below] {7};
	\draw (7,-6.2) -- (7,-5.8)	node[pos=0,below] {8};
	\draw (8,-6.2) -- (8,-5.8)	node[pos=0,below] {9};
	\draw (9,-6.2) -- (9,-5.8)  node[pos=0,below] {10};
	\draw (10,-6.2) -- (10,-5.8)node[pos=0,below] {11};
	\draw (11,-6.2) -- (11,-5.8)node[pos=0,below] {12};
	\draw (12,-6.2) -- (12,-5.8)node[pos=0,below] {13};
	\draw (13,-6.2) -- (13,-5.8)node[pos=0,below] {14};
	\draw (14,-6.2) -- (14,-5.8)node[pos=0,below] {15};
	\node[anchor=south, red] at (15,-6)  {\texttt{out}};
	
	\draw[lw,blue,->] (7,-6) -- (7,-5) node[pos=1, black, anchor=south east, xshift=0.125cm] {1};
	\draw[lw,blue,->] (7.1,-6) -- (7.1,-5) node[pos=1, black, anchor=south west, xshift=-0.125cm] {7};
	
	\draw[lw,blue,->] (8,-6) -- (8,-5) node[pos=1, black, anchor=south east, xshift=0.125cm] {2};
	\draw[lw,blue,->] (8.1,-6) -- (8.1,-5) node[pos=1, black, anchor=south west, xshift=-0.125cm] {8};
	
	\draw[lw,blue,->] (9,-6) -- (9,-5) node[pos=1, black, anchor=south east, xshift=0.125cm] {3};
	\draw[lw,blue,->] (9.1,-6) -- (9.1,-5) node[pos=1, black, anchor=south west, xshift=-0.125cm] {9};
	
	\draw[lw,blue,->] (10,-6) -- (10,-5) node[pos=1, black, anchor=south east, xshift=0.125cm] {4};
	\draw[lw,blue,->] (10.1,-6) -- (10.1,-5) node[pos=1, black, anchor=south west, xshift=-0.125cm] {10};
	
	\draw[lw,blue,->] (11,-6) -- (11,-5) node[pos=1, black, anchor=south east, xshift=0.125cm] {5};
	\draw[lw,blue,->] (11.1,-6) -- (11.1,-5) node[pos=1, black, anchor=south west, xshift=-0.125cm] {11};
	
	\draw[lw,blue,->] (12,-6) -- (12,-5) node[pos=1, black, anchor=south east, xshift=0.125cm] {6};
	\draw[lw,blue,->] (12.1,-6) -- (12.1,-5) node[pos=1, black, anchor=south west, xshift=-0.125cm] {12};

	\draw[lw,blue,->] (13,-6) -- (13,-5) node[pos=1, black, anchor=south] {13};

    \draw[lw,blue,->] (14,-6) -- (14,-5) node[pos=1, black, anchor=south] {14};

    \draw[-latex] (0,0) -- (7,-4);
	\draw[-latex] (0,0) -- (1,-1) node[pos=1,anchor=north west] (c1) {};
	\draw[red] (c1.north west) -- (c1.south east);
	\draw[red] (c1.south west) -- (c1.north east);

    \draw[-latex] (6,0) -- (13,-4);
	\draw[-latex] (6,0) -- (7,-2);
\end{tikzpicture}%
}%
    \caption{\label{fig:prob-statement:double-rate} An example trajectory on the toy example of Fig.~\ref{fig:prob-statement:toy-example} causing a double-rate output of the \acs{PEF}.}
\end{figure}
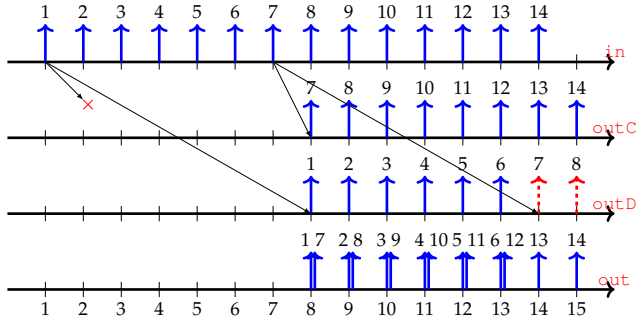


\vspace{\myvspacebeforesubsec}
\subsection{Interactions Between PREF and Other Devices}

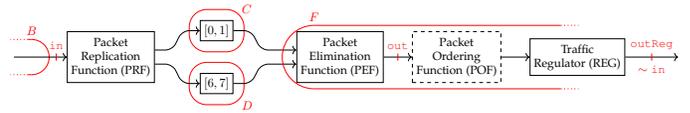
\begin{figure}\centering
    \resizebox*{\linewidth}{!}{\begin{tikzpicture}
	\node[draw] at (0,0) (s) {\makecell{Packet\\Replication\\Function (PRF)}};
	\node[draw] at (3,0.75) (a) {$[0,1]$};
	\node[draw] at (3,-0.75) (b) {$[6,7]$};
	\node[draw] at (6.5,0) (c) {\makecell{Packet\\Elimination\\Function (\acs{PEF})}};
    \node[draw, dashed, anchor=west] at ([xshift=0.8cm] c.east) (pof) {\makecell{Packet\\Ordering\\Function (\acs{POF})}};
    \node[draw, anchor=west] at ([xshift=0.8cm] pof.east) (reg) {\makecell{Traffic\\Regulator (\acs{REG})}};
	\node[draw, red, rounded corners=0.5cm, fit={(a)}, inner sep=0.3cm] (cNode) {};
	\node[draw, red, rounded corners=0.5cm, fit={(b)}, inner sep=0.3cm] (dNode) {};
	\node[anchor=west, red] at ([xshift=-0.2cm] dNode.south east) {$D$};
	\node[anchor=west, red] at ([xshift=-0.2cm] cNode.north east) {$C$};

	\draw[->, rounded corners=0.2cm] ([yshift=0.2cm]s.east) -- ([xshift=0.5cm,yshift=0.2cm] s.east) -- ([xshift=-0.5cm] a.west) -- (a.west);
	\draw[->, rounded corners=0.2cm] ([yshift=-0.2cm]s.east) -- ([xshift=0.5cm,yshift=-0.2cm] s.east) -- ([xshift=-0.5cm] b.west) -- (b.west);
	
	\draw[->, rounded corners=0.2cm] (a.east) -- ([xshift=0.5cm] a.east) -- ([xshift=-1cm,yshift=0.2cm] c.west) -- ([yshift=0.2cm]c.west) node[pos=0.5,anchor=center] (tA) {};
	\draw[->, rounded corners=0.2cm] (b.east) -- ([xshift=0.5cm] b.east) -- ([xshift=-1cm,yshift=-0.2cm] c.west) -- ([yshift=-0.2cm]c.west) node[pos=0.5,anchor=center] (tB) {};

	\draw[red] ([xshift=0.5cm] c.west) ++(90:-0.9) arc (90:-90:-0.9) node[pos=0, anchor=center] (mt) {} node[pos=1, anchor=center] (mtt) {};
	\draw[red] (mt.center) -- ++(7cm,0);
	\draw[red, dotted] (mt.center) ++(7cm,0) -- ++(0.5cm,0);
	\draw[red] (mtt.center) -- ++(7cm,0);
	\draw[red, dotted] (mtt.center) ++(7cm,0) -- ++(0.5cm,0);
	\node[red, anchor=south] at (mtt.center) {$F$};

	\draw[red] ([xshift=-1cm] s.west) ++(90:0.5) arc (90:-90:0.5) node[pos=0, anchor=center] (mmt) {} node[pos=1, anchor=center] (mmtt) {};
	\draw[red] (mmt.center) -- ++(-0.1cm,0);
	\draw[red, dotted] (mmt.center) ++(-0.1cm,0) -- ++(-0.5cm,0);
	\draw[red] (mmtt.center) -- ++(-0.1cm,0);
	\draw[red, dotted] (mmtt.center) ++(-0.1cm,0) -- ++(-0.5cm,0);
	\node[red, anchor=south] at (mmt.center) {$B$};


	\draw[->] ([xshift=-1.5cm] s.west) -- (s.west) node[pos=0.8,anchor=center] (tIn) {};
	\draw[-,red] ([yshift=0.1cm] tIn.center) -- ([yshift=-0.1cm] tIn.center) node[pos=0,above] {\texttt{in}};

    \draw[->] (c.east) -- (pof.west) node[pos=0.5,anchor=center] (tOut) {};
    \draw[->] (pof.east) -- (reg.west);
    \draw[->] (reg.east) -- ++(1.5cm,0) node[pos=0.5, anchor=center] (tRegOut) {};
	\draw[-,red] ([yshift=-0.1cm] tOut.center) -- ([yshift=0.1cm] tOut.center) node[pos=1,above] {\texttt{out}};
    \draw[-,red] ([yshift=-0.1cm] tRegOut.center) -- ([yshift=0.1cm] tRegOut.center) node[pos=1,above] {\texttt{outReg}} node[pos=0,below] {$\sim\texttt{in}$};
\end{tikzpicture}%
}%
    \caption{\label{fig:prob-statement:toy-example-with-reg} The toy example of Fig.~\ref{fig:prob-statement:toy-example} extended with \acs{POF} and \acs{REG} to deal with the mis-ordering and burstiness increase issues due to \acs{PEF}.}
\end{figure}

If either the end-to-end latency bound or the mis-ordering bound does not meet the system requirements, then we can use one of the devices specified by the working groups for tackling the corresponding issue.

For example, if the receiving application does not tolerate any mis-ordering, then the \ac{DetNet} \acf{POF} \cite[\S 3.2.2.2]{finnDeterministicNetworkingArchitecture2019} can be used after the \ac{PEF} to correct the mis-ordering introduced by \ac{PREF}.
Similarly, if the end-to-end latency of a flow does not meet its requirements due to a high worst-case delay in the third section of Figure~\ref{fig:prob-statement:three-sections}, then using traffic regulators \cite{ieeeDraftStandardLocal2019a} just after the \ac{PEF} appears as a natural choice: 
Traffic \acfp{REG} have been designed for removing the burstiness increase \cite{leboudecTheoryTrafficRegulators2018,mohammadpourLatencyBacklogBounds2018} thus for reducing the worst-case delay in downstream nodes.
They come in two flavors: \acfp{PFR} process each flow individually whereas \acfp{IR} process flow aggregates.

To the best of our knowledge, the interactions between \ac{PREF} and other devices such as \acp{POF} and \aclp{REG} have not yet been analyzed.
For instance, many properties of the regulators rely on the assumption that the upstream system is \ac{FIFO} \cite{leboudecTheoryTrafficRegulators2018}.
As observed on the toy example, this assumption does not hold with \ac{PREF}.



Assume, for example, that the traffic regulator in Fig.~\ref{fig:prob-statement:toy-example-with-reg} shapes the traffic back to the profile it had at the input ``\texttt{in}''.
In terms of burstiness, this makes the middle section in Fig.~\ref{fig:prob-statement:three-sections} transparent to the third section.
The regulator processes the traffic from the ``\texttt{out}'' line of Fig.~\ref{fig:prob-statement:double-rate} and forces the packets to be as spaced as in the ``\texttt{in}'' line by delaying and storing the packets if required.
Clearly, the upstream system between ``$\texttt{in}$'' and ``$\texttt{out}$'' in Fig.~\ref{fig:prob-statement:double-rate} is not \ac{FIFO}, because the packets exit the \ac{PEF} out of order.
Thus the properties of the regulators that depend on this assumption might not hold and the cohabitation of the \ac{PEF} and the \ac{REG} could negatively affect the latency bounds.
A \acf{POF} (dashed box in Fig.~\ref{fig:prob-statement:toy-example-with-reg}) can be used after the \ac{PEF} and before the \acl{REG} to force the upstream system to be \ac{FIFO}.
If such \ac{POF} is placed, then we would expect to retrieve all the properties of regulators.

In Sec. \ref{sec:regulators}, we analyze the interactions between \ac{PREF} and \aclp{REG}.
We observe that the conclusions depend on the type of the \acl{REG}: either \ac{PFR} or \ac{IR} (as with \ac{TSN} \acs{ATS}).

\vspace{\myvspacebeforesec}
\section{Related Work}\label{sec:related-work}

The most relevant timing analyses of redundancy mechanisms in time-sensitive networks can mainly be categorized according to the assumption of where to enable the packet replication and elimination functions.


The existing approaches in this area considering the packet replication and elimination only at the end-devices concern mainly \emph{\ac{HSR}} and \emph{\acf{PRP}}\cite[\S 4]{iecIEC6243932016}. Both mechanisms eliminate the duplicates only at the destination; thus their analysis does not require to bound the traffic at the output of the \ac{PEF} and discards the mis-ordering issue. 
In \cite{heiseSAFDXDeterministicHighavailability2014a}, worst-case delay bounds are computed in \ac{HSR}-based networks by using network calculus. The idea consists in taking the maximum of the delay bounds along each of the redundant paths. In \cite{taubrichFormalSpecificationAnalysis2007}, model checking is used to analyze how well \ac{PEF} algorithms based on sequence numbers can detect duplicates in \ac{AFDX}, a \ac{PRP}-based network. 


On the other hand, there exist only few seminal works in the literature considering the packet replication and elimination anywhere in the network. These works mainly concern \acs{FRER} \cite{IEEEStandardLocal2017a}, which is the first mechanism enabling such an assumption. As mentioned in \cite[\S C.9]{IEEEStandardLocal2017a} and further illustrated in Sec.~\ref{sec:prob-statement}, the elimination of duplicates within the network raises issues in computing the \acl{ETE} delay bounds. In \cite{hofmannChallengesLimitationsIEEE2020}, further concerns about \ac{FRER} have been discussed. In \cite{heiseTSimNetIndustrialTime2016}, a simulation framework based on OMNeT++  has been developed for \ac{TSN} mechanisms, including \ac{FRER}~\cite[\S 4.3.2]{heiseRealtimeGuaranteesDependability2018}. However there is no specific experiment for assessing the effect of \ac{FRER} on latency bounds. Furthermore, obtaining the worst-case delay bounds with simulators is a known intractable problem~\cite[\S I]{bouillard2018deterministic}. 

Thus, as stated above, there are no formal analyses of delay bounds of redundancy mechanisms, such as \ac{TSN} \ac{FRER} or \ac{DetNet} PREOF, when the packet replication and elimination is performed anywhere in the network, or on the interactions between redundancy and scheduling mechanisms. 

\vspace{\myvspacebeforesec}
\section{System Model}\label{sec:system-model}

Our system model is divided into three abstraction levels.
It results from an analysis of the \ac{TSN} and \ac{DetNet} documents and Appendix~\ref{sec:appendix:discussion-system-model} details its applicability for these standards. 
Notations used thorough the paper are listed in Table~\ref{tab:system-model:notations}.
\setlength\dashlinedash{0.2pt}
\setlength\dashlinegap{1.5pt}
\begin{table}\centering
    \caption{\label{tab:system-model:notations} Notations}
    \resizebox*{\linewidth}{!}{
    \begin{tabular}{r|l}
        Term & Definition \\
        \hline
        $\mathcal{G}$ & The graph of the network for the class of interest.\\
        \hdashline
        $f$ & A flow. \\
        \hdashline
        $\mathcal{G}(f)$ & The graph of flow $f$.\\
        \hline
        \makecell[r]{\acs{EP}-vertex\\in $\mathcal{G}(f)$} & \makecell[l]{A vertex at which the duplicates of $f$ have not\\ been eliminated yet.} \\
        \hdashline
        \makecell[r]{Diamond\\ancestor\\of $n$ in $\mathcal{G}(f)$} & \makecell[l]{A vertex that is not an \acs{EP}-vertex of $\mathcal{G}(f)$  and that is \\ contained in any paths of $f$ between its source and $n$.}\\
        \hdashline
        \makecell[r]{$d_f^{a\rightarrow n}$\\{}[resp., $D_f^{a\rightarrow n}$]} & \makecell[l]{Lower [resp., upper] delay bound for $f$ between $a$\\ and $n$, along any possible path $a\rightarrow n$ within $\mathcal{G}(f)$.} \\
        \hline
        $\texttt{PEF}_n(f)$ & \makecell[l]{Packet-elimination function at output-port $n$ \\ that eliminates the duplicates of flow $f$.} \\
        \hdashline
        $\texttt{POF}_n(\mathcal{F},o)$ & \makecell[l]{Packet-ordering function at $n$  that uses reference $o$ \\ to force  the  order of the data units of the aggregate $\mathcal{F}$.}\\
        \hdashline
        $\texttt{REG}_n(\mathcal{F},o)$ & \makecell[l]{Regulator (either interleaved or per-flow) that shapes \\ the flows  within $\mathcal{F}$ in a \ac{FIFO} manner.}\\
        \hdashline
        $\sigma_{n,f}$ & Shaping curve for $f$ at the regulator within $n$. \\
        \hline
        \makecell[r]{$\alpha_{f,n^*}$ \\ $[$resp., $\alpha_{f,\texttt{FUN}^*}$ $]$} & \makecell[l]{For $n$ a vertex in $\mathcal{G}$ [resp., $\texttt{FUN}$ a function], the arrival \\ curve of $f$ at the output of $n$  $[$resp., of the function \texttt{FUN}$]$.} \\
        \hline
        $\gamma_{r,b}:t\mapsto rt+b$ & Leaky-bucket arrival curve with rate $r$ and burst $b$ \\
        \hdashline
        $\delta_D:\left\lbrace\begin{aligned}t\mapsto\infty\\\text{if }t > D\end{aligned}\right.$ & Service-curve of a D-bounded-Delay element\\
        \hline
        $|x|^+$ & $=\max(0,x)$\\
        \hline
        t.u. & Time unit (arbitrary unit used in the examples) \\
        \hdashline 
        d.u. & Data unit (arbitrary unit used in the examples) \\
    \end{tabular}
    }
\end{table}
\vspace{\myvspacebeforesubsec}
\subsection{Network and Flow Model}\label{sec:system-model:flow-model}

\ul{Type of network:} We consider an asynchronous packet-switching full-duplex store-and-forward network that transports \emph{data units} between applications.
We assume that there is one or several classes of traffic and that flows are statically assigned to a class.
We focus on one class and denote by $\mathcal{G}$ the
underlying graph for this class \cite[Chap. 12]{bouillard2018deterministic}.
$\mathcal{G}$ contains one vertex per output port in the network (see Sec.~\ref{sec:system-model:device-model} for the exact mapping between the two notions) and $(a,b)$ is a directed edge of $\mathcal{G}$ if at least one flow crosses $b$ just after $a$.
The network does not need to be feed-forward, it can contain cyclic dependencies (\emph{i.e.}, $\mathcal{G}$ can contain cycles)\cite{thomasCyclicDependenciesRegulators2019}.

\begin{figure}\centering
    \resizebox*{0.7\linewidth}{!}{\begin{tikzpicture}[yscale=0.8,xscale=1.4]
    \tikzstyle{n} = [draw, circle]
    \def\sep{1}
    \node[n] at (0,0) (a) {A};
    \node[n] at (1,0) (b) {B};
    \node[n] at (2,1) (c) {C};
    \node[n] at (2,-1) (d) {D};
    \node[n] at (3,0) (f) {F};
    \node[n] at (4,0) (g) {G};
    \node[n] at (4,1) (e) {E};
    \draw[->] (a) -- (b);
    \draw[->] (b) -- (c);
    \draw[->] (b) -- (d);
    \draw[->] (c) -- (f);
    \draw[->] (d) -- (f);
    \draw[->] (f) -- (g);
    \draw[->] (c) -- (e);
\end{tikzpicture}%
}%
    \caption{\label{fig:system-model:example-flow-graph} Example of a flow graph $\mathcal{G}(f)$ (Sec.~\ref{sec:system-model:flow-model}) of flow f with a source $A$, destinations $E$ and $G$ and redundant paths to reach $G$.}
\end{figure}
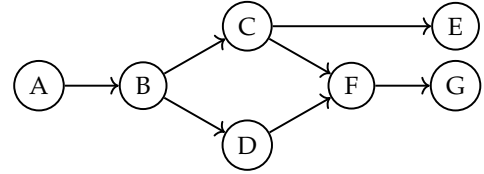
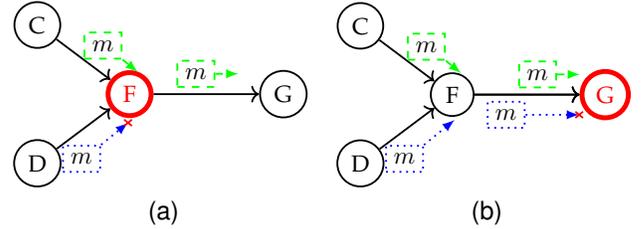
\begin{figure}\centering
    \def\lsize{0.48\linewidth}
    \subfloat[]{\resizebox*{0.45\linewidth}{!}{\begin{tikzpicture}[yscale=0.9,xscale=1.1]
    \tikzstyle{n} = [draw, circle]
    \tikzstyle{g} = [green, dashed]
    \tikzstyle{b} = [blue, dotted]
    \def\sep{1}
    \node[n] at (1.8,1.1) (c) {C};
    \node[n] at (1.8,-1.1) (d) {D};
    \node[n, red, line width=2pt] at (3,0) (f) {F};
    \node[n] at (5,0) (g) {G};
    \draw[->] (c) -- (f) node[pos=0.5, anchor=south west, draw, g, text=black] (mg) {$m$};
    \draw[-latex, g] (mg.south east) -- ++(0.2,-0.2);
    \draw[->] (d) -- (f) node[pos=0.1, anchor=north west, draw, b, text=black] (mb) {$m$};
    \draw[-latex, b] (mb.north east) -- ++(0.35,0.35) node[pos=1, anchor=center] (t) {};
    \draw[red] ([xshift=-0.05cm, yshift=-0.05cm] t.center) -- ([xshift=0.05cm, yshift=0.05cm] t.center);
    \draw[red] ([xshift=-0.05cm, yshift=0.05cm] t.center) -- ([xshift=0.05cm, yshift=-0.05cm] t.center);
    \draw[->] (f) -- (g) node[pos=0.4,anchor=center] (tt) {};
    \node[draw, g, anchor=south, text=black] at ([yshift=0.1cm] tt.center) (mmg) {$m$};
    \draw[-latex, g] (mmg.east) -- ++(0.3,0);
\end{tikzpicture}}\label{fig:system-model:example-flow-graph-different-pef-locations:f}}
    \hfil
    \subfloat[]{\resizebox*{0.45\linewidth}{!}{\begin{tikzpicture}[yscale=0.9,xscale=1.1]
    \tikzstyle{n} = [draw, circle]
    \def\sep{1}
    \tikzstyle{g} = [green, dashed]
    \tikzstyle{b} = [blue, dotted]
    \node[n] at (1.8,1.1) (c) {C};
    \node[n] at (1.8,-1.1) (d) {D};
    \node[n] at (3,0) (f) {F};
    \node[n, red, line width=2pt] at (5,0) (g) {G};
    \draw[->] (c) -- (f) node[pos=0.5, anchor=south west, draw, g, text=black] (mg) {$m$};
    \draw[-latex, g] (mg.south east) -- ++(0.2,-0.2);
    \draw[->] (d) -- (f) node[pos=0.1, anchor=north west, draw, b, text=black] (mb) {$m$};
    \draw[-latex, b] (mb.north east) -- ++(0.4,0.4) node[pos=1, anchor=center] (t) {};
    \draw[->] (f) -- (g) node[pos=0.6,anchor=center] (tt) {};
    \node[draw, g, text=black,anchor=south] at ([yshift=0.1cm] tt.center) (mmg) {$m$};
    \draw[-latex, g] (mmg.east) -- ++(0.3,0);
    \draw[->] (f) -- (g) node[pos=0.3,anchor=center] (ttt) {};
    \node[draw, b, text=black,anchor=north] at ([yshift=-0.1cm] ttt.center) (mmb) {$m$};
    \draw[-latex, b] (mmb.east) -- ++(0.7,0) node[pos=1, anchor=center] (u) {};
    \draw[red] ([xshift=-0.05cm, yshift=-0.05cm] u.center) -- ([xshift=0.05cm, yshift=0.05cm] u.center);
    \draw[red] ([xshift=-0.05cm, yshift=0.05cm] u.center) -- ([xshift=0.05cm, yshift=-0.05cm] u.center);
\end{tikzpicture}}\label{fig:system-model:example-flow-graph-different-pef-locations:g}}
    \caption{\label{fig:system-model:example-flow-graph-different-pef-locations} $F$ might receive the data unit $m$ twice (in the dashed green and the dotted blue packets). (a) $F$ contains a \acs{PEF}, it drops the dotted blue packet that contains the already-seen data unit $m$. (b) Only the destination $G$ contains a \acs{PEF}, it receives the data unit twice and drops the dotted blue packet.}
\end{figure}
\noindent\ul{\emph{Data unit} versus \emph{packet}:} At any time, a \emph{data unit} can be transported by several \emph{packets} located at several locations.
A \emph{flow} $f$ is a coherent sequence of data units that originate from a unique source and that follow a directed acyclic sub-graph of $\mathcal{G}$ to reach one or several destinations.
An example of such a flow graph, noted $\mathcal{G}(f)$, is shown in Fig.~\ref{fig:system-model:example-flow-graph}.

\noindent\ul{Flow constraints:}
We assume that each flow is constrained by a network-calculus arrival curve $\alpha_{f}^0$ at the output of its source application.
For an observation point $M$ (that can be a vertex or a function), we note $\alpha_{f,M}$ the arrival curve of $f$ at $M$.
For $n$ a vertex of $\mathcal{G}(f)$ [resp., for \texttt{FUN} a function], we note $\alpha_{f,n^*}$ [resp., $\alpha_{f,\texttt{FUN}^*}$] the arrival curve of $f$ at the output of vertex $n$ [resp., at the output of the function \texttt{FUN}].%

\noindent\ul{Position of PRF, PEF in a flow graph:}
When a vertex, such as $B$ in Fig.~\ref{fig:system-model:example-flow-graph}, has several children, we consider that an implicit \acs{PRF} has been installed on $B$ for the flow $f$: it sends a copy of each incoming data unit to each child.
When a vertex has several parents, such as $F$ in Fig.~\ref{fig:system-model:example-flow-graph}, this means that it can receive the same data unit several times, within different packets.
However, this does not necessarily mean that it implements a \acs{PEF}.
If a \ac{PEF} is present on such a vertex  (case of $F$ in Fig.~\ref{fig:system-model:example-flow-graph-different-pef-locations:f}), then it forwards only the first received packet that contains the data unit.
If the vertex does not contain a \ac{PEF} (case of $F$ in Fig.~\ref{fig:system-model:example-flow-graph-different-pef-locations:g}), then it forwards all the packets, and might consequently forward the same data unit several times.
Packets that transport already-seen data units at a given location are called \emph{duplicates}.
%

\noindent\ul{Assumption on the elimination of duplicates:}
When several paths of a flow merge, we assume that the duplicates are eliminated before the path can split again.
We believe that this assumption does not restrict the analysis of industrial systems.
Indeed, the main use-case for having a \ac{PEF} a few hops after the merge point (as in Fig.~\ref{fig:system-model:example-flow-graph-different-pef-locations:g}) is when the edge router does not support \ac{PEF}.
The edge router then forwards all the received packets to the end-system, that is responsible for removing the duplicates.

\noindent\ul{EP-vertex:}
\emph{\ac{EP}} vertices of $\mathcal{G}(f)$ are the only vertices that can observe a data unit of $f$ more than once.
Formally, if a vertex, that does not contain a \ac{PEF} for $f$, has several parents in $\mathcal{G}(f)$ (vertex $F$ in Fig.~\ref{fig:system-model:example-flow-graph-different-pef-locations:g}), then we qualify it as an \ac{EP}-vertex of $\mathcal{G}(f)$.
An \ac{EP}-vertex can have at most one child in $\mathcal{G}(f)$. Additionally, a vertex that does not contain a \ac{PEF} for $f$ and is a child of an \ac{EP}-vertex is also an \ac{EP}-vertex of $\mathcal{G}(f)$.

\noindent\ul{Diamond ancestor:} For any two vertices $a$ and $n$ in a flow graph $\mathcal{G}(f)$, we say that $a$ is a \emph{diamond ancestor} of $n$ in $\mathcal{G}(f)$ if $a$ is not an \ac{EP}-vertex in $\mathcal{G}(f)$ and all paths in $\mathcal{G}(f)$  from the source of $f$ to $n$ contain $a$.
In Fig.~\ref{fig:system-model:example-flow-graph}, $B$ is a diamond ancestor of $F$ because $B$ is not an \ac{EP}-vertex of $\mathcal{G}(f)$ and any paths from $A$ (source of $f$) to $F$ contain $B$.

\noindent\ul{Lost data unit:}
We say that a data-unit $m$ of a flow $f$ is \emph{lost for} a vertex $n$ [resp., for a function $\texttt{FUN}$] if the vertex $n$ in $\mathcal{G}(f)$ [resp., the function $\texttt{FUN}$] never observes the data unit $m$ in any packet.
In Fig.~\ref{fig:system-model:example-flow-graph}, if the link $B\rightarrow C$ fails, then a data unit $m$ is lost for $E$ but not necessarily for $G$.
The main purpose of \ac{PREF} is to reduce the probability of losing a data unit for any destinations of the flow.

\noindent\ul{Worst-case latency:}
Let $f$ be a flow and $d$ one of the destinations of $f$; the \ac{ETE} upper [resp., lower] latency bound of $f$ for $d$ is an upper bound [resp., lower bound] on the maximum [resp., minimum] delay that each data unit $m$ of $f$ takes to reach $d$, assuming that $m$ is not lost for $d$.

\vspace{\myvspacebeforesubsec}
\subsection{Device Model}\label{sec:system-model:device-model}
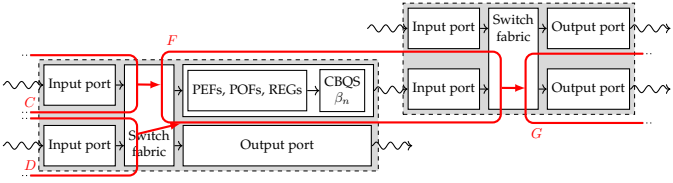
\begin{figure}\centering
    \resizebox*{\linewidth}{!}{	\begin{tikzpicture}[decoration=snake]
	\pgfdeclarelayer{bg}    
	\pgfsetlayers{bg,main}  
	\tikzstyle{mybox} = [draw, minimum height=1cm, fill=white]
	\tikzstyle{r} = [red, line width=1.5pt, rounded corners=0.2cm]

	\node[mybox, text width=1cm, minimum height=2.5cm] (SF) at (0,0){};
	\node[anchor=south] at (SF.south) {\makecell{Switch\\fabric}}; 
	\node[mybox, anchor=north west] (OP) at ([xshift=0.2cm]SF.north east) {

		\begin{tikzpicture}[decoration=snake]
    \tikzstyle{n} = [draw, minimum height=1cm, minimum width=1cm, anchor=west]
    \tikzstyle{xx} = [xshift=0.3cm]

    \node[n] at (0,0) (reg1) {\acsp{PEF}, \acsp{POF}, \acsp{REG}};
    \node[n] at ([xx] reg1.east) (cbqs) {\makecell{\acs{CBQS}\\$\beta_n$}};

    \draw[->] (reg1.east) -- (cbqs.west);
\end{tikzpicture}
		
	};
	\draw[->] (OP.west)++(-0.2cm,0) -- (OP.west);
	\node[mybox, anchor=north east] (INPUTI) at ([xshift=-0.2cm]SF.north west) {Input port};
	\draw[->, decorate] (INPUTI.west)++(-1cm,0) -- (INPUTI.west);
	\draw[->] (INPUTI.east) -- ++(0.2cm,0);
	\node[mybox, anchor=south east] (INPUTJ) at ([xshift=-0.2cm]SF.south west) {Input port};
	\draw[->, decorate] (INPUTJ.west)++(-1cm,0) -- (INPUTJ.west);
	\draw[->] (INPUTJ.east) -- ++(0.2cm,0);
	\node[mybox, anchor=south west, minimum width=4.65cm] (OPK) at ([xshift=0.2cm]SF.south east) {Output port};
	\draw[->] (OPK.west)++(-0.2cm,0) -- (OPK.west);
	\draw[->, decorate] (OPK.east) -- ++(1cm,0);

	\node[mybox, text width=1cm, minimum height=2.5cm] (SF2) at (9,1.4){};
	\node[anchor=north] at (SF2.north) {\makecell{Switch\\fabric}}; 
	\node[mybox, anchor=north west] (OP2) at ([xshift=0.2cm]SF2.north east) {Output port};
	\draw[->] (OP2.west)++(-0.2cm,0) -- (OP2.west);
	\draw[->, decorate] (OP2.east) -- ++(1cm,0);
	\node[mybox, anchor=north east] (INPUTI2) at ([xshift=-0.2cm]SF2.north west) {Input port};
	\draw[->, decorate] (INPUTI2.west)++(-1cm,0) -- (INPUTI2.west);
	\draw[->] (INPUTI2.east) -- ++(0.2cm,0);
	\node[mybox, anchor=south east] (INPUTJ2) at ([xshift=-0.2cm]SF2.south west) {Input port};
	\draw[->] (INPUTJ2.east) -- ++(0.2cm,0);
	\draw[->, decorate] (OP.east) -- (INPUTJ2.west);
	\node[mybox, anchor=south west] (OPK2) at ([xshift=0.2cm]SF2.south east) {Output port};
	\draw[->] (OPK2.west)++(-0.2cm,0) -- (OPK2.west);
	\draw[->, decorate] (OPK2.east) -- ++(1cm,0);

	\node[fit={([xshift=-0.2cm, yshift=0.2cm] SF.north east) ([xshift=-0.2cm] SF.east) ([yshift=-0.2cm, xshift=0.2cm] SF2.south west) ([xshift=0.2cm] SF2.west) ([xshift=-0.2cm, yshift=0.2cm] SF.north east)}, draw, r] (nodeF) {};
	\node[anchor=south west, r] at (nodeF.north west) {$F$};

	\node[fit={([xshift=0.2cm] SF.north west) ([xshift=0.2cm, yshift=0.2cm] SF.west) ([xshift=-0.2cm] INPUTI.south west) ([xshift=-0.2cm, yshift=0.1cm] INPUTI.north west)}] (nodeC) {};
	\node[fit={([xshift=0.2cm] SF.south west) ([xshift=0.2cm, yshift=-0.2cm] SF.west) ([xshift=-0.2cm, yshift=-0.1cm] INPUTJ.south west) ([xshift=-0.2cm] INPUTJ.north west)}] (nodeD) {};
	\draw[r] (nodeC.north west) -- (nodeC.north east) -- (nodeC.south east) -- (nodeC.south west);
	\draw[r] (nodeD.north west) -- (nodeD.north east) -- (nodeD.south east) -- (nodeD.south west);
	\draw[dotted] ([xshift=-0.2cm] nodeC.north west) -- (nodeC.north west);
	\draw[dotted] ([xshift=-0.2cm] nodeC.south west) -- (nodeC.south west);
	\draw[dotted] ([xshift=-0.2cm] nodeD.north west) -- (nodeD.north west);
	\draw[dotted] ([xshift=-0.2cm] nodeD.south west) -- (nodeD.south west);
	\node[r, anchor=south] at (nodeC.south west) {$C$};
	\node[r, anchor=south] at (nodeD.south west) {$D$};
	\draw[-latex, r] (nodeC) -- (nodeF);
	\draw[-latex, r] (nodeD) -- (nodeF);

	\node[fit={([xshift=-0.2cm, yshift=-0.2cm] SF2.south east) ([xshift=-0.2cm] SF2.east) ([xshift=0.2cm] OPK2.north east) ([xshift=0.2cm] OPK2.south east)}] (nodeG) {};
	\draw[r] (nodeG.north east) -- (nodeG.north west) -- (nodeG.south west) -- (nodeG.south east);
	\draw[dotted] ([xshift=0.2cm] nodeG.north east) -- (nodeG.north east);
	\draw[dotted] ([xshift=0.2cm] nodeG.south east) -- (nodeG.south east);
	\node[anchor=north west, r] at (nodeG.south west) {$G$};
	\draw[-latex, r] (nodeF) -- (nodeG);

	\begin{pgfonlayer}{bg}    
        \node[draw, dashed, fit={(SF) (INPUTI) (INPUTJ) (OP) (OPK)}, fill=gray!30] {}; 
		\node[draw, dashed, fit={(SF2) (INPUTI2) (INPUTJ2) (OP2) (OPK2)}, fill=gray!30] {}; 
    \end{pgfonlayer}%
	\end{tikzpicture}%
}%
    \caption{\label{fig:system-model:device-top-level-model} Model of two devices in the network (in gray dashed boxes) and their relation with the flow graph (vertices in thick red).} 
\end{figure}
\noindent\ul{Device:} The model for each \emph{device} in the network is illustrated in Fig.~\ref{fig:system-model:device-top-level-model}: it consists of \emph{input ports}, \emph{output ports}, and a \emph{switching fabric}.
The vertices in the network's graph $\mathcal{G}$, such as vertex $F$ in thick red in Fig.~\ref{fig:system-model:device-top-level-model}, are made of the output port on one device, followed by the input port on the subsequent device. The devices are connected through \emph{transmission links} that can lose packets.

\noindent\ul{Input port:} We assume that each \emph{input port} contains a store-and-forward step that we model as a network-calculus packetizer
\cite[\S 1.7.2]{leboudecNetworkCalculusTheory2001},
\cite[Thm.~1]{thomasCyclicDependenciesRegulators2019}. Any additional processing delay (\emph{e.g.}, decryption, CRC check, etc.) is assumed to be bounded between known values and is modeled using the network-calculus bounded-delay element~\cite[Prop.~1.3.3]{leboudecNetworkCalculusTheory2001}.

\noindent\ul{Switching fabric:} As illustrated in Fig.~\ref{fig:system-model:device-top-level-model}, the \emph{switching fabric} between vertices $C$ and $F$ forwards packets of flow $f$ from the input port within $C$ to the output port within $F$ if and only if $C\rightarrow F$ is an edge in $\mathcal{G}(f)$.
The switching fabric implements the \ac{PRF}.
When a packet is forwarded from one input port to two or more output ports, we say that the data unit contained in the incoming packet is replicated and transported by several new packets (one per recipient output port).
Any delay within the switching fabric is assumed to be bounded and is modeled by using the network-calculus bounded-delay element~\cite[Prop.~1.3.3]{leboudecNetworkCalculusTheory2001}.

\noindent\ul{Output port:} We model each \emph{output port} as in Fig.~\ref{fig:system-model:device-top-level-model}.
It contains a \ac{FIFO}-per-class \acf{CBQS}. We assume that, for each vertex $n$, we know a network-caclulus service curve $\beta_n$ that the \ac{CBQS} offers in a \ac{FIFO} manner to the class of interest. 
The service curve can be obtained through an analysis of the scheduling policy \cite{maile2020network} and includes any additional technological latency.
The \ac{CBQS} can be preceded by a set of optional functions.

\noindent\ul{Packetized streams:} Within a device, between the output of the input port (that contains the packetizer) and the input of the \ac{CBQS}, the stream of bits for each flow is packetized.

\vspace{\myvspacebeforesubsec}
\subsection{Model for the Functions}\label{sec:system-model:function-model}

\noindent\ul{PEF}: For a flow $f$ crossing $n$, the output port in $n$ can contain a \emph{\acf{PEF}} for flow $f$, noted $\texttt{PEF}_n(f)$.\label{sec:system-model:pef-model}
For each incoming packet of $f$, we assume that $\texttt{PEF}_n(f)$ determines without any delay if the data unit contained in the packet has already been observed by $\texttt{PEF}_n(f)$.
If so, the packet is identified as a duplicate and is discarded.
For the stream of packets that contains never-seen data units of $f$, the $\texttt{PEF}_n(f)$ is transparent: \ac{FIFO} and without any delay.
%
\begin{figure}\centering
    \resizebox*{\linewidth}{!}{\begin{tikzpicture}
    \def\mh{2cm}
    \def\ml{2cm}
    \def\mmar{0.2cm}
    \def\hdecal{0.2cm}
    \node[minimum width=1.6cm, minimum height=1.6cm, diamond, draw] at (0,0) (sf) {};
    \node[anchor=south] at (sf.north) (cond) {\makecell{Was $m_{-1}$ already forwarded ?}};
    \draw[->] ([xshift=-3cm] sf.west) -- (sf.west) node[pos=0.8,anchor=south] {$m$} node[pos=0.05,anchor=south west] {$m\in\mathcal{F}$};
    \node[database] at (3,-2) (db) {};
    \node[anchor=north] at (db.south) (dbname) {\makecell{Storage\\(non-FIFO)}};
    \node[draw] at (5,0) (fw) {Forward $m$};
    \draw[->] (sf.south) |- (db.west) node[pos=0, yshift=0.2cm, anchor=south] {No} node[pos=0.5, anchor=north] {\makecell{Store until:\\$m_{-1}$ is forwarded OR \\$T$ seconds have elapsed}};
    \node[draw, fit={(sf) (t) (cond) (db) (dbname) (fw)}] (bg) {};
    \node[anchor=north east] at (bg.north east) {$\texttt{POF}_n(\mathcal{F},o)$};
    \draw[->] (sf.east) -- (fw.west) node[pos=0, anchor=east, xshift=-0.2cm] {Yes} node[pos=0.8, above] {$m$};
    \draw[->] (db.north) |- (fw.west);
    \draw[->] (fw.east) -- ++(1.5cm,0);
    \draw[-latex, dotted] (fw.south) |- ([xshift=0.3cm] db.north) node [pos=0.56, anchor=center] {\makecell{Release $m_{+1}$\\without delay}};
\end{tikzpicture}%
}%
    \caption{\label{fig:system-model:pof-model} Functional model of the \acl{POF} $\texttt{POF}_n(\mathcal{F},o)$. For a data unit $m$, $m_{-1}$ [resp., $m_{+1}$] refers to the data unit of the aggregate $\mathcal{F}$ that exited the reference $o$ just before [resp., just after] $m$.}
\end{figure}
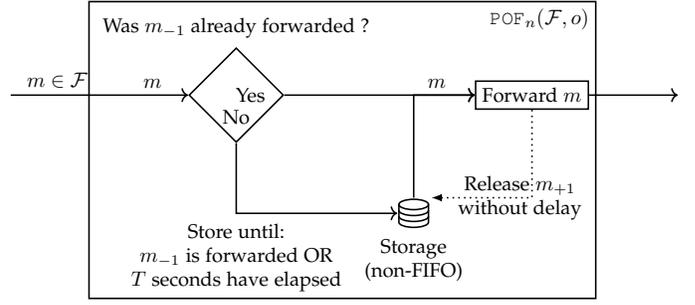

\noindent\ul{POF:} Consider a set of flows $\mathcal{F}$ crossing $n$ such that for each flow $f\in\mathcal{F}$, $o$ is a diamond ancestor of $n$ in $\mathcal{G}(f)$. 
The output port in $n$ can contain a \emph{\acf{POF}} for the aggregate $\mathcal{F}$ with reference $o$, noted $\texttt{POF}_n(\mathcal{F},o)$.
We assume that $\texttt{POF}_n(\mathcal{F},o)$ has the knowledge of the order in which the data units belonging to the aggregate $\mathcal{F}$ exited the reference $o$.
$\texttt{POF}_n(\mathcal{F},o)$ then enforces the same order at its own output, by delaying the packets that are out of order.

However, a data unit $m$ cannot be delayed by $\texttt{POF}_n(\mathcal{F},o)$ for a duration longer than the \ac{POF}'s timeout parameter $T$: After being stored for a duration $T$, $m$ is immediately released, even if the previously-expected data unit has not been received so far.
The timeout allows the \ac{POF} to recover from losses without blocking the following data units forever \cite{mohammadpourPacketReorderingTimeSensitive2020, vargaDeterministicNetworkingDetNet2021}.
We assume that the timeout value of every \ac{POF} conforms with the recommendations of \cite[\S IV.B]{mohammadpourPacketReorderingTimeSensitive2020}.
As a consequence, the timeout cannot only be triggered when one of the data units $m$ of $\mathcal{F}$ is \emph{lost for} the \ac{POF}.

The model of \ac{POF} is illustrated in Fig.~\ref{fig:system-model:pof-model}. A possible implementation is given in \cite[\S 3.4]{mohammadpourPacketReorderingTimeSensitive2020} and \cite{vargaDeterministicNetworkingDetNet2021}.
%
%
A \ac{POF} cannot be placed at an \ac{EP}-vertex: we always assume that the duplicates are eliminated before the flow is handed to the \ac{POF}, which is consistent with the assumptions in \cite[\S 4.1]{vargaDeterministicNetworkingDetNet2021}.
%
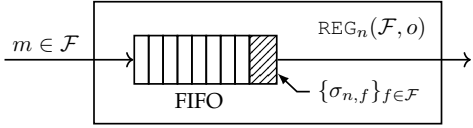
\begin{figure}\centering
    \resizebox*{0.7\linewidth}{!}{
\begin{tikzpicture} %
\tikzstyle{mn} = [draw, minimum height=1.5cm]
\tikzstyle{ns} = [draw, minimum height=0.75cm, anchor=west]
\tikzstyle{fil} = [draw, pattern=north east lines, text width=0.2cm]
\tikzstyle{hfil} = [fil, minimum height=0.75cm]
\tikzstyle{cor} = [xshift=-\pgflinewidth]

\node[ns] at (6,-0.05) (tt) {};
\node[ns] at ([cor]tt.east) (p) {};
\node[ns] at ([cor]p.east) (p) {};
\node[ns] at ([cor]p.east) (p) {};
\node[ns] at ([cor]p.east) (p) {};
\node[anchor=north] at (p.south) (fifoex) {FIFO};
\node[ns] at ([cor]p.east) (p) {};
\node[ns] at ([cor]p.east) (p) {};
\node[ns] at ([cor]p.east) (p) {};
\draw let \p1=(p.north east) in node at (\x1,0) (epfr) {};
\node[hfil, cor, anchor=west] at ([yshift=-0.05cm]epfr.center) (fpfr2) {};

\draw[->] ([xshift=-2cm] tt.west) -- (tt.west) node[pos=0, anchor=south west] {$m\in \mathcal{F}$}; 
\draw[->] (fpfr2.east) -- ++(3cm,0);

\node[anchor=north west] at ([xshift=0.5cm, yshift=0.2cm] fpfr2.south east)(sigma) {$\{\sigma_{n,f}\}_{f\in\mathcal{F}}$};
\draw[-latex] (sigma.west) -- ++(-0.2cm,0) -- (fpfr2);
\node[anchor=south west] at ([xshift=0.5cm, yshift=-0.2cm] fpfr2.north east) (func) {$\texttt{REG}_n(\mathcal{F},o)$};
\node[draw, fit={([xshift=-0.5cm] tt.west) (fifoex) (fpfr2) (sigma) (func)}] {};
\end{tikzpicture}%
}%
    \caption{\label{fig:system-model:reg-model} Model of a \acl{REG} $\texttt{REG}_{n}(\mathcal{F},o)$, with shaping curves $\{\sigma_{n,f}\}_f$.}
\end{figure}

\noindent\ul{REG:} Consider a set of flows $\mathcal{F}$ crossing $n$ such that, for each flow $f\in\mathcal{F}$, $o$ is a diamond ancestor of $n$ in $\mathcal{G}(f)$. The output port in $n$ can contain a \emph{\acf{REG}} for the aggregate $\mathcal{F}$ with reference $o$, noted $\texttt{REG}_n(\mathcal{F},o)$.
The regulator is configured with a set of shaping curves, one per flow $f$ of the aggregate $\mathcal{F}$, which we note $\{\sigma_{n,f}\}_{f\in\mathcal{F}}$.
For each $f \in \mathcal{F}$, $\sigma_{n,f}$ must be concave and must be an arrival curve of $f$ at the output of the reference vertex $o$.
The regulator then puts all the packets of the aggregate $\mathcal{F}$ in a \ac{FIFO} queue (Fig.~\ref{fig:system-model:reg-model}) and examines only the head-of-line packet.
It releases the head-of-line packet as soon as doing so does not violate the shaping curve $\sigma_{n,f}$, where $f$ is the flow of the head-of-line packet.
When the regulator processes a single flow, $\mathcal{F} = \{f\}$, we model it as a \emph{\acf{PFR}} \cite[\S 1.7.4]{leboudecNetworkCalculusTheory2001}.
When $\mathcal{F}$ contains two or more flows, we model it as an \emph{\acf{IR}} \cite{leboudecTheoryTrafficRegulators2018}.

We consider that each output port contains a forwarding pipeline before the \ac{CBQS} with the following optional functions, in this order:
$\acsp{PEF} \rightarrow \acsp{POF} \rightarrow \acsp{REG}$.

\begin{figure}[b]\centering
    \resizebox*{\linewidth}{!}{\begin{tikzpicture}[> = stealth']
    \tikzstyle{mn} = [draw, minimum width=1.8cm, minimum height=0.6cm]
    \tikzstyle{n} = [draw, minimum height=1.2cm]
    \tikzstyle{nn} = [n, anchor=east]
    \tikzstyle{xx} = [xshift=-0.5cm]
    \tikzstyle{yy} = [yshift=0.1cm]
    \tikzstyle{myy} = [xshift=-0.1cm]
    \node[n] at (0,0) (cbqs) {\acs{CBQS}};
    \node[nn] at ([xx] cbqs.west) (reg1) {\makecell{$\texttt{REG}_F$\\$(\{f,g\},B)$}};
    \node[nn] at ([xx] reg1.west) (pof1) {\makecell{$\texttt{POF}_F$\\$(\{f,g\},B)$}};
    \node[mn, anchor=north east] at ([xx] pof1.north west) (pef1) {$\texttt{PEF}_F(f)$};
    \node[mn, anchor=south east] at ([xx] pof1.south west) (pef2) {$\texttt{PEF}_F(g)$};
    \draw[->] (pef1) -- (pof1);
    \draw[->] (pef2) -- (pof1);
    \draw[->] (pof1) -- (reg1);
    \draw[->] (reg1) -- (cbqs);
\end{tikzpicture}%
}%
    \caption{\label{fig:system-model:example-of-functions} Example of an organization of the optional functions within an output port. After their respective \ac{PEF}, the two flows share the same \ac{POF} and the same \acf{REG}.} 
\end{figure}
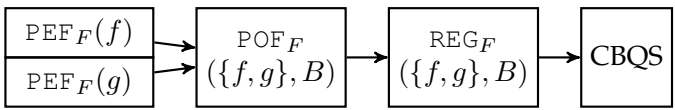
\begin{figure}\centering
    \resizebox*{\linewidth}{!}{\begin{tikzpicture}[> = stealth']
    \tikzstyle{mn} = [draw, minimum width=2.5cm, minimum height=0.6cm]
    \tikzstyle{mmn} = [draw, minimum width=1.7cm, minimum height=0.6cm]
    \tikzstyle{n} = [draw, minimum height=1.3cm]
    \tikzstyle{nn} = [n, anchor=east]
    \tikzstyle{xx} = [xshift=-0.5cm]
    \tikzstyle{yy} = [yshift=0.1cm]
    \tikzstyle{myy} = [xshift=-0.1cm]
    \node[n] at (0,0) (cbqs) {\acs{CBQS}};
    \node[mn, anchor=north east] at ([xx] cbqs.north west) (reg1) {$\texttt{REG}_F(\{f\},B)$};
    \node[mn, anchor=north east] at ([xx] reg1.north west) (pof1) {$\texttt{POF}_F(\{f\},B)$};
    \node[mmn, anchor=north east] at ([xx] pof1.north west) (pef1) {$\texttt{PEF}_F(f)$};
    \node[mn, anchor=south east] at ([xx] cbqs.south west) (reg2) {$\texttt{REG}_F(\{g\},B)$};
    \node[mn, anchor=south east] at ([xx] reg2.south west) (pof2) {$\texttt{POF}_F(\{g\},B)$};
    \node[mmn, anchor=south east] at ([xx] pof2.south west) (pef2) {$\texttt{PEF}_F(g)$};
    \draw[->] (pef1) -- (pof1);
    \draw[->] (pef2) -- (pof2);
    \draw[->] (pof1) -- (reg1);
    \draw[->] (pof2) -- (reg2);
    \draw[->] (reg1) -- (cbqs);
    \draw[->] (reg2) -- (cbqs);
\end{tikzpicture}%
}%
    \caption{\label{fig:system-model:example-of-functions-a} Example of an the organization of the optional functions within an output port with one \ac{POF} and one \ac{REG} per flow.} 
\end{figure}
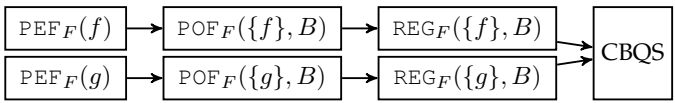
\textbf{Example:} Consider two flows $f,g$, both with the same flow graph of Fig.~\ref{fig:system-model:example-flow-graph} and a \ac{PEF} at $F$.
The output port $F$ processes streams of packets coming from both parents $C$ and $D$.
A first possible example of the organization of the functions before the \ac{CBQS} within vertex $F$ is shown in Fig.~\ref{fig:system-model:example-of-functions}.
Each flow is first processed by its respective \ac{PEF}, then both duplicate-free flows are reordered as an aggregate by using $\texttt{POF}_F(\{f,g\},B)$.
This function enforces the same order for the aggregate as the one at the output of $B$, \emph{i.e.}, before the redundant section.
Last, they are both processed by the same interleaved regulator that enforces two different contracts for $f$ and for $g$, but that keeps the aggregate $\{f,g\}$ \ac{FIFO}.
A variant of this situation is shown in Fig.~\ref{fig:system-model:example-of-functions-a}.
After elimination, each flow is now independent from the other one, where the \acp{POF} enforce per-flow order and the two \acsp{REG} are \acfp{PFR}.
%
This situation is different from Fig.~\ref{fig:system-model:example-of-functions} because a packet of $f$ cannot be delayed by a packet of $g$. In addition, this configuration could have a higher hardware cost than in Fig.~\ref{fig:system-model:example-of-functions}.

\noindent\ul{FIFO assumptions:} With the exception of \ac{POF}, each network element is assumed to be \ac{FIFO}
for the class of interest.
\noindent\ul{Assumptions on losses:}
With the exception of \ac{PEF}, each function, each \ac{CBQS}, each switching fabric, each input port and each internal connection within a device is assumed to be lossless (does not lose any packets).
Packets can be lost on the transmission links between devices.
This model covers various failures, including random media losses, the shutdown of an output port (equivalent to its out-going link losing all packets) and the shutdown of an input port (equivalent to its in-going link losing all packets).

As packets can be lost on transmission links, the network is not assumed to be lossless.
Of course, the latency bounds computed in this paper are only valid for the non-lost data units (the data units for which at least one replicate reaches the destination), but these bounds remain valid even if some other data units are lost in the network.




\vspace{\myvspacebeforesec}
\section{Toolbox for the Deterministic Analysis of Packet Replication and Elimination}\label{sec:toolbox}

Network calculus \cite{leboudecNetworkCalculusTheory2001} is a mathematical framework for computing deterministic latency bounds.
It relies on the concepts of arrival and service curves.
An arrival curve $\alpha_{f,M}$ at a specific observation point $M$ and for a specific flow~$f$ is a constraint on the maximum amount of traffic of flow~$f$ that can cross $M$ over any periods of time $[s,t]$, which is equivalent to: $\forall s\le t, R(t) - R(s) \le \alpha(t-s)$, with $R(t)$ the amount of data of flow $f$ crossing $M$ between $0$ and $t$. 
Also, a service curve $\beta_S$ of a specific network element $S$ is a constraint on the minimum amount of traffic that the network element must serve.
Network calculus gives delay and backlog bounds in network elements given the arrival-curve and service-curve constraints \cite{leboudecNetworkCalculusTheory2001,mohammadpourImprovedDelayBound2019a}.

In this section, we compute an upper bound of the burstiness increase caused by \ac{PREF} by computing an arrival curve $\alpha_{f,\text{PEF}^*}$ for the flow $f$ at the output of the \acl{PEF}.
The arrival curve $\alpha_{f,\text{PEF}^*}$ can then be combined with the service curves of the downstream elements (that can be found in \cite{maile2020network,zhao2020latency}) to compute a delay bound in the last section of Fig.~\ref{fig:prob-statement:three-sections}.
This delay bound is useful for validating the system's latency requirements.

We also quantify the amount of mis-ordering introduced by the redundancy.
This bound can be compared to the application's requirement to decide if reordering is required before delivering the data to the application.
If so, the same bound can be combined with the results of \cite{mohammadpourPacketReorderingTimeSensitive2020} to configure the \acf{POF} that corrects this mis-ordering.
The consequences of such reordering on the flow's delay and burstiness are also analyzed.

\vspace{\myvspacebeforesubsec}
\subsection{Output Arrival Curve of a \acs{PEF}}\label{sec:toolbox:pef-oac}
%
\begin{theorem}[Output arrival curve of a \ac{PEF}]\label{thm:toolbox:pef-oac}
    Let $\texttt{PEF}_n(f)$ be a \acl{PEF} for flow $f$ at the output port of vertex $n\in vertices(\mathcal{G}(f))$.
    Assume that $\alpha_{f,\text{PEF}^\text{in}}$ is an arrival curve of $f$ at the input of $\texttt{PEF}_n(f)$.
    Then
    \begin{enumerate}[1/]

        \item $\alpha_{f,\text{PEF}^\text{in}}$ is an arrival curve for the flow at the output of the \ac{PEF}.

        \item For every diamond ancestor $a$ of $n$ in $\mathcal{G}(f)$, assume that $\alpha_{f,a^*}$ is an arrival curve for $f$ at the output of $a$ and denote by $d_f^{a\rightarrow n}$ [resp., $D_f^{a \rightarrow n}$] a minimum [resp., maximum] delay bound for $f$ between the output of $a$ and the input of $\texttt{PEF}_n(f)$, along any possible paths $a\rightarrow n$ within the graph $\mathcal{G}(f)$.
        Then
        \begin{equation}\label{eq:thm:pef-oac:ancestor}\alpha_{f}^{a\rightarrow n} \triangleq   \alpha_{f,a^*} \oslash \delta_{(D_f^{a \rightarrow n}) - (d_f^{a\rightarrow n})}\end{equation}
        is an arrival curve for $f$ at the output of the \ac{PEF}.
    \end{enumerate}
    Furthermore, the min-plus convolution of all above arrival curves
    \begin{equation}
        \alpha_{f,\text{PEF}^*} = \alpha_{f,\text{PEF}^{\text{in}}} \otimes \alpha_{f}^{a_1\rightarrow n}  \otimes \alpha_{f}^{a_2\rightarrow n}  \otimes \alpha_{f}^{a_3\rightarrow n}  \otimes \dots
    \end{equation}
    for any set of diamond ancestors $a_1,a_2,a_3,\dots$ of $n$ in $\mathcal{G}(f)$ is also an arrival curve for $f$ at the output of the \ac{PEF}, where $\otimes$ denote the min-plus convolution\footnote{$f\otimes g: t\mapsto \inf_{s\ge0} (f(s) + g(t-s))$. The min-plus convolution is associative and commutative \cite[\S 2.1.3]{bouillard2018deterministic}.}.


    

\end{theorem}

The result is proved as follows:
Item 1/ is a direct consequence of the fact that the \ac{PEF} has no delay.
Item 2/ is obtained by considering the entire system made of the portion of the graph $\mathcal{G}(f)$ between the diamond ancestor $a$ and $n$.
This system is neither lossless nor \ac{FIFO}, but several classical network-calculus results remain applicable, as we discuss in Appendix~\ref{sec:appendix:non-fifo-non-lossless-results}.
$\alpha_{f,\text{PEF}^*}$ is finally obtained by applying \cite[Lemma 1.2.4]{leboudecNetworkCalculusTheory2001}.
A formal proof of Theorem~\ref{thm:toolbox:pef-oac} is given in Appendix~\ref{proof:appendix:pef-oac}.

\textbf{Application to the Toy Example:} An arrival curve $\alpha_{f,\text{PEF}^*}$ for $f$ at the output of the \ac{PEF} within $F$ (Fig.~\ref{fig:prob-statement:toy-example}) is shown in Fig.~\ref{fig:toolbox:toy-example-tihhtness} with a solid red line.

The first constituent, $\alpha_{f,\text{PEF}^\text{in}}$ is the arrival curve at $f$ at the input of the \ac{PEF} (as per Theorem~\ref{thm:toolbox:pef-oac}, Item 1).
To obtain it, we first observe that the periodic profile of the flow $f$ at the output of $B$ (as on the Line ``\texttt{in}'' of Fig.~\ref{fig:prob-statement:double-rate}) is constrained by the leaky-bucket arrival curve $\alpha_{f,B^*} = \gamma_{r_0,b_0}$ with a rate of one data unit per unit of time ($r_0 = 1$ d.u./t.u.) and a burst of one data unit ($b_0 = 1$ d.u.).
By using the jitter bound within $C$ and $D$ and Proposition~\ref{prop:appendix:ac-after-lossy} in Appendix~\ref{sec:appendix:non-fifo-non-lossless-results}, we obtain that the arrival curves of $f$ at the output of $C$ and $D$, $\alpha_{f,C^*}$ and $\alpha_{f,D^*}$, equal to the same leaky-bucket arrival curve $\gamma_{r_0,2b_0}$ with a burst $2b_0$ of two units of data.
As $f$ enters $F$ from both $C$ and $D$, we obtain  $\alpha_{f,\text{PEF}^\text{in}} = \alpha_{f,C^*} + \alpha_{f,D^*} = \gamma_{2r_0,4b_0}$, a leaky-bucket arrival curve with a rate $2r_0$ and a burst $4b_0$.

The second constituent of $\alpha_{f,\text{PEF}^*}$ in Fig.~\ref{fig:toolbox:toy-example-tihhtness} is obtained by applying the Equation (\ref{eq:thm:pef-oac:ancestor}) of Theorem~\ref{thm:toolbox:pef-oac}, Item 2/ with $a=B$.
From Fig.~\ref{fig:prob-statement:toy-example}, we obtain that a delay lower-bound [resp., an upper-bound] for $f$ from $B$ to $F$ along any possible paths within $\mathcal{G}(f)$ is $d_{f}^{B\rightarrow F} = 0$ t.u. (through $C$) [resp., $D_{f}^{B\rightarrow F} = 7$ t.u, through $D$].
We obtain $\alpha_{f}^{B\rightarrow F} = \alpha_{f,B^*} \oslash \delta_{D_{f}^{B\rightarrow F} - d_{f}^{B\rightarrow F}} = \gamma_{r_0,b_0}\oslash\delta_{7}$, \emph{i.e.}, $ \alpha_{f}^{B\rightarrow F} = \gamma_{r_0,8b_0}$.

If we assumes that the \ac{PEF} does not delete any packet, as in the \emph{intuitive} approach mentioned in Section~\ref{sec:prob-statement}, we only know that $f$ has the arrival curve $\alpha_{f,\text{PEF}^{\text{in}}}$ at the output of the \ac{PEF} (Item 1 of the Theorem).
This arrival curve shows that the traffic can exhibits a burst of $4b_0$ and a rate $2r_0$ twice as big as the normal source rate.

But our theorem goes beyond the intuitive approach: its second item applied with $a=B$ provides a second arrival curve for $f$: $\alpha_{f}^{B\rightarrow F}$.
In the network-calculus framework, we can combine the knowledge of two arrival curves by computing their min-plus convolution: $\alpha_{f,\text{PEF}^*} = \alpha_{f,\text{PEF}^{\text{in}}} \otimes \alpha_{f}^{B\rightarrow F}$ is also an arrival curve for $f$ at the output of the \ac{PEF}.
With the leaky-bucket arrival curves of the toy example, the min-plus convolution is simply the minimum of the two curves, shown with a solid red line on Fig.~\ref{fig:toolbox:toy-example-tihhtness}.
We observe that Theorem~\ref{thm:toolbox:pef-oac} provides a better upper-bound of the traffic than the intuitive approach.
For example, $\alpha_{f,\text{PEF}^*}$ indicates that the double rate $2r_0$ is only a peak rate that the traffic cannot exhibits forever: flow $f$ keeps a sustained rate $r_0$, but with a much higher burst $8b_0$.
In network calculus, the arrival curves that describe flows with a peek rate ($2 r_0$) and a sustained rate ($r_0$) are called \emph{\acl{VBR}} (\acs{VBR}) arrival curves.
Theorem~\ref{thm:toolbox:pef-oac} provides for the toy example the best possible \acs{VBR} arrival curve, as we prove later.
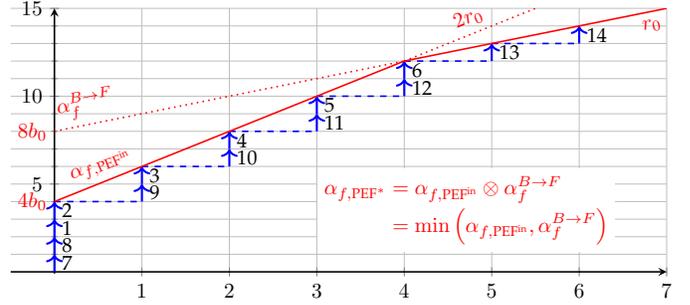
\begin{figure}\centering
    \resizebox*{\linewidth}{!}{\begin{tikzpicture}
	\tikzstyle{lw} = [line width=1pt]
	\begin{axis}[
		axis x line=center,
		axis y line=center,
        ytick={0,5,10,15},
		xmin=-0.5,
		xmax=7,
		ymin=-0.1,
		ymax=15,
        width=12.5cm,
        height=6cm,
		minor y tick num=4,
		grid=both
	]	
	\draw[->,blue,lw] (axis cs:0,0) -- (axis cs:0,1) node[pos=0.5,right, black] {7};
	\draw[->,blue,lw] (axis cs:0,1) -- (axis cs:0,2) node[pos=0.5,right, black] {8};
	\draw[->,blue,lw] (axis cs:0,2) -- (axis cs:0,3) node[pos=0.5,right, black] {1};
	\draw[->,blue,lw] (axis cs:0,3) -- (axis cs:0,4) node[pos=0.5,right, black] {2};

    \draw[dashed, blue] (axis cs:0,4) -- (axis cs:1,4);
	
	\draw[->,blue,lw] (axis cs:1,4) -- (axis cs:1,5) node[pos=0.5,right, black] {9};
	\draw[->,blue,lw] (axis cs:1,5) -- (axis cs:1,6) node[pos=0.5,right, black] {3};

    \draw[dashed, blue] (axis cs:1,6) -- (axis cs:2,6);
	
	\draw[->,blue,lw] (axis cs:2,6) -- (axis cs:2,7) node[pos=0.5,right, black] {10};
	\draw[->,blue,lw] (axis cs:2,7) -- (axis cs:2,8) node[pos=0.5,right, black] {4};

    \draw[dashed, blue] (axis cs:2,8) -- (axis cs:3,8);
	
	\draw[->,blue,lw] (axis cs:3,8) -- (axis cs:3,9) node[pos=0.5,right, black] {11};
	\draw[->,blue,lw] (axis cs:3,9) -- (axis cs:3,10) node[pos=0.5,right, black] {5};

    \draw[dashed, blue] (axis cs:3,10) -- (axis cs:4,10);
	
	\draw[->,blue,lw] (axis cs:4,10) -- (axis cs:4,11) node[pos=0.5,right, black] {12};
	\draw[->,blue,lw] (axis cs:4,11) -- (axis cs:4,12) node[pos=0.5,right, black] {6};

    \draw[dashed, blue] (axis cs:4,12)  -- (axis cs:5,12);
	
	\draw[->,blue,lw] (axis cs:5,12) -- (axis cs:5,13) node[pos=0.5,right, black] {13};

    \draw[dashed, blue] (axis cs:5,13)  -- (axis cs:6,13);
	
	\draw[->,blue,lw] (axis cs:6,13) -- (axis cs:6,14) node[pos=0.5,right, black] {14};
	
	\addplot[red, domain=0:4]{4+2*x} node[pos=0.15, above, sloped] {$\alpha_{f,\text{PEF}^\text{in}}$};
	\addplot[red, dotted, domain=4:8]{4+2*x} node[pos=0.2,above,sloped]{$2r_0$};;
	
	\addplot[red, domain=4:8]{8+x} node[pos=0.7,below,sloped]{$r_0$};
	\addplot[red, dotted, domain=0:4]{8+x} node[pos=0.1,above,sloped] {$\alpha_{f}^{B\rightarrow F}$};
	
	\node[anchor=east, red] at (axis cs:0,4) {$4b_0$};
	
	\node[anchor=east, red] at (axis cs:0,8) {$8b_0$};

	\node[red, fill=white, anchor=west] at (axis cs:3,3.5) {$\begin{aligned}\alpha_{f,\text{PEF}^*}&=\alpha_{f,\text{PEF}^\text{in}}\otimes \alpha_{f}^{B\rightarrow F} \\ &= \min\left(\alpha_{f,\text{PEF}^\text{in}},\alpha_{f}^{B\rightarrow F}\right)\end{aligned}$};
		
	\end{axis}%
\end{tikzpicture}%
}%
    \caption{\label{fig:toolbox:toy-example-tihhtness} Solid red: $\alpha_{f,\text{PEF}^*}$, arrival curve of $f$ on the toy example, at the output of the \acl{PEF} $\texttt{PEF}_F(f)$, obtained by applying Theorem~\ref{thm:toolbox:pef-oac}. Dashed blue: Cumulative arrival function obtained with the trajectory of Fig.~\ref{fig:toolbox:toy-example-tight}, which shows the tightness of the result.}
\end{figure}

\textbf{Remark:} 
Theorem~\ref{thm:toolbox:pef-oac} does not require to identify pairs of replication/elimination functions, with one \ac{PRF} and one \ac{PEF} in each pair.
Therefore, Theorem~\ref{thm:toolbox:pef-oac} is suited for complex flow graphs, including graphs with repeated patterns of redundancy, with meshes, as well as graphs where the \acl{PEF} is not located at the merge point of the paths.
When pairs of \ac{PRF}/\ac{PEF} can be identified as in Fig.~\ref{fig:prob-statement:three-sections}, the following simpler corollary 
can be used.

\begin{corollary}[Application of the theorem to a unique redundant section with parallel systems]\label{cor:toolbox:thm-simple-application}
    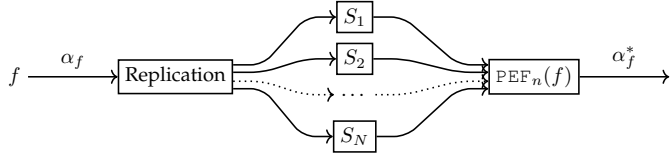
\begin{figure}
        \resizebox{\linewidth}{!}{\begin{tikzpicture}
	\node[draw] at (0,0) (s) {Replication};
	\node[draw] at (3,1) (a) {$S_1$};
	\node[draw] at (3,0.3) (b) {$S_2$};
	\node at (3,-0.3) (c) {\ldots};
	\node[draw] at (3,-1) (d) {$S_N$};
	
	\node[draw] at (6,0) (recov) {$\texttt{PEF}_n(f)$};
	
	\draw[->, rounded corners=0.2cm] ([yshift=0.2cm]s.east) -- ([xshift=0.5cm,yshift=0.2cm] s.east) -- ([xshift=-0.5cm] a.west) -- (a.west);
	\draw[->, rounded corners=0.2cm] ([yshift=0.075cm]s.east) -- ([xshift=0.5cm,yshift=0.075cm] s.east) -- ([xshift=-0.5cm] b.west) -- (b.west);
	\draw[->, rounded corners=0.2cm, dotted] ([yshift=-0.075cm]s.east) -- ([xshift=0.5cm,yshift=-0.075cm] s.east) -- ([xshift=-0.5cm] c.west) -- (c.west);
	\draw[->, rounded corners=0.2cm] ([yshift=-0.2cm]s.east) -- ([xshift=0.5cm,yshift=-0.2cm] s.east) -- ([xshift=-0.5cm] d.west) -- (d.west);

	w
	\draw[->, rounded corners=0.2cm] (a.east) -- ([xshift=0.5cm] a.east) -- ([xshift=-0.5cm,yshift=0.2cm] recov.west) -- ([yshift=0.2cm]recov.west) node[pos=0.5,anchor=center] (tA) {};
	\draw[->, rounded corners=0.2cm] (b.east) -- ([xshift=0.5cm] b.east) -- ([xshift=-0.5cm,yshift=0.075cm] recov.west) -- ([yshift=0.075cm]recov.west);
	\draw[->, rounded corners=0.2cm, dotted] (c.east) -- ([xshift=0.5cm] c.east) -- ([xshift=-0.5cm,yshift=-0.075cm] recov.west) -- ([yshift=-0.075cm]recov.west);
	\draw[->, rounded corners=0.2cm] (d.east) -- ([xshift=0.5cm] d.east) -- ([xshift=-0.5cm,yshift=-0.2cm] recov.west) -- ([yshift=-0.2cm]recov.west);
	
	\draw[->] ([xshift=-1.5cm] s.west) -- (s.west) node[pos=0,left] {$f$} node[pos=0.5,above] {$\alpha_f$};
	\draw[->] (recov.east) -- ([xshift=1.5cm] recov.east) node[pos=0.5,above] {$\alpha_f^*$};
\end{tikzpicture}%
}%
        \caption{\label{fig:toolbox:oac-prop-simplified-figure} Notations of Corollary~\ref{cor:toolbox:thm-simple-application}. Flow $f$ is replicated and sent to $N$ parallel systems. Corollary~\ref{cor:toolbox:thm-simple-application} gives the arrival curve $\alpha_{f}^*$ at the output of the \acl{PEF} $\texttt{PEF}_n(f)$.}
    \end{figure}
    Consider a flow $f$ with an arrival curve $\alpha_f$ that is replicated and sent into $N$ systems $\{S_i\}_{i\in \llbracket 1 , N\rrbracket}$ and then processed by a \acl{PEF} $\texttt{PEF}(f)$, as in Fig.~\ref{fig:toolbox:oac-prop-simplified-figure}.
    Note that each $S_i$ is not necessary a single network element but can be any combination of network elements.
    Assume that the packets forwarded through $S_i$ (\emph{i.e.}, the ones not lost) have a delay through $S_i$ that is bounded within $[d_i, D_i]$. Then,
    \begin{equation}\resizebox*{\linewidth}{!}{$\label{eq:toolbox:simplified-proposition}
        \alpha_f^* = \left(\sum_{i\in\llbracket 1, N\rrbracket} \alpha_f \oslash \delta_{(D_i - d_i)}\right) \otimes \left( \alpha_f \oslash \delta_{\left(\max\limits_{i\in\llbracket 1, N\rrbracket} D_i - \min\limits_{j\in\llbracket 1, N\rrbracket}d_j\right)} \right)$}
    \end{equation} is an arrival curve for $f$ at the output of $\texttt{PEF}(f)$.
\end{corollary}

Corollary~\ref{cor:toolbox:thm-simple-application} is a direct application of Theorem~\ref{thm:toolbox:pef-oac}. A formal proof is given in Appendix~\ref{proof:appendix:toolbox:cor}. 
The corollary is of interest for two reasons.
First, its simpler notation is likely to cover many industrial applications containing a unique redundant portion with parallel systems.
Second, Corollary~\ref{cor:toolbox:thm-simple-application} is tight in the following sense.

\begin{proposition}[The result in Corollary~\ref{cor:toolbox:thm-simple-application} is tight with $N=2$ and leaky-bucket-constrained flows, in the family of \ac{VBR} arrival curves.]\label{prop:toolbox:simplified-result-is-tight}
		
    \textbf{For any} leaky-bucket arrival curve $\gamma_{r,b}$, for any set of values $d_1,D_1,d_2,D_2 \in \mathbb{R}$ such that  $d_1 \leq D_1$ and $ d_2 \leq D_2$,

    \textbf{there exists} a flow $f$ with arrival-curve $\alpha_f = \gamma_{r,b}$ and no minimum packet length whose content is replicated and sent to two systems $S_1$ and $S_2$ in which the packets of $f$ suffer a delay bounded in $[d_1,D_1]$ and $[d_2,D_2]$ respectively; the sum of the outputs of the two systems is then processed by a \acl{PEF} $\texttt{PEF}_n(f)$,

    \textbf{such that}, the arrival curve $\alpha_f^*$ defined in (\ref{eq:toolbox:simplified-proposition}) is the best \ac{VBR} arrival curve for $f$ at the output of $\texttt{PEF}_n(f)$.
\end{proposition}

Note that, due to the inherent nature of the \ac{PEF} processing packets, there could exist staircase arrival-curves that fit the worst-case traffic even better than the arrival curve provided in Corollary~\ref{cor:toolbox:thm-simple-application}.
However, deterministic computational tools process concave piecewise-linear arrival-curves better than staircase arrival-curves~\cite{mifdaouiAccuracyComplexityTradeoffsCompositional2017}.
Proposition~\ref{prop:toolbox:simplified-result-is-tight} proves that we obtain the best arrival-curve in the family of concave piecewise-linear arrival-curves with two segments or less.

\textbf{Intuition of the Proof with the Toy Example:}
We give an intuition of the proof of Proposition~\ref{prop:toolbox:simplified-result-is-tight} by using the toy example of Fig.~\ref{fig:prob-statement:toy-example}.
Our goal is to obtain a cumulative function $R^*(t)$ at the output of \ac{PEF} such that $t\mapsto R^*(t) - R^*(s)$ ``perfectly fits'' the arrival curve $\gamma_{2r_0,4b_0} \otimes \gamma_{r_0,8b_0}$, for some observation starting time $s$ (as in Fig.~\ref{fig:toolbox:toy-example-tihhtness}).
In the scenario of Fig.~\ref{fig:prob-statement:double-rate}, we already achieved a peak rate of $2r_0$ by using a disconnection of the short link for a duration equal to the delay difference between the two paths.
To obtain the worst-case burst, we now simply need to use the jitter within each path and synchronize the moments when the maximum burst on each path reaches the \ac{PEF}.

This is done by using the trajectory shown in Fig.~\ref{fig:toolbox:toy-example-tight}.
Here, Packet $1$ suffers the maximal delay on the long path and the following packets suffer only 6 t.u.
This causes Packets 1 and 2 to exit $D$ at the same time.
We do the same with Packets 7 and 8 through $C$ and we synchronize these two events at the same time, so that four packets simultaneously exit the \ac{PEF} at t.u. 8.
In the figure we spread the packets within t.u. 8 for ease of reading, but they exit at the exact same time (t.u. 8).
Because of this, we can also put an arbitrary order of arrivals among them (in a real-life system it means that there exists a very small difference in their reception instants).

If we start counting the packets at Time Unit 8, we observe the cumulative arrival function shown in dashed blue in Fig.~\ref{fig:toolbox:toy-example-tihhtness}, for which it is clear that the arrival curve in solid red is the best concave piecewise-linear envelope with two segments.
The formal proof of Proposition~\ref{prop:toolbox:simplified-result-is-tight} in Appendix~\ref{proof:appendix:pef-oac:simplified-result-is-tight} extends the intuition for any choice of values for $r,b,d_1,d_2,D_1,D_2$ (assuming no minimal packet length).

\begin{figure}\centering
    \resizebox*{0.95\linewidth}{!}{\begin{tikzpicture}
	\tikzstyle{lw} = [line width=2pt]
	\tikzstyle{every node}=[font=\Large]
	
	\draw [->,lw] (-1,0) -- (15,0);
	
	\draw (0,-0.2) -- (0,0.2);
	\draw (1,-0.2) -- (1,0.2);
	\draw (2,-0.2) -- (2,0.2);
	\draw (3,-0.2) -- (3,0.2);
	\draw (4,-0.2) -- (4,0.2);
	\draw (5,-0.2) -- (5,0.2);
	\draw (6,-0.2) -- (6,0.2);
	\draw (7,-0.2) -- (7,0.2);
	\draw (8,-0.2) -- (8,0.2);
	\draw (9,-0.2) -- (9,0.2);
	\draw (10,-0.2) -- (10,0.2);
	\draw (11,-0.2) -- (11,0.2);
	\draw (12,-0.2) -- (12,0.2);
	\draw (13,-0.2) -- (13,0.2);
	\draw (14,-0.2) -- (14,0.2);
	
	\draw[lw,blue,->] (0,0) -- (0,1) node[pos=1,above, black] {1};
	\draw[lw,blue,->] (1,0) -- (1,1) node[pos=1,above, black] {2};
	\draw[lw,blue,->] (2,0) -- (2,1) node[pos=1,above, black] {3};
	\draw[lw,blue,->] (3,0) -- (3,1) node[pos=1,above, black] {4};
	\draw[lw,blue,->] (4,0) -- (4,1) node[pos=1,above, black] {5};
	\draw[lw,blue,->] (5,0) -- (5,1) node[pos=1,above, black] {6};
	\draw[lw,blue,->] (6,0) -- (6,1) node[pos=1,above, black] {7};
	\draw[lw,blue,->] (7,0) -- (7,1) node[pos=1,above, black] {8};
	\draw[lw,blue,->] (8,0) -- (8,1) node[pos=1,above, black] {9};
	\draw[lw,blue,->] (9,0) -- (9,1) node[pos=1,above, black] {10};
	\draw[lw,blue,->] (10,0) -- (10,1) node[pos=1,above, black] {11};
	\draw[lw,blue,->] (11,0) -- (11,1) node[pos=1,above, black] {12};
	\draw[lw,blue,->] (12,0) -- (12,1) node[pos=1,above, black] {13};
	\draw[lw,blue,->] (13,0) -- (13,1) node[pos=1,above, black] {14};
    \node at (14,0.5) {\ldots};

	\node[anchor=south, red] at (15,0)  {\texttt{in}};
	
	\draw [->,lw] (-1,-2) -- (15,-2);
	\draw (0,-2.2) -- (0,-1.8);
	\draw (1,-2.2) -- (1,-1.8);
	\draw (2,-2.2) -- (2,-1.8);
	\draw (3,-2.2) -- (3,-1.8);
	\draw (4,-2.2) -- (4,-1.8);
	\draw (5,-2.2) -- (5,-1.8);
	\draw (6,-2.2) -- (6,-1.8);
	\draw (7,-2.2) -- (7,-1.8);
	\draw (8,-2.2) -- (8,-1.8);
	\draw (9,-2.2) -- (9,-1.8);
	\draw (10,-2.2) -- (10,-1.8);
	\draw (11,-2.2) -- (11,-1.8);
	\draw (12,-2.2) -- (12,-1.8);
	\draw (13,-2.2) -- (13,-1.8);
	\draw (14,-2.2) -- (14,-1.8);
	\node[anchor=south, red] at (15,-2)  {\texttt{outC}};
	
	\draw[lw,blue,->] (7,-2) -- (7,-1) node[pos=1, anchor=south, black] {7,8};
    \draw[lw,blue,->] (7.1,-2) -- (7.1,-1) node[pos=1, anchor=south, black] {};
	\draw[lw,blue,->] (8,-2) -- (8,-1) node[pos=1, anchor=south, black] {9};
	\draw[lw,blue,->] (9,-2) -- (9,-1) node[pos=1, anchor=south, black] {10};
	\draw[lw,blue,->] (10,-2) -- (10,-1) node[pos=1, anchor=south, black] {11};
	\draw[lw,blue,->] (11,-2) -- (11,-1) node[pos=1, anchor=south, black] {12};
	\draw[lw,blue,->] (12,-2) -- (12,-1) node[pos=1, anchor=south, black] {13};
	\draw[lw,blue,->] (13,-2) -- (13,-1) node[pos=1, anchor=south, black] {14};
    \node at (14,-1.5) {\ldots};
	
	\draw [->,lw] (-1,-4) -- (15,-4);
	\draw (0,-4.2) -- (0,-3.8);
	\draw (1,-4.2) -- (1,-3.8);
	\draw (2,-4.2) -- (2,-3.8);
	\draw (3,-4.2) -- (3,-3.8);
	\draw (4,-4.2) -- (4,-3.8);
	\draw (5,-4.2) -- (5,-3.8);
	\draw (6,-4.2) -- (6,-3.8);
	\draw (7,-4.2) -- (7,-3.8);
	\draw (8,-4.2) -- (8,-3.8);
	\draw (9,-4.2) -- (9,-3.8);
	\draw (10,-4.2) -- (10,-3.8);
	\draw (11,-4.2) -- (11,-3.8);
	\draw (12,-4.2) -- (12,-3.8);
	\draw (13,-4.2) -- (13,-3.8);
	\draw (14,-4.2) -- (14,-3.8);
	\node[anchor=south, red] at (15,-4)  {\texttt{outD}};

    \draw[lw,blue,->] (7,-4) -- (7,-3) node[pos=1, anchor=south, black] {1,2};
	\draw[lw,blue,->] (7.1,-4) -- (7.1,-3) node[pos=1, anchor=south, black] {};
    \draw[lw,blue,->] (8,-4) -- (8,-3) node[pos=1, anchor=south, black] {3};
	\draw[lw,blue,->] (9,-4) -- (9,-3) node[pos=1, anchor=south, black] {4};
	\draw[lw,blue,->] (10,-4) -- (10,-3) node[pos=1, anchor=south, black] {5};
	\draw[lw,blue,->] (11,-4) -- (11,-3) node[pos=1, anchor=south, black] {6};
	\draw[lw,red,->] (12,-4) -- (12,-3) node[pos=1, anchor=south, black] {7};
	\draw[lw,red,->] (13,-4) -- (13,-3) node[pos=1, anchor=south, black] {8};
    \node at (14,-3.5) {\ldots};

	\draw [->,lw] (-1,-6) -- (15,-6);
	\draw (0,-6.2) -- (0,-5.8)	node[pos=0,below] {1};
	\draw (1,-6.2) -- (1,-5.8)	node[pos=0,below] {2};
	\draw (2,-6.2) -- (2,-5.8)	node[pos=0,below] {3};
	\draw (3,-6.2) -- (3,-5.8)	node[pos=0,below] {4};
	\draw (4,-6.2) -- (4,-5.8)	node[pos=0,below] {5};
	\draw (5,-6.2) -- (5,-5.8)	node[pos=0,below] {6};
	\draw (6,-6.2) -- (6,-5.8)	node[pos=0,below] {7};
	\draw (7,-6.2) -- (7,-5.8)	node[pos=0,below] {8};
	\draw (8,-6.2) -- (8,-5.8)	node[pos=0,below] {9};
	\draw (9,-6.2) -- (9,-5.8)  node[pos=0,below] {10};
	\draw (10,-6.2) -- (10,-5.8)node[pos=0,below] {11};
	\draw (11,-6.2) -- (11,-5.8)node[pos=0,below] {12};
	\draw (12,-6.2) -- (12,-5.8)node[pos=0,below] {13};
	\draw (13,-6.2) -- (13,-5.8)node[pos=0,below] {14};
	\draw (14,-6.2) -- (14,-5.8)node[pos=0,below] {15};
	\node[anchor=south, red] at (15,-6)  {\texttt{out}};
	
	\draw[lw,blue,->] (7,-6) -- (7,-5) node[pos=1, black, anchor=south, xshift=-0.2cm] {7,8,1,2};
	\draw[lw,blue,->] (7.1,-6) -- (7.1,-5);
    \draw[lw,blue,->] (7.2,-6) -- (7.2,-5);
    \draw[lw,blue,->] (7.3,-6) -- (7.3,-5);

	\draw[lw,blue,->] (8,-6) -- (8,-5) node[pos=1, black, anchor=south] {9,3};
	\draw[lw,blue,->] (8.1,-6) -- (8.1,-5);
	
	\draw[lw,blue,->] (9,-6) -- (9,-5) node[pos=1, black, anchor=south] {10,4};
	\draw[lw,blue,->] (9.1,-6) -- (9.1,-5);
	
	\draw[lw,blue,->] (10,-6) -- (10,-5) node[pos=1, black, anchor=south] {11,5};
	\draw[lw,blue,->] (10.1,-6) -- (10.1,-5);
	
	\draw[lw,blue,->] (11,-6) -- (11,-5) node[pos=1, black, anchor=south] {12,6};
	\draw[lw,blue,->] (11.1,-6) -- (11.1,-5);
	
	\draw[lw,blue,->] (12,-6) -- (12,-5) node[pos=1, black, anchor=south] {13};

	\draw[lw,blue,->] (13,-6) -- (13,-5) node[pos=1, black, anchor=south] {14};

    \node at (14,-5.5) {\ldots};

    \draw[-latex, dashed] (0,0) -- (7,-4);
    \draw[-latex] (1,0) -- (7,-4);
    \draw[-latex] (2,0) -- (8,-4);

	\draw[-latex, dashed] (6,0) -- (7,-2);
    \draw[-latex] (7,0) -- (7,-2);
    \draw[-latex] (8,0) -- (8,-2);
\end{tikzpicture}%
}%
    \caption{\label{fig:toolbox:toy-example-tight} Trajectory showing that the results of Corollary~\ref{cor:toolbox:thm-simple-application} is tight for the toy example. The cumulative function of $f$, starting at Time Unit 8 in the above trajectory, is given as a dashed blue line in Fig.~\ref{fig:toolbox:toy-example-tihhtness}.}
\end{figure}
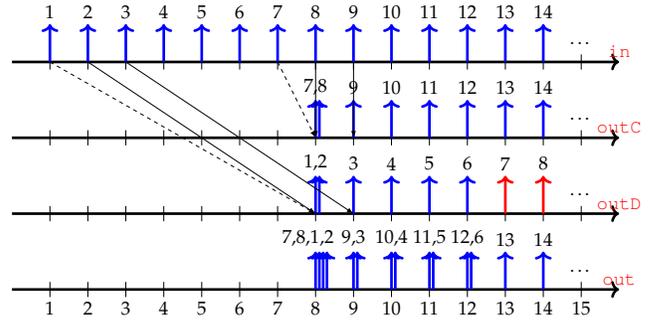
%


\vspace{\myvspacebeforesubsec}
\subsection{Reordering Introduced by the Packet Replication and Elimination Functions}\label{sec:toolbox:reordering}

In Sec.~\ref{sec:toolbox:pef-oac} we provide a characterization of the traffic at the output of a \ac{PEF} in the form of an arrival curve.
The arrival curve can then be used to compute delay and backlog bounds on subsequent vertices, from which we can obtain the \ac{ETE} delay bounds.
However, as we can observe in the toy example (Figures~\ref{fig:prob-statement:double-rate} and \ref{fig:toolbox:toy-example-tight}), the data units at the output of the \ac{PEF} are out-of-order compared to the input.
The mis-ordering of the flow's data units cannot be captured by arrival curves.
As described in Sec.~\ref{sec:prob-statement}, it still has an effect on the performances of time-sensitive networks \cite{mohammadpourPacketReorderingTimeSensitive2020}.

Two metrics are of interest when quantifying mis-ordering in time-sensitive networks: the \ac{RTO} and the \ac{RBO}\cite{mohammadpourPacketReorderingTimeSensitive2020,rfc4737}.
In this paper we focus on the mis-ordering as a consequence of the redundancy.
Thus we are only interested in defining reordering metrics after the \ac{PEF}, relative to a reference order defined before the \ac{PRF}.

For a flow $f$ and two vertices $n$ [resp.,  $o$] containing the observation points $v$ [resp., $w$] such that $f$ is packetized at $v$ and $w$, $n$ is not an \ac{EP}-vertex of $\mathcal{G}(f)$ and $o$ is a diamond ancestor of $n$ in $\mathcal{G}(f)$, we denote by $\lambda_v(f,w)$ the \ac{RTO} of the data units of flow $f$ at the observation point $v$, with respect to their order at $w$, as defined in \cite{mohammadpourPacketReorderingTimeSensitive2020,rfc4737}.
With the restrictions on $v$ and $w$, $\lambda_v(f,w)$ is well defined from \cite{mohammadpourPacketReorderingTimeSensitive2020,rfc4737} because each data unit of $f$ is observed at most once at $w$ and $v$, thus the arrival instant of each data unit at $w$ and $v$ is well defined.
Similarly, with $v$ and $w$ meeting the same conditions, we denote by $\pi_v(f,w)$ the \ac{RBO}, as defined in \cite{mohammadpourPacketReorderingTimeSensitive2020,rfc4737} of the data units of flow $f$ at $v$ with respect to their order at the reference $w$.

If $\texttt{POF}_n(\{f\},o)$ is a \acl{POF} that forces the data units of $f$ to be in the same order as their order at the output of $o$, then with $v$ being the input of $\texttt{POF}_n$, $\lambda_v(f,o^*)$ gives the minimum value for the timeout parameter $T$ of \ac{POF} algorithm and $\pi_v(f,o^*)$ gives its required buffer size~\cite[\S IV.B]{mohammadpourPacketReorderingTimeSensitive2020}.
In general, if a destination $d$ does not support any mis-ordering, then a function $\texttt{POF}_d(\{f\},\text{source}(f))$ that uses the reference $o=\text{source}(f)$ is placed just before delivery to the application.
The \acl{ETE} \ac{RTO} and \ac{RBO} $\lambda_d(f,\text{source}(f)),\pi_d(f,\text{source}(f))$ must be obtained to correctly configure this \ac{POF}.
\begin{proposition}[\label{prop:toolbox:rto-vs-rbo}$\ac{RBO} \le \alpha(\acs{RTO})$]
    For a flow $f$, and two observations points $v,w$ meeting the above conditions, if $\lambda_v(f,w)<+\infty$, then
    \begin{equation}\label{eq:toolbox:rto-vs-rbo}
        \pi_v(f,w) \le \alpha_{v,f}(\lambda_v(f,w))
    \end{equation}
\end{proposition}
The result is directly obtained by writing the definitions of the two notions. 
Its formal proof is in Appendix~\ref{proof:appendix:toolbox:rto-vs-rbo}.
Proposition~\ref{prop:toolbox:rto-vs-rbo} combined with our results from Sec.~\ref{sec:toolbox:pef-oac} show that we can focus on the effect of the \ac{PEF} on the \ac{RTO} to also obtain a bound on the \ac{RBO}.

\begin{theorem}[\ac{RTO} at the output of a \ac{PEF}]\label{thm:toolbox:reordering}
    Consider a flow $f$, a vertex $n$ containing a \acl{PEF} $\texttt{PEF}_n(f)$ and a diamond ancestor $a$ of $n$ in $\mathcal{G}(f)$.
    Denote by $d_f^{a\rightarrow n}$ [resp., $D_f^{a \rightarrow n}$] a lower [resp., upper] delay bound for $f$ between the output of $a$ and the input of $\texttt{PEF}_n(f)$, along any possible path in the graph $\mathcal{G}(f)$.
    Then $\lambda_{\texttt{PEF}_n(f)^*}(f,a)$, the \acf{RTO} of $f$ at the output of the \ac{PEF}, with respect to $a$, verifies
    \begin{equation}
        \lambda_{\texttt{PEF}_n(f)*}(f,a^*) \le \left| D_f^{a\rightarrow n} - d_f^{a\rightarrow n}  -\alpha_{a*}^{\downarrow}(2L^\text{min}) \right|^+
    \end{equation}
    where $|x|^+\triangleq\max(0,x)$, $\alpha_{f,a^*}$ is an arrival curve for $f$ at the output of the input port within $a$ and $\alpha_{f,a^*}^{\downarrow}$ is its lower pseudo-inverse\footnote{For $f:\mathbb{R}\rightarrow\mathbb{R}\cup\{-\infty, +\infty\}$ a wide-sense increasing function, its lower pseudo inverse $f^{\downarrow}$ is defined by $f^{\downarrow}(y) = \inf\{x|f(x)\ge y\}$.} defined in \cite[\S 10]{liebeherrDualityMaxPlusMinPlus2017b}.
\end{theorem}
Theorem~\ref{thm:toolbox:reordering} is a direct application of \cite[Thm. 5]{mohammadpourPacketReorderingTimeSensitive2020} for the system located between the diamond ancestor and the output of the \ac{PEF}, see Appendix~\ref{proof:thm:toolbox:reordering}.

\textbf{Application to the Toy Example:}
The lower-pseudo inverse of $\alpha_{f,B^*}=\gamma_{r_0,b_0}$ in the toy example of Fig.~\ref{fig:prob-statement:toy-example} is $\alpha_{f,B^*}^{\downarrow}: x \mapsto |x - b_0|^+ / r_0$.
In the toy example, all packets have the same size of one d.u., so $\alpha_{f,B^*}^{\downarrow}(2L_{\min}) = 1$ t.u.
Applying Theorem~\ref{thm:toolbox:reordering} proves that the \ac{RTO} at the output of the \ac{PEF} within $F$ in Fig.~\ref{fig:prob-statement:toy-example} is bounded by 6 t.u.
In the trajectory of Fig.~\ref{fig:toolbox:toy-example-tight}, we observe that d.u. 6 is late by 4 t.u. with respect to d.u. 7.
The worst-case \ac{RTO} is hence comprised between 4 and 6 t.u.
%
%
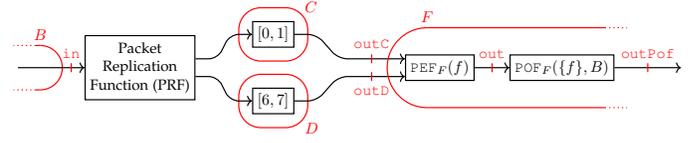
\begin{figure}\centering
    \resizebox*{\linewidth}{!}{\begin{tikzpicture}
	\node[draw] at (0,0) (s) {\makecell{Packet\\Replication\\Function (PRF)}};
	\node[draw] at (3,0.75) (a) {$[0,1]$};
	\node[draw] at (3,-0.75) (b) {$[6,7]$};
	\node[draw] at (6.75,0) (c) {$\texttt{PEF}_F(f)$};
    \node[draw] at (9.5,0) (pof) {$\texttt{POF}_F(\{f\},B)$};

	\node[draw, red, rounded corners=0.5cm, fit={(a)}, inner sep=0.3cm] (cNode) {};
	\node[draw, red, rounded corners=0.5cm, fit={(b)}, inner sep=0.3cm] (dNode) {};
	\node[anchor=west, red] at ([xshift=-0.2cm] dNode.south east) {$D$};
	\node[anchor=west, red] at ([xshift=-0.2cm] cNode.north east) {$C$};

	\draw[->, rounded corners=0.2cm] ([yshift=0.2cm]s.east) -- ([xshift=0.5cm,yshift=0.2cm] s.east) -- ([xshift=-0.5cm] a.west) -- (a.west);
	\draw[->, rounded corners=0.2cm] ([yshift=-0.2cm]s.east) -- ([xshift=0.5cm,yshift=-0.2cm] s.east) -- ([xshift=-0.5cm] b.west) -- (b.west);
	
	\draw[->, rounded corners=0.2cm] (a.east) -- ([xshift=0.5cm] a.east) -- ([xshift=-1.5cm,yshift=0.2cm] c.west) -- ([yshift=0.2cm]c.west) node[pos=0.5,anchor=center] (tA) {};
	\draw[->, rounded corners=0.2cm] (b.east) -- ([xshift=0.5cm] b.east) -- ([xshift=-1.5cm,yshift=-0.2cm] c.west) -- ([yshift=-0.2cm]c.west) node[pos=0.5,anchor=center] (tB) {};
	\draw[-,red] ([yshift=-0.1cm] tA.center) -- ([yshift=0.1cm] tA.center) node[pos=1,above] {\texttt{outC}};
	\draw[-,red] ([yshift=-0.1cm] tB.center) -- ([yshift=0.1cm] tB.center) node[pos=0,below] {\texttt{outD}};

	\draw[red] ([xshift=0.5cm] c.west) ++(90:-0.9) arc (90:-90:-0.9) node[pos=0, anchor=center] (mt) {} node[pos=1, anchor=center] (mtt) {};
	\draw[red] (mt.center) -- ++(4cm,0);
	\draw[red, dotted] (mt.center) ++(4cm,0) -- ++(0.5cm,0);
	\draw[red] (mtt.center) -- ++(4cm,0);
	\draw[red, dotted] (mtt.center) ++(4cm,0) -- ++(0.5cm,0);
	\node[red, anchor=south] at (mtt.center) {$F$};

	\draw[red] ([xshift=-1cm] s.west) ++(90:0.5) arc (90:-90:0.5) node[pos=0, anchor=center] (mmt) {} node[pos=1, anchor=center] (mmtt) {};
	\draw[red] (mmt.center) -- ++(-0.1cm,0);
	\draw[red, dotted] (mmt.center) ++(-0.1cm,0) -- ++(-0.5cm,0);
	\draw[red] (mmtt.center) -- ++(-0.1cm,0);
	\draw[red, dotted] (mmtt.center) ++(-0.1cm,0) -- ++(-0.5cm,0);
	\node[red, anchor=south] at (mmt.center) {$B$};


	\draw[->] ([xshift=-1.5cm] s.west) -- (s.west) node[pos=0.8,anchor=center] (tIn) {};
	\draw[-,red] ([yshift=0.1cm] tIn.center) -- ([yshift=-0.1cm] tIn.center) node[pos=0,above] {\texttt{in}};
	\draw[->] (c.east) -- (pof.west) node[pos=0.5,anchor=center] (tOut) {};
	\draw[-,red] ([yshift=-0.1cm] tOut.center) -- ([yshift=0.1cm] tOut.center) node[pos=1,above] {\texttt{out}};
    \draw[->] (pof.east) -- ([xshift=1.5cm] pof.east) node[pos=0.5,anchor=center] (tOutPof) {};
	\draw[-,red] ([yshift=-0.1cm] tOutPof.center) -- ([yshift=0.1cm] tOutPof.center) node[pos=1,above] {\texttt{outPof}};
\end{tikzpicture}%
}%
    \caption{\label{fig:toolbox:toy-example-with-pof} Toy example of Fig.~\ref{fig:prob-statement:toy-example}, with a \acf{POF} placed after the \ac{PEF} to correct the mis-ordering caused by the redundancy.}
\end{figure}
\begin{figure}\centering
    \resizebox*{\linewidth}{!}{\begin{tikzpicture}[xscale=2]
	\tikzstyle{lw} = [line width=2pt]
	\tikzstyle{every node}=[font=\Large]

	\draw [->,lw] (6,-6) -- (15,-6);
	\draw (6,-6.2) -- (6,-5.8)	node[pos=0,below] {7};
	\draw (7,-6.2) -- (7,-5.8)	node[pos=0,below] {8};
	\draw (8,-6.2) -- (8,-5.8)	node[pos=0,below] {9};
	\draw (9,-6.2) -- (9,-5.8)  node[pos=0,below] {10};
	\draw (10,-6.2) -- (10,-5.8)node[pos=0,below] {11};
	\draw (11,-6.2) -- (11,-5.8)node[pos=0,below] {12};
	\draw (12,-6.2) -- (12,-5.8)node[pos=0,below] {13};
	\draw (13,-6.2) -- (13,-5.8)node[pos=0,below] {14};
	\draw (14,-6.2) -- (14,-5.8)node[pos=0,below] {15};
	\node[anchor=south, red] at (15,-6)  {\texttt{outPof}};
	
    \draw[lw,blue,->] (7,-6) -- (7,-5) node[pos=1, black, anchor=south] {1};
    \draw[lw,blue,->] (7.3,-6) -- (7.3,-5) node[pos=1, black, anchor=south] {2};

	\draw[lw,blue,->] (8,-6) -- (8,-5) node[pos=1, black, anchor=south] {3};
	
	\draw[lw,blue,->] (9,-6) -- (9,-5) node[pos=1, black, anchor=south] {4};
	
	\draw[lw,blue,->] (10,-6) -- (10,-5) node[pos=1, black, anchor=south] {5};
	
	\draw[lw,blue,->] (11,-6) -- (11,-5);
	\draw[lw,blue,->] (11.1,-6) -- (11.1,-5);
    \draw[lw,blue,->] (11.2,-6) -- (11.2,-5);
    \draw[lw,blue,->] (11.3,-6) -- (11.3,-5) node[pos=1, black, anchor=south] {\makecell{6,7,8,9,\\10,11,12}};
    \draw[lw,blue,->] (11.4,-6) -- (11.4,-5);
    \draw[lw,blue,->] (11.5,-6) -- (11.5,-5);
    \draw[lw,blue,->] (11.6,-6) -- (11.6,-5);

	\draw[lw,blue,->] (12,-6) -- (12,-5) node[pos=1, black, anchor=south] {13};

	\draw[lw,blue,->] (13,-6) -- (13,-5) node[pos=1, black, anchor=south] {14};

    \node at (14,-5.5) {\ldots};
\end{tikzpicture}%
}%
    \caption{\label{fig:toolbox:toy-example-pof-output} Trajectory of the packets at the output of the \ac{POF} of Fig.~\ref{fig:toolbox:toy-example-with-pof} when the \ac{POF} processes the packets from the trajectory of Fig.~\ref{fig:toolbox:toy-example-tight}.}
\end{figure}
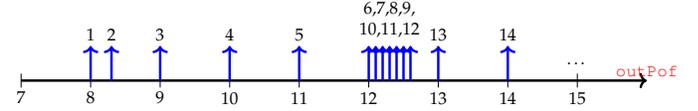

Assume now that we place, after the \ac{PEF}, the function $\texttt{POF}_F(\{f\},B)$, a \acl{POF} enforcing for $f$ the order defined at $B$ (Fig.~\ref{fig:toolbox:toy-example-with-pof}).
With Theorem~\ref{thm:toolbox:reordering}, we know that its timeout $T$ should be of at least 6 time units and it requires a buffer of at least 14 d.u. (Proposition~\ref{prop:toolbox:rto-vs-rbo} and Fig.~\ref{fig:toolbox:toy-example-tihhtness}).
In the trajectory of Fig.~\ref{fig:toolbox:toy-example-tight}, the \ac{POF} receives the traffic from the Line ``\texttt{out}'' and forces the data units to be in the same order as on the Line ``\texttt{in}''.
The resulting output is given in Fig.~\ref{fig:toolbox:toy-example-pof-output}.
We observe two main characteristics of the \ac{POF}; they have been widely studied in \cite{mohammadpourPacketReorderingTimeSensitive2020}.

First, we note that all data units continue to have a delay upper-bounded by 7 t.u.
Indeed, none of the data units has been lost for the \ac{POF} thus the \ac{POF} does not increase the \acf{ETE} latency of the data units \cite[Thm.~4]{mohammadpourPacketReorderingTimeSensitive2020}.

Second, we observe that the traffic at the output of the \ac{POF} (Fig.~\ref{fig:toolbox:toy-example-pof-output}) is much more bursty than the traffic at the output of the \ac{PEF} (Line ``\texttt{out}'', Fig.~\ref{fig:toolbox:toy-example-tight}).
We observe that seven d.u. exit the \ac{POF} at the same time (t.u. 12).
The traffic is hence no more constrained by $\alpha_{PEF^*} = \gamma_{2r_0,4b_0} \otimes \gamma_{r_0,8b_0}$, the arrival curve of the flow at the output of the \ac{PEF}, obtained by applying Theorem~\ref{thm:toolbox:pef-oac} (Sec.~\ref{sec:toolbox:pef-oac}).
We apply Corollary~1 of \cite{mohammadpourPacketReorderingTimeSensitive2020}: 
If none of the data units is lost for the \ac{POF} (at least one replicate of each data unit reaches the \ac{PEF}), then $\alpha_{POF^*} = \gamma_{r_0,8b_0}$ is an arrival curve of $f$ at the output of the \ac{POF}.
The trajectory in Fig.~\ref{fig:toolbox:toy-example-pof-output} is indeed $\gamma_{r_0,8b_0}$-constrained.
If both replicates of a data unit can be lost, then $\gamma_{r_0,8b_0+Tr_0}$ is an arrival curve for $f$ at the output of the \ac{POF}, with $T$ being the timeout parameter of the \ac{POF}.

Placing a \ac{POF} after a \ac{PEF} hence comes with benefits and drawbacks, as summarized on the first line of Table~\ref{tab:benefits-drawbacks}. 

\vspace{\myvspacebeforesec}
\section{Analysis of the Interactions between \ac{PREF} and Traffic Regulators}\label{sec:regulators}
Sec.~\ref{sec:toolbox:reordering} shows that a \acf{POF} can be used after a \ac{PEF} to remove the mis-ordering caused by the redundancy.
Similarly, \aclp{REG} can be used after a \ac{PEF} to remove the burstiness increase caused by the redundancy, especially if the downstream systems cannot support the worst-case traffic of the \ac{PEF} output (Theorem~\ref{thm:toolbox:pef-oac}).
%

Traffic regulators come in two flavors: \acfp{PFR} and \acfp{IR}.
Both are configured with per-flow contracts, $\{\sigma_{f,n}\}_{f}$, and force each flow $f$ to be $\sigma_{f,n}$-compliant, delaying the packets if required.

Hence, when the shaping curve $\sigma_{f,n}$ for each flow $f$ equals the arrival curve that the flow had before the redundant section, then the regulators remove any burstiness increase caused by the redundancy, thus making the redundancy transparent to the downstream nodes.
However, regulators are themselves queuing systems and their effect on the worst-case \ac{ETE} delay should be accounted for.

In this section, we first analyze the interactions between \ac{PEF} and a \acl{REG} placed directly after.
We evaluate how these interactions affect the \ac{ETE} delay guarantees of the flows, and we show that the conclusions highly depend on the nature of the regulator (either \ac{PFR} or \ac{IR}).
We last analyze the effect of a \ac{POF} placed after the \ac{PEF} and before the \ac{REG}.
%
\vspace{\myvspacebeforesubsec}
\subsection{Delay Bound Analysis of \acs{PREF} Combined with Per-Flow Regulators}\label{sec:regulators:pfr}

\begin{figure}\centering
    \resizebox*{\linewidth}{!}{\begin{tikzpicture}
    \def\mfact{1.6}
    \draw[dashed] (0.25,0) ++(90:\mfact*0.2) arc (90:-90:\mfact*0.2);
    \node at (0.25,0) {$a$};
    \draw[dashed] (0.25,\mfact*0.2) -- ++(-0.2cm,0);
    \draw[dashed] (0.25,-\mfact*0.2) -- ++(-0.2cm,0);

    \node at (\mfact*1.25,0) (entry) {};
    \node[circle, draw, dashed] at (\mfact*1.5,\mfact*0.5) (n1){};
    \node[circle, draw, dashed] at (\mfact*2,0) (n2){};
    \node[circle, draw, dashed] at (\mfact*1.75,\mfact*-0.5)(n3) {};  
    \node[circle, draw, dashed] at (\mfact*2.5,\mfact*0.25)(n4) {};
    \node[circle, draw, dashed] at (\mfact*3,\mfact*-0.25) (n5){};
    \node at (\mfact*3.25,0) (exit) {};
    \node[draw, dotted, anchor=north, rotate=90] at (\mfact*4.5,0) (pef) {$\texttt{PEF}_n(f)$};
    \node[draw, anchor=north, rotate=90] at ([xshift=0.5cm] pef.south) (pfr) {\resizebox*{1.6cm}{!}{$\texttt{REG}_n(\{f\},a)$}};
    \node[draw, fit={(entry) (n1) (n2) (n3) (n4) (n5) (exit) (pef)}] (sys) {};
    \draw[dashed] (\mfact*4.6,0) ++(90:\mfact*-0.6) arc (90:-90:\mfact*-0.6);
    \draw[dashed] (\mfact*4.6,\mfact*0.6) -- ++(1.7cm,0);
    \draw[dashed] (\mfact*4.6,\mfact*0.6) ++(1.7cm,0) -- ++ (0.5cm,0);
    \draw[dashed] (\mfact*4.6,\mfact*-0.6) -- ++(1.7cm,0);
    \draw[dashed] (\mfact*4.6,\mfact*-0.6) ++(1.7cm,0) -- ++ (0.5cm,0);

    \draw[->] (0.25+\mfact*0.2,0) -- (sys.west);
    \draw[->, dotted] (sys.west) -- (n1);
    \draw[->, dotted] (sys.west) -- (n3);
    \draw[->, dotted] (n3) -- (n5);
    \draw[->, dotted] (n3) -- (n2);
    \draw[->, dotted] (n1) -- (n2);
    \draw[->, dotted] (n2) -- (n4);
    \draw[->, dotted] (n4) -- (\mfact*4.6-\mfact*0.6,0);
    \draw[->, dotted] (n5) -- (\mfact*4.6-\mfact*0.6,0);

    \draw[->, dotted] (\mfact*4.6-\mfact*0.6,0) -- (pef.north);
    \draw[->, dotted] (pef.south) -- (sys.east);
    \draw[->] (sys.east) -- (pfr.north);
    \draw[->, dotted] (pfr.south) -- ++(1cm,0);

    \node[anchor=south] at (sys.north) (sysTxt) {$\mathcal{S}$};
    \node[draw, fit={(entry) (n1) (n2) (n3) (n4) (n5) (exit) (pef) (pfr) (sys) (sysTxt)}] (sysPrime) {};
    \node[anchor=south] at (sysPrime.north) (sysPrimeTxt) {$\mathcal{S}'$};
    \node[anchor=north west] at (\mfact*4.5-\mfact*0.6, \mfact*0.6) {$n$};

    \draw[red, dashdotted] (\mfact*0.75,\mfact*1) -- (\mfact*0.75,\mfact*-1.5);
    \draw[red, dashdotted] (\mfact*5.0,\mfact*1) -- (\mfact*5.0,\mfact*-1.5);

    \draw[latex-latex, red] (\mfact*0.75,\mfact*-1.1) -- (\mfact*5.0,\mfact*-1.1) node[pos=0, left] {$[d,D]$};
    \draw[red, dashdotted] (\mfact*5.75,\mfact*1) -- (\mfact*5.75,\mfact*-1.5);
    \draw[latex-latex, red] (\mfact*0.75,\mfact*-1.25) -- (\mfact*5.75,\mfact*-1.25) node[pos=1, right] {$[d',D']$};

    \node[anchor=south west] at ([xshift=0.1cm, yshift=0.75cm] pfr.south east) (conf) {$\sigma_{n,f} \triangleq \alpha_{f,a^*}$};
    \draw[-latex] (conf.west) -| (pfr.east);

\end{tikzpicture}%
}%
    \caption{\label{fig:regulators:pfr-overview} Notations for the analysis of the interactions between \ac{PEF} and a \ac{PFR} for a flow $f$. Vertices of $\mathcal{G}(f)$ are shown in dashed circles/ovals and edges are shown with dotted arrows.}
\end{figure}
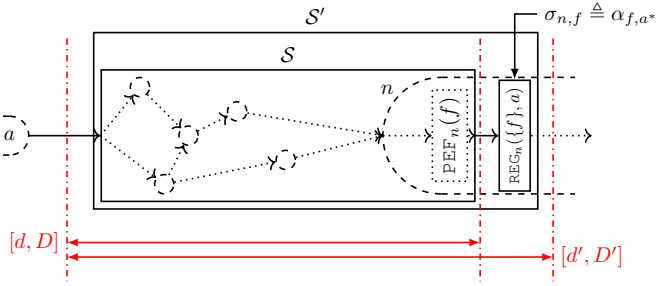
Consider a vertex $n$ containing a function $\texttt{PEF}_n(f)$ and consider a diamond ancestor $a$ of $n$ in $\mathcal{G}(f)$  (Fig.~\ref{fig:regulators:pfr-overview}).
Between $a$ and $n$, the flow follows $\mathcal{G}(f)$, with potentially multiple vertices and multiple paths.
Consider the system $\mathcal{S}$ between the output of $a$ and the output of $\texttt{PEF}_n(f)$ (solid box in Fig.~\ref{fig:regulators:pfr-overview}).
Due to all the possible paths with different lengths, $\mathcal{S}$ is neither \ac{FIFO} nor lossless in the general case.
We denote by $d$ [resp., $D$] a delay lower-bound [resp., upper-bound] for each forwarded d.u. through $\mathcal{S}$.
The delays $d$ and $D$ are well-defined because the data units are seen at most once at the output of the \ac{PEF}.
Note that the \ac{PEF} has no delay, hence $d$ [resp., $D$] verifies $d=d_f^{a\rightarrow n}$ [resp., $D=D_f^{a\rightarrow n}$], \emph{i.e.}, a delay bound along any possible paths $a\rightarrow n$ is a delay bound through $\mathcal{S}$.

After $\mathcal{S}$, and still within vertex $n$ (dashed oval on the right of Fig.~\ref{fig:regulators:pfr-overview}), we place a \ac{PFR}: $\texttt{REG}_n(\{f\},a)$ with shaping curve $\sigma_{n,f}\triangleq \alpha_{f,a^*}$.
We now consider the system $\mathcal{S}'$ made of $\mathcal{S}$ followed by the \ac{PFR}, and we are interested in the delay bounds $[d',D']$ for the non-lost data units through $\mathcal{S}'$.
%
If $\mathcal{S}$ was \ac{FIFO}, we could use the essential \emph{shaping-for-free} property of regulators~\cite{leboudecTheoryTrafficRegulators2018,leboudecNetworkCalculusTheory2001}: As $f$ is $\sigma_{f,n}$-constrained at the input of $\mathcal{S}$, the regulator would not have increased the \ac{ETE} delay of the data units;  we write this as $D'=D$.
But, as $\mathcal{S}$ is not \ac{FIFO}, the \ac{PFR} does not guarantee the \emph{shaping-for-free} property, as we show on the toy example.
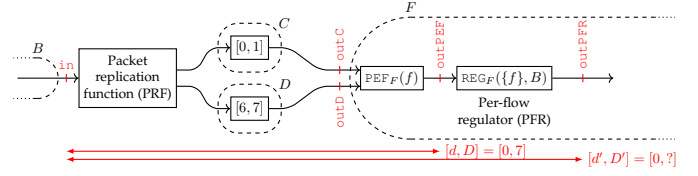
\begin{figure}\centering
    \resizebox*{\linewidth}{!}{\begin{tikzpicture}
	\node[draw] at (0,0) (s) {\makecell{Packet\\replication\\function (\acs{PRF})}};
	\node[draw] at (3,0.75) (a) {$[0,1]$};
	\node[draw] at (3,-0.75) (b) {$[6,7]$};
	\node[draw] at (6.5,0) (c) {$\texttt{PEF}_F(f)$};
    \node[draw, anchor=west, label=below:{\makecell{Per-flow\\regulator (PFR)}}] at ([xshift=0.8cm] c.east) (reg) {$\texttt{REG}_F(\{f\},B)$};
	\node[draw, dashed, rounded corners=0.5cm, fit={(a)}, inner sep=0.3cm] (cNode) {};
	\node[draw, dashed, rounded corners=0.5cm, fit={(b)}, inner sep=0.3cm] (dNode) {};
	\node[anchor=west] at ([xshift=-0.2cm] dNode.north east) {$D$};
	\node[anchor=west] at ([xshift=-0.2cm] cNode.north east) {$C$};

	\draw[->, rounded corners=0.2cm] ([yshift=0.2cm]s.east) -- ([xshift=0.5cm,yshift=0.2cm] s.east) -- ([xshift=-0.5cm] a.west) -- (a.west);
	\draw[->, rounded corners=0.2cm] ([yshift=-0.2cm]s.east) -- ([xshift=0.5cm,yshift=-0.2cm] s.east) -- ([xshift=-0.5cm] b.west) -- (b.west);
	
	\draw[->, rounded corners=0.2cm] (a.east) -- ([xshift=0.5cm] a.east) -- ([xshift=-1cm,yshift=0.2cm] c.west) -- ([yshift=0.2cm]c.west) node[pos=0.5,anchor=center] (tA) {};
	\draw[->, rounded corners=0.2cm] (b.east) -- ([xshift=0.5cm] b.east) -- ([xshift=-1cm,yshift=-0.2cm] c.west) -- ([yshift=-0.2cm]c.west) node[pos=0.5,anchor=center] (tB) {};
	\draw[-,red] ([yshift=-0.1cm] tA.center) -- ([yshift=0.1cm] tA.center) node[pos=1,right, rotate=90] {\texttt{outC}};
	\draw[-,red] ([yshift=-0.1cm] tB.center) -- ([yshift=0.1cm] tB.center) node[pos=0,left, rotate=90] {\texttt{outD}};

	\draw[dashed] ([xshift=1.25cm] c.west) ++(90:-1.5) arc (90:-90:-1.5) node[pos=0, anchor=center] (mt) {} node[pos=1, anchor=center] (mtt) {};
	\draw[dashed] (mt.center) -- ++(6cm,0);
	\draw[dotted] (mt.center) ++(6cm,0) -- ++(0.5cm,0);
	\draw[dashed] (mtt.center) -- ++(6cm,0);
	\draw[dotted] (mtt.center) ++(6cm,0) -- ++(0.5cm,0);
	\node[anchor=south] at (mtt.center) {$F$};

	\draw[dashed] ([xshift=-1cm] s.west) ++(90:0.5) arc (90:-90:0.5) node[pos=0, anchor=center] (mmt) {} node[pos=1, anchor=center] (mmtt) {};
	\draw[] (mmt.center) -- ++(-0.1cm,0);
	\draw[dotted] (mmt.center) ++(-0.1cm,0) -- ++(-0.5cm,0);
	\draw[] (mmtt.center) -- ++(-0.1cm,0);
	\draw[dotted] (mmtt.center) ++(-0.1cm,0) -- ++(-0.5cm,0);
	\node[anchor=south] at (mmt.center) {$B$};


	\draw[->] ([xshift=-1.5cm] s.west) -- (s.west) node[pos=0.8,anchor=center] (tIn) {};
	\draw[-,red] ([yshift=0.1cm] tIn.center) -- ([yshift=-0.1cm] tIn.center) node[pos=0,above] {\texttt{in}} node[pos=0, anchor=center] (targetIn) {};

    \draw[->] (c.east) -- (reg.west) node[pos=0.5,anchor=center] (tOut) {};
    \draw[->] (reg.east) -- ++(1.5cm,0) node[pos=0.5, anchor=center] (tRegOut) {};
	\draw[-,red] ([yshift=-0.1cm] tOut.center) -- ([yshift=0.1cm] tOut.center) node[pos=1,right, rotate=90] {\texttt{outPEF}} node[pos=0, anchor=center] (targetOutPEF) {};
    \draw[-,red] ([yshift=-0.1cm] tRegOut.center) -- ([yshift=0.1cm] tRegOut.center) node[pos=1, right, rotate=90] {\texttt{outPFR}} node[pos=0,anchor=center] (targetOutPFR) {};

    \draw[latex-latex, red] let \p1 = (targetIn.center), \p2 = (targetOutPEF.center) in (\x1,-1.8) -- (\x2,-1.8) node[pos=1, right] {$[d,D] = [0,7]$};
    \draw[latex-latex, red] let \p1 = (targetIn.center), \p2 = (targetOutPFR.center) in (\x1,-2) -- (\x2,-2) node[pos=1,right] {$[d',D'] = [0,?]$};
\end{tikzpicture}%
}%
    \caption{\label{fig:regulators:toy-example-notations} Toy example of Fig.~\ref{fig:prob-statement:toy-example} with a \acf{PFR} placed after the \ac{PEF} to remove the burstiness increase caused by the redundancy.}
\end{figure}
\begin{figure}\centering
    \resizebox*{\linewidth}{!}{

\begin{tikzpicture}[xscale=1.4]
\tikzstyle{pkn} = [pos=1, anchor=south]
\tikzstyle{delay} = [,pos=0.5,left,black]
\tikzstyle{ab} = [pos=0, anchor=north, black]
\tikzstyle{pk} = [->, blue, line width=1.25pt]
\tikzstyle{nb} = [pos=0.5, above, yshift=0.1cm]
\tikzstyle{dld} = [xshift=-0cm]
\tikzstyle{drd} = [xshift=0cm]
\tikzstyle{six} = [[->, red, line width=1.75pt]
\def\xdev{1cm}
\def\yline{0cm}
\def\ygraduate{0.2cm}
\def\ph{1cm}
\def\mxmax{22}
\def\mnpacket{14}
\def\mxminmin{5}
\def\mxmin{0}



\def\yline{-2cm}
\draw[-latex] (\mxminmin*\xdev,\yline) -- ([xshift=\xdev]\mxmax*\xdev, \yline) node[ab,pos=1] {time} node[anchor=south east, pos=1, red] {\texttt{outC}} ;
\foreach \xhere in {5,6,...,\mxmax}	\draw (\xhere*\xdev,\yline-\ygraduate) -- (\xhere*\xdev,\yline+\ygraduate);
\foreach \i in {1,2,...,6}{
	\pgfmathsetmacro\xhere{(\i + 7)}
	\draw[pk]  (\xhere*\xdev,\yline) -- (\xhere*\xdev,\ph+\yline) node[pos=0, anchor=center] (cpied\i) {} node[pkn] (c\i) {\i} node[delay] (delayc\i) {7};
}
\foreach \i in {7}{
	\pgfmathsetmacro\xhere{(\i + 7)}
	\draw[pk, dotted]  (\xhere*\xdev,\yline) -- (\xhere*\xdev,\ph+\yline)  node[pkn] {\i} node[delay] {7};
}

\node[anchor=east] at (7.5*\xdev,\yline+0.5cm) {\color{red}\texttt{in} \color{black} $\rightarrow$ \color{red}\texttt{outC} \color{black}:};

\def\yline{-4cm}
\draw[-latex] (\mxminmin*\xdev,\yline) -- ([xshift=\xdev]\mxmax*\xdev, \yline) node[ab,pos=1] {time} node[anchor=south east, pos=1, red] {\texttt{outD}} ;
\foreach \xhere in {5,6,...,\mxmax}	\draw (\xhere*\xdev,\yline-\ygraduate) -- (\xhere*\xdev,\yline+\ygraduate);

\draw[pk]  (8*\xdev,\yline) -- (8*\xdev,\ph+\yline)  node[pkn] {7} node[delay] {1};
\draw[pk]  (8.2*\xdev,\yline) -- (8.2*\xdev,\ph+\yline)  node[pkn] {8} node[delay, right] {0};

\foreach \i in {9,10,...,14}{
	\pgfmathsetmacro\xhere{(\i)}
	\draw[pk]  (\i*\xdev,\yline) -- (\i*\xdev,\ph+\yline)  node[pkn] {\i} node[delay] {0};
}

\node[anchor=east] at (7.5*\xdev,\yline+0.5cm) {\color{red}\texttt{in} \color{black} $\rightarrow$ \color{red}\texttt{outD} \color{black}:};

\def\yline{-6cm}
\draw[-latex] (\mxminmin*\xdev,\yline) -- ([xshift=\xdev]\mxmax*\xdev, \yline) node[ab,pos=1] {time} node[anchor=south east, pos=1, red] {\texttt{outPEF}} ;
\foreach \xhere in {5,6,...,\mxmax}	\draw (\xhere*\xdev,\yline-\ygraduate) -- (\xhere*\xdev,\yline+\ygraduate);

\draw[pk]  (8*\xdev,\yline) -- (8*\xdev,\ph+\yline)  node[pkn] {7};
\draw[pk]  (8.2*\xdev,\yline) -- (8.2*\xdev,\ph+\yline)  node[pkn] {1};
\draw[pk]  (8.4*\xdev,\yline) -- (8.4*\xdev,\ph+\yline)  node[pkn] {8};
\foreach \i in {9,10,...,14}{
	\pgfmathsetmacro\xhere{(\i)}
	\pgfmathtruncatemacro\delayHere{0}
	\draw[pk]  (\i*\xdev,\yline) -- (\i*\xdev,\ph+\yline)  node[pkn] {\i};
}
\foreach \i in {2,3,...,6}{
	\pgfmathsetmacro\xhere{(\i + 7 + 0.4)}
	\draw[pk]  (\xhere*\xdev,\yline) -- (\xhere*\xdev,\ph+\yline)  node[pkn] {\i};
}

\def\yline{-8cm}
\draw[-latex] (\mxminmin*\xdev,\yline) -- ([xshift=\xdev]\mxmax*\xdev, \yline) node[ab,pos=1] {time} node[anchor=south east, pos=1, red] {\texttt{outPFR}} ;
\foreach \xhere in {5,6,...,\mxmax}	\draw (\xhere*\xdev,\yline-\ygraduate) -- (\xhere*\xdev,\yline+\ygraduate) node[ab] (xpfr\xhere) {\xhere};

\edef\absoluteHere{7}
\foreach \i in {7,1,8,9,2,10,3,11,4,12,5,13,6,14}{
	\pgfmathparse{\absoluteHere + 1}
	\xdef\absoluteHere{\pgfmathresult}
	\pgfmathtruncatemacro\delayHere{\absoluteHere - \i}
	\draw[pk] (\absoluteHere*\xdev,\yline) -- (\absoluteHere*\xdev,\ph+\yline) node[pkn] (pfr\i) {\i} node[delay] (delaypfr\i) {\delayHere};
	
}
\node[anchor=east] at (7.5*\xdev,\yline+0.5cm) {\color{red}\texttt{in} \color{black} $\rightarrow$ \color{red}\texttt{outPFR} \color{black}:};

\node[draw, fit={(xpfr20) (pfr6) (delaypfr6)}, line width=2pt] {};
\node[draw, fit={(cpied6) ([yshift=0.05cm] c6.center) (delayc6)}, line width=2pt] {};
\end{tikzpicture}%
}%
    \caption{\label{fig:regulators:traj-2d-pfr} An acceptable trajectory on the toy example, which shows that the delay bound $D'$ through $\mathcal{S}'$ is at least 14~t.u. The delay of the data units from ``\texttt{in}'' to the observation points are given on the left of the packets.}
\end{figure}
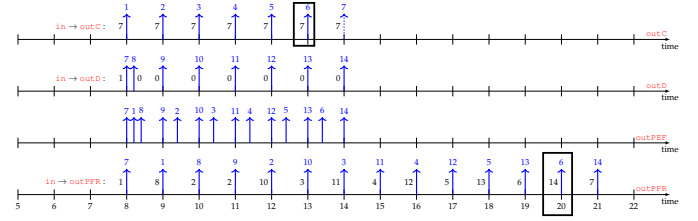

\textbf{Application to the Toy Example}:
Fig.~\ref{fig:regulators:toy-example-notations} considers the toy example from Fig.~\ref{fig:prob-statement:toy-example}, to which we add the \ac{PFR} $\texttt{REG}_n(\{f\},B)$ within vertex $F$ (dashed oval on the right in Fig.~\ref{fig:regulators:toy-example-notations}), just after the function $\texttt{PEF}_F(f)$.
With the above notations, system $\mathcal{S}$ is between the observation points ``\texttt{in}'' and ``\texttt{outPEF}'', with the delay bounds $[d,D]=[0,7]$~t.u.
System $\mathcal{S}'$ is between the observation points ``\texttt{in}'' and ``\texttt{outPFR}'', and we seek to obtain a delay-bound $D'$ for $\mathcal{S}'$.

Fig.~\ref{fig:regulators:traj-2d-pfr} presents an acceptable trajectory at the different observation points using the same input ``\texttt{in}'' as in Fig.~\ref{fig:prob-statement:double-rate}.
The path through vertex $C$ forwards all packets with a constant delay of 7~t.u., whereas the path through vertex $D$ drops Packets $1$ to $6$, then forwards Packet $7$ with a delay of 1~t.u. (its worst-case delay) and finally forwards the following packets with a delay of 0~t.u. (its best-case delay).
The line ``\texttt{outPEF}'' gives the resulting trajectory at the output of the \ac{PEF} that removes any duplicates.

Based on its input (``\texttt{outPEF}'')  and on its shaping curve ($\sigma_{f,F} = \alpha_{f,B^*} = \gamma_{r_0,b_0}$), the \ac{PFR} outputs the packets as shown on the Line ``\texttt{outPFR}''.
Recall that the \ac{PFR} $\texttt{REG}_n(\{f\},a)$ is \emph{itself} a \ac{FIFO} system (model in Sec.~\ref{sec:system-model:function-model}).

We observe that the d.u. $6$ suffers through $\mathcal{S}'$ a total delay of 14~t.u.; this is twice the delay upper-bound $D$ through $\mathcal{S}$ alone.
We note that this high delay for d.u. 6 can be explained by the time needed by the \ac{PFR} to process d.u. $1$ to $5$ and $7$ to $13$ that arrived before d.u. 6 and to pace them as required by the shaping curve.
This is done even though d.u.s $7$ to $13$ are out of order (``too early'') with respect to d.u. 6.
At ``\texttt{outPFR}'', the packet containing d.u. 6 is late with respect to d.u. 7 by 12~t.u.
Hence, the \acf{RTO} of the flow through $\mathcal{S}'$ (\emph{i.e.}, at the output of $\mathcal{S}'$, using the input of $\mathcal{S}'$ as reference) is at least 12~t.u., while it was bounded by only 6~t.u. through $\mathcal{S}$ alone (Sec.~\ref{sec:toolbox:reordering}).

We observe that the output of the \ac{PEF} is bursty and out of order, and the \ac{PFR} placed afterwards paces the packets to remove the burstiness.
But, by doing so, the \ac{PFR} worsens the mis-ordering of the packets (12 instead of 6) and increases the delay of the late packets (Packet 6), thus increasing the worst-case \ac{ETE} delay (at least 14~t.u.).
As such, the regulator comes with a \emph{delay penalty}.
With \acp{PFR} configured with leaky-bucket shaping curves, we can upper-bound this delay penalty for any networks.
\begin{theorem}[\label{thm:regulators:pfr-2d} Bound on the delay penalty of a \ac{PFR} placed after a \ac{PEF}]
    Assume that the \ac{PFR} $\texttt{REG}_n(\{f\},a)$ is configured with a leaky-bucket shaping curve $\sigma_{n,f}=\gamma_{r,b}$, and that $\sigma_{n,f}$ is an arrival curve of $f$ at the input of $\mathcal{S}$.
    If $d$ [resp., $D$] is a lower [resp., an upper] bound on the delay of $f$ through the system $\mathcal{S}$ (Fig.~\ref{fig:regulators:pfr-overview}), then $d'=d$ [resp., $D'= 2D - d$] is a lower [resp., an upper] bound on the delay of $f$ through $\mathcal{S}'$.
\end{theorem}

The proof combines Theorem~\ref{thm:toolbox:pef-oac} with the service-curve characterization of a \ac{PFR} \cite[\S 1.7.3]{leboudecNetworkCalculusTheory2001} to obtain a delay bound within the \ac{PFR}, see Appendix~\ref{proof:regulators:pfr-2d}.
Combined with \cite[Thm.~7]{mohammadpourPacketReorderingTimeSensitive2020}, we directly obtain the following result.

\begin{corollary}[\label{cor:regulators:pfr-rto}Bound on the \ac{RTO} at the output of a \ac{PFR} placed after a \ac{PEF}]
    With the notations of Theorem~\ref{thm:regulators:pfr-2d}, the \ac{RTO} of $f$ at the output of $\texttt{PFR}_n(\{f\}, a)$, with reference $a$,  verifies
    \begin{equation*}
        \lambda_{n,\texttt{PFR}^*}(f,a)  \le \lambda_{n,\texttt{PEF}^*}(f,a) + D - d
    \end{equation*}
    with $\lambda_{n,\texttt{PEF}^*}(f,a)$ the \ac{RTO} of $f$ at the output of the \ac{PEF}, again with respect to the order of the data units at $a$.
\end{corollary}


\textbf{Application to the Toy Example}:
Applying Theorem~\ref{thm:regulators:pfr-2d} shows that $2D-d =14$~t.u. is an upper delay bound through $\mathcal{S}'$. As it is achieved by d.u. 6 in Fig.~\ref{fig:regulators:traj-2d-pfr}, it is also the worst-case delay. Applying Corollary~\ref{cor:regulators:pfr-rto} to the toy example gives that 13~t.u. is an upper-bound on the \ac{RTO} of the flow at the output of the \ac{PFR}, with respect to the order of the packets at $B$. Data Unit 6 in the trajectory achieves a reordering offset of 12~t.u. (with respect to d.u. 7), thus the worst-case \ac{RTO} at the output of $\mathcal{S}'$ in the toy example is between 12 and 13~t.u.

When a \ac{PFR} is used after a \ac{PEF}, the current subsection shows that the \emph{shaping-for-free} property does not hold, but Theorem~\ref{thm:regulators:pfr-2d} captures the delay penalty by using the service-curve characterization of \acp{PFR}, combined with the arrival curve obtained from Theorem~\ref{thm:toolbox:pef-oac}.
As we do not know any service-curve characterization for an \ac{IR}, we cannot apply the Theorem~\ref{thm:regulators:pfr-2d} to \acfp{IR}.
%
\vspace{\myvspacebeforesubsec}
\subsection{Instability of the Interleaved Regulator Placed after a Set of \acp{PEF}}\label{sec:regulators:ir}

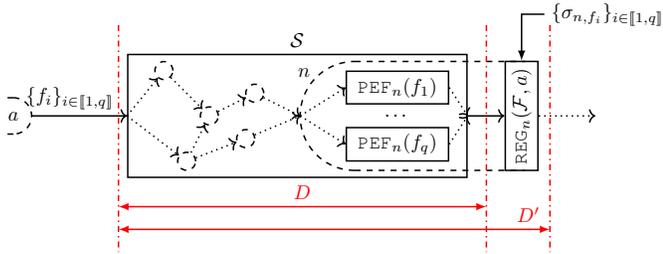
\begin{figure}\centering
    \resizebox*{\linewidth}{!}{\begin{tikzpicture}
    \def\mfact{1.6}

    \def\mma{-0.25}
    \draw[dashed] (\mma,0) ++(90:\mfact*0.2) arc (90:-90:\mfact*0.2);
    \node at (\mma,0) {$a$};
    \draw[dashed] (\mma,\mfact*0.2) -- ++(-0.2cm,0);
    \draw[dashed] (\mma,-\mfact*0.2) -- ++(-0.2cm,0);

    \node at (\mfact*1.25,0) (entry) {};
    \node[circle, draw, dashed] at (\mfact*1.5,\mfact*0.5) (n1){};
    \node[circle, draw, dashed] at (\mfact*2,0) (n2){};
    \node[circle, draw, dashed] at (\mfact*1.75,\mfact*-0.5)(n3) {};  
    \node[circle, draw, dashed] at (\mfact*2.5,\mfact*0.25)(n4) {};
    \node[circle, draw, dashed] at (\mfact*2.4,\mfact*-0.25) (n5){};
    \node at (\mfact*3.25,0) (exit) {};
    \node[draw, anchor=west, minimum width=1.8cm] at ([yshift=0.5cm)]\mfact*3.5,0) (pef) {$\texttt{PEF}_n(f_1)$};
    \path let \p1 = (pef) in node at (\x1,0) {\ldots};
    \node[draw, anchor=west, minimum width=1.8cm] at ([yshift=-0.5cm)]\mfact*3.5,0) (pef2) {$\texttt{PEF}_n(f_q)$};
    \node[draw, anchor=north, rotate=90] at (\mfact*5.25,0) (pfr) {$\texttt{REG}_n(\mathcal{F},a)$};
    \node[draw, fit={(entry) (n1) (n2) (n3) (n4) (n5) (exit) (pef) ([xshift=0.2cm]pef2.east)}] (sys) {};
    \draw[dashed] (\mfact*3.6,0) ++(90:\mfact*-0.6) arc (90:-90:\mfact*-0.6);
    \draw[dashed] (\mfact*3.6,\mfact*0.6) -- ++(3cm,0);
    \draw[dashed] (\mfact*3.6,\mfact*-0.6) -- ++(3cm,0);

    \draw[->] (\mma+\mfact*0.2,0) -- (sys.west) node[pos=0.4,above] {$\{f_i\}_{i\in\llbracket1,q\rrbracket}$};
    \draw[->, dotted] (sys.west) -- (n1);
    \draw[->, dotted] (sys.west) -- (n3);
    \draw[->, dotted] (n3) -- (n5);
    \draw[->, dotted] (n3) -- (n2);
    \draw[->, dotted] (n1) -- (n2);
    \draw[->, dotted] (n2) -- (n4);
    \draw[->, dotted] (n4) -- (\mfact*3.6-\mfact*0.6,0);
    \draw[->, dotted] (n5) -- (\mfact*3.6-\mfact*0.6,0);

    \draw[->, dotted] (\mfact*3.6-\mfact*0.6,0) -- (pef.west);
    \draw[->, dotted] (\mfact*3.6-\mfact*0.6,0) -- (pef2.west);
    \draw[->, dotted] (pef.east) -- (sys.east);
    \draw[->, dotted] (pef2.east) -- (sys.east);
    \draw[->] (sys.east) -- (pfr.north);
    \draw[->, dotted] (pfr.south) -- ++(1cm,0);

    \node[anchor=south] at (sys.north) {$\mathcal{S}$};
    \node[anchor=north west] at (\mfact*3.5-\mfact*0.6, \mfact*0.6) {$n$};

    \draw[red, dashdotted] (\mfact*1,\mfact*1) -- (\mfact*1,\mfact*-1.5);
    \draw[red, dashdotted] (\mfact*5.05,\mfact*1) -- (\mfact*5.05,\mfact*-1.5);

    \draw[latex-latex, red] (\mfact*1,\mfact*-1) -- (\mfact*5.05,\mfact*-1) node[pos=0.5, above] {$D$};
    \draw[red, dashdotted] (\mfact*5.75,\mfact*1) -- (\mfact*5.75,\mfact*-1.5);
    \draw[latex-latex, red] (\mfact*1,\mfact*-1.25) -- (\mfact*5.75,\mfact*-1.25) node[pos=0.95, above] {$D'$};

    \node[anchor=south west] at ([xshift=0.1cm, yshift=0.5cm] pfr.south east) (conf) {$\{\sigma_{n,f_i}\}_{i\in\llbracket1,q\rrbracket}$};
    \draw[-latex] (conf.west) -| (pfr.east);
\end{tikzpicture}%
}%
    \caption{\label{fig:regulators:ir-overview} Notations for the analysis of the interactions between \acp{PEF} and an \acf{IR} for an aggregate of flows $\mathcal{F}=\{f_i\}_{i\in\llbracket 1, q\rrbracket}$.}
\end{figure}
With an \acf{IR}, several flows $\mathcal{F} = \{f_i\}_{1\le i \le q}$, sharing the same redundant section $a \rightarrow n$ are processed by the same \ac{IR} $\texttt{REG}_n(\mathcal{F},a)$, after their respective elimination function $\texttt{PEF}_n(f_i)$ for $i\in \llbracket 1,m\rrbracket$ (see Fig.~\ref{fig:regulators:ir-overview}).

When the aggregate contains a unique flow, then the \ac{IR} is a \ac{PFR}.
Therefore, we do not expect the \emph{shaping-for-free} property to be valid with the \ac{IR} either.
However, as opposed to the \ac{PFR}, we exhibit an adversarial model in which any \ac{IR} placed after the \acp{PEF} and processing several flows yields unbounded latencies.
\begin{theorem}[\label{thm:regulators:ir-instable}Instability of the \ac{IR} placed after the \acp{PEF}]
    Consider a network with graph $\mathcal{G}$ and consider $q\in\mathbb{N}$ flows $f_1,\ldots,f_q$  (see Fig.~\ref{fig:regulators:ir-overview}).
    Take two vertices $a$ and $n$ such that, for each $i\in\llbracket 1,q\rrbracket$, $a$ is a diamond ancestor of $n$ in $\mathcal{G}(f_i)$. Assume that

    \begin{enumerate}[(a)]
        \item for each $i\in\llbracket 1,q\rrbracket$, vertex $n$ contains $\texttt{PEF}_n(f_i)$, a \ac{PEF} for $f_i$,
        \item vertex $n$ contains $\texttt{REG}_n(\{f_i\}_{i\in\llbracket1,q\rrbracket},a)$, an \acf{IR} for the aggregate, placed after the \acp{PEF}, with the same leaky-bucket shaping curve for each flow: $\forall i\in \llbracket 1,q \rrbracket, \sigma_{f_i,n} = \gamma_{r,b}$,
        \item all graphs $\{\mathcal{G}(f_i)\}_{i\in\llbracket1,q\rrbracket}$ share at least two different paths $P_1,P_2$ to reach $n$ from $a$.
    \end{enumerate}

    For $q \in \mathbb{N}$ and $r,b,d_1,d_2,D_1,D_2 \in \mathbb{R}^+$ with $d_1 \le D_1$, $d_2 \le D_2$ and $D_1 \le D_2$ (flipping the indexes if required), if
    \begin{enumerate}[(a)] \setcounter{enumi}{3}
        \item $b$ is greater than the minimum packet length, 
        \item $d_1,D_1,d_2, D_2$ are not all equal, and
        \item $q \ge q_{\min}$ with
    \end{enumerate}
    \begin{equation*}
        q_{\min} \triangleq \left\lfloor \frac{2r\left|d_2-D_1\right|^+}{b} + 2\right\rfloor + 1
    \end{equation*}
    then there exists an adversarial traffic arrival at $a$ for each of the $q$ flows and an adversarial implementation of the paths $\{P_j\}_{j}$ such that

    \begin{enumerate}[1/]
        \item each flow $f_i$ is $\gamma_{r,b}$-constrained at $a$, \label{item:thm:ir:source} \label{item:thm:ir:first}
        \item for each data unit $m$ belonging to one of the flows $\{f_i\}_{i\in\llbracket 1,q\rrbracket}$, if $m$ is not lost on $P_1$ [resp., on $P_2$], then its delay along $P_1$ [resp., along $P_2$] is within $[d_1,D_1]$ [resp., within $[d_2,D_2]$], \label{item:thm:ir:path-delay}
        \item flows $\{f_i\}_{i}$ have an unbounded latency within the \ac{IR}, \label{item:thm:ir:undbounded}
        \item $P_1$ and $P_2$ are both \ac{FIFO}, \label{item:thm:ir:both-fifo}
        \item the system $\mathcal{S}$ made of the sub-graph of $\mathcal{G}$ between $a$ and the output of the \acp{PEF} (Fig.~\ref{fig:regulators:ir-overview}) remains lossless and \ac{FIFO}-per-flow for each $f_i$. \label{item:thm:ir:lossless-fifo-per-flow} \label{item:thm:ir:last} 
    \end{enumerate}
\end{theorem}

The proof is in Appendix~\ref{proof:regulators:ir-instable}.
It relies on the trajectory developed for the proof of \cite[Prop.~7.3]{thomasTimeSynchronizationIssues2020}.
The main idea is to use the mis-ordering caused by \ac{PREF} and the property that the \ac{IR} looks only at the head-of-line packet to generate blocking situations with always-increasing packet delays.

Note that only Properties~\ref{item:thm:ir:first} to \ref{item:thm:ir:undbounded} of Theorem~\ref{thm:regulators:ir-instable} are required to prove the validity of the adversarial model.
However, our adversarial model provides additional Properties~\ref{item:thm:ir:both-fifo} and \ref{item:thm:ir:last}; they are of interest when considering the solutions for preventing the instability, as we illustrate in Sec.~\ref{sec:regulators:reorder}.
Theorem~\ref{thm:regulators:ir-instable} also provides a mean to obtain the following wider result, whose proof is in Appendix~\ref{proof:cor:regulators:ir-after-non-fifo}.



 \begin{corollary}[\label{cor:regulators:ir-after-non-fifo} Instability of the \acl{IR} after a non-\ac{FIFO} system, even if the system is \ac{FIFO}-per-flow and lossless]
     For any $D_{\max}>0$, $r>0$, $b$ greater than the minimum packet length, and for any \ac{IR} that processes 3 or more flows $\{f_i\}_i$ using the same leaky-bucket shaping curve $\gamma_{r,b}$, there exists a lossless \ac{FIFO}-per-flow system $\mathcal{S}$ and a $\gamma_{r,b}$-constrained adversarial generation of each flow at the input of $\mathcal{S}$ such that, when the \ac{IR} is placed after $\mathcal{S}$, the delay of the flows through $\mathcal{S}$ is upper-bounded by $D_{\max}$ but the delay of the flows through the \ac{IR} is not bounded.
 \end{corollary}

\vspace{\myvspacebeforesubsec}
\subsection{Effect of the Packet-Ordering Function on the Combination of a \acs{PEF} with Traffic Regulators}\label{sec:regulators:reorder}
\begin{table*}\centering
    \caption{\label{tab:benefits-drawbacks} Benefits and Drawbacks of Several Configurations, Compared to the Situation with the \ac{PEF}(s) only.}
    \resizebox*{\linewidth}{!}{\begin{tabular}{r|l|l}
    Configuration & Benefits with respect to the \ac{PEF} alone & Drawbacks with respect to the \ac{PEF} alone 
    \\
    
    \hline
    
    
    PEF + POF
    
     & 
     
     \begin{minipage}{0.45\textwidth}
        $\bullet$ Destination receives the data units in order.\\
        $\bullet$ Reordering-for-free: the \ac{POF} does not increase the end-to-end delay bounds (when at least one replicate per d.u. is received).
       \end{minipage}
     
     &
     \begin{minipage}{0.4\textwidth}
     $\bullet$ The \ac{POF} worsens the arrival curve; this can lead to higher delay bounds in downstream nodes.\\
     $\bullet$ Increased hardware complexity (Figure~\ref{fig:system-model:pof-model}).

     \end{minipage}
     
    \\
     
  \hline
  
     PEF + REG
     
     &

     \begin{minipage}[c]{0.45\textwidth}
        $\bullet$ Output traffic keeps the arrival constraints it had before the redundant section, resulting in smaller delay bounds in downstream nodes.     
       \end{minipage}

     &
     \begin{minipage}[c]{0.4\textwidth}
     $\bullet$ Delay penalty due to mis-ordering:\\
      With PFR: delay penalty with a guaranteed maximum delay.
      With IR: unbounded delay.\\
     $\bullet$ Increased hardware complexity (Figure~\ref{fig:system-model:reg-model}).
     \end{minipage}

    \\

  \hline
  
   PEF + POF + REG 
   
   &

     \begin{minipage}[c]{0.45\textwidth}
        $\bullet$ Destination receives the data units in order.\\
        $\bullet$ Reordering-for-free and shaping-for-free: [POF + REG] does not increase the delay bounds (when at least one replicate per data unit is received).  \\
        $\bullet$ Output traffic keeps the same arrival constraints as              it had before the redundant section.
       \end{minipage}
     &
     \begin{minipage}[c]{0.4\textwidth}
     $\bullet$ Increased hardware complexity (Figures~\ref{fig:system-model:pof-model} and \ref{fig:system-model:reg-model}).
     \end{minipage}  
\end{tabular} }
\end{table*}

Table~\ref{tab:benefits-drawbacks} summarizes the benefits and drawbacks of using regulators after a \ac{PEF}, as analyzed in Sections~\ref{sec:regulators:pfr} and \ref{sec:regulators:ir}.
We observe that the drawbacks of the \aclp{REG} appear symmetrical with respect to those of the \ac{POF}.
For example, a main issue of the \ac{POF} is the burstiness of the traffic at its output; this can be corrected by using a regulator.
A main issue of the \acp{REG} is the delay penalty caused by the out-of-order input; this can be solved by placing a \ac{POF} just before.

The combination \ac{PEF} + \ac{POF} + \ac{REG} appears as a potential solution for keeping the benefits of both the \ac{POF} and the \ac{REG} without their main drawbacks.
We first analyze this new configuration on the toy example.

\textbf{Application to the Toy Example:}
Let us first add a \ac{POF} before the \ac{PFR} of the single-flow situation in Fig.~\ref{fig:regulators:toy-example-notations}.
It gives the situation presented in Fig.~\ref{fig:regulators:toy-example-pef-pof-pfr}.
The \ac{POF} enforces the order of the data units as seen at $B$.
Assume for example that it receives the traffic defined by the line ``\texttt{outPEF}'' of Fig.~\ref{fig:regulators:traj-2d-pfr}.
Then the \ac{POF} outputs the data units as on Line ``\texttt{outPOF}'' of Fig.~\ref{fig:regulators:traj-toy-example-pof-pfr}.
The \ac{PFR} further processes this trajectory to spread the data units as per the flow's contract and outputs them as on the Line ``\texttt{outPFR}'' of Fig.~\ref{fig:regulators:traj-toy-example-pof-pfr}.
The resulting traffic is compliant with the initial arrival curve $\alpha_{r_0,b_0}$. 
We observe that all the data units have kept an \ac{ETE} delay below 7 t.u.
\begin{figure}\centering
    \resizebox*{\linewidth}{!}{\begin{tikzpicture}
	\node[draw] at (0,0) (s) {\makecell{Packet\\replication\\function (\acs{PRF})}};
	\node[draw] at (3,0.75) (a) {$[0,1]$};
	\node[draw] at (3,-0.75) (b) {$[6,7]$};
	\node[draw] at (6.5,0) (c) {$\texttt{PEF}_F(f)$};
    \node[draw, anchor=west] at ([xshift=0.8cm] c.east) (pof) {$\texttt{POF}_F(\{f\},B)$};
    \node[draw, anchor=west] at ([xshift=0.8cm] pof.east) (reg) {$\texttt{REG}_F(\{f\},B)$};
	\node[draw, dashed, rounded corners=0.5cm, fit={(a)}, inner sep=0.3cm] (cNode) {};
	\node[draw, dashed, rounded corners=0.5cm, fit={(b)}, inner sep=0.3cm] (dNode) {};
	\node[anchor=west] at ([xshift=-0.2cm] dNode.north east) {$D$};
	\node[anchor=west] at ([xshift=-0.2cm] cNode.north east) {$C$};

	\draw[->, rounded corners=0.2cm] ([yshift=0.2cm]s.east) -- ([xshift=0.5cm,yshift=0.2cm] s.east) -- ([xshift=-0.5cm] a.west) -- (a.west);
	\draw[->, rounded corners=0.2cm] ([yshift=-0.2cm]s.east) -- ([xshift=0.5cm,yshift=-0.2cm] s.east) -- ([xshift=-0.5cm] b.west) -- (b.west);
	
	\draw[->, rounded corners=0.2cm] (a.east) -- ([xshift=0.5cm] a.east) -- ([xshift=-1cm,yshift=0.2cm] c.west) -- ([yshift=0.2cm]c.west) node[pos=0.5,anchor=center] (tA) {};
	\draw[->, rounded corners=0.2cm] (b.east) -- ([xshift=0.5cm] b.east) -- ([xshift=-1cm,yshift=-0.2cm] c.west) -- ([yshift=-0.2cm]c.west) node[pos=0.5,anchor=center] (tB) {};
	\draw[-,red] ([yshift=-0.1cm] tA.center) -- ([yshift=0.1cm] tA.center) node[pos=1,right, rotate=90] {\texttt{outC}};
	\draw[-,red] ([yshift=-0.1cm] tB.center) -- ([yshift=0.1cm] tB.center) node[pos=0,left, rotate=90] {\texttt{outD}};

	\draw[dashed] ([xshift=1.25cm] c.west) ++(90:-1.5) arc (90:-90:-1.5) node[pos=0, anchor=center] (mt) {} node[pos=0.9,below] (Ftxt) {$F$} node[pos=1, anchor=center] (mtt) {};
	\draw[dashed] (mt.center) -- ++(6cm,0);
	\draw[dotted] (mt.center) ++(6cm,0) -- ++(0.5cm,0);
	\draw[dashed] (mtt.center) -- ++(6cm,0);
	\draw[dotted] (mtt.center) ++(6cm,0) -- ++(0.5cm,0);

	\draw[dashed] ([xshift=-1cm] s.west) ++(90:0.5) arc (90:-90:0.5) node[pos=0, anchor=center] (mmt) {} node[pos=1, anchor=center] (mmtt) {};
	\draw[] (mmt.center) -- ++(-0.1cm,0);
	\draw[dotted] (mmt.center) ++(-0.1cm,0) -- ++(-0.5cm,0);
	\draw[] (mmtt.center) -- ++(-0.1cm,0);
	\draw[dotted] (mmtt.center) ++(-0.1cm,0) -- ++(-0.5cm,0);
	\node[anchor=south] at (mmt.center) {$B$};

	\draw[->] ([xshift=-1.5cm] s.west) -- (s.west) node[pos=0.8,anchor=center] (tIn) {};
	\draw[-,red] ([yshift=0.1cm] tIn.center) -- ([yshift=-0.1cm] tIn.center) node[pos=0,above] {\texttt{in}} node[pos=0, anchor=center] (targetIn) {};

    \draw[->] (c.east) -- (pof.west) node[pos=0.5,anchor=center] (tOut) {};
    \draw[->] (pof.east) -- (reg.west) node[pos=0.5,anchor=center] (tOutPof) {};
    \draw[->] (reg.east) -- ++(1cm,0) node[pos=0.5, anchor=center] (tRegOut) {};
    \draw[-,red] ([yshift=-0.1cm] tOutPof.center) -- ([yshift=0.1cm] tOutPof.center) node[pos=1,right, rotate=90] {\texttt{outPOF}} node[pos=0, anchor=center] (targetOutPOF) {};
	\draw[-,red] ([yshift=-0.1cm] tOut.center) -- ([yshift=0.1cm] tOut.center) node[pos=1,right, rotate=90] {\texttt{outPEF}} node[pos=0, anchor=center] (targetOutPEF) {};
    \draw[-,red] ([yshift=-0.1cm] tRegOut.center) -- ([yshift=0.1cm] tRegOut.center) node[pos=1, right, rotate=90] {\texttt{outPFR}} node[pos=0,anchor=center] (targetOutPFR) {};


\end{tikzpicture}%
}%
    \caption{\label{fig:regulators:toy-example-pef-pof-pfr} Toy example of Fig.~\ref{fig:prob-statement:toy-example}, to which we added a \ac{POF} followed by a \ac{PFR}.}
\end{figure}
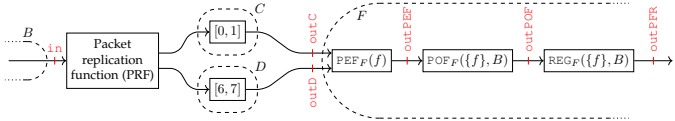
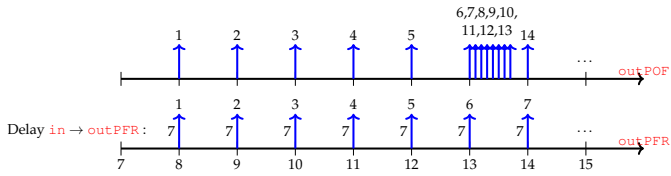
\begin{figure}\centering
    \resizebox*{\linewidth}{!}{\begin{tikzpicture}[xscale=2,yscale=1.2]
	\tikzstyle{lw} = [line width=2pt]
	\tikzstyle{every node}=[font=\Large]

	\draw [->,lw] (6,-6) -- (15,-6);
	\draw (6,-6.2) -- (6,-5.8)	;
	\draw (7,-6.2) -- (7,-5.8)	;
	\draw (8,-6.2) -- (8,-5.8)	;
	\draw (9,-6.2) -- (9,-5.8)  ;
	\draw (10,-6.2) -- (10,-5.8);
	\draw (11,-6.2) -- (11,-5.8);
	\draw (12,-6.2) -- (12,-5.8);
	\draw (13,-6.2) -- (13,-5.8);
	\draw (14,-6.2) -- (14,-5.8);
	\node[anchor=south, red] at (15,-6)  {\texttt{outPOF}};
	
    \draw[lw,blue,->] (7,-6) -- (7,-5) node[pos=1, black, anchor=south] {1};
    \draw[lw,blue,->] (8,-6) -- (8,-5) node[pos=1, black, anchor=south] {2};
	\draw[lw,blue,->] (9,-6) -- (9,-5) node[pos=1, black, anchor=south] {3};
	\draw[lw,blue,->] (10,-6) -- (10,-5) node[pos=1, black, anchor=south] {4};
	\draw[lw,blue,->] (11,-6) -- (11,-5) node[pos=1, black, anchor=south] {5};
	
	\draw[lw,blue,->] (12,-6) --   (12,-5);
	\draw[lw,blue,->] (12.1,-6) -- (12.1,-5);
    \draw[lw,blue,->] (12.2,-6) -- (12.2,-5);
    \draw[lw,blue,->] (12.3,-6) -- (12.3,-5) node[pos=1, black, anchor=south] {\makecell{6,7,8,9,10,\\11,12,13}};
    \draw[lw,blue,->] (12.4,-6) -- (12.4,-5);
    \draw[lw,blue,->] (12.5,-6) -- (12.5,-5);
    \draw[lw,blue,->] (12.6,-6) -- (12.6,-5);
	\draw[lw,blue,->] (12.7,-6) -- (12.7,-5);
	\draw[lw,blue,->] (13,-6) -- (13,-5) node[pos=1, black, anchor=south] {14};
    \node at (14,-5.5) {\ldots};

	\draw [->,lw] 	(6,-8) -- (15,-8);
	\draw (6, -8.2) -- (6, -7.8)	node[pos=0,below] {7};
	\draw (7, -8.2) -- (7, -7.8)	node[pos=0,below] {8};
	\draw (8, -8.2) -- (8, -7.8)	node[pos=0,below] {9};
	\draw (9, -8.2) -- (9, -7.8)  node[pos=0,below] {10};
	\draw (10,-8.2) -- (10,-7.8)node[pos=0,below] {11};
	\draw (11,-8.2) -- (11,-7.8)node[pos=0,below] {12};
	\draw (12,-8.2) -- (12,-7.8)node[pos=0,below] {13};
	\draw (13,-8.2) -- (13,-7.8)node[pos=0,below] {14};
	\draw (14,-8.2) -- (14,-7.8)node[pos=0,below] {15};
	\node[anchor=south, red] at (15,-8)  {\texttt{outPFR}};
	
	\draw[lw,blue,->] (7,-8) -- (7,-7) node[pos=1, black, anchor=south] {1} node[pos=0.5,left, black] {7};
	\draw[lw,blue,->] (8,-8) -- (8,-7) node[pos=1, black, anchor=south] {2} node[pos=0.5,left, black] {7};
	\draw[lw,blue,->] (9,-8) -- (9,-7) node[pos=1, black, anchor=south] {3} node[pos=0.5,left, black] {7};
	\draw[lw,blue,->] (10,-8) -- (10,-7) node[pos=1, black, anchor=south] {4} node[pos=0.5,left, black] {7};
	\draw[lw,blue,->] (11,-8) -- (11,-7) node[pos=1, black, anchor=south] {5} node[pos=0.5,left, black] {7};
	
	\draw[lw,blue,->] (12,	-8) --   (12,-7)  node[pos=1, black, anchor=south] {6} node[pos=0.5,left, black] {7};
	\draw[lw,blue,->] (13,-8) -- (13,-7) node[pos=1, black, anchor=south] {7} node[pos=0.5,left, black] {7};
	\node at (14,-7.5) {\ldots};

	\node[anchor=east] at (6.5,-7.5) {Delay \color{red}\texttt{in} \color{black} $\rightarrow$ \color{red}\texttt{outPFR} \color{black}:};
\end{tikzpicture}%
}%
    \caption{\label{fig:regulators:traj-toy-example-pof-pfr} Output of the \ac{POF} and of the \ac{PFR} of Fig.~\ref{fig:regulators:toy-example-pef-pof-pfr} when they process the trajectory of Fig.~\ref{fig:regulators:traj-2d-pfr}.}
\end{figure}


When using an interleaved regulator, Property 5/ of Theorem~\ref{thm:regulators:ir-instable} shows that the re-sequencing must be performed globally on the aggregate processed by the \ac{IR}, and not for each flow individually.
%
%
The above observations are summarized in the following result, valid for both \acp{PFR} and \acp{IR}.

\begin{theorem}[\label{thm:regulators:preof-for-free} Elimination-resequencing-reshaping is for free]
    Consider a network with graph $\mathcal{G}$ and consider a set of one or more flows $\mathcal{F}$.
    Take $a$ and $n$ two vertices of $\mathcal{G}$ such that for each flow $f\in\mathcal{F}$, $a$ is a diamond ancestor of $n$ in $\mathcal{G}(f)$  (see Fig.~\ref{fig:regulators:preof-for-free}).
    Assume that the \ac{CBQS} within $n$ is preceded by the following functions, in this order: a set of parallel \aclp{PEF} $\{\texttt{PEF}_n(f)\}_{f\in\mathcal{F}}$, followed by a unique \acl{POF} with configuration $\texttt{POF}_n(\mathcal{F},a)$, and finally a \acl{REG} with configuration $\texttt{REG}_n(\mathcal{F},a)$. 
    Denote by $d$ [resp., $D$] a lower bound [resp., an upper bound] for the delay of the non-lost data units of $\mathcal{F}$ through the system $\mathcal{S}$ between $a$ and the output of the \acp{PEF}.

    $\bullet$ If $\mathcal{S}$ is lossless for $\mathcal{F}$ (\ie for every data unit $m$ of the aggregate, at least one packet containing $m$ reaches the \acp{PEF}), then $d$ [resp., $D$] is also a lower bound [resp., an upper bound] for the delay of the non-lost data units through $\mathcal{S}'$, which we note $[d',D']=[d,D]$.

    $\bullet$ Otherwise, denote by $T$ the timeout value of the \ac{POF} \cite[\S III.D]{mohammadpourPacketReorderingTimeSensitive2020}. Then $d$ [resp., $D+T$] is a lower bound [resp., an upper bound] for the delay of the data units through $\mathcal{S}'$, \emph{i.e.,} $[d',D']=[d,D+T]$.

\end{theorem}

\begin{figure}\centering
    \resizebox*{\linewidth}{!}{\begin{tikzpicture}[xscale=1.3]
    \def\mfact{1.6}

    \def\mma{-0.5}
    \draw[dashed] (\mma,0) ++(90:\mfact*0.2) arc (90:-90:\mfact*0.2);
    \node at (\mma,0) {$a$};
    \draw[dashed] (\mma,\mfact*0.2) -- ++(-0.2cm,0);
    \draw[dashed] (\mma,-\mfact*0.2) -- ++(-0.2cm,0);

    \node at (\mfact*0.5,0) (entry) {};
    \node[circle, draw, dashed] at (\mfact*0.5,\mfact*0.5) (n1){};
    \node[circle, draw, dashed] at (\mfact*1,0) (n2){};
    \node[circle, draw, dashed] at (\mfact*0.75,\mfact*-0.5)(n3) {};  
    \node[circle, draw, dashed] at (\mfact*1.5,\mfact*0.25)(n4) {};
    \node[circle, draw, dashed] at (\mfact*1.4,\mfact*-0.25) (n5){};
    \node at (\mfact*2.25,0) (exit) {};
    \node[draw, anchor=west, minimum width=1.8cm] at ([yshift=0.5cm)]\mfact*2.5,0) (pef) {$\texttt{PEF}_n(f_1)$};
    \path let \p1 = (pef) in node at (\x1,0) (myExit) {\ldots};
    \node[draw, anchor=west, minimum width=1.8cm] at ([yshift=-0.5cm)]\mfact*2.5,0) (pef2) {$\texttt{PEF}_n(f_m)$};

    \node[draw, anchor=north, rotate=90] at ([xshift=1.5cm] myExit.east) (pof) {$\texttt{POF}_n(\mathcal{F},a)$};

    \node[draw, anchor=north, rotate=90] at (\mfact*5.1,0) (pfr) {$\texttt{REG}_n(\mathcal{F},a)$};
    \node[draw, thin, fit={(entry) (n1) (n2) (n3) (n4) (n5) (exit) (pef) ([xshift=0.2cm]pef2.east)}] (sys) {};
    \draw[dashed] (\mfact*2.6,0) ++(90:\mfact*-0.65) arc (90:-90:\mfact*-0.65);
    \draw[dashed] (\mfact*2.6,\mfact*0.65) -- ++(5cm,0);
    \draw[dashed] (\mfact*2.6,\mfact*-0.65) -- ++(5cm,0);

    \node[thin, densely dotted, draw, fit={(entry) (n1) (n2) (n3) (n4) (n5) (exit) (pef) (sys) (pof) ([xshift=0.2cm]pef2.east)}] (sysDagger) {};
    \node[draw, thin, fit={(entry) (n1) (n2) (n3) (n4) (n5) (exit) (pef) (sys) (pof) (sysDagger) (pfr) ([xshift=0.2cm]pef2.east)}] (sysPrime) {};

    \node[thin, draw, anchor=south west, fill=white] at ([xshift=-\pgflinewidth] sys.south east) {$\mathcal{S}$};
    \node[thin, densely dotted, draw, anchor=south west, fill=white] at ([xshift=-\pgflinewidth] sysDagger.south east) {$\mathcal{S}^{\dagger}$};
    \node[thin, draw, anchor=south west, fill=white] at ([xshift=-\pgflinewidth] sysPrime.south east) {$\mathcal{S}'$};

    \draw[->] (\mma+\mfact*0.2,0) -- (sys.west) node[pos=0.4,above] {$\mathcal{F}$};
    \draw[->, dotted] (sys.west) -- (n1);
    \draw[->, dotted] (sys.west) -- (n3);
    \draw[->, dotted] (n3) -- (n5);
    \draw[->, dotted] (n3) -- (n2);
    \draw[->, dotted] (n1) -- (n2);
    \draw[->, dotted] (n2) -- (n4);
    \draw[-, dotted] (n4) -- (\mfact*2.6-\mfact*0.65,0);
    \draw[-, dotted] (n5) -- (\mfact*2.6-\mfact*0.65,0);

    \draw[->, dotted] (\mfact*2.6-\mfact*0.65,0) -- (pef.west);
    \draw[->, dotted] (\mfact*2.6-\mfact*0.65,0) -- (pef2.west);
    \draw[-, dotted] (pef.east) -- (sys.east);
    \draw[-, dotted] (pef2.east) -- (sys.east);
    \draw[->] (sys.east) -- (pof.north);
    \draw[->] (pof.south) -- (pfr.north);
    \draw[->, dotted] (pfr.south) -- ++(1cm,0);

    \node[anchor=north west] at (\mfact*2.5-\mfact*0.6, \mfact*0.6) {$n$};

    \draw[red, dashdotted] (\mfact*0.25,\mfact*1) -- (\mfact*0.25,\mfact*-1.1);
    \draw[red, dashdotted] (\mfact*3.835,\mfact*1) -- (\mfact*3.835,\mfact*-1.1);
    \draw[latex-latex, red] (\mfact*0.25,\mfact*-0.9) -- (\mfact*3.835,\mfact*-0.9) node[pos=0, left] {$[d,D]$};
    \draw[red, dashdotted] (\mfact*5.45,\mfact*1) -- (\mfact*5.45,\mfact*-1.1);
    \draw[latex-latex, red] (\mfact*0.25,\mfact*-1) -- (\mfact*5.45,\mfact*-1) node[pos=1, right] {$[d',D']$};
\end{tikzpicture}%
}%
    \caption{\label{fig:regulators:preof-for-free} Notations of Theorem~\ref{thm:regulators:preof-for-free}. An aggregate re-sequencing followed by a \ac{REG} is placed after the \acp{PEF}. We are interested in the delay bounds through system $\mathcal{S}'$.}
\end{figure}
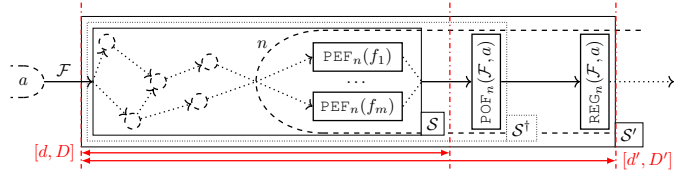

The proof in Appendix~\ref{proof:thm:regulators:preof-for-free} first applies \cite[Theorem~4]{mohammadpourPacketReorderingTimeSensitive2020} to obtain the delay bounds through the system $\mathcal{S}^\dagger$ on Fig.~\ref{fig:regulators:preof-for-free}.
This system is \ac{FIFO} thus \cite[Thm. 5]{leboudecNetworkCalculusTheory2001} can be applied.

Therefore, the ``\ac{PEF} + \ac{POF} + \ac{REG}'' configuration provides all the benefits on the network performance bounds associated with the ``\ac{PEF} + \ac{POF}'' and the ``\ac{PEF} + \ac{REG}'' configurations, removing most of their drawbacks.
This is summarized on the last line of Table~\ref{tab:benefits-drawbacks}.
Only the hardware cost remains a drawback, as the models of Figures~\ref{fig:system-model:pof-model} and \ref{fig:system-model:reg-model} must be implemented.

%

\vspace{\myvspacebeforesec}
\section{Evaluation of the Framework on an Industrial Use-Case}\label{sec:evaluation}

\label{sec:volvo}

In this section, we use a modified version of FP-TFA~\cite[\S VI]{thomasCyclicDependenciesRegulators2019} that implements the results from Sections~\ref{sec:toolbox} and \ref{sec:regulators} to compute end-to-end delay bounds in a representative industrial use-case that contains \ac{PREF}.
FP-TFA has been chosen because it can compute delay bounds for general topologies, \emph{i.e.} even for those with cyclic dependencies \cite{thomasCyclicDependenciesRegulators2019}.

\textbf{Network Description:}
We consider the Volvo core \ac{TSN} network~\cite{navetEarlystageBottleneckIdentification2020}.
Its physical topology is given in Fig.~\ref{fig:volvo:phy-topo}.
The network contains two redundant control units \texttt{P1} and \texttt{P2}~\cite[Page~4]{navetEarlystageBottleneckIdentification2020}.
Each of the four \acp{MCU} acts as a gateway between the core \ac{TSN} network and the local networks running on legacy protocols.
We hence assume that the \acp{MCU} are legacy devices that support only 100Mbps full-duplex links and cannot implement the recent technologies of \ac{TSN} or \ac{DetNet}, such as \ac{PREOF}.
We assume that their applications cannot handle any duplicate.
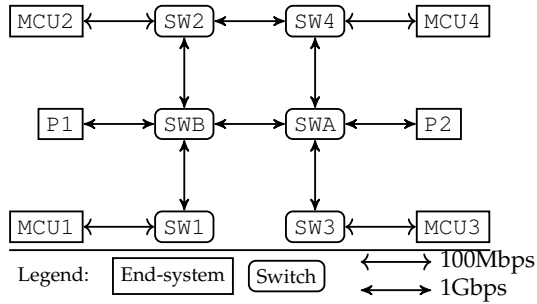
\begin{figure}\centering
    \resizebox*{0.8\linewidth}{!}{\begin{tikzpicture}
    \tikzstyle{es} = [draw]
    \tikzstyle{sw} = [draw, rounded corners=0.1cm]
    \tikzstyle{1g} = [stealth'-stealth']
    \tikzstyle{100m} = [<->]

    \node[sw] at (0,0) (swb) {\texttt{SWB}};
    \node[sw, anchor=west] at ([xshift=1cm] swb.east) (swa) {\texttt{SWA}};
    \node[sw, anchor=south] at ([yshift=1cm] swb.north) (sw2) {\texttt{SW2}};
    \node[es, anchor=east] at ([xshift=-1cm] sw2.west) (mcu2) {\texttt{MCU2}};
    \node[sw, anchor=south] at ([yshift=1cm] swa.north) (sw4) {\texttt{SW4}};
    \node[es, anchor=west] at ([xshift=1cm] sw4.east) (mcu4) {\texttt{MCU4}};
    \node[es, anchor=west] at ([xshift=1cm] swa.east) (p2) {\texttt{P2}};
    \node[sw, anchor=north] at ([yshift=-1cm] swa.south) (sw3) {\texttt{SW3}};
    \node[es, anchor=west] at ([xshift=1cm] sw3.east) (mcu3) {\texttt{MCU3}};
    \node[es, anchor=east] at ([xshift=-1cm] swb.west) (p1) {\texttt{P1}};
    \node[sw, anchor=north] at ([yshift=-1cm] swb.south) (sw1) {\texttt{SW1}};
    \node[es, anchor=east] at ([xshift=-1cm] sw1.west) (mcu1) {\texttt{MCU1}};

    \draw[100m] (mcu1) -- (sw1);
    \draw[100m] (sw3) -- (mcu3);
    \draw[100m] (sw4) -- (mcu4);
    \draw[100m] (sw2) -- (mcu2);
    \draw[1g] (sw1) -- (swb);
    \draw[1g] (swb) -- (p1);
    \draw[1g] (swb) -- (swa);
    \draw[1g] (swa) -- (sw3);
    \draw[1g] (swa) -- (p2);
    \draw[1g] (swa) -- (sw4);
    \draw[1g] (sw4) -- (sw2);
    \draw[1g] (sw2) -- (swb);

    \draw ([yshift=-0.1cm] mcu1.south west) -- ([yshift=-0.1cm] mcu3.south east);
    \node[anchor=north west] at ([yshift=-0.2cm] mcu1.south west) (legend) {\footnotesize Legend:};
    \node[anchor=west, es] at ([xshift=0.2cm] legend.east) (les) {\footnotesize End-system};
    \node[anchor=west, sw] at ([xshift=0.2cm] les.east) (lsw) {\footnotesize Switch};
    \draw[100m] ([xshift=0.5cm, yshift=0.2cm] lsw.east) -- ++(1cm,0) node[pos=1,right] {100Mbps};
    \draw[1g] ([xshift=0.5cm, yshift=-0.2cm] lsw.east) -- ++(1cm,0) node[pos=1, right] {1Gbps};
\end{tikzpicture}%
}%
    \caption{\label{fig:volvo:phy-topo} Simplified physical topology of the Volvo core \ac{TSN} Network. From~\cite{navetEarlystageBottleneckIdentification2020}.}
\end{figure}


\textbf{Flow Description:}
We focus on the \emph{Command and Control} class and consider four different periodic traffic profiles within the class. Their characteristics are based on~\cite[Page~13]{navetEarlystageBottleneckIdentification2020} and listed in Table~\ref{tab:volvo:traffic-profiles}.
%
For each traffic profile and for each \texttt{MCU}, there exist a multicast flow that carries the sensor data from the \ac{MCU} to both \texttt{P1} and \texttt{P2} and a unicast flow per control unit (2 in total) that carries the commands from the control unit to the \ac{MCU} (see Table~\ref{tab:volvo:traffic-path}). 

To meet stringent loss-ratio requirements, flows are redounded by using \ac{PREF}, whenever two alternative paths can be found for a $($source$,$ destination$)$ tuple.
In total, the network contains 48 flows, including 40 redounded flows, 16 of which are also multicast.
\begin{table}\centering
    \caption{\label{tab:volvo:traffic-profiles} Traffic Profiles. Realistic Use-Case Based on the Values for \emph{Command and Control} Flows in \cite[Page~13]{navetEarlystageBottleneckIdentification2020}.
    }
    \resizebox*{0.8\linewidth}{!}{\begin{tabular}{l|r|r|r}
    Name & \multicolumn{1}{l|}{Payload size} & \multicolumn{1}{l|}{Period at source} & \multicolumn{1}{l}{Deadline} \\
    \hline
    \texttt{S}  & 64B & 0.5ms & 0.2ms\\
    \texttt{M1} & 92B & 2ms & 0.8ms \\
    \texttt{M2} & 120B & 3.5ms & 1.4ms\\
    \texttt{B} & 150B & 5ms & 2ms \\
\end{tabular}}
\end{table}

\begin{table}\centering
    \caption{\label{tab:volvo:traffic-path} Flow Path for $i \in \{\texttt{1},\texttt{2},\texttt{3},\texttt{4}\}$, $p\in \{\texttt{S},\texttt{M1},\texttt{M2},\texttt{B}\}$.}
    \resizebox*{\linewidth}{!}{\begin{tabular}{r|l|l|l}
    Name & Source & Dest. & Redundancy \\
    \hline
    \texttt{C\_MCU$i$\_P12\_$p$} & \texttt{MCU$i$} & \texttt{P1}, \texttt{P2} & \makecell{For \texttt{C\_MCU3\_P12\_$p$} [resp., \texttt{C\_MCU4\_P12\_$p$}],\\  dest. \texttt{P2} [resp., \texttt{P1}] is not protected} \\
    \texttt{C\_P1\_MCU$i$\_$p$}   & \texttt{P1}   & \texttt{MCU$i$}   & Except for \texttt{C\_P1\_MCU1\_$p$} \\
    \texttt{C\_P2\_MCU$i$\_$p$}   & \texttt{P2}   & \texttt{MCU$i$}   & Except for \texttt{C\_P2\_MCU3\_$p$}\\

\end{tabular}}
\end{table}

\textbf{Service Description:}
As the class of interest is of highest priority, each \ac{CBQS} offers to the aggregate a service rate equal to the capacity of the transmission link (either 100Mbps or 1Gbps). 
We also assume that the technological latency within each output port is below $2\mu s$, and we neglect input-port and switching-fabric latencies. 

\textbf{Comparison of the Analytical Models:}
We first set the load of the network at 5.2\%. We compare the \emph{intuitive approach} from Sec.~\ref{sec:prob-statement:pref-issues} with the \emph{tight model} that relies on Theorem~\ref{thm:toolbox:pef-oac}.
In Fig.~\ref{fig:volvo:model-comparison-results}, we provide the deterministic lower and upper bounds of the latency of each flow for each of its destinations. 
The delay upper-bounds are obtained by using the fix-point version of FP-TFA~\cite[\S~VI.C]{thomasCyclicDependenciesRegulators2019}, modified for taking into account the effect of \ac{PREF} with either the intuitive approach or the tight model.
The exact best-case and worst-case latencies for the flow are guaranteed to be within the provided interval, thus the smaller the guaranteed interval the better the model.

We observe that an analysis of the network by using the tight model concludes that all flows meet their deadline, whereas the same analysis that uses the intuitive approach shows that four flows may violate their deadlines. 
The delay bounds for all flows, including those that are not redounded by \ac{PREF}, are improved with the tight model.
For example, the flow in a box in Fig.~\ref{fig:volvo:model-comparison-results}, from \texttt{P2} to \texttt{MCU3}, is not redounded, but the tight model still computes a guaranteed delay interval tighter than with the intuitive approach.
Indeed, the flow shares the link $\texttt{SWA}\rightarrow\texttt{SW3}$ and $\texttt{SW3}\rightarrow\texttt{MCU3}$ with redounded flows, for which the burst bounds obtained with the tight model are smaller. 
Hence, the delay that this flow suffers in \texttt{SWA} and \texttt{SW3} has a better bound with the tight model than with the intuitive one.

\begin{figure}\centering
    \resizebox*{\linewidth}{!}{\input{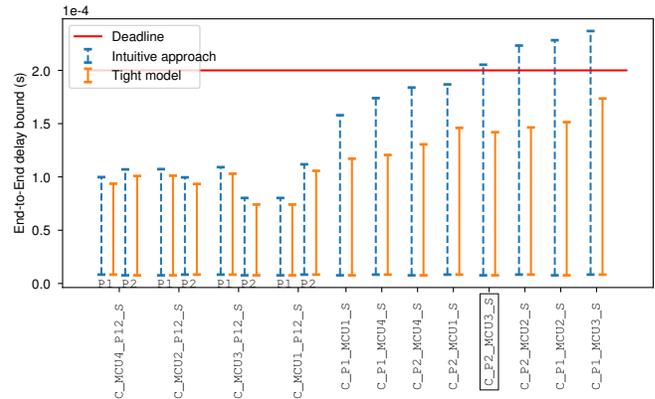}}%
    \vspace{-0.6cm}
    \caption{\label{fig:volvo:model-comparison-results} Comparison of the guaranteed \acf{ETE} latency intervals (upper and lower bounds) for each flow and each destination, obtained by using either the intuitive approach or the tight model.
    }
\end{figure}

\textbf{Comparison of the Technological Solutions:}
\begin{figure}\centering
    \resizebox*{\linewidth}{!}{\input{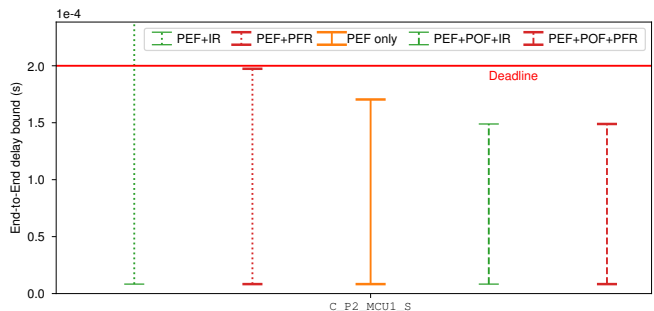}}%
    \vspace{-0.6cm}
    \caption{\label{fig:volvo:tech-comparison-resutls}
        Comparison of the guaranteed \ac{ETE} latency intervals with several technological choices.
        Delay bounds without any \ac{REG} are shown in the middle.
        The bars on the left are the guaranteed intervals when the flows are regulated after the \ac{PEF}, but without any \ac{POF}.
        When a \ac{POF} is additionally placed between the \ac{PEF} and the \ac{REG}, we obtain the results on the right of the baseline.
    }
\end{figure}
Fig.~\ref{fig:volvo:model-comparison-results} shows that, at low network load, the network edges withstand the peak rate an increased burstiness at the output of the \acp{PEF}, even if they rely only on 100Mbps links.

We now consider the same network but we increase the load up to 88\% by reducing the period of each flow.
We focus on the four redounded flows from \texttt{P2} to \texttt{MCU1}.
Each of them is processed by a \ac{PEF} within \texttt{SWB} to eliminate the duplicates coming from \texttt{SW2} and \texttt{SWA} and each of them present a peak rate and and increased burstiness after its \ac{PEF}.

We evaluate the opportunity to shape the four flows with their source profile before they compete with the four other flows coming from \texttt{P1} in the output port of \texttt{SWB}.
We can either use four \acfp{PFR} 
(each processing a unique flow), or we can use a unique \acf{IR}, because they all share the same reference point $\texttt{P2}$.

Fig.~\ref{fig:volvo:tech-comparison-resutls} focuses on flow \texttt{C\_P2\_MCU1\_S}.
The baseline guaranteed delay interval (in the middle) is obtained from the application of the tight model without any regulator.
We note that the flow is schedulable, but as the network load is higher, its safety margin is reduced with respect to Fig.~\ref{fig:volvo:model-comparison-results}.

The dotted bars on the left of the baseline represent the guaranteed delay intervals obtained when the flows are processed, either with an \ac{IR} (far-left), or with four independent \acp{PFR}, but without using any \ac{POF}.
For the \ac{IR}, no guarantee can be obtained per Theorem~\ref{thm:regulators:ir-instable}.
For the \ac{PFR}, the flow remain schedulable but its safety margin is drastically reduced by the delay penalty of the \ac{PFR} (Theorem~\ref{thm:regulators:pfr-2d}).

The dashed bars on the right of the baseline represent the guaranteed delay bounds when using the combination \ac{POF}+\ac{REG} after the \ac{PEF}, assuming that for each data unit, at least one replicate is not lost.
On the far-right, the delay bounds with four per-flow \acp{POF} placed before the \acp{PFR} (as in Fig.~\ref{fig:system-model:example-of-functions-a}), and the other bar represents the delay bounds with a unique \ac{POF} for the aggregate before the \ac{IR} (as in Fig.~\ref{fig:system-model:example-of-functions}).
The shaping-for-free property holds in both cases, thus their delay bounds are equal.
They represent a 13\% improvement with respect to the baseline.
Indeed, the regulators reduce the downstream burst, thus reducing the worst-case delay in the low-capacity link \texttt{SW1}$\rightarrow$\texttt{MCU1}.

	\section{Conclusion}
\vspace{\myvspacebeforesec}
\section{Conclusion}

We provide a toolbox of network-calculus results that give theoretical foundations for the worst-case analysis of \ac{DetNet} PREOF (\emph{\acl{PREOF}}) and \ac{TSN} \acs{FRER} (\emph{\acl{FRER}}).
The toolbox contains an output-arrival-curve characterization of the packet-elimination function that is tighter than any other variable-bit-rate or leaky-bucket arrival curves. It also contains a quantification of the amount of mis-ordering caused by the redundancy.

We further 
analyze the interactions between the packet-elimination function, the packet-ordering function and traffic regulators.
We show that the latter can cancel the burstiness increase caused by the redundancy.
But when traffic regulators are placed immediately after the \acl{PEF}, they do not enjoy the shaping-for-free property: Per-flow regulators induce a delay penalty that we upper-bound, whereas \aclp{IR} (such as \acs{TSN} \emph{Asynchronous Traffic Shapers}) induce unbounded latencies.
Shaping-for-free can be retrieved if the data units are reordered after the elimination function and prior to shaping.

The users of \ac{TSN} \ac{FRER} and \ac{TSN} \ac{ATS} are invited to bear in mind the conflicting interactions outlined in this paper, as no \acl{POF} is available within \ac{TSN} at the time of writing.

We finally apply our theoretical and practical results on a representative industrial use-case.
The latency bounds obtained with the toolbox are significantly tighter than those obtained with an intuitive approach.
We also highlight the end-to-end latency gain obtained on the use-case when traffic regulators are placed after the redundant section with a reordering function in between.




	
	
	%
	\vspace{-0.05cm}
	\bibliography{../biblio/bib-zotero,../biblio/bib-jylb,../biblio/bib-ludo}

	\begin{acronym}[CP-OFDMX]
	\acro{XXXXX}{\tagg{List of acronyms for Holly}\tagg{Will be removed in final version}}
	\acro{ACP}{aggregate computation pipeline}
	\acro{AFDX}{Avionics Full-dupleX switched Ethernet}
	\acro{ATS}{asynchronous traffic shaping}
	\acro{AVB}{Audio Video Bridging}
	\acro{CAN}{Controller Area Network}
	\acro{CBQS}{class-based queuing subsystem}
	\acro{CBS}{credit-based scheduler}
	\acro{CDT}{control-data traffic}
	\acro{CEV}{crew exploration vehicle}
	\acro{COTS}{commercial off the shelf}
	\acro{DAG}{directed acyclic graph}
	\acro{DetNet}{deterministic networking}
	\acro{DNC}{deterministic network calculus}
	\acro{ETE}{end-to-end}
	\acro{EP}{elimination-pending}
	\acro{FIFO}{first in, first out}
	\acro{FP}{fixed-priority}
	\acro{FRER}{frame replication and elimination for redundancy}
	\acro{HSR}{High-availability Seamless Redundancy}
	\acro{HTTP}{Hypertext Transfer Protocol}
	\acro{IEC}{International Electrotechnical Committee}
	\acro{IEEE}{Institute of Electrical and Electronics Engineers}
	\acro{IETF}{Internet Engineering Task Force}
	\acro{IMA}{Integrated Modular Avionics}
	\acro{iPRP}{IP parallel redundancy protocol}
	\acro{IR}{interleaved regulator}
	\acro{LCAN}{low-cost acyclic network}
	\acro{MCU}{micro-controller unit}
	\acro{MFAS}{minimum feedback arc set}
	\acro{MFVS}{minimum feedback vertex set}
	\acro{MOST}{Media Oriented Systems Transport}
	\acro{NC}{network calculus}
	\acro{NoC}{networks on chip}
	\acro{OSI}{Open Systems Interconnection}
	\acro{PBOO}{pay burst only once}
	\acro{PFR}{per-flow regulator}
	\acro{PMOC}{pay multiplexing only at convergence points}
	\acro{PEF}{packet-elimination function}
	\acro{POF}{packet-ordering function}
	\acro{PRF}{packet-replication function}
	\acro{PREF}[PREFs]{packet replication and elimination functions}
	\acro{PREOF}[PREOFs]{packet replication, elimination and ordering functions}
	\acro{PRP}{Parallel Redundancy Protocol}
	\acro{QoS}{quality of service}
	\acro{RAMS}{Reliability, Availability, Maintainability, and Safety}
	\acro{RBO}{reordering byte offset}
	\acro{REG}{regulator}
	\acro{RSTP}{Rapid Spanning Tree Protocol}
	\acro{RTE}{Real-Time Ethernet}
	\acro{RTO}{reordering late time offset}
	\acro{SFA}{single-flow analysis}
	\acro{SNC}{stochastic network calculus}
	\acro{TAS}{Time-Aware Shaping}
	\acro{TCP}{Transmission Control Protocol}
	\acro{TFA}{total-flow analysis}
	\acro{TP}{turn prohibition}
	\acro{TSN}{time-sensitive networking}
	\acro{VBR}{variable-bit-rate}
	\acro{VIU}{vehicle interface unit}
\end{acronym}


%

\begin{IEEEbiography}[{\includegraphics[width=1in,clip,trim=0 350 0 0,height=1.25in]{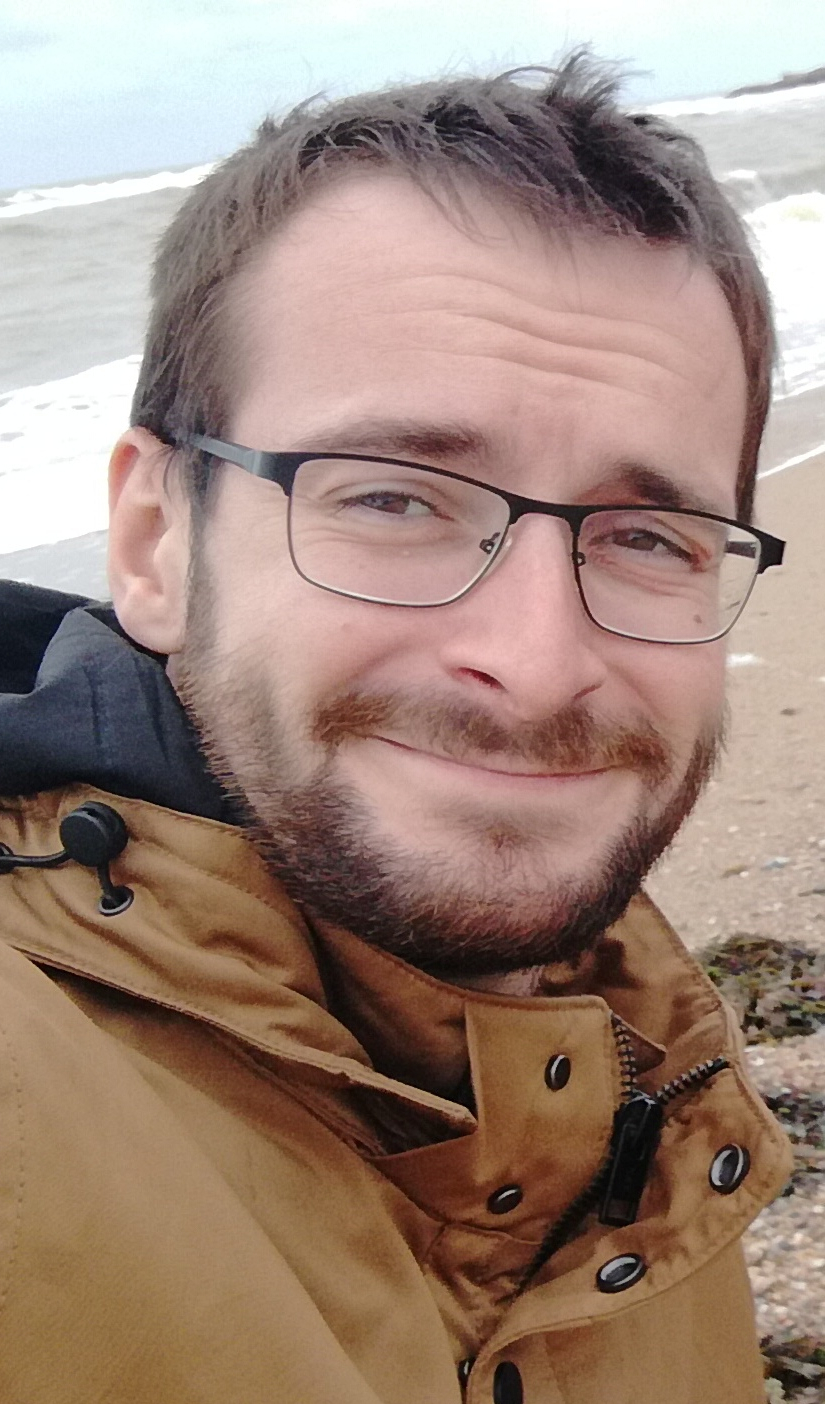}}]{Ludovic Thomas}
	(Graduate Student Member, IEEE)
	obtained his Master's degree in Aerospace Engineering from ISAE-SUPAERO in 2018.
	He performed his final-year internship at the French aerospace agency where he studied the effects of fully-encrypted transport-layer protocols (such as QUIC) on the performance of satellite Internet access. 
	He is currently a PhD student at ISAE-SUPAERO in collaboration with EPFL and focuses on the deterministic performance analysis of Time-Sensitive Networking (TSN) technologies using the network-calculus framework.
	His research interests include the performance analysis of computer networks using measurements, simulations and deterministic approaches.
\end{IEEEbiography}

\begin{IEEEbiography}[{\includegraphics[width=1in,clip,trim=20 0 25 0,height=1.25in]{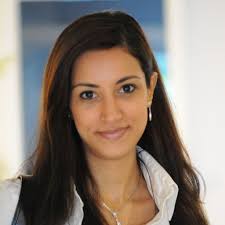}}]{Ahlem Mifdaoui}
	(Member, IEEE) received the M.E. degree in computer science and air traffic management from the Ecole Nationale de Aviation civile (ENAC) Toulouse, in 2004, and the Ph.D. degree in computer science and telecommunication from the Institut National Polytechnique of Toulouse (INPT), in 2007. 
	She has been a Full Professor with the Department of Complex System Engineering, ISAE-Supaero, University of Toulouse, since 2017. 
	She has been the head of the embedded system master’s degree, since 2016. Her main research interests are real-time networks and cyber-physical systems. 
	Particular attention is given to the performance analysis of safety-critical applications such as avionics and satellites. 
	Since 2008, she has been successfully coordinated more than ten national projects and supervised almost ten Ph.D. degree students and postdoctoral researchers. Moreover, she has been involved in more than 30 conference program committees such as RTSS, ECRTS, and RTAS, and is a regular Reviewer of many journals in the field such as the IEEE Transactions on Computers and Real-Time Systems.
\end{IEEEbiography}

\begin{IEEEbiography}[{\includegraphics[width=1in,clip,trim=40 0 60 0,height=1.25in]{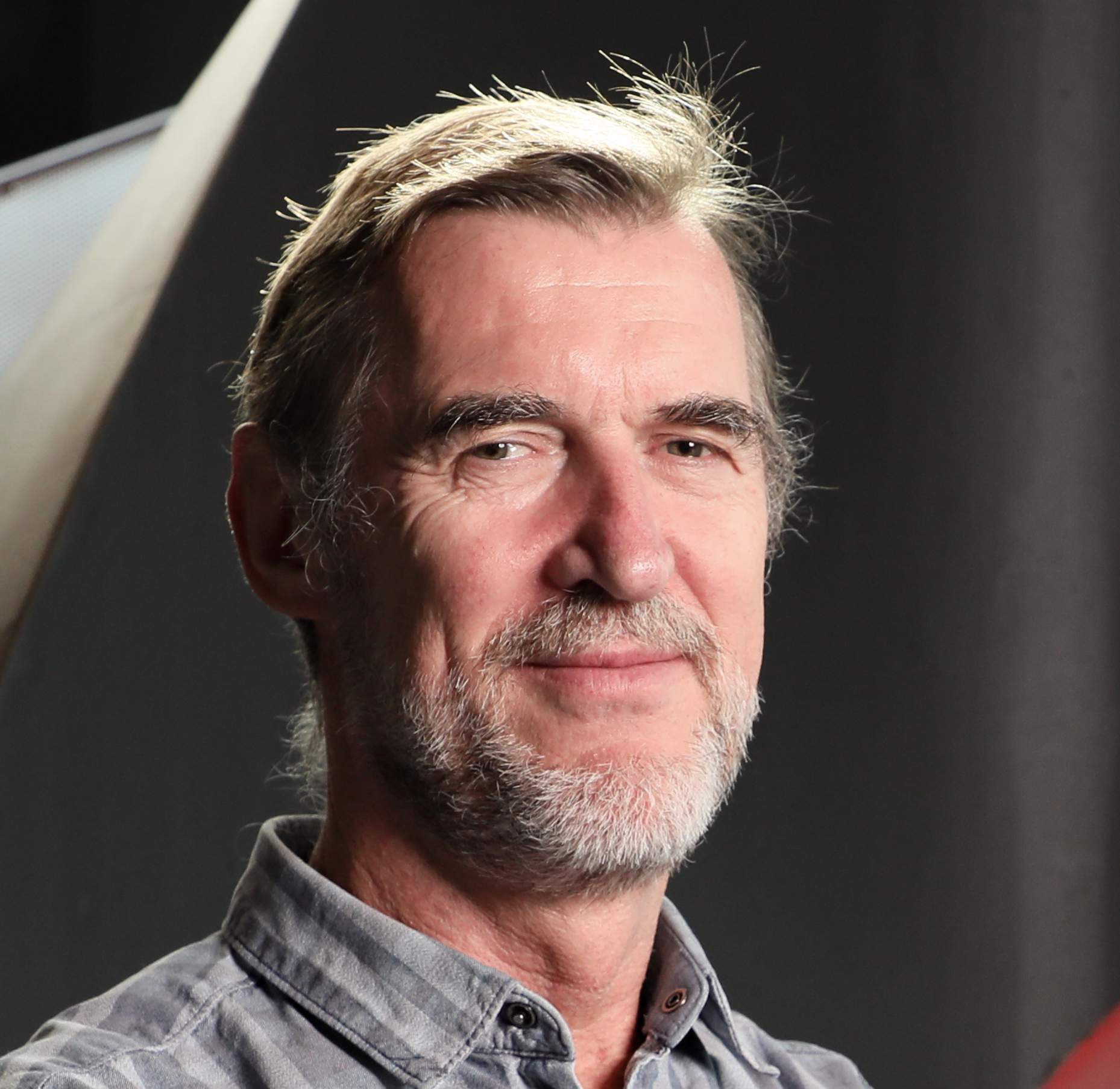}}]{Jean-Yves Le Boudec}
	(Fellow, IEEE) received the Agrégation degree in mathematics from the École Normale Supérieure de Saint-Cloud, Paris, in 1980, and the doctorate degree from the University of Rennes, France, in 1984.
	From 1984 to 1987, he was with INSA/IRISA, Rennes. 
	In 1987, he joined Bell Northern Research, Ottawa, Canada, as a member of scientific staff at the Network and Product Traffic Design Department. 
	In 1988, he joined the IBM Zurich Research Laboratory where he was the Manager of the Customer Premises Network Department. 
	In 1994, he became an Associate Professor at the École Polytechnique Fédérale de Lausanne (EPFL), where he is a Professor. 
	He is a coauthor of a book on network calculus, which serves as a foundation for deterministic networking, an introductory textbook on information sciences, and the author of the book \emph{Performance Evaluation}. 
	His research interests are in the performance and architecture of communication systems and smart grids.
\end{IEEEbiography}






	\appendices

\section{Discussion on the Relationship between our System Model and the \acs{TSN} and \acs{DetNet} Documents.}\label{sec:appendix:discussion-system-model}

The system model proposed in Section~\ref{sec:system-model} results from an analysis of both the \ac{DetNet} \ac{PREOF} \cite{finnDeterministicNetworkingArchitecture2019} and the \ac{TSN} \ac{FRER} \cite{IEEEStandardLocal2017a} documents.
The present appendix can be used by the \ac{TSN} and the \ac{DetNet} communities for evaluating the applicability of our results in the \ac{TSN} and \ac{DetNet} contexts.
The appendix highlights the similarities and the differences between the terms, the notions and the assumptions used in our system model with those that are used in the \ac{DetNet} and \ac{TSN} documents.

The \ac{IETF} \ac{DetNet} documents focus on the network-layer mechanisms whereas \ac{IEEE} \ac{TSN} documents focus on the link-layer mechanisms. A notion can hence have different terms depending on the considered layer.
For example, within \ac{DetNet}, data units are encapsulated within \emph{packets} and a coherent sequence of them that originates also from a single source is a \emph{\ac{DetNet} flow}\cite[\S 2.1]{finnDeterministicNetworkingArchitecture2019}. Within \ac{TSN}, data units are carried by \emph{frames}, and a coherent sequence of them that originates from the same source is a \emph{stream}\cite[\S 3]{IEEEStandardLocal2017a}. In the paper, we use the terms \emph{packets} and \emph{flows}.

\subsection{Directed Acyclic Graphs (\acsp{DAG}) versus \emph{Compound} and \emph{Member} Flows.}\label{sec:appendix:discussion-system-model:dag}

When defining the redundancy that protects a flow, the working groups of \acs{DetNet} and \ac{TSN} use the terms \emph{compound} and \emph{member} flows.
For example, \ac{DetNet} \ac{PREOF} specifies that, 
``\emph{A \acs{DetNet} compound flow is a \acs{DetNet} flow that has been separated into multiple duplicate \acs{DetNet} member flows for service protection [\ldots]. Member flows are merged back into a single \acs{DetNet} compound flow [\ldots]}''\cite[\S 2.1]{finnDeterministicNetworkingArchitecture2019}.
Similarly, \ac{TSN} \ac{FRER} indicates that, ``\emph{A Compound Stream is a Stream composed of one or more Member Streams linked together via Frame Replication and Elimination for Reliability (\acs{FRER})}''\cite[\S 3]{IEEEStandardLocal2017a}.

These two adjectives are not used in our paper. 
Indeed, the model of Section~\ref{sec:system-model} is based on \acfp{DAG} and on the knowledge of positions of the elimination points.
This model is compatible with the \emph{compound} and \emph{member} terms but is also more versatile.
For example, take the flow $f$ defined by the graph in Figure~\ref{fig:system-model:example-flow-graph} and by the knowledge that $F$ contains a \ac{PEF} for $f$.
If we focus on destination $G$, we could say that the compound flow with path $A\rightarrow B$ is \emph{separated} into \emph{duplicate member flows} with paths $B\rightarrow C \rightarrow F$ and $B\rightarrow D \rightarrow F$ and that these member flows are \emph{merged back} into the \emph{compound flow} for path $F\rightarrow G$.
But the previous distinction between \emph{compound} and \emph{member} flows cannot be applied for destination $E$. 
Neither the \ac{DetNet} documents nor the \ac{TSN} documents clarify that the \emph{compound}/\emph{member} distinction depends on the destination that is considered. By focusing on \acp{DAG} and by not using the two adjectives, we remove the above ambiguity.

\subsection{\emph{Replicates} versus \emph{Duplicates}}

Neither \cite{IEEEStandardLocal2017a} nor \cite{finnDeterministicNetworkingArchitecture2019} provides a formal definition for \emph{replicates} and \emph{duplicates}.
The \ac{TSN} \ac{FRER} standard even use both ``\emph{eliminate duplicate packets}'' \cite[\S 7.1.1]{IEEEStandardLocal2017a} and 
``\emph{eliminates the replicates}'' \cite[\S 1.6]{IEEEStandardLocal2017a}.
However, both documents seem to adhere to the following convention.

$-$ \emph{Replicates} are defined as identical copies of the same packet (of the same piece of data).

$-$ \emph{Duplicates} are defined with respect to a given location or for a given function (\emph{e.g.}, for a \ac{PEF}): 
A packet is a duplicate at a given location [resp., for a given function] if an identical copy of itself (another replicate of the same piece of data) has been observed previously at the location [resp., by the function].

We re-use the same convention in our paper.
The \ac{PEF} is hence a function that forwards the first replicate and drops the duplicates, it \emph{eliminates the duplicates}.

\subsection{Considered Types of Failures}\label{sec:appendix:discussion-system-model:protect}

Within \ac{DetNet}, \acs{PREOF} provide \emph{service protection} that ``\emph{aims to mitigate or eliminate packet loss due to equipment failures, including random media and/or memory faults.}''\cite[\S 3.2.2.]{finnDeterministicNetworkingArchitecture2019}.
Similarly, ``\emph{\acs{FRER} can substantially reduce the probability of packet loss due to equipment failures}''\cite[\S 1.2]{IEEEStandardLocal2017a}.

In the paper, we consider only failures that cause packets to be lost on the transmission links. 
As described in Section~\ref{sec:system-model:function-model}, this can model various real-life failures that lose packets but only if they do not affect the service provided to the non-lost packets.
For example, random media packet losses fall within our model.
Similarly, we can model a device that shuts down: in this case, all its links lose all packets.
However, failures that cause a network element to provide less service than its minimum-service contract or that cause a source to generate more traffic than its maximum-traffic contract are not considered in our model but could be considered using tools from \emph{\acf{SNC}}\footnote{Y. Jiang and Y. Liu, Stochastic Network Calculus. London: Springer-Verlag, 2008. \url{https://www.springer.com/gp/book/9781848001268}.}.

\subsection{Packet Replication Function (\acs{PRF}) and Multicast Mechanisms}\label{sec:appendix:discussion-system-model:prf}

In \ac{DetNet}, ``\emph{Flow replication [\ldots] can be performed by, for example, techniques similar to ordinary multicast replication}''\cite[\S 4.1.1]{finnDeterministicNetworkingArchitecture2019}. In \ac{TSN}, ``\emph{no explicit Stream splitting function [\ldots] is required. Frames in a single Compound Stream can be replicated using the normal multicast mechanisms [\ldots]}''\cite[\S 8.1]{IEEEStandardLocal2017a}.

Our system model follows the same rationale. The \acf{PRF} is implemented by the switching fabric that already handles the duplication of packets for multicast flows. The rationale also motivates the choice of \acp{DAG} for modeling flow paths. Indeed, \acp{DAG} represent a natural extension of multicast trees.

\subsection{Packet Ordering Function (\acs{POF}) and its Position with Respect to the \acs{PEF}}\label{sec:appendix:discussion-system-model:pof}

In \cite[\S 2.1]{finnDeterministicNetworkingArchitecture2019}, a \acf{POF} is defined as a function that ``\emph{reorders packets within a \acs{DetNet} flow that are received out of order}''.
In \ac{TSN}, there exists no function with similar goals as of March 2022.
In fact, in-order-delivery was a goal in an early draft version of the \ac{FRER} standard, but it was latter removed due to hardware considerations\footnote{See Comment 29 at \url{https://www.ieee802.org/1/files/private/cb-drafts/d0/802-1CB-d0-3-dis.pdf}. For the credentials, consult \url{https://www.ietf.org/proceedings/52/slides/bridge-0/tsld003.htm}.}.

In our model, we extend the definition of the \ac{DetNet} \ac{POF} and allow the function to consider an aggregate of flows.
As we show in Section~\ref{sec:regulators}, this aggregate reordering is necessary for obtaining delay bounds when \acfp{IR} are placed after the elimination function, and per-flow reordering is not sufficient to guarantee bounded latencies when an \ac{IR} is placed after the \ac{PEF}.

In \cite[\S 3.2.2.2]{finnDeterministicNetworkingArchitecture2019}, the \ac{DetNet} working group states that the ``\emph{order in which a \ac{DetNet} node applies \ac{PEF} [and] \ac{POF} [\ldots] to a \ac{DetNet} flow is left open for implementations}''.
This is however contradicted by both \cite[\S 3.2.2.1.]{finnDeterministicNetworkingArchitecture2019} and \cite[\S 7.1.1.m]{IEEEStandardLocal2017a} where packet mis-ordering is seen as a side-effect of the \ac{PEF}, for which a \ac{POF} is a remedy when placed after the \ac{PEF}.
Additionally, the ongoing draft for the \ac{POF} states that, ``\emph{the [\acs{POF}] algorithm assumes that a Packet Elimination Function (\acs{PEF}) is performed on the incoming packets before they are handed to the \acs{POF} function.
Hence, the sequence of incoming packets can be out of order or incomplete but cannot contain duplicate packets}'' \cite[\S 4.1]{vargaDeterministicNetworkingDetNet2021}. In our model, we follow the same assumption.

The internal algorithm of the \acf{POF} also relies on a timeout parameter, that is denoted by $T$ in our paper and in \cite{mohammadpourPacketReorderingTimeSensitive2020} and that is called ``\texttt{POFMaxDelay}'' in \cite{vargaDeterministicNetworkingDetNet2021}.
As its name suggest, it corresponds to the maximum delay that a data unit can spend in the \ac{POF}, even if the previously-expected data unit has not been received so far.
The timeout prevents the \ac{POF} from holding forever a data unit if the previously-expected data unit has been lost.
To prevent spurious transmission of out-of-order data units, the timeout $T$ cannot take any value.
In our paper, we always assume that $T$ follows the recommendations of \cite[\S IV.B]{mohammadpourPacketReorderingTimeSensitive2020}, which are also consistent with \cite[\S 4.3]{vargaDeterministicNetworkingDetNet2021}. 
This configuration depends on the value of the \acf{RTO} at the location of the \ac{POF} and our Theorem~\ref{thm:toolbox:reordering} provides this value for when the \ac{POF} is placed after a \ac{PEF}.

The implications of an incomplete sequence of data units for a \ac{POF} are further analyzed in \cite{mohammadpourPacketReorderingTimeSensitive2020}.
In our Theorem~\ref{thm:regulators:preof-for-free}, we reuse the results from \cite{mohammadpourPacketReorderingTimeSensitive2020} and apply them in the case of a \ac{POF} placed just after a \ac{PEF}.

\subsection{Traffic Regulators (\acsp{REG}) and their Position with Respect to the Other Functions}\label{sec:appendix:discussion-system-model:reg}

Within \ac{TSN}, \emph{\acf{ATS}}~\cite{IEEEStandardLocal2020} is a building block that implements the \acf{IR} model within the \ac{TSN} bridges.
In \ac{ATS}, each \ac{IR} (called an \emph{ATS Scheduler Group} in \cite{IEEEStandardLocal2020}) is in the form $\texttt{REG}_n(\mathcal{F},o)$ where $\mathcal{F}$ is the set of flows that enter $n$ from $p$, $p$ is a direct parent of $n$, and $o$ is a direct parent of $p$. 
As such, each \emph{ATS Scheduler Group} can only cancel the burstiness increase within the direct upstream parent $p$.
Our model authorizes more flexibility when defining the aggregate $\mathcal{F}$ and the reference~$o$. 
Within our model, an \ac{IR} can cancel the burstiness increase caused by any system made of several network elements.

For example, in the industrial application in Section~\ref{sec:volvo}, the \ac{IR} placed in \texttt{SWB} cancel the burstiness increase caused by the entire system that is located between \texttt{P2} and \texttt{SWB} (Figure~\ref{fig:volvo:phy-topo}) and that includes a \acf{PRF}, two redundant paths with several devices for each, and a \acf{PEF}.

In the technical documents, there exists an uncertainty on the relative order of the \ac{REG} and the other \ac{PREOF} functions within a device (especially with respect to \ac{PEF}).
The \ac{TSN} implementation of \ac{PEF} (called \ac{FRER}, see Table~\ref{tab:introduction:comparison-names}) is defined in \cite{IEEEStandardLocal2017a} whereas the \ac{TSN} implementation of \aclp{REG} (\ac{ATS}) is defined in \cite{IEEEStandardLocal2020}.
Their pipeline models, \cite[Figure~8.2]{IEEEStandardLocal2017a} and \cite[Figure~8.13]{IEEEStandardLocal2020}, place their respective mechanisms exactly at the same position in the forwarding process (between IEEE802.1Q 8.6.5 \emph{Flow metering} and IEEE802.1Q 8.6.6 \emph{Queuing frames}).
As of August 2021, no information has been provided on their relative order.

In our paper, traffic regulators are of particular interest when they are placed after the \acp{PEF}, because they can shape the traffic back to the profile it had at the input of the redundant section (second section in Figure~\ref{fig:prob-statement:three-sections}).
Analyzing the interactions between \ac{PREOF} and traffic regulators in this configuration is one of our major objectives in the paper, while placing a regulator before the \ac{PEF} is equivalent to shaping the traffic within a sub-path of a multicast flow, a situation widely analyzed \cite{leboudecTheoryTrafficRegulators2018, mohammadpourLatencyBacklogBounds2018,thomasCyclicDependenciesRegulators2019}.
If regulators are placed before the \ac{PEF}, then their effects on the arrival curve of the flow at the input of the \ac{PEF}, $\alpha_{\text{PEF}^{\text{in}}}$, can be computed. 
Then Item 1/ of Theorem~\ref{thm:toolbox:pef-oac} can be applied to obtain the arrival curve of the flow at the output of the \ac{PEF}, as affected by the regulators placed before.

Within \ac{DetNet}, traffic shaping is also mentioned as one of the mechanisms for providing bounded delivery~\cite[\S 4.5]{finnDeterministicNetworkingArchitecture2019}. The \ac{DetNet} working group refers to the traffic shapers of DiffServ \cite{blackArchitectureDifferentiatedServices1998}. Thus, the traffic shapers that are mentioned in \ac{DetNet} follow the \acf{PFR} model, but \ac{DetNet} also indicates that the ``\emph{actual queuing and shaping mechanisms are
typically provided by the underlying subnet}''~\cite[\S 4.1.1.]{finnDeterministicNetworkingArchitecture2019}. There exists the same uncertainty as in \ac{TSN} on the relative order of the traffic shaper with the other functions of \ac{DetNet}.


\section{What Network Calculus Results Remain Valid for Non-lossless Non-FIFO Systems ?}\label{sec:appendix:non-fifo-non-lossless-results}

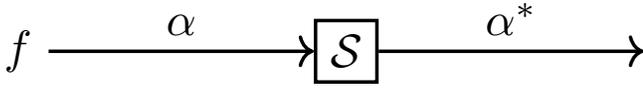
\begin{figure}\centering
    \resizebox*{\linewidth}{!}{\begin{tikzpicture}
    \node[draw] at (0,0) (s) {$\mathcal{S}$};
    \draw[->] ([xshift=-2cm] s.west) -- (s.west) node[pos=0, left] {$f$} node[pos=0.5,above] {$\alpha$};
    \draw[->] (s.east) -- ++(2cm,0) node[pos=0.5,above]{$\alpha^*$};
\end{tikzpicture}}
    \caption{\label{fig:appendix:nc-nonfifo-nonlossless:notations-system} Notations for Appendix~\ref{sec:appendix:non-fifo-non-lossless-results}. Flow $f$ with input arrival curve $\alpha$ enters system $\mathcal{S}$ and exits with an output arrival curve $\alpha^*$.}
\end{figure}

A main aspect of the network-calculus framework is the capacity to combine its results in order to analyze a specific property in a network. 
Assume for example that a flow $f$ goes through a lossless system $\mathcal{S}$ (Figure~\ref{fig:appendix:nc-nonfifo-nonlossless:notations-system}) and assume that the delay of each bit of $f$ through $\mathcal{S}$ is upper-bounded by $D$.
Furthermore, note $\alpha$ the arrival curve of $f$ at the input of $\mathcal{S}$.
Then, if we are interested in an arrival curve $\alpha^*$ of $f$ at the output of $\mathcal{S}$, we can use Proposition~1.3.7 of \cite{leboudecNetworkCalculusTheory2001} to obtain that $\delta_D$ is a service curve of $\mathcal{S}$ for $f$ and combine it with Theorem~1.4.3 of \cite{leboudecNetworkCalculusTheory2001} to finally obtain $\alpha^* = \alpha \oslash \delta_D$ is an arrival curve of $f$ at the output of the system.

Clearly, if the system $\mathcal{S}$ can lose packets (is not lossless) but continues to provide the guaranteed delay bound $D$ to the non-lost bits, then we expect the output traffic (that contains only the bits of non-lost packets) to be also $\alpha^*$-constrained.
Yet, \cite[Prop.~1.3.7]{leboudecNetworkCalculusTheory2001} applies only to lossless systems, thus making it impossible to reuse the same combination of results.

In this section, we provide a set of network-calculus results that remain valid with lossless and non-\ac{FIFO} systems.
In the literature, Ciucu \emph{et al.} introduce\footnote{F. Ciucu, J. Schmitt, and H. Wang, “On expressing networks with flow transformations in convolution-form,” in 2011 Proceedings IEEE INFOCOM, pp. 1979–1987, Apr. 2011.} the concept of loss processes to model non-lossless but \ac{FIFO} systems using the \acf{SNC} framework. They focus on the service curves and on the concatenation property of services curves~\cite[Thm.~1.4.6]{leboudecNetworkCalculusTheory2001}, whereas we focus on obtaining an output arrival curve in the framework of \acf{DNC}.
In the \ac{DNC} framework, Mohammadpour and Le Boudec obtain an output arrival curve for a flow at the output of a non-\ac{FIFO} system, when the jitter of each packet is constrained~\cite[Lemma~1]{mohammadpourPacketReorderingTimeSensitive2020}.

Our following result can be distinguished from the previous work as it applies to any systems that do not need to be \ac{FIFO} or lossless and in which the delay of each non-lost bit is constrained in a bounded interval.

\begin{proposition}[Arrival curve of a flow at the output of a system with bounded delay]\label{prop:appendix:ac-after-lossy}
    Consider a flow $f$ entering a system $\mathcal{S}$. 
    Assume that each bit of $f$ that exits $\mathcal{S}$ suffers a delay within $\mathcal{S}$ that is bounded within $[d,D]$. Finally assume that $\alpha$ is an arrival curve for $f$ at the input of $\mathcal{S}$.
    $\mathcal{S}$ does not need to be \ac{FIFO} or lossless.

    Then, $\alpha^* = \alpha\oslash\delta_{D-d}$ is an arrival curve for $f$ at the output of $\mathcal{S}$.
\end{proposition}

\begin{proof}
    Denote by $R_f$ the cumulative process of $f$ at the input of $\mathcal{S}$ (Figure~\ref{fig:appendix:nc-nonfifo-nonlossless:notations-system}).
    We decompose $R_f = R_a + R_b$, with $R_a$ the cumulative process, at the input, for the stream of bits of $f$ that are not lost inside $\mathcal{S}$ and $R_b = R_f - R_a$.
    $R_a(t)$ is defined as the number of bits of $f$ that eventually exit $\mathcal{S}$ (that are not lost inside it) and that are observed at the input of $\mathcal{S}$ during the interval $[0,t]$.
    Note that the cumulative functions $R_a$ and $R_b$ are unknown in general:
    When a bit is observed at the input of $\mathcal{S}$, the real-life observer cannot infer whether it will be lost within $\mathcal{S}$ or not. 
    However, we can still work on the unknown functions $R_a$ and $R_b$.

    We denote by $R_a^*$ [resp., $R_b^*$] the output cumulative function related to the input process $R_a$ [resp., $R_b$].
    Hence, $R_a^*(t)$ is defined as the number of bits of $f$ that exit $\mathcal{S}$ (are not lost inside it) and that are seen at the output of $\mathcal{S}$ during the interval $[0,t]$.

    All cumulative functions are positive, wide-sense increasing and defined for $t\ge0$~\cite[\S 1.1.1]{leboudecNetworkCalculusTheory2001}.
    We extend their definition domain by using the convention that all cumulative functions equal zero in $\mathbb{R}-$: $\forall t \le 0, R_a(t) = R_b(t) = R_a^*(t) = R_b^*(t) = 0$.

    As $\alpha$ is an arrival curve for $f$ at the input of $\mathcal{S}$, the input process $R_f$ is $\alpha$-constrained\cite[Definition~1.2.1]{leboudecNetworkCalculusTheory2001}, thus for all $s\le t$,
    \begin{equation*}
        \begin{aligned}
            R_f(t) - R_f(s) &\le \alpha(t-s) \\
            R_a(t) + R_b(t) - R_a(s) + R_b(s) &\le \alpha(t-s) \\
            R_a(t) - R_a(s) &\le \alpha(t-s) + R_b(s) - R_b(t) \\
        \end{aligned}
    \end{equation*}
    As $s\le t$, and $R_b$ is wide-sense increasing, $R_b(s) - R_b(t) \le 0$ and $R_a(t) - R_a(s) \le \alpha(t-s)$ which shows that the cumulative process $R_a$ is also $\alpha$-constrained.

    By definition of $R_b$ and $R_b^*$, $\forall t, R_b^*(t) = 0$ and
    \begin{equation}\label{eq:appendix:toolbox:RaStarIsRfStar}
        R_a^*(t) = R_f^*(t)
    \end{equation}

    The system $\mathcal{S}$ is not \ac{FIFO} but the non-lost bits have a maximum delay of $D$.
    Hence, all the bits of $R_a$ that have entered $\mathcal{S}$ at $t$ have exited $\mathcal{S}$ by $t+D$,
    \begin{equation}\label{eq:appendix:toolbox:r-and-D}
        R_a^*(t+D) \ge R_a(t)
    \end{equation}
    Equation~(\ref{eq:appendix:toolbox:r-and-D}) is valid for $t<0$ because $R_a(t) = 0$ for $t<0$ and $R_a^*$ is a positive function.
    Similarly, the minimum delay of each data unit within $\mathcal{S}$ is $d$.
    As such, all the data units that have exited $\mathcal{S}$ by $t+d$ must have entered $\mathcal{S}$ before $t$,
    \begin{equation}\label{eq:appendix:toolbox:r-and-d}
        R_a^*(t+d) \le R_a(t)
    \end{equation}
    Equation~(\ref{eq:appendix:toolbox:r-and-d}) is again valid for $t <0$: $R_a(t) = 0$ but $R_a^*(t+d)$ also equals zero because the minimum time that a bit needs to reach the output is $d$.
    
    Then, $\forall t\ge s$,
    \begin{align*}
        R_a^*(t) &- R_a^*(s) \\
        & \le R_a(t-d) - R_a(s-D) \qquad\text{\Comment{(\ref{eq:appendix:toolbox:r-and-D}) and (\ref{eq:appendix:toolbox:r-and-d})}} \\
        & \le \alpha(t-s+(D-d)) \qquad \text{\Comment{}} R_a \text{ is }\alpha\text{-constrained}\\
        & \le (\alpha\oslash\delta_{D-d})(t-s)
    \end{align*}
    Combining the above result with (\ref{eq:appendix:toolbox:RaStarIsRfStar}) shows, $\forall t\ge s$
    \begin{equation}
        R_f^*(t) - R_f^*(s) \le (\alpha\oslash\delta_{D-d})(t-s)
    \end{equation}
    which proves that $\alpha^* =\alpha\oslash\delta_{D-d}$ is an arrival curve for $f$ at the output of $\mathcal{S}$.
\end{proof}

For a system with a constant delay or without any delay, Proposition~\ref{prop:appendix:ac-after-lossy} is simplified as follows.
\begin{corollary}[\label{cor:appendix:lossy-with-fixed-delay-keeps-arrival-curves}A system with constant delay keeps the arrival curves]
    If $\mathcal{S}$ is a system in which the non-lost bits of any flow have a constant delay, then an arrival curve for a flow or aggregate of flows at the input of $\mathcal{S}$ is also an arrival curve for the same flow or aggregate of flows at the output of $\mathcal{S}$.    
\end{corollary}
\begin{proof}
    For the aggregate, we simply need to consider the whole aggregate as a unique flow when applying Proposition~\ref{prop:appendix:ac-after-lossy}.
\end{proof}
\section{Proofs}
\subsection{Proof of Theorem~\ref{thm:toolbox:pef-oac}}\label{sec:appendix:pef-oac}

\begin{proof}[Proof of Theorem~\ref{thm:toolbox:pef-oac}]\label{proof:appendix:pef-oac}

	Consider a vertex $n$ and a flow $f$ such that $n$ contains a \ac{PEF} for $f$, noted $\texttt{PEF}_n(f)$.
	We first note that flow $f$ is packetized at both the input and the output of $\texttt{PEF}_n(f)$.

	\ul{Proof of Item 1/}
	As per the model in Section~\ref{sec:system-model:pef-model}, $\texttt{PEF}_n(f)$ is a network element that can lose packets but does not have any delay for the forwarded packets. As a consequence, it does not have any delay for the forwarded bits either (both its input and its output are packetized).
	Applying Corollary~4 proves that an arrival curve for $f$ at the input of the \ac{PEF}, $\alpha_{f,\text{PEF}^{\text{in}}}$, is also an arrival curve for $f$ at the output of the \ac{PEF}.

	\ul{Item 2/}
	Consider a diamond ancestor $a$ of $n$.
	The observation point $a^*$ is located at the output of the input port within $a$. 
	As such, flow $f$ is packetized at the observation point $a^*$.
	As such, a bound on the per-bit delay between $a^*$ and either the \ac{PEF}'s input or the \ac{PEF}'s output $\texttt{PEF}^*$ is also a per-packet delay bound on the delay between the same two observation points, and \emph{vice versa}.

	Denote by $\mathcal{P}^{\mathcal{G}(f)}_{a,n}$ the set of all possible paths from $a$ to $n$ in $\mathcal{G}(f)$ and consider a data unit $m$ of $f$ such that $m$ is not lost for $n$.
	By definition of the diamond ancestor, $a$ is not an \ac{EP}-vertex for $f$, thus the data unit $m$ is observed exactly once at $a$.

	Denote by $\{P^m_i\}_{i\in\mathcal{I}(m)}$ the set of packets containing $m$ that reach $\texttt{PEF}_n(f)$, with $\mathcal{I}(m)$ a set to index them.
	$\mathcal{I}(m)$ is not empty because $m$ is not lost for $n$ (at least one packet containing $m$ reaches $n$).
	Furthermore, $\mathcal{I}(m)$ is a finite set because $\mathcal{G}(f)$ is finite and acyclic: $m$ is replicated a finite number of times.
	
	For $i$ in $\mathcal{I}(m)$, call $\rho_i$ the path within $\mathcal{G}(f)$ that the packet $P^m_i$ took from the source of $f$ to $n$. 
	By definition of a diamond ancestor of $n$, this path crosses $a$ and by definition of $\mathcal{P}^{\mathcal{G}(f)}_{a,n}$, there exists a path $p_i\in \mathcal{P}^{\mathcal{G}(f)}_{a,n}$ such that packet $P^m_i$ took path $p_i$ between $a$ and $n$ ($p_i$ is a sub-path of $\rho_i$).
	
	Denote by $d^m_i$ the delay of packet $P^m_i$ between the output of $a$ and the input of $\texttt{PEF}_n(f)$.
	By definition of the notations $D_f^{a\rightarrow n}$ and $d_f^{a\rightarrow n}$ used in Theorem~\ref{thm:toolbox:pef-oac}, 
	
	\begin{equation}\label{eq:appendix:pef-oac:bounding-packet-bound}
		d_f^{a\rightarrow n} \le d^m_i \le D_f^{a\rightarrow n}
	\end{equation}

	
	The values $d_f^{a\rightarrow n}$ and $D_f^{a\rightarrow n}$ can be seen as the lower and upper-bound of the non-lost data units through the system located between the output of $a$ and the input of $\texttt{PEF}_n(f)$. 
	This system is represented with a cloud in Figure~\ref{fig:appendix:system-a-to-pef}.
	
	The data unit $m$ exits $\texttt{PEF}_n(f)$ as soon as one of the the packets $\{P^m_i\}_{i\in\mathcal{I}(m)}$ reaches the \acl{PEF}.
	If we denote by $d^{m}_{a\rightarrow \texttt{PEF}^*}$ the delay of the data unit $m$ from the output of $a$ to the output of the \ac{PEF}, we have, $\forall m\in f$, $m$ not lost for $n$,
	\begin{equation}\label{eq:appendix:pef-oac:relation-packet-data-out}
		\exists i \in\mathcal{I}(m), \quad d^{m}_{a\rightarrow \texttt{PEF}^*} = d_i^m
	\end{equation}
	Combining Equations~(\ref{eq:appendix:pef-oac:bounding-packet-bound}) and (\ref{eq:appendix:pef-oac:relation-packet-data-out}) gives
	\begin{equation}\label{eq:appendix:pef-oac:bounds-on-s}
		\forall m \in f, m \text{ not lost for } n, \quad d_f^{a\rightarrow n} \le d^{m}_{a^*\rightarrow \texttt{PEF}^*} \le D_f^{a\rightarrow n}
	\end{equation}
	
	Equation (\ref{eq:appendix:pef-oac:bounds-on-s}) proves that any non-lost data units of $f$ for $n$ suffer through the system $S$ in Figure~\ref{fig:appendix:system-a-to-pef} a delay bounded in $[d_f^{a\rightarrow n},D_f^{a\rightarrow n}]$.
	As both $a^*$ and the output of the \ac{PEF} are packetized, this also proves that each bit of $f$ that is not lost within $\mathcal{S}$ (neither in the cloud of Figure~\ref{fig:appendix:system-a-to-pef} nor in the \ac{PEF}) suffers a delay through $\mathcal{S}$ that is bounded within $[d_f^{a\rightarrow n},D_f^{a\rightarrow n}]$.
	We apply Proposition~\ref{prop:appendix:ac-after-lossy} and obtain that $\alpha_{f,a^*}\oslash\delta_{(D_f^{a\rightarrow n})-(d_f^{a\rightarrow n})}$ is an arrival curve for $f$ at the output of the \ac{PEF}.

	\begin{figure}\centering
		\resizebox*{\linewidth}{!}{\begin{tikzpicture}
    \def\myy{0.7cm}
    \node[circle, draw, red] at (0,0) (a) {$a$};
    \draw[-] (a.east) -- ++(1.5cm,0) node[pos=0.5,above] {$\alpha_{f,a*}$} node[pos=1, anchor=center](tt){};
    \asymCloud{([xshift=2.5cm] tt)}{}{0.6}{0}
    \draw[latex-latex] (3,0.8) -- (5,0.8) node[pos=0.8,above] (dCloud) {$[d_f^{a\rightarrow n},D_f^{a\rightarrow n}]$ (Eq.~\ref{eq:appendix:pef-oac:bounding-packet-bound})};
    \node[draw, rotate=90] at (7,0)(pef)  {$\texttt{PEF}_n(f)$};
    \node[draw, fit={(pef) ([xshift=0.1cm] tt.center) (dCloud)}] (sb) {};
    \node[anchor=north east] at (sb.north west) (s) {$S$};
    \draw[->] (tt.center) -- (2.95,0);
    \draw[->] (5,0) -- (pef.north);
    \draw[->] (pef.south) -- ++(1.5cm,0);
    \draw[latex-latex] ([yshift=-0.2cm] sb.south west) -- ([yshift=-0.2cm] sb.south east) node[pos=0.5, below] {$[d_f^{a\rightarrow n},D_f^{a\rightarrow n}]$ (Eq.~\ref{eq:appendix:pef-oac:bounds-on-s})};
\end{tikzpicture}}
		\caption{\label{fig:appendix:system-a-to-pef} Notations for Appendix~\ref{sec:appendix:pef-oac}: System from diamond ancestor $a$ to the \ac{PEF}, focusing on the non-lost data units.}
	\end{figure}
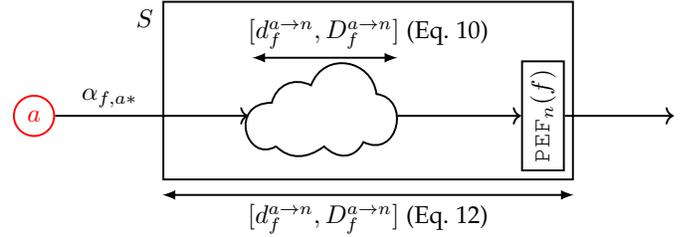
\end{proof}

\subsection{Proof of Corollary~\ref{cor:toolbox:thm-simple-application}}\label{sec:appendix:toolbox:cor}
\begin{proof}[Proof of Corollary~\ref{cor:toolbox:thm-simple-application}]\label{proof:appendix:toolbox:cor}
    The replication is performed by the switching fabric, both its input and output are hence packetized. 
    The same remark applies for the \ac{PEF} thus we conclude that both the input and the output of each system $S_i$ is also packetized.
    Therefore, the per-bit delay of flow $f$ through $S_i$ is also bounded by $[d_i,D_i]$. 

    For each $i \in \llbracket 1,N \rrbracket$, the application of Proposition~\ref{prop:appendix:ac-after-lossy} gives the arrival curve for $f$ at the output of $S_i$
    $$\forall \in \llbracket 1,N \rrbracket, \quad \alpha_{f,S_i^*} = \alpha_f \oslash \delta_{D_i-d_1}$$
    We then obtain $\alpha_{f,\text{PEF}^{\text{in}}}$
    \begin{align*}
        \alpha_{f,\text{PEF}^{\text{in}}} &= \sum_{i\in\llbracket 1 ,n \rrbracket} \alpha_{f,S_i^*} \\
        & =\sum_{i\in\llbracket 1 ,n \rrbracket}\alpha_f \oslash \delta_{D_i-d_1}
    \end{align*}

    We then apply Equation~(\ref{eq:thm:pef-oac:ancestor}) with the ancestor $a$ being the input of the replication function in Figure~\ref{fig:toolbox:oac-prop-simplified-figure}.
    A lower delay bound for $f$ from the ancestor $a$ to the input of the \ac{PEF} along any possible paths (\ie through any $S_i$) is $$d_f^{a\rightarrow n} = \min\{d_i; i\in\llbracket 1,N\rrbracket\}$$
    Similarly, 
    $$D_f^{a\rightarrow n} = \max\{D_i; i\in\llbracket 1,N\rrbracket\}$$
    and (\ref{eq:thm:pef-oac:ancestor}) can be written
    \begin{align*}
        \alpha_f^{a\rightarrow n} &= \alpha_{f,a^*} \oslash \delta_{D_f^{a\rightarrow n} - d_f^{a\rightarrow n}} \\
        &= \alpha_f \oslash \delta_{\max_i D_i - \min_i d_i}
    \end{align*}
    We apply Theorem~\ref{thm:toolbox:pef-oac}: $\alpha_{f,\text{PEF}^*} = \alpha_{f,\text{PEF}^{\text{in}}} \otimes \alpha_f^{a\rightarrow n}$ is an arrival curve for $f$ at the output of the \ac{PEF}.
    Replacing with the above expressions for $\alpha_{f,\text{PEF}^{\text{in}}}$ and $\alpha_f^{a\rightarrow n}$ gives Equation~(\ref{eq:toolbox:simplified-proposition}) of Corollary~\ref{cor:toolbox:thm-simple-application}.
\end{proof}

\subsection{Proof of Proposition~\ref{prop:toolbox:simplified-result-is-tight}}
\label{sec:appendix:pef-oac:simplified-result-is-tight}

\begin{proof}[Proof of Proposition~\ref{prop:toolbox:simplified-result-is-tight}]\label{proof:appendix:pef-oac:simplified-result-is-tight}
	
	Take a leaky-bucket arrival curve $\gamma_{r,b}$. 
	And $[d_1,D_1]$, $[d_2,D_2]$ two intervals of $\mathbb{R}^+$.
	
	We first prove the result when $d_2 - D_1 \ge b/r$, we prove the other situation afterwards.

	\subsubsection{Case $d_2 - D_1 \ge b/r$:}
	Applying Corollary~\ref{cor:toolbox:thm-simple-application} with $N=2$ systems $S_1$, $S_2$ with bounded intervals $[d_1,D_1]$ and $[d_2,D_2]$ gives that
	\begin{equation}\label{eq:appendix:tight:alpha-star}\begin{aligned}
		\alpha^* &= \left(\sum_{i\in\llbracket 1,2\rrbracket } \gamma_{r,b}\oslash\delta_{D_i-d_i}\right) \otimes \left( \gamma_{r,b} \oslash\delta_{D_2-d_1}\right)\\
		&=\left(\gamma_{r,b+r(D_1-d_1)} + \gamma_{r,b+r(D_2-d_1)}\right) \otimes \gamma_{r,b+r(D_2-d_1)} \\
		&=\gamma_{2r,2b+r(D_1-d_1+D_2-d_2)}\otimes \gamma_{r,b+r(D_2-d_1)} \\
	\end{aligned}\end{equation}
	is an arrival curve for $f$ after the \ac{PEF} in Figure~\ref{fig:toolbox:oac-prop-simplified-figure}.

	In the following, we exhibit a trajectory with a $\gamma_{r,b}$-constrained source for $f$ and no minimal packet length.
	We exhibit also two systems $S_1, S_2$, in which the delay of the non-lost data-units is in the intervals $[d_1, D_1]$ and $[d_2,D_2]$.
	The proof operates in several steps, as follows:

	\paragraph{Definition of several constants used in the proof}
	We define
	\begin{equation}\label{eq:proof:frer:tight:chi_values}
		\chi^1 \triangleq \left\lceil\frac{r (D_1-d_1) }{b}\right\rceil \qquad \text{and} \qquad \chi^2 \triangleq \left\lceil\frac{r (D_2-d_2) }{b}\right\rceil
	\end{equation}
	(Note that $\chi^1\ge1$ and $\chi^2\ge1$)
	
	Last, we define
	\begin{equation}\label{eq:proof:frer:tight:psi_values}
		\psi = \left\lceil\frac{r (d_2-D_1) -b}{b}\right\rceil
	\end{equation}	
	and we also have $\psi \ge 1$.

	\paragraph{Description of the traffic generation at the source}\label{par:appendix:tightness:traffic-gen}

	For the sake of clarity, we classify the data-units generated by the source into four categories: $I,B,S$ and $X$. 
	The category of a data-unit defines the role that the data-unit has in the trajectory. 
	Each of the three first categories ($I,B,S$) has two sub-categories that we distinguish by using a superscript (e.g., $I^1$ and $I^2$). 
	This sub-category notion is used in order to distinguish the role of every system ($S_1$ or $S_2$) in the trajectory. 
	
	Subcategories do not infer the order with which data-units are generated. The notions of categories and subcategories are only used in the proof, they are not related to any physical property of the packets (neither to their length nor to any field in their header).
	
	\ul{Category $I$:}
	The source generates two ``initiator" data-units: $I^2$ [resp., $I^1$] at absolute time $0$ [resp., $(D_2 - D_1)$], of length $b$ (see Table~\ref{tab:proof:frer:oac:tight:initiators_generation}). 
	Figure~\ref{fig:proof:frer:oac:tight:timeline:initiators_only} shows the timeline of the data-units $I$ out of the source.
	
	\begin{figure}
		\resizebox{\linewidth}{!}{

\begin{tikzpicture}
	\def\xdev{1cm}
	\def\yline{0cm}
	\def\ph{1cm}
	
	\tikzstyle{ab} = [pos=0, rotate=90, anchor=east, black]
	\tikzstyle{pkn} = [black, pos=1,above]
	\tikzstyle{pk} = [->, blue, line width=1.25pt]
	
	\draw[-latex] (-2*\xdev,\yline) -- (10*\xdev, \yline) node[anchor=north east, black] {time};
	\draw[-latex] (-1*\xdev,\yline-0.5cm) -- (-1*\xdev,\yline+2*\ph) node[pos=1, anchor=north east] {size};
	
	\draw[pk]  (0,\yline) -- (0,\ph+\yline) node[ab] {$0$} node[pkn] {$I^2$};
	\draw[pk]  (5*\xdev,\yline) -- (5*\xdev,\ph+\yline) node[ab] {$D_2-D_1$} node[pkn] {$I^1$};
	
	\draw[dotted] (-1*\xdev, \ph) -- (5*\xdev,\ph) node[pos=0,left] {$b$};

\end{tikzpicture}}
		\caption{\label{fig:proof:frer:oac:tight:timeline:initiators_only} Source output in the trajectory achieving the tightness of Corollary~\ref{cor:toolbox:thm-simple-application}. Two "initiator" data units are sent at $0$ and at $D_2-D_1$.}		
	\end{figure}
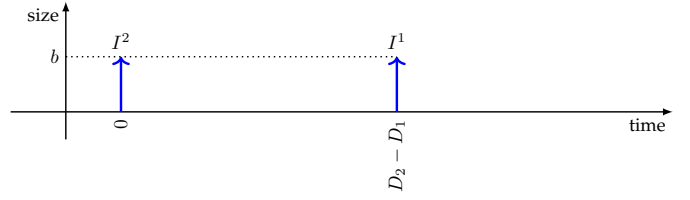

	\begin{table}\centering
		\caption{\label{tab:proof:frer:oac:tight:initiators_generation} Generation of the Data-Units of Category $I$ in the Trajectory that Achieves the Tightness of Corollary~\ref{cor:toolbox:thm-simple-application}.}
\begin{tabular}{r|c|c}
	Data unit $m$ & Size, $\text{size}(m)$ & Generation time, $\mathcal{G}(m)$ \\
	\hline
	$I^2$ & $b$ & $0$ \\
	$I^1$ & $b$ & $D_2-D_1$ \\
\end{tabular}
	\end{table}

	\emph{Note: The role of the two data-units of category $I$ is to initiate the backlog period. In the next parts of the proof, we create a situation where $I^1$ and $I^2$ exit the \ac{PEF} of Figure~\ref{fig:toolbox:oac-prop-simplified-figure} at the same time, creating the $2b$ part of the burst in the term $\gamma_{2r,\mathbf{2b}+r(D_1-d_1+D_2-d_2)}$ of (\ref{eq:appendix:tight:alpha-star}).}

	\ul{Category $B$:}	
	In addition to the data-units of category $I$, the source generates $\chi^2$ data-units of subcategory $B^2$ and $\chi^1$ data-units of subcategory $B^1$, as described in Table~\ref{tab:proof:frer:oac:tight:bursty_generation}.
	\begin{table}\centering
		\caption{\label{tab:proof:frer:oac:tight:bursty_generation} Generation of the Data-Units of Category $B$ in the Trajectory that Achieves the Tightness of Corollary~\ref{cor:toolbox:thm-simple-application}.}
		\resizebox*{\linewidth}{!}{
\begin{tabular}{r|c|c}
	Data unit $m$ & Size, $\text{size}(m)$ & Generation time, $\mathcal{G}(m)$ \\
	\hline
	$\forall k \in \llbracket 1, \chi^2-1\rrbracket, \qquad B_k^2$ & $b$ & $k\frac{b}{r}$ \\
	$B^2_{\chi^2}$ & $r(D_2-d_2) - (\chi^2-1)b$ & $(D_2-d_2)$\\
	\hline
	$\forall k \in \llbracket 1, \chi^1-1\rrbracket, \qquad B_k^1$ & $b$ & $(D_2-D_1) + k\frac{b}{r}$ \\
	$B^1_{\chi^1}$ & $r(D_1-d_1) - (\chi^1-1)b$ & $(D_2-D_1)+(D_1-d_1)$\\
\end{tabular}}
	\end{table}
	A possible output of the source when combining categories $I$ and $B$ is shown in Figure~\ref{fig:proof:frer:oac:tight:timeline:bursty_1}. In the proposed situation, we have $\chi_1=1$ (\emph{i.e.}, $r (D_1-d_1)  \le b$). Then the interval $\llbracket 1, \chi^1 - 1\rrbracket$ in Table~\ref{tab:proof:frer:oac:tight:bursty_generation} is empty and category $B^1$ contains a unique data-unit $B^1_1 = B^1_{\chi^1}$ of size $r (D_1-d_1) $ and sent at time $ (D_2-D_1)  +  (D_1-d_1)  = D_2-d_1$. 
	In Figure~\ref{fig:proof:frer:oac:tight:timeline:bursty_1}, $\chi^2$ equals $2$, and category $B_2$ is made of two data-units: $B_1^2$, of size $b$, released at time $b/r$; and $B_2^2$, of size $r (D_2-d_2) -b$, released at time $ (D_2-d_2) $.
	
	Note that for any value of $\chi^1$, $\chi^2$, by Table~\ref{tab:proof:frer:oac:tight:bursty_generation},
	\begin{equation}\label{eq:proof:frer:oac:tight:sum_cat_b}
		\forall j \in\lbrace 1, 2 \rbrace \qquad \sum_{k\in\llbracket 1, \chi^j \rrbracket} \text{size}(B_k^j) = r(D_j-d_j)
	\end{equation}
	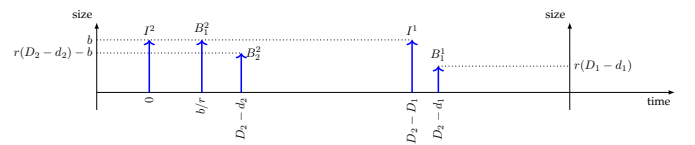
\begin{figure}
		\resizebox{\linewidth}{!}{

\begin{tikzpicture}
	\tikzstyle{ab} = [pos=0, rotate=90, anchor=east, black]
	\tikzstyle{pkn} = [black, pos=1,above]
	\tikzstyle{pk} = [->, blue, line width=1.25pt]
	\def\xdev{1.5cm}
	\def\yline{0cm}
	\def\ph{1.5cm}
	
	\draw[-latex] (-1*\xdev,\yline) -- (10*\xdev, \yline) node[black, anchor=north east, pos=1] {time};
	\draw[-latex] (-1*\xdev,\yline-0.5cm) -- (-1*\xdev,\yline+\ph+0.5cm) node[pos=1, anchor=south east] {size};
	\draw[-latex] (8*\xdev,\yline-0.5cm) -- (8*\xdev,\yline+\ph+0.5cm) node[pos=1, anchor=south east] {size};
	\draw[pk]  (0,\yline) -- (0,\ph+\yline) node[ab] {$0$} node[pkn] {$I^2$};
	\draw[pk]  (5*\xdev,\yline) -- (5*\xdev,\ph+\yline) node[ab] {$D_2-D_1$} node[pkn] {$I^1$};
	
	\draw[pk]  (1*\xdev,\yline) -- (1*\xdev,\ph+\yline) node[ab] {$b/r$} node[pkn] {$B^2_1$};
	\draw[pk]  (1.75*\xdev,\yline) -- (1.75*\xdev,0.75*\ph+\yline) node[ab] {$D_2-d_2$} node[pkn,anchor=west] {$B^2_2$};
	
	\draw[pk]  (5.5*\xdev,\yline) -- (5.5*\xdev,0.5*\ph+\yline) node[ab] {$D_2-d_1$} node[pkn] {$B^1_1$};

	\draw[dotted] (-1*\xdev, \ph) -- (5*\xdev,\ph) node[pos=0,left] {$b$};
	\draw[dotted] (5.5*\xdev, 0.5*\ph+\yline) -- (8*\xdev,0.5*\ph+\yline) node[pos=1,right] {$r(D_1-d_1)$};
	\draw[dotted] (-1*\xdev, 0.75*\ph) -- (1.75*\xdev,0.75*\ph) node[pos=0,left] {$r(D_2-d_2)-b$};
	
\end{tikzpicture}}
		\caption{\label{fig:proof:frer:oac:tight:timeline:bursty_1}Example of the source output in the trajectory achieving the tightness, focusing on categories $I$ and $B$.}		
	\end{figure}
	
	\emph{Note: The role of the data-units of category $B$ is to participate in the burst term of $\gamma_{2r,2b+r(D_1-d_1+D_2-d_2)}$ in (\ref{eq:appendix:tight:alpha-star}).
	In the next parts of the proof, we create a situation where all data-units of category $B$ (both subcategories $B^1$ and $B^2$) are released at the same time, simultaneously with data-units $I^1$ and $I^2$. This give the part $r(D_1-d_1+D_2-d_2)$ in the burst term of $\gamma_{2r,2b+\mathbf{r(D_1-d_1+D_2-d_2)}}$.} 
	
	We now prove that data-units of subcategory $B^1$ [resp., $B^2$] are generated after data-units of subcategory $I^1$ [resp., $I^2$] and in the order of their lower-script index.

	$\bullet$ If $\chi^2 = 1$, then the first data-unit of $B^2$ is sent at $ D_2-d_2 $ and $ D_2-d_2 \ge0$, so data-unit of $B^2$ is generated after the data-unit of $I^2$.
	
	$\bullet$ If $\chi^2 \ge 1$, then the first data-unit of $B^2$ is sent at $b/r\ge0$, so data-units of $B^2$ are generated after the data-units of $I^2$. Also, the data-units $(B^2_k)_{\kappa\in\llbracket 1, \chi^2\rrbracket}$  are generated in the same order as their index: this is clear for indexes up to $\chi^2 - 1$. For the order between $B^2_{\chi^2-1}$ and $B^2_{\chi^2}$, we note that $\left\lceil\frac{r (D_2-d_2) }{b}\right\rceil \frac{b}{r} - \frac{b}{r} \le D_2 - d_2$ by property of the ceiling function.
	
	$\bullet$ If $\chi^1 = 1$, then the first data-unit of $B^1$ is sent at $(D_2-d_1)$ thus after the data-unit of $I^1$ (as $D_1 \ge d_1$).
	
	$\bullet$ If $\chi^1 \ge 1$, then by using the same reasoning for $B^2$, we obtain that data-units of $B^1$ are generated after the initiator $I^1$ and they are released in the order of their lower-script index.
	
	\ul{Category $S$:} In addition to the data-units of categories $I$ and $B$, the source generates $\psi$ data-units of subcategory $S^2$ and $\psi$ data-units of subcategory $S^1$, as described in Table~\ref{tab:proof:frer:oac:tight:s_generation}.
	A possible output of the source when adding category $S$ to Figure~\ref{fig:proof:frer:oac:tight:timeline:bursty_1} is shown in Figure~\ref{fig:proof:frer:oac:tight:timeline:three_categories}. In the proposed situation, we have $\psi = 2$ and subcategories $S^1$ and $S^2$ are both made of two data-units: $S^1_1$ [resp., $S^2_1$], of size $b$, released $\frac{b}{r}$ after the last data-unit of $B^1$ [resp., $B^2$] and $S^1_2$ [resp., $S^2_2$], of size $r (d_2-D_1) -2b$, released at $(D_2+d_2) - (D_1+d_1)-\frac{b}{r}$ [resp., $D_2-D_1-\frac{b}{r}$].
	\begin{table}\centering
		\caption{\label{tab:proof:frer:oac:tight:s_generation} Generation of the Data-Units of Category $S$ in the Trajectory that Achieves the Tightness of Corollary~\ref{cor:toolbox:thm-simple-application}.}
		\resizebox*{\linewidth}{!}{
\begin{tabular}{r|c|c}
	Data unit $m$ & Size, $\text{size}(m)$ & Generation time, $\mathcal{G}(m)$ \\
	\hline
	$\forall k \in \llbracket 1, \psi-1\rrbracket, \qquad S_k^2$ & $b$ & $ (D_2-d_2)+k\frac{b}{r}$ \\
	$S^2_{\psi}$ & $r(d_2-D_1)-\psi b$  & $(D_2-D_1-\frac{b}{r})$\\
	\hline
	$\forall k \in \llbracket 1, \psi-1\rrbracket, \qquad S_k^1$ & $b$ & $(D_2-d_1) + k\frac{b}{r}$ \\
	$S^1_{\psi}$ & $r(d_2-D_1)-\psi b$ & $(D_2+d_2)-(D_1+d_1)-\frac{b}{r}$\\
\end{tabular}}
	\end{table}
	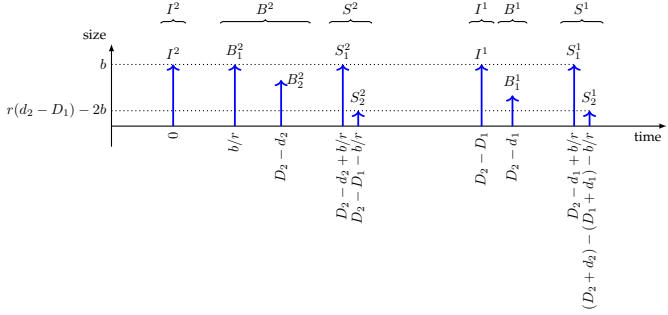
\begin{figure}
		\resizebox{\linewidth}{!}{

\begin{tikzpicture}
	\tikzstyle{ab} = [pos=0, rotate=90, anchor=east, black]
	\tikzstyle{pkn} = [black, pos=1,above]
	\tikzstyle{pk} = [->, blue, line width=1.25pt]
	\tikzstyle{nb} = [pos=0.5, above, yshift=0.1cm]
	\tikzstyle{dld} = [xshift=-0cm]
	\tikzstyle{drd} = [xshift=0cm]
	\def\xdev{1.5cm}
	\def\yline{0cm}
	\def\ph{1.5cm}
	\def\ybrace{\yline+2.5cm}

	\draw[-latex] (-1*\xdev,\yline) -- (8*\xdev, \yline) node[black, anchor=north east] {time};
	\draw[-latex] (-1*\xdev,\yline-0.5cm) -- (-1*\xdev,\yline+\ph+0.5cm) node[pos=1, anchor=south east] {size};
	\draw[pk]  (0,\yline) -- (0,\ph+\yline) node[ab] {$0$} node[pkn] (i2l) {$I^2$};
	\draw[pk]  (5*\xdev,\yline) -- (5*\xdev,\ph+\yline) node[ab] {$D_2-D_1$} node[pkn] (i1l) {$I^1$};
	
	\draw[pk]  (1*\xdev,\yline) -- (1*\xdev,\ph+\yline) node[ab] {$b/r$} node[pkn] (b21l) {$B^2_1$};
	\draw[pk]  (1.75*\xdev,\yline) -- (1.75*\xdev,0.75*\ph+\yline) node[ab] {$D_2-d_2$} node[pkn,anchor=west] (b22l) {$B^2_2$};
	
	\draw[pk]  (5.5*\xdev,\yline) -- (5.5*\xdev,0.5*\ph+\yline) node[ab] {$D_2-d_1$} node[pkn] (b11l) {$B^1_1$};
	
	\draw[pk]  (2.75*\xdev,\yline) -- (2.75*\xdev,\ph+\yline) node[ab] {$D_2-d_2+b/r$} node[pkn] (s21l) {$S^2_1$};
	\draw[pk]  (3*\xdev,\yline) -- (3*\xdev,0.25*\ph+\yline) node[ab] {$D_2-D_1-b/r$} node[pkn] (s22l) {$S^2_2$};

	\draw[pk]  (6.5*\xdev,\yline) -- (6.5*\xdev,\ph+\yline) node[ab] {$D_2-d_1+b/r$} node[pkn] (s11l) {$S^1_1$};
	\draw[pk]  (6.75*\xdev,\yline) -- (6.75*\xdev,0.25*\ph+\yline) node[ab] {$(D_2+d_2)-(D_1+d_1)-b/r$} node[pkn] (s12l) {$S^1_2$};
	
	\draw[dotted] (-1*\xdev, \ph) -- (6.5*\xdev,\ph) node[pos=0,left] {$b$};
	\draw[dotted] (-1*\xdev, 0.25*\ph) -- (6.75*\xdev,0.25*\ph) node[pos=0,left] {$r(d_2-D_1)-2b$};

	\draw[decorate, decoration={brace}] let \p1 = (i2l.north west), \p2 = (i2l.north east) in ([dld] \x1,\ybrace) -- ([drd] \x2,\ybrace) node[nb] {$I^2$};
	\draw[decorate, decoration={brace}] let \p1 = (b21l.north west), \p2 = (b22l.north east) in ([dld] \x1,\ybrace) -- ([drd] \x2,\ybrace) node[nb] {$B^2$};
	\draw[decorate, decoration={brace}] let \p1 = (s21l.north west), \p2 = (s22l.north east) in ([dld] \x1,\ybrace) -- ([drd] \x2,\ybrace) node[nb]  {$S^2$};
	
	\draw[decorate, decoration={brace}] let \p1 = (i1l.north west), \p2 = (i1l.north east) in ([dld] \x1,\ybrace) -- ([drd] \x2,\ybrace) node[nb] {$I^1$};
	\draw[decorate, decoration={brace}] let \p1 = (b11l.north west), \p2 = (b11l.north east) in ([dld] \x1,\ybrace) -- ([drd] \x2,\ybrace) node[nb] {$B^1$};
	\draw[decorate, decoration={brace}] let \p1 = (s11l.north west), \p2 = (s12l.north east) in ([dld] \x1,\ybrace) -- ([drd] \x2,\ybrace) node[nb] {$S^1$};
	
\end{tikzpicture}}
		\caption{\label{fig:proof:frer:oac:tight:timeline:three_categories}Source output, with the three categories $I,B$ and $S$ of data-units}		
	\end{figure}

	Note that for any value of $\psi$, by Table~\ref{tab:proof:frer:oac:tight:s_generation},
	\begin{equation}\label{eq:proof:frer:oac:tight:sum_cat_s}
		\forall j \in\lbrace 1, 2 \rbrace \qquad \sum_{k\in\llbracket 1, \psi \rrbracket} \text{size}(S_k^j) = r (d_2-D_1)  - b
	\end{equation}
	
	\emph{Note: 
	The role of the data-units of category $S$ is to participate in the peak-rate term of $\gamma_{\mathbf{2r},2b+r(D_1-d_1+D_2-d_2)}$ in (\ref{eq:appendix:tight:alpha-star}). The output traffic should maintain the peak rate for a sufficient duration so that the obtained cumulative output intersects with the curve $\gamma_{r,b+r(D_2-d_1)}$. 
	In the next parts of the proof, we create a situation where each data unit of subcategory $S^1$ is released at the same time as its peer of subcategory $S^2$. 
	This creates a peak rate $2r$ for a duration of at least $d_2-D_1-b/r$.
	The resulting cumulative output intersects the curve $\gamma_{r,b+r(D_2-d_1)}$.
	} 
	
	We now check the order of the data data units of category $S$.

	$\bullet$ If $\psi = 1$, then the first data-unit of $S^2$ is sent at $D_2-D_1-b/r$, whereas the last data-unit of $B^2$ was sent at $ D_2-d_2 $. By assumption, $d_2-D_1 \ge b/r$, so $D_2-D_1-b/r \ge D_2 - d_2$ and the first data-unit of $S^2$ is sent after the last data-unit of $B^2$.
	
	$\bullet$ If $\psi \ge 1$, then the first data-unit of $S^2$ is sent $b/r$ after the last data-unit of $B^2$. Also, the data-units $(S^2_k)_{\kappa\in\llbracket 1, \psi\rrbracket}$  are generated in the same order as their index: this is clear for indexes up to $\psi - 1$.
	For the order between $S^2_{\psi-1}$ and $S^2_{\psi}$, we note that $(\psi-1) \frac{b}{r} \le  (d_2-D_1)  - \frac{b}{r}$ by properties of the ceiling function and so $  (D_2-d_2)  + (\psi-1)\frac{b}{r} \le  (D_2-d_2) + (d_2 - D_1) - \frac{b}{r}$, \emph{i.e.}, $ D_2-d_2  + (\psi-1) \frac{b}{r} \le D_2-D_1 - \frac{b}{r}$, \emph{i.e},  $S^2_{\psi-1}$ is sent before $S^2_{\psi}$
	
	We then apply the same principles for $S^1$. We thus prove that data-units of subcategory $S^1$ [resp., $S^2$] are generated after data-units of subcategory $B^1$ [resp., $B^2$] and in the order of their index.  We also observe that the last data-unit of subcategory $S^2$ is sent at $D_2 - D_1 - \frac{b}{r}$, whereas the data-unit of subcategory $I^1$ is sent at $D_2-D_1$, hence data-units of subcategory $S^2$ are sent before the data-unit of subcategory $I^1$.
	
	\ul{Category $X$:} After the data-units of subcategory $S^1$, the source generates for eternity data-units $(X_n)_{n\in\mathbb{N}^*}$ of size $b$ with a period $b/r$ (see Table~\ref{tab:proof:frer:oac:tight:x_generation}). The first one of these data-units is sent $b/r$ after the last data-unit of $S^1$.
	Figure~\ref{fig:proof:frer:oac:tight:timeline:four_categories} presents the output of the source with all four categories.
	\begin{table}\centering
		\caption{\label{tab:proof:frer:oac:tight:x_generation} Generation of the Data-units of Category $X$ in the Trajectory that Achieves the Tightness of Corollary~\ref{cor:toolbox:thm-simple-application}.}
		\resizebox{\linewidth}{!}{
\begin{tabular}{r|c|c}
	Data unit $m$ & Size, $\text{size}(m)$ & Generation time, $\mathcal{G}(m)$ \\
	\hline
	$\forall k \in \mathbb{N}^*, \qquad X_k$ & $b$ & $(D_2+d_2)-(D_1+d_1) + (k-1)\frac{b}{r}$\\
\end{tabular}}
	\end{table}
	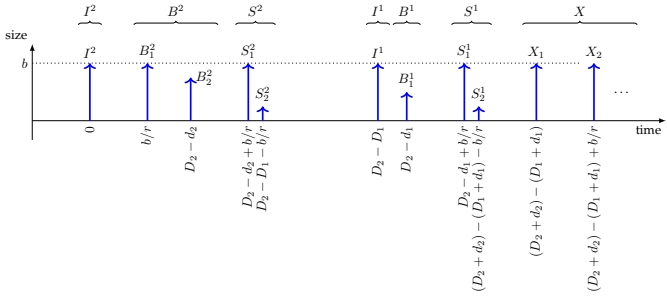
\begin{figure}
		\resizebox{\linewidth}{!}{

\begin{tikzpicture}
	\tikzstyle{ab} = [pos=0, rotate=90, anchor=east, black]
	\tikzstyle{pkn} = [black, pos=1,above]
	\tikzstyle{pk} = [->, blue, line width=1.25pt]
	\tikzstyle{nb} = [pos=0.5, above, yshift=0.1cm]
	\tikzstyle{dld} = [xshift=-0cm]
	\tikzstyle{drd} = [xshift=0cm]
	\def\xdev{1.5cm}
	\def\yline{0cm}
	\def\ph{1.5cm}
	\def\ybrace{\yline+2.5cm}

	\draw[-latex] (-1*\xdev,\yline) -- (10*\xdev, \yline) node[black, anchor=north east] {time};
	\draw[-latex] (-1*\xdev,\yline-0.5cm) -- (-1*\xdev,\yline+\ph+0.5cm) node[pos=1, anchor=south east] {size};
	\draw[pk]  (0,\yline) -- (0,\ph+\yline) node[ab] {$0$} node[pkn] (i2l) {$I^2$};
	\draw[pk]  (5*\xdev,\yline) -- (5*\xdev,\ph+\yline) node[ab] {$D_2-D_1$} node[pkn] (i1l) {$I^1$};
	
	\draw[pk]  (1*\xdev,\yline) -- (1*\xdev,\ph+\yline) node[ab] {$b/r$} node[pkn] (b21l) {$B^2_1$};
	\draw[pk]  (1.75*\xdev,\yline) -- (1.75*\xdev,0.75*\ph+\yline) node[ab] {$D_2-d_2$} node[pkn,anchor=west] (b22l) {$B^2_2$};
	
	\draw[pk]  (5.5*\xdev,\yline) -- (5.5*\xdev,0.5*\ph+\yline) node[ab] {$D_2-d_1$} node[pkn] (b11l) {$B^1_1$};
	
	\draw[pk]  (2.75*\xdev,\yline) -- (2.75*\xdev,\ph+\yline) node[ab] {$D_2-d_2+b/r$} node[pkn] (s21l) {$S^2_1$};
	\draw[pk]  (3*\xdev,\yline) -- (3*\xdev,0.25*\ph+\yline) node[ab] {$D_2-D_1-b/r$} node[pkn] (s22l) {$S^2_2$};

	\draw[pk]  (6.5*\xdev,\yline) -- (6.5*\xdev,\ph+\yline) node[ab] {$D_2-d_1+b/r$} node[pkn] (s11l) {$S^1_1$};
	\draw[pk]  (6.75*\xdev,\yline) -- (6.75*\xdev,0.25*\ph+\yline) node[ab] {$(D_2+d_2)-(D_1+d_1)-b/r$} node[pkn] (s12l) {$S^1_2$};

	\draw[pk]  (7.75*\xdev,\yline) -- (7.75*\xdev,\ph+\yline) node[ab] {$(D_2+d_2)-(D_1+d_1)$} node[pkn] (x1) {$X_1$};
	\draw[pk]  (8.75*\xdev,\yline) -- (8.75*\xdev,\ph+\yline) node[ab] {$(D_2+d_2)-(D_1+d_1)+b/r$} node[pkn] (x2) {$X_2$};
    \node at (9.25*\xdev,0.5*\ph+\yline) {\ldots};
	
	\draw[dotted] (-1*\xdev, \ph) -- (8.5*\xdev,\ph) node[pos=0,left] {$b$};

	\draw[decorate, decoration={brace}] let \p1 = (i2l.north west), \p2 = (i2l.north east) in ([dld] \x1,\ybrace) -- ([drd] \x2,\ybrace) node[nb] {$I^2$};
	\draw[decorate, decoration={brace}] let \p1 = (b21l.north west), \p2 = (b22l.north east) in ([dld] \x1,\ybrace) -- ([drd] \x2,\ybrace) node[nb] {$B^2$};
	\draw[decorate, decoration={brace}] let \p1 = (s21l.north west), \p2 = (s22l.north east) in ([dld] \x1,\ybrace) -- ([drd] \x2,\ybrace) node[nb]  {$S^2$};
	
	\draw[decorate, decoration={brace}] let \p1 = (i1l.north west), \p2 = (i1l.north east) in ([dld] \x1,\ybrace) -- ([drd] \x2,\ybrace) node[nb] {$I^1$};
	\draw[decorate, decoration={brace}] let \p1 = (b11l.north west), \p2 = (b11l.north east) in ([dld] \x1,\ybrace) -- ([drd] \x2,\ybrace) node[nb] {$B^1$};
	\draw[decorate, decoration={brace}] let \p1 = (s11l.north west), \p2 = (s12l.north east) in ([dld] \x1,\ybrace) -- ([drd] \x2,\ybrace) node[nb] {$S^1$};

    \draw[decorate, decoration={brace}] let \p1 = (x1.north west), \p2 = (x2.north east) in ([dld] \x1,\ybrace) -- (9.5*\xdev,\ybrace) node[nb] {$X$};
\end{tikzpicture}}
		\caption{\label{fig:proof:frer:oac:tight:timeline:four_categories}Source output, with the four categories of data units: $I,B, S$ and $X$ of data-units, in the trajectory achieving the tightness of Corollary~\ref{cor:toolbox:thm-simple-application}.}		
	\end{figure}
	
	\emph{Note: The role of category $X$ is to generate the sustained rate term in the curve $\gamma_{\mathbf{r},b+r(D_2-d_1)}$ in (\ref{eq:appendix:tight:alpha-star}).}

	\paragraph{Properties of the traffic generation at the source}
	Now that we have described the profile of the traffic generated by the source, we can prove that the generation is $\gamma_{r,b}$-compliant.
	This is done with the following lemmas.
	
	\begin{lemma}[Size of the data-units]\label{lemma:frer:oac:tight:packet_sizes}
		For any data-unit $P$ described in Paragraph~\ref{par:appendix:tightness:traffic-gen}, $ 0 \le \text{size}(P) \le b$
	\end{lemma}

	\begin{proof}[Proof of Lemma~\ref{lemma:frer:oac:tight:packet_sizes}]
		We focus on certain data-units.

		$\bullet$ $B^2_{\chi^2}$: By property of the ceiling function:
		\begin{equation*}\resizebox{\linewidth}{!}{$\begin{aligned}
			 \frac{r (D_2-d_2) }{b} &\le \chi^2 &&\le \frac{r (D_2-d_2) }{b} + 1 \\
			 r (D_2-d_2)  - b &\le (\chi^2 - 1)b &&\le r (D_2-d_2)  \qquad\text{\Comment{}}b > 0 \\
			 -r (D_2-d_2)  + b &\ge -(\chi^2 - 1)b &&\ge -r (D_2-d_2)  \\
			 b &\ge r (D_2-d_2) -(\chi^2 - 1)b &&\ge 0 \\
		\end{aligned}$}\end{equation*}
		
		$\bullet$ $B^1_{\chi^1}$: same idea
		
		$\bullet$ $S^1_{\psi}$ and $S^2_{\psi}$ (they have the same size): By property of the ceiling function:
		\begin{equation*}\resizebox{\linewidth}{!}{$\begin{aligned}
			\frac{r}{b}  (d_2-D_1)  - 1 &\le \psi &&\le \frac{r}{b} (d_2-D_1)\\
			-r (d_2-D_1)  + b & \ge -\psi b &&\ge -r (d_2-D_1)  \qquad\text{\Comment{}} b > 0\\
			 b & \ge r (d_2-D_1)  - \psi b &&\ge 0
		\end{aligned}$}\end{equation*}

		All the other data-units have the same size, equal to the burst $b$.
	\end{proof}
	
	\begin{lemma}[Minimum time distance between two successive data-units]\label{lemma:frer:oac:tight:min_dist}
		Consider two successive data-units $m,m'$ in the traffic described in Paragraph~\ref{par:appendix:tightness:traffic-gen}, \emph{i.e.}, $m'$ is the first data-unit sent after $m$. Note $\mathcal{G}(m)$ and $\mathcal{G}(m')$ the time at which they are generated.
		
		Then $\mathcal{G}(m') - \mathcal{G}(m) \ge \frac{\text{size}(m')}{r}$
	\end{lemma}

	\begin{proof}[Proof of Lemma~\ref{lemma:frer:oac:tight:min_dist}]
		We simply describe all the possible combinations:
		
		$\bullet$ Case $m = I^2$ and $m' = B^2_1$:
		
		We have $\mathcal{G}(m)=0$ (see Table~\ref{tab:proof:frer:oac:tight:initiators_generation}).
		If $\chi^2 = 1$, then $\mathcal{G}(m') - \mathcal{G}(m) =  (D_2-d_2) $, $\text{size}(m') = r (D_2-d_2) $ (see Table~\ref{tab:proof:frer:oac:tight:bursty_generation}) and the result holds.
		If $\chi^2 \ge 2$, then $\mathcal{G}(m') - \mathcal{G}(m) = b/r$, $\text{size}(m') = b$ and the result holds.
		
		$\bullet$ Case $m,m' \in B^2$ (when $\chi^2 \ge 2$):
		
		There exists $k \in \llbracket 1, \chi^2-1 \rrbracket$ such that $m=B^2_k$ and $m'=B^2_{k+1}$. If $k \le \chi^2 - 2$, then $\mathcal{G}(m') - \mathcal{G}(m) = b/r$, $\text{size}(m') = b$ and the result holds. If $k = \chi^2 - 1$, then:
		\begin{align*}
			\mathcal{G}(m') - \mathcal{G}(m) &=  (D_2-d_2)  - k\frac{b}{r} \\
			& =  (D_2-d_2)  - (\chi^2 - 1) \frac{b}{r}\\	
			& = r \cdot \text{size}(m')		
		\end{align*}
		
		$\bullet$ Case $m = B^2_{\chi^2}$ and $m'=S^2_1$:
		
		We have $\mathcal{G}(m) = D_2 - d_2$ (Table~\ref{tab:proof:frer:oac:tight:bursty_generation}). 
		If $\psi = 1$, then $m' = S^2_{\psi}$ and
		\begin{align*}
			 \mathcal{G}(m') - \mathcal{G}(m) &= D_2 - D_1 - \frac{b}{r} -  (D_2 - d_2)\\ 
			 &= d_2 - D_1 - \frac{b}{r} \\
			 &= \frac{r (d_2-D_1)  - \psi b}{r} \\
			 &= \text{size}(m')/r
		\end{align*}
		If $\psi \ge 2$, then $\mathcal{G}(m') - \mathcal{G}(m)=b/r$ (see Tables~\ref{tab:proof:frer:oac:tight:bursty_generation} and \ref{tab:proof:frer:oac:tight:s_generation}) and $\text{size}(m') = b$ so the result holds.
		
		$\bullet$ Case $m,m' \in S^2$ (when $\psi \ge 2$):
		
		There exists $k\in\llbracket 1, \psi -1 \rrbracket$ such that $m=S^2_k$ and $m'=S^2_{k+1}$. If $k\le \psi - 2$, then $\mathcal{G}(m') - \mathcal{G}(m) = b/r$ and $\text{size}(m') = b$ so the result holds. If $k=\psi-1$, then
		\begin{align*}
			\mathcal{G}(m') - \mathcal{G}(m) &= d_2 - D_1 - \psi \frac{b}{r}\\
			& = \text{size}(S^2_{\psi})/r = \text{size}(m')/r \\
		\end{align*}
	
		$\bullet$ Case $m = S^2_{\psi}$ and $m'=I^1$:
		
		Then we have $\mathcal{G}(m) = D_2-D_1-\frac{b}{r}$ (Table~\ref{tab:proof:frer:oac:tight:s_generation}) and $\mathcal{G}(m') = (D_2-D_1)$ (Table~\ref{tab:proof:frer:oac:tight:initiators_generation}). So $\mathcal{G}(m') - \mathcal{G}(m) = b/r = \text{size}(m')/r$.
		
		$\bullet$ Case $m=I^1$ and $m'\in B^1$:
		
		We have $\mathcal{G}(m)=(D_2-D_1)$ (see Table~\ref{tab:proof:frer:oac:tight:initiators_generation}).
		If $\chi^1 = 1$, then $\mathcal{G}(m') - \mathcal{G}(m) = (D_1-d_1)$, $\text{size}(m') = r (D_1-d_1) $ (see Table~\ref{tab:proof:frer:oac:tight:bursty_generation}) and the result holds.
		If $\chi^1 \ge 2$, then $\mathcal{G}(m') - \mathcal{G}(m) = b/r$, $\text{size}(m') = b$ and the result holds also.
		
		$\bullet$ Case $m,m' \in B^1$ (when $\chi^1 \ge 2$):
		
		There exists $k \in \llbracket 1, \chi^1-1 \rrbracket$ such that $m=B^1_k$ and $m'=B^1_{k+1}$. If $k \le \chi^1 - 2$, then $\mathcal{G}(m') - \mathcal{G}(m) = b/r$, $\text{size}(m') = b$ and the result holds. If $k = \chi^1 - 1$, then
		\begin{equation*}\resizebox*{\linewidth}{!}{$\begin{aligned}
		\mathcal{G}(m') - \mathcal{G}(m) &=  (D_2-D_1)  +  (D_1-d_1)  -  (D_2-D_1)  -  k\frac{b}{r} \\
		& =  (D_1-d_1)  - (\chi^1 - 1) \frac{b}{r}\\	
		& = r \cdot \text{size}(m')		
		\end{aligned}$}\end{equation*}

		$\bullet$ Case $m = B^1_{\chi^1}$ and $m'=S^1_1$:
		
		We have $\mathcal{G}(m) = (D_2-d_1)$ (Table~\ref{tab:proof:frer:oac:tight:bursty_generation}). 
		If $\psi = 1$, then $m' = S^1_{\psi}$ and
		\begin{align*}
		\mathcal{G}(m') - \mathcal{G}(m) &= D_2 + d_2 - D_1 - d_1 - b/r - D_2 + d_1 \\ 
		&= d_2 - D_1 - b/r \\
		&= \text{size}(m')/r
		\end{align*}
		If $\psi \ge 2$, then $\mathcal{G}(m') - \mathcal{G}(m)=b/r$ (see Tables~\ref{tab:proof:frer:oac:tight:bursty_generation} and \ref{tab:proof:frer:oac:tight:s_generation}) and $\text{size}(m') = b$ so the result holds.
		
		$\bullet$ Case $m,m' \in S^1$ (when $\psi \ge 2$):
		
		There exists $k\in\llbracket 1, \psi -1 \rrbracket$ such that $m=S^1_k$ and $m'=S^1_{k+1}$. If $k\le \psi - 2$, then $\mathcal{G}(m') - \mathcal{G}(m) = b/r$ and $\text{size}(m') = b$ so the result holds. If $k=\psi-1$, then
		\begin{align*}
		\mathcal{G}(m') - \mathcal{G}(m) &= d_2 - D_1 - \psi b/r\\
		&= \text{size}(m')/r \\
		\end{align*}
		
		$\bullet$ Case $m \in S^1$, $m' \in X$: clear (Table~\ref{tab:proof:frer:oac:tight:x_generation})
		
		$\bullet$ Case $m,m' \in X$: clear as well (Table~\ref{tab:proof:frer:oac:tight:x_generation})
	\end{proof}
	
	\begin{lemma}[The source described in Paragraph~\ref{par:appendix:tightness:traffic-gen} complies with the arrival curve $\gamma_{r,b}$]\label{lemma:fre:oac:tight:source_complies}
		The traffic generation described in Paragraph~\ref{par:appendix:tightness:traffic-gen} is $\gamma_{r,b}$-constrained.
	\end{lemma}
	\begin{proof}[Proof of Lemma~\ref{lemma:fre:oac:tight:source_complies}]
		Consider any set of $n$ consecutive data-units generated by the source $(P_v)_{v\in\llbracket 1,n\rrbracket}$. Then
		\begin{align*}
			\mathcal{G}(P_n) - \mathcal{G}(P_1) &\ge \sum_{v\in\llbracket 1,n-1 \rrbracket} \mathcal{G}(P_{v+1}) - \mathcal{G}(P_v)\\
			&\ge  \sum_{v\in\llbracket 1,n-1 \rrbracket} \frac{\text{size}(P_{v+1})}{r} \quad\text{\Comment{Lemma~\ref{lemma:frer:oac:tight:min_dist}}}\\
			r(\mathcal{G}(P_n) - \mathcal{G}(P_1)) + b &\ge \sum_{v\in\llbracket 1,n-1 \rrbracket} \text{size}(P_{v+1}) + b\\
			&\ge \sum_{v\in\llbracket 2,n \rrbracket} \text{size}(P_{v}) + b
		\end{align*}
		Applying Lemma~\ref{lemma:frer:oac:tight:min_dist}, $\text{size}(P_1) \le b$, so we obtain
		\begin{equation}\label{eq:proof:oac:tight:source_max_plus_ac}
			r(\mathcal{G}(P_n) - \mathcal{G}(P_1)) + b \ge \sum_{v\in\llbracket 1,n \rrbracket} \text{size}(P_{v})
		\end{equation}
		Equation~(\ref{eq:proof:oac:tight:source_max_plus_ac}) is the max-plus representation of a packetized flow constrained by an arrival curve $\gamma_{r,b}$~\cite[\S 3]{leboudecNetworkCalculusTheory2001}.
	\end{proof}

	\paragraph{Description of the systems $S_1$, $S_2$}
	For a data unit $m$, we note $\mathcal{E}_1(m)$ [resp., $\mathcal{E}_2(m)$] the absolute time at which the packet transporting $m$ through $S_1$ [resp., through $S_2$], exits $S_1$ [resp., exits $S_2$]. 
	For $j\in\lbrace1,2\rbrace$, we note $\mathcal{E}_j(m) = +\infty$ if and only if the packet transporting data unit $m$ through $S_j$ is lost by system $S_j$.
	Systems $S_1$ and $S_2$ release the packets generated by the source at the time instants shown in Table~\ref{tab:proof:frer:oac:tight:output_time}.
	\begin{table*}\centering
		\caption{\label{tab:proof:frer:oac:tight:output_time} Absolute Release Time for Each Packet at the Output of Each System $S_1$, $S_2$}
		\resizebox{\linewidth}{!}{
\begin{tabular}{r|c|c|c}
	Data unit, $m$ & $\mathcal{G}(m)$, generation time & $\mathcal{E}_1(m)$, exit time out of $S_1$ for the packet  & $\mathcal{E}_2(m)$, exit time out of $S_2$ for the packet\\
	& & transporting $m$ through $S_1$ & transporting $m$ through $S_2$\\
	\hline
	$I^2$ & $0$ & $+\infty$ & $D_2$ \\
	$\forall k \in \llbracket 1, \chi^2-1\rrbracket, \qquad B_k^2$ & $k\frac{b}{r}$ & $+\infty$ & $D_2$ \\
	$B^2_{\chi^2}$ & $D_2-d_2$ & $+\infty$ & $D_2$ \\
	$\forall k \in \llbracket 1, \psi-1\rrbracket, \qquad S_k^2$ & $ D_2-d_2+k\frac{b}{r}$ & $+\infty$ & $D_2+k\frac{b}{r}$ \\
	$S_{\psi}^2$ & $D_2-D_1-\frac{b}{r}$ & $+\infty$ & $D_2 + d_2-D_1-\frac{b}{r}$ \\
	\hline
	$I^1$ & $D_2-D_1$ & $D_2$ & $+\infty$ \\
	$\forall k \in \llbracket 1, \chi^1-1\rrbracket, \qquad B_k^1$ & $D_2-D_1 + k\frac{b}{r}$ & $D_2$ & $+\infty$ \\
	$B_{\chi^1}^1$ & $D_2 - d_1$ & $D_2$ & $+\infty$ \\
	$\forall k \in \llbracket 1, \psi-1\rrbracket, \qquad S_k^1$ & $D_2 - d_1 + k\frac{b}{r}$ & $D_2+k\frac{b}{r}$ & $+\infty$ \\
	$S_{\psi}^1$ & $D_2 + d_2 - D_1 - d_1 - \frac{b}{r}$ &$D_2 + d_2-D_1-\frac{b}{r}$ & $+\infty$ \\
	$\forall k \in \mathbb{N}^*, \qquad X_k$ & $D_2 + d_2 - D_1 - d_1 + (k-1)\frac{b}{r}$ & $D_2 + d_2-D_1+ (k-1)\frac{b}{r}$ & $+\infty$\\
\end{tabular}}
	\end{table*}

	\emph{Remark: We chose a system such that any packet transporting a data unit of category $I^2,B^2$ or $S^2$ is lost within $S_1$ ($\mathcal{E}_1(m)=+\infty$) and any packet transporting a data unit of category $I^1,B^1,S^1$ or $X$ is lost within $S_2$ ($\mathcal{E}_2(m)=+\infty$). This scheme keeps the proof simple but note that a similar proof could be obtained assuming $S_2$ is lossless. In that case, we would only need to make sure that, for $m$ in category $I^1,B^1,S^1$ or $X$, the packet transporting $m$ through $S_2$ exits $S_2$ after the packet transporting $m$ through $S_1$ exits $S_1$, \emph{i.e.}, $\mathcal{E}_2(m) \ge \mathcal{E}_1(m)$ for any $m$.}
	
	As an illustration, Figure~\ref{fig:proof:frer:oac:tight:cumul-out} shows the obtained cumulative function at the output of the \ac{PEF} when applying the exit time instants of Table~\ref{tab:proof:frer:oac:tight:output_time} on the example of Figure~\ref{fig:proof:frer:oac:tight:timeline:four_categories}. Dashed boxes represent values of interest that are further detailed in Paragraph~\ref{par:tight:oaf}. We also plot on top of it the arrival curve at the output of the \ac{PEF}, obtained in (\ref{eq:appendix:tight:alpha-star}). In the following paragraphs, we prove that there exists no better \ac{VBR}-arrival curve than this one for the shown output cumulative function.

	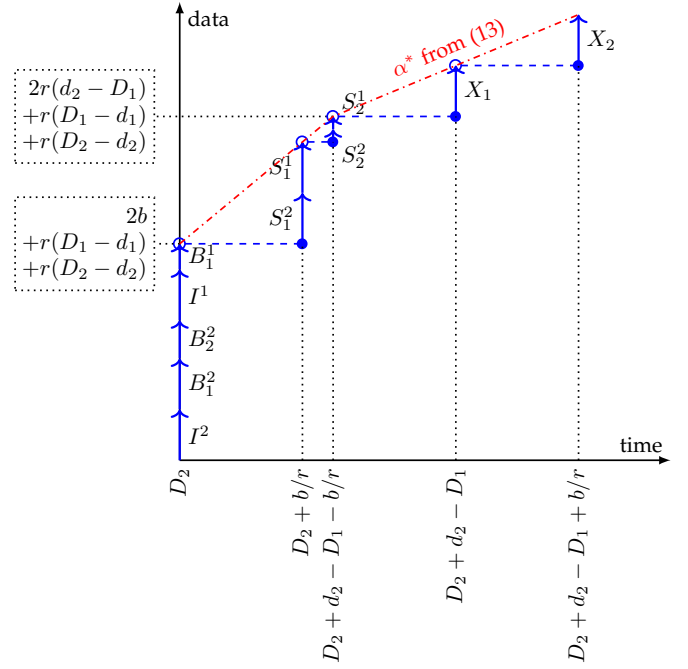
\begin{figure}\centering
		\resizebox*{\linewidth}{!}{\begin{tikzpicture}
	\tikzstyle{lw} = [line width=1pt]
	\tikzstyle{das} = [dashed, blue, o - Circle ]
	\tikzstyle{xx} = [xshift=-4\pgflinewidth]
	\tikzstyle{xxx} = [xshift=3\pgflinewidth]
	\begin{axis}[
        scale=1.6,
		axis x line=none,
		axis y line=none,
        ticks=none,
		xmin=-1.4,
		xmax=4,
		ymin=-4,
		ymax=9,
        height=8cm,
        width=8cm,
	]

    \node[anchor=south east] at (axis cs:4,0) {time};
    \node[anchor=north west] at (axis cs:0,9) {data};
    \node[anchor=east, rotate=90] at (axis cs:0,0) {$D_2$};
    \draw[-latex] (axis cs:0,0) -- (axis cs:4,0);
    \draw[-latex] (axis cs:0,0) -- (axis cs:0,9);

	\draw[->,blue,lw] (axis cs:0,0) -- (axis cs:0,1) node[pos=0.5,right, black] {$I^2$};
	\draw[->,blue,lw] (axis cs:0,1) -- (axis cs:0,2) node[pos=0.5,right, black] {$B^2_1$};
	\draw[->,blue,lw] (axis cs:0,2) -- (axis cs:0,2.75) node[pos=0.5,right, black] {$B^2_2$};
	\draw[->,blue,lw] (axis cs:0,2.75) -- (axis cs:0,3.75) node[pos=0.5,right, black] {$I^1$};
    \draw[->,blue,lw] (axis cs:0,3.75) -- (axis cs:0,4.25) node[pos=0.5,right, black] {$B^1_1$};

    \draw[das] ([xx] axis cs:0,4.25) -- ([xxx] axis cs:1,4.25);
    \draw[dotted] (axis cs:0,4.25) -- (axis cs:-0.2,4.25) node[pos=1, draw, dotted, anchor=east] {\makecell[r]{$2b$\\$+r(D_1-d_1)$\\$+r(D_2-d_2)$}};
	
	\draw[->,blue,lw] (axis cs:1,4.25) -- (axis cs:1,5.25) node[pos=0.5,left, black] {$S^2_1$};
	\draw[->,blue,lw] (axis cs:1,5.25) -- (axis cs:1,6.25) node[pos=0.5,left, black] {$S^1_1$};

    \draw[dotted] (axis cs:1.25,6.75) -- (axis cs:-0.2,6.75) node[pos=1, draw, dotted, anchor=east] {\makecell[r]{$2r(d_2-D_1)$\\$+r(D_1-d_1)$\\$+r(D_2-d_2)$}};

    \draw[dotted] (axis cs:1,0) -- (axis cs:1,4.25) node[pos=0,rotate=90, anchor=east] {$D_2 + b/r$};

    \draw[das] ([xx] axis cs:1,6.25) -- ([xxx] axis cs:1.25,6.25);
	
	\draw[->,blue,lw] (axis cs:1.25,6.25) -- (axis cs:1.25,6.5) node[pos=0.5,anchor=north west, black] {$S^2_2$};
	\draw[->,blue,lw] (axis cs:1.25,6.5) -- (axis cs:1.25,6.75) node[pos=0.5,anchor=south west, black] {$S^1_2$};

    \draw[dotted] (axis cs:1.25,0) -- (axis cs:1.25,6.25) node[pos=0,rotate=90, anchor=east] {$D_2 + d_2-D_1-b/r$};

    \draw[das] ([xx] axis cs:1.25,6.75) -- ( [xxx]axis cs:2.25,6.75);
	\draw[->,blue,lw] (axis cs:2.25,6.75) -- (axis cs:2.25,7.75) node[pos=0.5,right, black] {$X_1$};
    \draw[dotted] (axis cs:2.25,0) -- (axis cs:2.25,6.75) node[pos=0,rotate=90, anchor=east] {$D_2 + d_2-D_1$};
    \draw[das] ([xx] axis cs:2.25,7.75) -- ([xxx] axis cs:3.25,7.75);
    \draw[->,blue,lw] (axis cs:3.25,7.75) -- (axis cs:3.25,8.75) node[pos=0.5,right, black] {$X_2$};
    \draw[dotted] (axis cs:3.25,0) -- (axis cs:3.25,7.75) node[pos=0,rotate=90, anchor=east] {$D_2 + d_2-D_1 + b/r$};

    \draw[red, dashdotted] (axis cs:0,4.25) -- (axis cs:1.25,6.75) -- (axis cs:3.25,8.75) node[pos=0.5, sloped, above] {$\alpha^*$ from (\ref{eq:appendix:tight:alpha-star})};
	
	
	
		
	\end{axis}
\end{tikzpicture}}
		\caption{\label{fig:proof:frer:oac:tight:cumul-out} Dashed blue: Cumulative function $R^*$ at the output of the \ac{PEF}, from the example trajectory of Figure~\ref{fig:proof:frer:oac:tight:timeline:four_categories}. Dashed boxes: Values of interest of the cumulative function, further detailed in Paragraph~\ref{par:tight:oaf}. Red dashdotted: Arrival curve obtained from Corollary~\ref{cor:toolbox:thm-simple-application} and recalled in Equation~(\ref{eq:appendix:tight:alpha-star}).}
	\end{figure}

	\paragraph{Properties of the systems $S_1$, $S_2$} We now show the following properties of the above-described systems $S_1$, $S_2$.
	\begin{lemma}[The delay bounds through $S_1$, $S_2$]\label{lemma:fre:oac:tight:path_delay_bounds}
		The delay of any non-lost packet through $S_1$ is bounded between $d_1$ and $D_1$.		
		The delay of any non-lost packet through $S_2$  is bounded between $d_2$ and $D_2$.  
	\end{lemma}
	\begin{proof}[Proof of Lemma~\ref{lemma:fre:oac:tight:path_delay_bounds}]
		From Table~\ref{tab:proof:frer:oac:tight:output_time}, the result is clear for packets transporting data-units of categories $I,S$ and $X$ 
		
		We prove it for packets transporting data-units of category $B$:
		
		\ul{Through $S_1$:} for $m$ a data unit of category $B^2$, the packet transporting $m$ through $S_1$ is lost. For $m$ a data unit of category $B^1$, the packet transporting $m$ through $S_1$ verifies
		
		\begin{equation*}
			\mathcal{G}(B^1_{\chi^1}) \ge \mathcal{G}(m) \ge \mathcal{G}(I^1)
		\end{equation*}
		because data units of $B^1$ are sent after $I^1$ and before $B^1_{\chi^1}$.
		\begin{equation}\label{eq:appendix:tight:delay-through-s1}
			D_2 - (D_2 - d_1) \le \mathcal{E}_1(m) - \mathcal{G}(m) \le D_2 -  (D_2-D_1) 
		\end{equation}
		per Table~\ref{tab:proof:frer:oac:tight:output_time}.
		Equation~(\ref{eq:appendix:tight:delay-through-s1}) proves that any packet transporting a data unit of type $B^1$ through $S_1$ has a delay through $S_1$ bounded in $[d_1,D_1]$.
				
		\ul{Through $S_2$:} for $m$ a data unit of category $B^1$, the packet transporting $m$ through $S_2$ is lost. For $m$ a data unit of category $B^2$, the packet transporting $m$ through $S_2$ verifies
		
		\begin{equation*}
			\mathcal{G}(B^2_{\chi^2}) \ge \mathcal{G}(m) \ge \mathcal{G}(I^2)
		\end{equation*}
		because data units $B^2$ are sent after $I^2$ and before $B^2_{\chi^2}$.
		\begin{equation}\label{eq:appendix:tight:delay-through-s2}
			D_2 - (D_2 - d_2) \le \mathcal{E}_1(m) - \mathcal{G}(m) \le D_2 - 0
		\end{equation}
		Equation~(\ref{eq:appendix:tight:delay-through-s2}) proves that any packet transporting a data unit of type $B^2$ through $S_2$ has a delay through $S_2$ bounded in $[d_2,D_2]$.

	\end{proof}

	\paragraph{Properties of the output cumulative function}\label{par:tight:oaf}
	
	Call $R^*$ the output cumulative function of the flow at the output of the \ac{PEF}.
	Any data unit $m$ is released as soon as the first packet containing $m$ is received from either $S_1$ or $S_2$.
	Therefore, for any time instant $t$,
	\begin{equation}\label{eq:appendix:tight:compute-r-star}
		R^*(t) =  \sum_{\left\lbrace m \left| \begin{aligned}\mathcal{E}_1(m) < t\\\text{or }\mathcal{E}_2(m) < t\end{aligned}\right.\right\rbrace} \text{size}(m)
	\end{equation}
	
	We apply Equation~(\ref{eq:appendix:tight:compute-r-star}) to obtain the value of the cumulative function at several time-instants of interest.
	We start with $t=D_2$.
	
	\begin{equation*}
		R^*(D_2) =  \sum_{\left\lbrace m \left| \begin{aligned}\mathcal{E}_1(m) < D_2\\\text{or }\mathcal{E}_2(m) < D_2\end{aligned}\right.\right\rbrace} \text{size}(m)
	\end{equation*}
	Per Table~\ref{tab:proof:frer:oac:tight:output_time}, we obtain
	\begin{equation}\label{eq:proof:frer:oac:tight:RatD2}
		R^*(D_2) = 0
	\end{equation}
	We then continue with $D_2+\epsilon$ for $\epsilon > 0$,
	\begin{align*}
	& R^*(D_2+\epsilon) \\
	& =  \sum_{\min(\mathcal{E}_1(m),\mathcal{E}_2(m)) < D_2+\epsilon} \text{size}(m) \\
	& \ge \sum_{m\in I}\text{size}(m) + \sum_{m\in B^1}\text{size}(m) + \sum_{m\in B^2}\text{size}(m) \\
	& \text{\Comment{From~Table~\ref{tab:proof:frer:oac:tight:output_time}}}\\
	\end{align*}
	With (\ref{eq:proof:frer:oac:tight:sum_cat_b}), we obtain
	\begin{equation}\label{eq:proof:frer:oac:tight:RatD2Plus}
		\forall \epsilon > 0, \qquad R^*(D_2+\epsilon)	\ge 2b + r (D_1-d_1)  + r (D_2-d_2) 
	\end{equation}
	And finally, $\forall \epsilon>0$,	
	\begin{align*}
	&R^*(D_2+(d_2-D_1)-b/r+\epsilon) \\
	&=   \sum_{\min(\mathcal{E}_1(m),\mathcal{E}_2(m)) <D_2+(d_2-D_1)-b/r+\epsilon} \text{size}(m) \\
	&\ge \sum_{m\in I}\text{size}(m) + \sum_{m\in B}\text{size}(m) \\
	&\qquad + \sum_{m\in S_1}\text{size}(m) + \sum_{m\in S_2}\text{size}(m)\textbf{}\\
	&\text{\Comment{From~Table~\ref{tab:proof:frer:oac:tight:output_time}}}\\
	\end{align*}
	With (\ref{eq:proof:frer:oac:tight:sum_cat_s}), we obtain $\forall \epsilon > 0 $
	\begin{equation}\label{eq:proof:frer:oac:tight:RatD2PlusTs}\begin{aligned}
		&R^*(D_2+(d_2-D_1)-b/r+\epsilon)	\\&\ge r (D_1-d_1)  + r (D_2-d_2)  + 2 r (d_2-D_1)
	\end{aligned}\end{equation}

	\paragraph{Properties of any candidate arrival curve for $f$}

	Consider any \ac{VBR} arrival curve $\alpha'= \text{SPEC}(M',p',r',b')$ defined in \cite[\S 1.2]{leboudecNetworkCalculusTheory2001} and assume that $\alpha'$ is an arrival curve for $f$ at the output of the \ac{PEF}.

	Consider also the piecewise-linear function $\alpha^{\dagger}$ defined on $\mathbb{R}+$ by
	\begin{equation}
		\alpha^{\dagger}: t\mapsto \min(M'+\rho' t, b' + r' t)
	\end{equation}

	By definition, we have, for all $t\ge0$
	\begin{equation}\label{eq:tight:primevsdagger}
		\alpha'(t) = \left\lbrace\begin{aligned} \alpha^{\dagger}(t) &\quad\text{ if }t>0 \\ 0 &\quad\text{ if } t=0\end{aligned}\right.
	\end{equation}

	Note that $\alpha^{\dagger}$ is concave and wide-sense increasing. We also have the following result
	\begin{lemma}\label{lemma:tight:alphadagger-evolv}
		For any $s \ge t\ge0$, $\alpha^{\dagger}(s) - \alpha^{\dagger}(t) \ge r'(s-t)$
	\end{lemma}
	\begin{proof}[Proof of Lemma~\ref{lemma:tight:alphadagger-evolv}]
		We simply break in all the possible cases:
		$\bullet$ If both $s \le \frac{M'-b'}{\rho'-r'}$ and $t \le \frac{M'-b'}{\rho'-r'}$ then $\alpha^{\dagger}(s) - \alpha^{\dagger}(t) = \rho'(s-t) \ge r'(s-t)$ because $r'\ge \rho'$.
		
		$\bullet$ If $t \le \frac{M'-b'}{\rho'-r'}$, and $s \ge \frac{M'-b'}{\rho'-r'}$, then $\alpha^{\dagger}(s) - \alpha^{\dagger}(t) = b' - M' + r' s - \rho' t \ge r' s - \rho' t \ge r' (s-t)$ because $b' \ge M'$ and $\rho' \ge r'$.
		
		$\bullet$ If both  $s \ge \frac{M'-b'}{\rho'-r'}$ and $t \ge \frac{M'-b'}{\rho'-r'}$ then $\alpha^{\dagger}(s) - \alpha^{\dagger}(t) = r'(s-t)$.
	\end{proof}

	We observe that after $D_2+(d_2-D_1)$, the output traffic $R^*$ is made of the data units of category $X$ with a size $b$ and a period $b/r$. Therefore, the long-term rate of the flow at the output of the \ac{PEF} is exactly $r$ and any piece-wise linear arrival curve for this flow must have a long-term rate at least as big as $r$. For the \ac{VBR} arrival curve $\alpha'$, this gives
	\begin{equation}\label{eq:appendix:tight:rprime}
		r'\ge r
	\end{equation} 
	
	Then, as $\alpha'$ is an arrival curve for $f$ at the output of the \ac{PEF}, by \cite[Definition 1.2.1]{leboudecNetworkCalculusTheory2001}, for any $t,s \ge 0$, $R^*(t+s) - R(t) \le \alpha'(s)$. 
	
	In particular, $\forall \epsilon >0$
	\begin{equation*}
		\alpha'(\epsilon) \ge R^*(D_2+\epsilon) - R^*(D_2)
	\end{equation*}
	With (\ref{eq:proof:frer:oac:tight:RatD2}) and (\ref{eq:proof:frer:oac:tight:RatD2Plus}) this gives, $\forall \epsilon > 0$,
	\begin{equation}\label{eq:proof:frer:oac:tight:alphaM}
		\alpha'(\epsilon) \ge 2b+r(D_1-d_1+D_2-d_2)
	\end{equation}
	This is valid for any choice of $\epsilon > 0$ thus $\lim_{t\rightarrow 0}\alpha'(t) \ge 2b+r(D_1-d_1+D_2-d_2)$, which gives:
	\begin{equation}\label{eq:tight:daggerAt0}
		\alpha^{\dagger}(0) \ge 2b+r(D_1-d_1+D_2-d_2)
	\end{equation}

	Similarly, $\forall \epsilon > 0$, 
	\begin{equation}\resizebox*{\linewidth}{!}{$
		\alpha'\left(d_2-D_1 - \frac{b}{r} + \epsilon\right) \ge R^*\left(D_2+d_2-D_1-\frac{b}{r} + \epsilon\right) - R^*(D_2)
	$}\end{equation}
	with (\ref{eq:proof:frer:oac:tight:RatD2}) and (\ref{eq:proof:frer:oac:tight:RatD2PlusTs}), we obtain, $\forall \epsilon >0$,
	\begin{equation}\resizebox*{\linewidth}{!}{$
		\alpha'\left(d_2-D_1 - \frac{b}{r} + \epsilon\right) \ge r(D_1-d_1) + r(D_2-d_2) + 2r(d_2-D_1) 
	$}\end{equation}
	this is again valid for any value $\epsilon>0$ so $\lim_{\epsilon \rightarrow 0}\alpha'(d_2-D_1 - \frac{b}{r} + \epsilon) \ge r(D_1-d_1) + r(D_2-d_2) + 2r(d_2-D_1)$, which gives
	\begin{equation}\label{eq:tight:daggerAtX}\resizebox*{\linewidth}{!}{$\begin{aligned}
		\alpha^{\dagger}\left(d_2-D_1 - \frac{b}{r}\right)  &= \alpha'\left(d_2-D_1 - \frac{b}{r}\right) \\ &\ge r(D_1-d_1) + r(D_2-d_2) + 2r(d_2-D_1)
	\end{aligned}$}\end{equation}

	And using the above properties, we can prove the following result
	\begin{lemma}\label{lemma:tight:primeisgreater}
		For any $t > 0$, $\alpha^{\dagger}(t)\ge \alpha^*(t)$ where $\alpha^*$ is the \ac{VBR} obtained by applying Corollary~\ref{cor:toolbox:thm-simple-application} and given in (\ref{eq:appendix:tight:alpha-star}).
	\end{lemma}
	\begin{proof}[Proof of Lemma~\ref{lemma:tight:primeisgreater}]

		$\bullet$ If $t > d_2-D_1-\frac{b}{r}$, then
		\begin{equation*}\resizebox*{\linewidth}{!}{$\begin{aligned}
			\alpha^{\dagger}(t) &= \alpha^{\dagger}\left(d_2-D_1-\frac{b}{r}\right) + \alpha^{\dagger}(t) - \alpha^{\dagger}\left(d_2-D_1-\frac{b}{r}\right)\\
			& \ge \alpha^{\dagger}\left(d_2-D_1-\frac{b}{r}\right) + r' (t - d_2 + D_1) + b \quad &&\text{\Comment{Lemma~\ref{lemma:tight:alphadagger-evolv}}} \\
			& \ge \alpha^{\dagger}\left(d_2-D_1-\frac{b}{r}\right) + r (t - d_2 + D_1) + b \quad &&\text{\Comment{(\ref{eq:appendix:tight:rprime})}} \\
			& \ge r(D_1-d_1) + r(D_2-d_2)+2r(d_2-D_1) + r (t - d_2 + D_1) + b \quad &&\text{\Comment{(\ref{eq:tight:daggerAtX})}} \\
			& \ge b + rD_2 - rd_1 + rt = \alpha^*(t) \quad && \text{\Comment(\ref{eq:appendix:tight:alpha-star})} \\
		\end{aligned}$}\end{equation*}

		$\bullet$ If $0<t<d_2-D_1-\frac{b}{r}$, then we use the fact that $\alpha^{\dagger}$ is concave on $\mathcal{R}+$.
		
		Define 
		\begin{equation}
			x = \frac{t}{d_2-D_1-\frac{b}{r}}
		\end{equation}
		then, by definition of a concave function,
		\begin{equation*}\resizebox*{\linewidth}{!}{$\begin{aligned}
			\alpha^{\dagger}(t) &\ge x \alpha^{\dagger}\left(d_2-D_1-\frac{b}{r}\right) + (1-x) \alpha^{\dagger}(0)\\
			&\ge x \left(r(D_1-d_1) + r (D_2-d_2)+2r(d_2-D_1)\right) \quad && \text{\Comment(\ref{eq:tight:daggerAtX})} \\ &\qquad+  (1-x)\left(2b+r(D_1-d_1+D_2-d_2)\right) \quad && \text{\Comment(\ref{eq:tight:daggerAt0})}\\
			&\ge 2 r x\left(d_2-D_1-\frac{b}{r}\right) + 2b + r(D_1-d_1) + r(D_2-d_2) \\
			&\ge 2rt + 2b + r(D_1-d_1) + r(D_2-d_2) = \alpha^*(t) && \text{\Comment(\ref{eq:tight:daggerAt0}) and (\ref{eq:appendix:tight:alpha-star})} \\
		\end{aligned}$}\end{equation*}

	\end{proof}

	By (\ref{eq:tight:primevsdagger}), Lemma~\ref{lemma:tight:primeisgreater} proves that $\forall t > 0$, $\alpha'(t)\ge\alpha^*(t)$ and by definition of an arrival curve, $\alpha'(0)=\alpha^*(0)=0$. 
	We hence have proved that $\alpha^*$ is a better arrival curve for $f$ at the output of the \ac{PEF} than $\alpha'$.
	This is valid for any \ac{VBR} curve $\alpha'$ that is an arrival curve for $f$ at the output of the \ac{PEF}.
	Therefore $\alpha^*$ obtained using Corollary~\ref{cor:toolbox:thm-simple-application} is the best \ac{VBR} arrival curve for $f$ at the output of the \ac{PEF}.
	
	\subsubsection{Case $d_2 - D_1 \le b/r$:}
	In this case, the leaky-bucket $\gamma_{2r,2b+r(D_1-d_1+D_2-d_2)}$ is always larger than $\gamma_{r,b+r(D_2-d_1)}$. Thus the application of Corollary~\ref{cor:toolbox:thm-simple-application} gives that the leaky-bucket $\alpha^* = \gamma_{r,b+r(D_2-d_1)}$ is an arrival curve for $f$ at the output of the \ac{PEF} in Figure~\ref{fig:toolbox:oac-prop-simplified-figure}.

	Using the same rationale as for the previous case, we can use a greedy source that generates packets with a long-term rate of $r$, thus any arrival curve for $f$ at the output of the \ac{PEF} must also have a long-term rate larger than $r$.

	Therefore, the proof of tightness needs only to exhibit a trajectory that creates a burst as big as $b+r(D_2-d_1)$.
	This is done as follows.

	\paragraph{Definition of several constants}
	We define
	\begin{equation}
		\chi^1 \triangleq \left\lceil \frac{r(D_1-d_1)}{b}\right\rceil \quad\text{ and }\quad \chi^2 \triangleq \left\lceil \frac{r(D_2-D_1)}{b}\right\rceil
	\end{equation}
	Note that both $\chi^1 \ge 1$ and $\chi^2 \ge 1$.
	We further consider a time instant $t_0$ such that $t_0 > D_2 - D_1$.

	\paragraph{Description of the traffic generation at the source}\label{par:appendix:tightness:second:traffic-gen}
	We classify the data units generated by the source into two categories: $I$, $B$.
	Category $B$ is then subdivided into subcategories $B^1$ and $B^2$. 
	The category of a data unit defines the role that the data unit has in the trajectory.
	This notion is only used in the proof and does not relate to any physical property of the data units.

	\ul{Category $I$}
	The source generates a unique data unit $I$ at absolute time $t_0$, of length $b$ (see Table~\ref{tab:proof:frer:oac:tight:second:initiators_generation}).
	
	\begin{table}\centering
		\caption{\label{tab:proof:frer:oac:tight:second:initiators_generation} Generation of the Data-Unit of Category $I$ in the Trajectory that Achieves the Tightness of Corollary~\ref{cor:toolbox:thm-simple-application}, when $d_2 - D_1 \le b/r$.}
\begin{tabular}{r|c|c}
	Data unit $m$ & Size, $\text{size}(m)$ & Generation time, $\mathcal{G}(m)$ \\
	\hline
	$I$ & $b$ & $t_0$ \\
\end{tabular}
	\end{table}

	\emph{Note:} The role of the data unit $I$ is to create the term $b$ of the burst $b+r(D_2-d_1)$.

	\ul{Category $B$}
	In addition, the source generates $\chi^1$ data units of subcategory $B^1$ and $\chi^2$ data units of subcategory $B^2$, as described in Table~\ref{tab:proof:frer:oac:tight:second:bursty_generation}.

	A possible output of the source when combining categories $I$ and $B$ is shown in Figure~\ref{fig:proof:frer:oac:tight:second:timeline:bursty_1}. In the proposed situation, we have $\chi^1=\chi^2=2$.
	Both subcategories $B^1$ and $B^2$ are made of two data units.
	The data units of $B^2$ are sent before $I$ whereas the data units of $B^1$ are sent after $I$.

	We note that, for any value of $\chi^1$, $\chi^2$,
	\begin{equation}\begin{aligned}\label{eq:tight:second:sums}
		&\sum_{k\in\llbracket 1,\chi^1 \rrbracket} \text{size}(B_k^1) = r(D_1-d_1) \\ \quad\text{ and }\quad &\sum_{k\in\llbracket 1,\chi^2 \rrbracket} \text{size}(B_k^2) = r(D_2-D_1) 
	\end{aligned}\end{equation}

	\begin{table}\centering
		\caption{\label{tab:proof:frer:oac:tight:second:bursty_generation} Generation of the Data-Units of Category $B$ in the Trajectory that Achieves the Tightness of Corollary~\ref{cor:toolbox:thm-simple-application}, when $d_2 - D_1 \le b/r$.}
		\resizebox*{\linewidth}{!}{
\begin{tabular}{r|c|c}
	Data unit $m$ & Size, $\text{size}(m)$ & Generation time, $\mathcal{G}(m)$ \\
	\hline
	$\forall k \in \llbracket 1, \chi^2-1\rrbracket, \qquad B_k^2$ & $b$ & $t_0 - D_2+D_1 + (k-1)\frac{b}{r}$ \\
	$B^2_{\chi^2}$ & $r(D_2-D_1) - (\chi^2-1)b$ & $t_0-b/r$\\
	\hline
	$\forall k \in \llbracket 1, \chi^1-1\rrbracket, \qquad B_k^1$ & $b$ & $t_0 + k\frac{b}{r}$ \\
	$B^1_{\chi^1}$ & $r(D_1-d_1) - (\chi^1-1)b$ & $t_0+ (D_1-d_1)$\\
\end{tabular}}
	\end{table}
	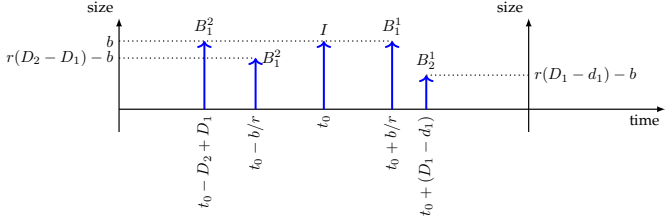
\begin{figure}
		\resizebox{\linewidth}{!}{

\begin{tikzpicture}
	\tikzstyle{ab} = [pos=0, rotate=90, anchor=east, black]
	\tikzstyle{pkn} = [black, pos=1,above]
	\tikzstyle{pk} = [->, blue, line width=1.25pt]
	\def\xdev{1.5cm}
	\def\yline{0cm}
	\def\ph{1.5cm}
	
	\draw[-latex] (1*\xdev,\yline) -- (9*\xdev, \yline) node[black, anchor=north east, pos=1] {time};
	\draw[-latex] (1*\xdev,\yline-0.5cm) -- (1*\xdev,\yline+\ph+0.5cm) node[pos=1, anchor=south east] {size};
	\draw[-latex] (7*\xdev,\yline-0.5cm) -- (7*\xdev,\yline+\ph+0.5cm) node[pos=1, anchor=south east] {size};
	\draw[pk]  (4*\xdev,\yline) -- (4*\xdev,\ph+\yline) node[ab] {$t_0$} node[pkn] {$I$};

	\draw[pk]  (3*\xdev,\yline) -- (3*\xdev,0.75*\ph+\yline) node[ab] {$t_0-b/r$} node[pkn, anchor=west] {$B^2_1$};
	\draw[pk]  (2.25*\xdev,\yline) -- (2.25*\xdev,\ph+\yline) node[ab] {$t_0-D_2+D_1$} node[pkn] {$B^2_1$};
	
	\draw[pk]  (5*\xdev,\yline) -- (5*\xdev,\ph+\yline) node[ab] {$t_0+b/r$} node[pkn] {$B^1_1$};
	\draw[pk]  (5.5*\xdev,\yline) -- (5.5*\xdev,0.5*\ph+\yline) node[ab] {$t_0+(D_1-d_1)$} node[pkn] {$B^1_2$};

	\draw[dotted] (1*\xdev, \ph) -- (5*\xdev,\ph) node[pos=0,left] {$b$};
	\draw[dotted] (5.5*\xdev, 0.5*\ph+\yline) -- (7*\xdev,0.5*\ph+\yline) node[pos=1,right] {$r(D_1-d_1)-b$};
	\draw[dotted] (1*\xdev, 0.75*\ph) -- (3*\xdev,0.75*\ph) node[pos=0,left] {$r(D_2-D_1)-b$};
	
\end{tikzpicture}}
		\caption{\label{fig:proof:frer:oac:tight:second:timeline:bursty_1}Example of the source output in the trajectory achieving the tightness, focusing on categories $I$ and $B$.}		
	\end{figure}

	\paragraph{Properties of the traffic generation at the source}
	As for the previous case, we prove that Lemmas~\ref{lemma:frer:oac:tight:packet_sizes} and ~\ref{lemma:frer:oac:tight:min_dist} hold for the traffic described in \ref{par:appendix:tightness:second:traffic-gen}.
	This is clear for most pair of data units, let us show it for example for data units $B_{\chi^2-1}^2$ and $B_{\chi^2}^2$:
	\begin{equation*}\resizebox{\linewidth}{!}{$\begin{aligned}
		\mathcal{G}(B_{\chi^2}^2) - \mathcal{G}(B_{\chi^2-1}^2) &= t_0-\frac{b}{r} - t_0 +D_2-D_1 - (\chi^2-1)\frac{b}{r} + \frac{b}{r}\\
		& = D_2-D_1-(\chi^2-1)\frac{b}{r}\\
		& \ge \frac{\text{size}(B_{\chi^2}^2)}{r}
	\end{aligned}$}\end{equation*}
	Therefore, Lemma~\ref{lemma:fre:oac:tight:source_complies} also holds for the traffic described in \ref{par:appendix:tightness:second:traffic-gen}.
	The traffic described in the trajectory is $\gamma_{r,b}$-constrained.

	\paragraph{Description of the systems $S_1$, $S_2$}
	With the same notations and conventions as for the previous case, the systems $S_1$ and $S_2$ release the packets containing the different data units at the time instants shown in Table~\ref{tab:proof:frer:oac:tight:second:output_time}.

	\begin{table*}\centering
		\caption{\label{tab:proof:frer:oac:tight:second:output_time} Absolute Release Time for Each Packet at the Output of Each System $S_1$, $S_2$, in the Case $d_2 - D_1 \le b/r$.}
		\resizebox{\linewidth}{!}{
\begin{tabular}{r|c|c|c}
	Data unit, $m$ & $\mathcal{G}(m)$, Generation time & $\mathcal{E}_1(m)$, exit time out of $S_1$ for the packet  & $\mathcal{E}_2(m)$, exit time out of $S_2$ for the packet\\
	& & transporting $m$ through $S_1$ & transporting $m$ through $S_2$\\
	\hline
	$\forall k \in \llbracket 1, \chi^2-1\rrbracket, \qquad B_k^2$ & $t_0 - D_2+D_1 + (k-1)\frac{b}{r}$ & $+\infty$ & $t_0+D_1$ \\
	$B^2_{\chi^2}$ & $t_0-b/r$ & $+\infty$ & $t_0+D_1$ \\
    \hline
	$I$ & $t_0$  & $t_0+D_1$  & $+\infty$ \\
	$\forall k \in \llbracket 1, \chi^1-1\rrbracket, \qquad B_k^1$ & $t_0 + k\frac{b}{r}$ & $t_0+D_1$ & $+\infty$ \\
	$B_{\chi^1}^1$ & $t_0+ (D_1-d_1)$ & $t_0+D_1$ & $+\infty$ \\
\end{tabular}}
	\end{table*}

	\paragraph{Properties of the systems $S_1$, $S_2$}
	As for the previous case, we can also prove that Lemma~\ref{lemma:fre:oac:tight:path_delay_bounds} holds for the systems $S_1$, $S_2$ described above.
	This is clear from Table~\ref{tab:proof:frer:oac:tight:second:output_time} for most data units. For example, the delay of the packet transporting the data unit $B^2_{\chi^2}$ through $S_2$ is at least $d_2$ because, by assumption, $d_2-D_1\le\frac{b}{r}$. 
	
	\paragraph{Properties of the output cumulative function}
	In the trajectory of Table~\ref{tab:proof:frer:oac:tight:second:output_time}, all data units exit the \ac{PEF} at $t_0+D_1$.
	We hence have created a burst of size
	\begin{equation*}\begin{aligned}
		&\text{size}(I) + \sum_{k\in\llbracket 1,\chi^1 \rrbracket} \text{size}(B_k^1) + \sum_{k\in\llbracket 1,\chi^2 \rrbracket} \text{size}(B_k^2)\\
		&= b + r(D_1-d_1) + r(D_2-D_1)\quad&&\text{\Comment (\ref{eq:tight:second:sums})} \\
		&= b + r (D_2-d_1) \\
	\end{aligned}\end{equation*}
	Therefore, any curve that is an arrival curve of $f$ at the output of the \ac{PEF} should have a limit at $0$ at least larger than $b + r (D_2-d_1)$ and a long-term rate at least larger than $r$.
	Thus any such curve that is in addition concave on $\mathbb{R}^*+$ must hence be larger than the leaky-bucket arrival curve $\gamma_{r,b+r(D_2-d_1)}$.
	In particular, any \ac{VBR} arrival curve (concave on $\mathbb{R}^*+$ by definition) is larger than $\gamma_{r,b+r(D_2-d_1)}$.

\end{proof}
\subsection{Proof of Proposition~\ref{prop:toolbox:rto-vs-rbo}}\label{sec:appendix:toolbox:rto-vs-rbo}
\begin{proof}[Proof of Proposition~\ref{prop:toolbox:rto-vs-rbo}]\label{proof:appendix:toolbox:rto-vs-rbo}
    Consider the flow $f$ and two observation points $v,w$ such that $v$ is in vertex $n$, $w$ is in vertex $o$, $n$ is not an \ac{EP}-vertex of $\mathcal{G}(f)$, $o$ is a diamond ancestor of $n$ in $\mathcal{G}(f)$, and the flow $f$ is packetized at $v,w$.

    As $o$ is a diamond ancestor, it is not an \ac{EP}-vertex, thus each data unit of $f$ is observed at most once at $w$.
    As done in \cite{mohammadpourPacketReorderingTimeSensitive2020}, the $k$-th data unit of $f$ is defined as the data unit of $f$ that crosses $w$ in the $k$-th position.

    We note $E_k$ the arrival time of the $k$-th data unit of $f$ at $v$, with the convention that $E_k = +\infty$ if the $k$-th data unit of $f$ is lost for $n$.
    $E_k$ is correctly defined because $n$ is not an \ac{EP}-vertex, thus the $k$-th data unit of $f$ can cross the observation point $v$ at most once. Furthermore, $f$ is packetized at $v$, thus all the bits of the $k$-th data unit cross $v$ at the same time.

    Then the reordering offset of the $k$-th data unit of $f$ \cite[Eq.~(4)]{mohammadpourPacketReorderingTimeSensitive2020}, \cite{rfc4737} is defined by
    \begin{equation}
        \Pi_k = \sum_{j|j>k,E_j < E_k} l_j 
    \end{equation}
    with $l_j$ the size of the packet transporting the $j$-th data unit of $f$.

    Denote by $R$ the cumulative arrival function of flow $f$ at observation point $v$.
    By definition, $R$ is the number of bits of flow $f$ that cross $v$ over the time interval $[0,t[$.
    Thus for any non-lost data unit $k$, $R(E_k)$ is the number of bits of $f$ that cross $v$ strictly before\footnote{Here we use the traditional convention that cumulative functions are left-continuous. A discussion of this assumption is available in \cite[\S 1.1.1]{leboudecNetworkCalculusTheory2001}.} $E_k$.
    As $f$ is packetized at $v$, $R(E_k)$ is hence the sum of the length of the packets for all data units that arrived before the $k$-th data unit, excepted the $k$-th data unit itself.

    Thus $\Pi_k$ can be written
    \begin{equation*}\begin{aligned}
        \Pi_k &= \sum_{j|j>k,E_j < E_k} l_j \\
        & = R(E_k) - R(\min_{j>k} E_j)\\
        & \le \alpha_{f,v}(E_k - \min_{j>k} E_j)
    \end{aligned}\end{equation*}
    because $\alpha_{f,v}$ is an arrival curve of $f$ at $v$.
    By definition, $E_k - \min_{j>k} E_j$ is the reordering late offset of data unit $k$ that we denote by $\Lambda_k$ \cite[Eq.~(2)]{mohammadpourPacketReorderingTimeSensitive2020}.
    We hence obtain that for all $k$ such that $E_k < +\infty$,
    \begin{equation}\label{eq:rto-rbo}
        \Pi_k \le \alpha_{f,v}(\Lambda_k)
    \end{equation}
    The \ac{RTO} and the \ac{RBO} of flow $f$ are defined \cite[\S C]{mohammadpourPacketReorderingTimeSensitive2020} by 
    \begin{equation}
        \pi_{v}(f,o) \triangleq \sup_{k|E_k < +\infty}\Pi_k \quad\text{and}\quad \lambda_{v}(f,o) \triangleq \sup_{k|E_k < +\infty}\Lambda_k
    \end{equation}
    Equation (\ref{eq:rto-rbo}) is valid for any $k$ such that $E_k < +\infty$,  $\alpha_{f,v}$ is a wide-sense increasing function and $\sup_{k|E_k < +\infty}\Lambda_k$ is bounded by assumption.
    We hence obtain  $\pi_{v}(f,o) \le \alpha_{f,v}(\lambda_{v}(f,o))$.

\end{proof}
\subsection{Proof of Theorem~\ref{thm:toolbox:reordering}}\label{sec:appendix:thm:toolbox:reordering}
\begin{proof}[Proof of Theorem~\ref{thm:toolbox:reordering}]\label{proof:thm:toolbox:reordering}

The section of the network between the diamond ancestor $a$ and the vertex $n$ that contains the \ac{PEF} is a system (neither \ac{FIFO} nor lossless in general) with a jitter for $f$ bounded by  $D_f^{a\rightarrow n} - d_f^{a\rightarrow n}$.
We apply \cite[Thm~5]{mohammadpourPacketReorderingTimeSensitive2020} to obtain the result.
\end{proof}
\subsection{Proof of Theorem~\ref{thm:regulators:pfr-2d}}
\begin{proof}[Proof of Theorem~\ref{thm:regulators:pfr-2d}]\label{proof:regulators:pfr-2d}
    Applying Item 2/ of Theorem~\ref{thm:toolbox:pef-oac} with diamond ancestor $a$ gives that $\gamma_{r,b} \oslash \delta_{D-d} = \gamma_{r,b+r(D-d)}$ is an arrival-curve for $f$ at the input of the \ac{PFR}.
    From \cite[\S 1.7.4]{leboudecNetworkCalculusTheory2001}, a \ac{PFR} with concave shaping curve $\sigma$ is a network element that offers $\sigma$ as a service curve. The \ac{PFR} is \ac{FIFO} and lossless, thus we can apply \cite[Thm.~1.4.2]{leboudecNetworkCalculusTheory2001} and we obtain that $D-d$ (the maximal horizontal distance between the input arrival curve $\gamma_{r,b+r(D-d)}$ and the service curve $\gamma_{r,b}$) is an upper-bound on the delay of $f$ through the \ac{PFR}. Adding the already-known delay bounds for $\mathcal{S}$ gives those for $\mathcal{S}'$.
\end{proof}
\subsection{Proof of Theorem~\ref{thm:regulators:ir-instable}}
	\begin{proof}[Proof of Theorem~\ref{thm:regulators:ir-instable}]\label{proof:regulators:ir-instable}
		Consider a system defined by Figure~\ref{fig:regulators:ir-overview} and by Conditions~(a) to (c) of Theorem~\ref{thm:regulators:ir-instable}.
		Take any $r>0$, $b>0$ and $d_1,D_1,d_2,D_2$ such that Conditions~(d) to (f) of Theorem~\ref{thm:regulators:ir-instable} are met.
		We first describe the adversarial model applied when $D_1 < d_2$.
		\subsubsection{Adversarial model for the case $D_1 < d_2$}
		We exhibit an  adversarial model $\mathcal{M}_{D_1 < d_2}$ for the sources and for the paths $\{P_j\}_j$ such that Properties~\ref{item:thm:ir:first} to \ref{item:thm:ir:last} of Theorem~\ref{thm:regulators:ir-instable} hold for $\mathcal{M}_{D_1 < d_2}$.
		
		\paragraph{Constants of $\mathcal{M}_{D_1 < d_2}$}\label{par:adv:constants}
		We define
		\begin{equation}\label{eq:proof:adv:inter:j}
			J \triangleq d_2 - D_1			
		\end{equation}
		And
		\begin{equation}\label{eq:proof:adv:inter:dD}
			D \triangleq d_2 \qquad d \triangleq D_1		
		\end{equation}
		thus $d < D$. Note that $q \ge q_{\min}$ can be written		
		\begin{equation*}\label{eq:proof:adv:inter:qmin}
			q \ge q_{\min} = \left\lfloor \frac{2rJ}{b}+2\right\rfloor + 1
		\end{equation*}
		With $J>0$.
		Note that $q > \frac{2rJ}{b}+2$ thus $(q-2)\frac{b}{r} > 2J$. Therefore, take any $\epsilon$ such that
		\begin{equation}\label{eq:proof:adv:inter:epsilon}
			\min\left(\frac{b}{r} - \frac{2}{q-2} J,J\right) > \epsilon > 0
		\end{equation}
		We further define
		\begin{equation}\label{eq:proof:adv:inter:I}
			I \triangleq \max\left(\frac{q}{q-2}J,\frac{b}{r}\right)
		\end{equation}
		\begin{equation}\label{eq:proof:adv:inter:phi}
			\phi \triangleq I - J + \epsilon
		\end{equation}
		and
		\begin{equation}\label{eq:proof:adv:inter:tau}
			\tau \triangleq q \phi
		\end{equation}
		Finally, we consider a starting instant $x_1 > 0$ and for $i\in\llbracket1,q\rrbracket$, we define
		\begin{equation}\label{eq:proof:adv:inter:xi}
			x_i \triangleq (i-1) \phi + x_1
		\end{equation}
		
		\paragraph{Properties on the constants of $\mathcal{M}_{D_1 < d_2}$}\label{par:adv:constants-prop}
		For $q > 3$, $\frac{q}{q-2}>1$ thus by (\ref{eq:proof:adv:inter:I})
		\begin{equation}\label{eq:proof:adv:inter:i-greater-than-j}
			I > J > 0
		\end{equation}
		thus we also have $\phi > 0$ and $\tau > 0$ by (\ref{eq:proof:adv:inter:phi}) and (\ref{eq:proof:adv:inter:tau}).
		Furthermore,
		\begin{equation*}\begin{aligned}
			\tau - I &= q(I-J) + q\epsilon - I && \quad\text{\Comment{} by }(\ref{eq:proof:adv:inter:phi}),(\ref{eq:proof:adv:inter:tau}) \\
			& = (q-2) I -q J + q\epsilon + I && \\
			& \ge qJ - qJ + q\epsilon + I && \quad\text{\Comment{} by }(\ref{eq:proof:adv:inter:I})\\
			& > I && \quad\text{\Comment{} by (\ref{eq:proof:adv:inter:epsilon})}\\
		\end{aligned}\end{equation*}
		combined again with (\ref{eq:proof:adv:inter:I}), this gives
		\begin{equation}\label{eq:proof:adv:tauMoinsI}
			\tau - I > \frac{b}{r}
		\end{equation}
		For $\phi$, we first have
		\begin{equation}\label{eq:proof:adv:phiRel1}
			\phi < I
		\end{equation}
		because $\epsilon < J$ and
		\begin{equation}\label{eq:proof:adv:phiRel2}
			\phi < I + \frac{b}{r} - \frac{q}{q-2} J
		\end{equation}
		because $\epsilon < \frac{b}{r} - \frac{2}{q-2}J$. By (\ref{eq:proof:adv:inter:I}), $I$ can take only one of two values. If $I = \frac{b}{r}$, then (\ref{eq:proof:adv:phiRel1}) gives $\phi < \frac{b}{r}$. If $I = \frac{q}{q-2}J$, then (\ref{eq:proof:adv:phiRel2}) gives $\phi < \frac{b}{r}$. We hence prove
		\begin{equation}\label{eq:proof:adv:IvsBR}
			\phi < \frac{b}{r}
		\end{equation}

		\paragraph{Adversarial traffic generation at the source in $\mathcal{M}_{D_1 < d_2}$}\label{par:adv:source}

		For each $i \in \llbracket 1,q\rrbracket$, the source $a$ in Figure~\ref{fig:regulators:ir-overview} sends\footnote{If $b$ is larger than the maximal packet length, then the source sends several data units simultaneously such that the sum of their length equal $b$. In this case, $m_{i,k}^1$ and $m_{i,k}^2$ represent the set of these data units.} a data unit $m_{i,k}^1$, of size $b$ at the time instant $x_i + k\tau$ and $m_{j,k}^2$ of size $b$ at the time instant $x_i + k\tau+I$ for all $k\in\mathbb{N}$.

		Figure \ref{fig:proof:adv:timeline-local} presents the traffic at the output of the source, focusing on two successive flows: $f_i$ and $f_{i+1}$ (with $i~\le~q~-~1$). Their source profiles are periodic with a period $\tau$ and Figure~\ref{fig:proof:adv:timeline-local} focuses on the $k$-th period.
		For the flow $f_i$ (solid-blue data units), the source generates the data unit $m_{i,k}^1$ at time $x_i + k \tau$, then sends $m_{i,k}^2$ after a duration $I$ and it finally waits for the next period $(k+1)$ before it restarts the same profile and sends $m_{i,k+1}^1$.
		The source profile for flow $f_{i+1}$ (dashed-red data units) is identical, but shifted by $\phi$ with respect to the source profile for $f_i$, because $x_{i+1} = x_i + \phi$ by (\ref{eq:proof:adv:inter:xi}).
		By (\ref{eq:proof:adv:inter:phi}) and (\ref{eq:proof:adv:inter:epsilon}), $\phi < I$ thus $m_{i,k+1}^1$ is sent before $m_{i,k}^2$ as shown in the figure.
		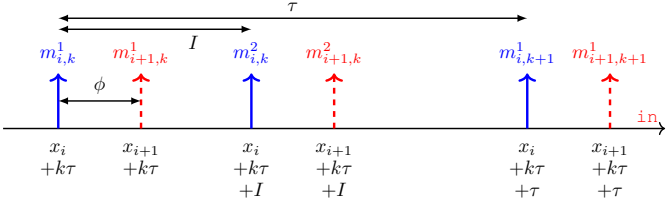
\begin{figure}\centering %
			\resizebox{\linewidth}{!}{

\begin{tikzpicture}

	\tikzstyle{pkn} = [pos=1, anchor=south]
	\tikzstyle{pki} = [pos=0, anchor=north, black]
	\tikzset{
		f/.cd,
		1/.style={blue},
		2/.style={red, dashed},
		3/.style={green, dotted},
		4/.style={black, dashdotted},
		5/.style={black},
		6/.style={black},
		7/.style={black},
		8/.style={black},
		9/.style={black},
		10/.style={black},
		11/.style={black},
		12/.style={black},
		13/.style={black},
		14/.style={black},
	}
	\tikzset{
		role/.cd,
		1/.style={->, line width=1.25pt},
		2/.style={->>, line width=1.25pt}
	}

	\draw[->] (0,0) -- (12,0) node[pos=1, anchor=east, black] {};
	\node[red, anchor=south east] at (12,0) {\texttt{in}};
	\draw[f/1, role/1] (1,0) -- (1,1) node[pos=0, black, anchor=north] {\makecell{$x_{i}$\\$+k\tau$}}                   node[pkn] {$m_{i,k}^1$}     ;
	\draw[f/1, role/1] (4.5,0) -- (4.5,1) node[pos=0, black, anchor=north] {\makecell{$x_{i}$\\$+k\tau$\\$+I$}}         node[pkn] {$m_{i,k}^2$}     ;
	\draw[f/1, role/1] (9.5,0) -- (9.5,1) node[pos=0, black, anchor=north] {\makecell{$x_{i}$\\$+k\tau$\\$+\tau$}}      node[pkn] {$m_{i,k+1}^1$}   ;

	\draw[f/2, role/1] (2.5,0) -- (2.5,1) node[pos=0,   black, anchor=north] {\makecell{$x_{i+1}$\\$+k\tau$}}           node[pkn] {$m_{i+1,k}^1$}  ;
	\draw[f/2, role/1] (6,0) -- (6,1) node[pos=0,   black, anchor=north] {\makecell{$x_{i+1}$\\$+k\tau$\\$+I$}}         node[pkn] {$m_{i+1,k}^2$}  ;
	\draw[f/2, role/1] (11,0) -- (11,1) node[pos=0, black, anchor=north] {\makecell{$x_{i+1}$\\$+k\tau$\\$+\tau$}}      node[pkn] {$m_{i+1,k+1}^1$};
	
	\draw[latex-latex, line width=0.75pt] (1,0.5) -- (2.5,0.5) node[midway, above] {$\phi$};
    \draw[latex-latex, line width=0.75pt] (1,1.8) -- (4.5,1.8) node[pos=0.7, below] {$I$};
	\draw[latex-latex, line width=0.75pt] (1,2) -- (9.5,2) node[midway, above] {$\tau$};
\end{tikzpicture}} %
			\caption{\label{fig:proof:adv:timeline-local} Generation of data units for flows $f_i$ and $f_{i+1}$ ($i \le q-1$). Their traffic profile is periodic with period $\tau$. The source sends a data unit for $f_i$ at $x_i + k\tau$ for $k\in\mathbb{N}$, then it sends another data unit after a duration $I$ and finally restarts at the next period. The profile for $f_{i+1}$ is identical and shifted by $\phi$ with respect to the one of $f_i$ ($x_{i+1} = x_i + \phi$).}
		\end{figure}

		\paragraph{Properties on the traffic generation at the source in $\mathcal{M}_{D_1 < d_2}$}\label{par:adv:source-prop}
		By (\ref{eq:proof:adv:inter:I}) and (\ref{eq:proof:adv:tauMoinsI}), the minimum time elapsed at the source between any two data units of $f_i$ is larger than $b/r$, which shows that Property~\ref{item:thm:ir:source} of Theorem~\ref{thm:regulators:ir-instable} holds.

		\paragraph{Adversarial paths in  $\mathcal{M}_{D_1 < d_2}$}\label{par:adv:paths}
		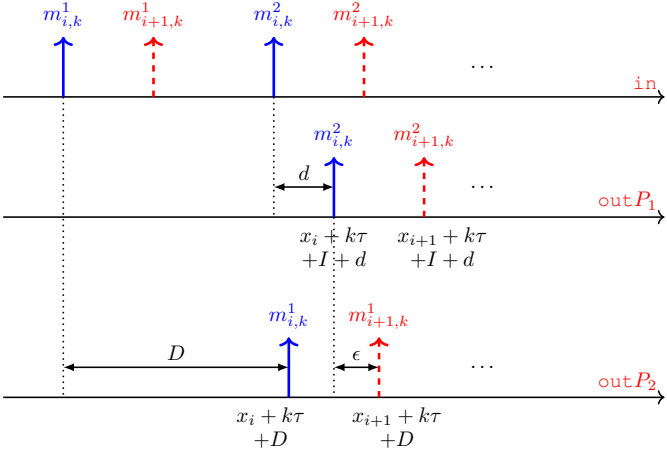
\begin{figure}\centering 
			\resizebox{\linewidth}{!}{

\begin{tikzpicture}

	\tikzstyle{pkn} = [pos=1, anchor=south]
	\tikzstyle{pki} = [pos=0, anchor=north, black]
	\tikzset{
		f/.cd,
		1/.style={blue},
		2/.style={red, dashed},
		3/.style={green, dotted},
		4/.style={black, dashdotted},
		5/.style={black},
		6/.style={black},
		7/.style={black},
		8/.style={black},
		9/.style={black},
		10/.style={black},
		11/.style={black},
		12/.style={black},
		13/.style={black},
		14/.style={black},
	}
	\tikzset{
		role/.cd,
		1/.style={->, line width=1.25pt},
		2/.style={->>, line width=1.25pt}
	}

	\def\mxmax{11}
	\draw[->] (0,0) -- (\mxmax,0) node[pos=1, anchor=east, black] {};
	\node[red, anchor=south east] at (\mxmax,0) {\texttt{in}};
	\draw[f/1, role/1] (1,0) -- (1,1)       node[pkn] {$m_{i,k}^1$}     ;
	\draw[f/1, role/1] (4.5,0) -- (4.5,1)   node[pkn] {$m_{i,k}^2$}     ;

    \draw[f/2, role/1] (2.5,0) -- (2.5,1)     node[pkn] {$m_{i+1,k}^1$}  ;
	\draw[f/2, role/1] (6,0) -- (6,1)         node[pkn] {$m_{i+1,k}^2$}  ;
	\node at (8,0.5) {\ldots};

	\draw[->] (0,-2) -- (\mxmax,-2) node[pos=1, anchor=east, black] {};
	\node[red, anchor=south east] at (\mxmax,-2) {\texttt{out}$P_1$};
	\draw[f/1, role/1] (5.5,-2) -- (5.5, -1)  				node[pkn] {$m_{i,k}^2$};
	\draw[f/2, role/1] (7,		-2) -- 	(7,		-1) 	  	node[pkn] {$m_{i+1,k}^2$};
	\node at (8,-1.5) {\ldots};

	\node[black, anchor=north] at (5.5,-2) {\makecell{$x_{i}+k\tau$\\$+I+d$}};
	\node[black, anchor=north] at (7.3,-2) {\makecell{$x_{i+1}+k\tau$\\$+I+d$}};  

	\draw[->] (0,-5) -- (\mxmax,-5) node[pos=1, anchor=east, black] {};
	\node[red, anchor=south east] at (\mxmax,-5) {\texttt{out}$P_2$};
	\draw[f/1, role/1] (4.75,-5) -- (4.75, -4)  				node[pkn] {$m_{i,k}^1$};
	\draw[f/2, role/1] (6.25,		-5) -- 	(6.25,		-4) 	  	node[pkn] {$m_{i+1,k}^1$};
	\node at (8,-4.5) {\ldots};

	\node[black, anchor=north] at (4.45,-5) {\makecell{$x_{i}+k\tau$\\$+D$}};
	\node[black, anchor=north] at (6.55,-5) {\makecell{$x_{i+1}+k\tau$\\$+D$}};  

	\draw[line width=0.75pt, dotted] (1,0) -- (1,-5);
	\draw[line width=0.75pt, dotted] (4.5,0) -- (4.5,-2);
	\draw[latex-latex, line width=0.75pt] (1,-4.5) -- (4.75,-4.5) node[pos=0.5,above] {$D$};
	\draw[latex-latex, line width=0.75pt] (4.5,-1.5) -- (5.5,-1.5) node[pos=0.5,above] {$d$};

	\draw[line width=0.75pt, dotted] (5.5,-2) -- (5.5,-5);
	\draw[latex-latex, line width=0.75pt] (5.5,-4.5) -- (6.25,-4.5) node[pos=0.5,above]{$\epsilon$};
\end{tikzpicture}} %
			\caption{\label{fig:proof:adv:timeline-paths} Traffic profile for the flows $f_i$ and $f_{i+1}$ at the output of the two adversarial paths $P_1$ and $P_2$, in the model $\mathcal{M}_{D_1 < d_2}$. $P_1$ drops all $m^1$ data units whereas $P_2$ drops all $m^2$ data units.}
		\end{figure}
		\begin{itemize}[-]
			\item For any $k\in\mathbb{N}$ and any $i\in\llbracket 1,q \rrbracket$, path $P_1$ drops the packet containing the data unit $m^1_{i,k}$ and forwards the packet containing the data unit $m^2_{i,k}$ with a delay $d$.
			\item For any $k\in\mathbb{N}$ and any $i\in\llbracket 1,q \rrbracket$, path $P_2$ forwards the packet containing the data unit $m^1_{i,k}$ with a delay $D$ and drops the packet containing the data unit $m^2_{i,k}$.
			\item Any other path $P_j$ with $j \ge 3$ drops all packets.
		\end{itemize}
		Figure~\ref{fig:proof:adv:timeline-paths} shows the trajectory at the output the two adversarial paths, focusing on period $k$ and on flows $f_i$ and $f_{i+1}$. Path $P1$ drops the packets containing the data units $m_{i,k}^1$ and $m_{i+1,k}^1$. It forwards those that contain $m_{i,k}^2$ and $m_{i+1,k}^2$ with a delay $d$. 
		Similarly, $P_2$ drops $m_{i,k}^2$ and $m_{i+1,k}^2$ but forwards $m_{i,k}^1$ and $m_{i+1,k}^1$ with a delay $D$.

		\paragraph{Properties of the paths in $\mathcal{M}_{D_1 < d_2}$}\label{par:adv:paths-prop}
		The delay of the non-lost packets through $P_1$ [resp., through $P_2$] equals $d=D_1$ [resp., $D=d_2$] that belongs to $[d_1,D_1]$ [resp., to $[d_2,D_2]$ ]. Thus the adversarial paths meet Properties~\ref{item:thm:ir:path-delay} and \ref{item:thm:ir:both-fifo} of Theorem~\ref{thm:regulators:ir-instable}.

		\paragraph{Effect of the \acsp{PEF} in $\mathcal{M}_{D_1 < d_2}$}
		The set of parallel \acp{PEF} in Figure~\ref{fig:regulators:ir-overview} receive the sum of the two paths outputs. As per its model in Section~\ref{sec:system-model:function-model}, each \ac{PEF} forwards the first packet containing the data unit.
		For $i \in \llbracket 1,q\rrbracket$, $k \in \mathbb{N}$ and $w \in \{1,2\}$, we denote by $A_{i,k}^w$ the time instant at which the unique packet containing the data unit $m_{i,k}^w$ exits the set of parallel \acp{PEF}.

		By construction of the adversarial paths, for $i \in \llbracket 1,q\rrbracket$ and $k \in \mathbb{N}$, only path $P_1$ forwards a packet containing the data unit $m_{i,k}^2$, released $d$ after its emission by the source. Thus $m_{i,k}^2$ exits the \acp{PEF} as soon as the packet exits $P_1$. We obtain
		\begin{equation}\label{eq:proof:adv:A2}
			\forall i \in \llbracket 1,q \rrbracket, \forall k \in \mathbb{N} \quad A_{i,k}^2 = x_i + k \tau + I +d 
		\end{equation}
		Similarly with $m_{i,k}^1$ that is only forwarded by $P_2$,
		\begin{equation}\label{eq:proof:adv:A1}
			\forall i \in \llbracket 1,q \rrbracket, \forall k \in \mathbb{N} \quad A_{i,k}^1 = x_i + k \tau + D 
		\end{equation}
		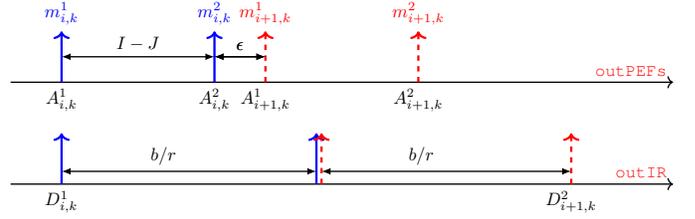
\begin{figure}\centering 
			\resizebox{\linewidth}{!}{

\begin{tikzpicture}
	\tikzstyle{pkn} = [pos=1, anchor=south]
	\tikzstyle{pki} = [pos=0, anchor=north, black]
	\tikzset{
		f/.cd,
		1/.style={blue},
		2/.style={red, dashed},
		3/.style={green, dotted},
		4/.style={black, dashdotted},
		5/.style={black},
		6/.style={black},
		7/.style={black},
		8/.style={black},
		9/.style={black},
		10/.style={black},
		11/.style={black},
		12/.style={black},
		13/.style={black},
		14/.style={black},
	}
	\tikzset{
		role/.cd,
		1/.style={->, line width=1.25pt},
		2/.style={->>, line width=1.25pt}
	}
    \def\mxmax{13}

	\draw[->] (0,0) -- (\mxmax,0) node[pos=1, anchor=east, black] {};
    \node[red, anchor=south east] at (\mxmax,0) {\texttt{outPEFs}};
	\draw[f/1, role/1] (1,0) -- (1, 1)  			node[pkn] {$m_{i,k}^1$} node[pos=0, below, black] {$A_{i,k}^1$};
	\draw[f/1, role/1] (4,		0) -- 	(4,		1) 	node[pkn] {$m_{i,k}^2$} node[pos=0, below, black] {$A_{i,k}^2$}	;
	\draw[f/2, role/1] (5,0) -- (5,1)  				node[pkn] {$m_{i+1,k}^1$} node[pos=0, below, black] {$A_{i+1,k}^1$};
	\draw[f/2, role/1] (8,0) -- 	(8,1) 	  		node[pkn] {$m_{i+1,k}^2$} node[pos=0, below, black] {$A_{i+1,k}^2$};
	\draw[latex-latex] (1,0.5) -- (4,0.5) node[pos=0.5, above] {$I-J$};
	\draw[latex-latex] (4,0.5) -- (5,0.5) node[pos=0.5, above] {$\epsilon$};

	\def\myyy{-2}
	\draw[->] (0,\myyy) -- (\mxmax,\myyy) node[pos=1, anchor=east, black] {};
    \node[red, anchor=south east] at (\mxmax,\myyy) {\texttt{outIR}};
	\draw[f/1, role/1] (1,\myyy) -- (1,\myyy+1)  			node[pos=0, below, black] {$D_{i,k}^1$};
	\draw[f/1, role/1] (6,		\myyy) -- 	(6,		\myyy+1);
	\draw[f/2, role/1] (6.1,\myyy) -- (6.1,\myyy+1) ;
	\draw[f/2, role/1] (11,\myyy) -- 	(11,\myyy+1) 	  	 	node[pos=0, below, black] {$D_{i+1,k}^2$};
	\draw[latex-latex] (1,\myyy+0.25) -- (6,\myyy+0.25) node[pos=0.4, above] {$b/r$};
	\draw[latex-latex] (6.1,\myyy+0.25) -- (11,\myyy+0.25) node[pos=0.4, above] {$b/r$};
	\draw[latex-latex] (4,0.5) -- (5,0.5) node[pos=0.5, above] {$\epsilon$};

\end{tikzpicture}} %
			\caption{\label{fig:proof:adv:timeline-input-ir} Traffic profile for the flows $f_i$ and $f_{i+1}$ at the input of the \ac{IR} (above) and at its output (below). To ease the lecture, the scale is not the same as in Figures~\ref{fig:proof:adv:timeline-local} and \ref{fig:proof:adv:timeline-paths}.}
		\end{figure}

		The top line of Figure~\ref{fig:proof:adv:timeline-input-ir} shows the trajectory at the output of the \acp{PEF} focusing on flows $f_i$ and $f_{i+1}$ and on the $k$-th period of the profile.
		Note that \texttt{outPEFs} is the output of the system denoted by $S$ in Section~\ref{sec:regulators:ir} and is also the input of the \ac{IR} (Figure~\ref{fig:regulators:ir-overview}).

		\paragraph{Properties of the system $S$ between \texttt{in} and \texttt{outPEFs} in $\mathcal{M}_{D_1 < d_2}$}
		For $i \in \llbracket 1,q\rrbracket$ and $n \in \mathbb{N}$ we note that
		\begin{align*}
			A_{i,k}^2 - A_{i,k}^1 &= I + d - D && \text{ \Comment{} (\ref{eq:proof:adv:A2}) and (\ref{eq:proof:adv:A1})}\\
			& = I - J &&\text{ \Comment{} (\ref{eq:proof:adv:inter:j})}\\
			&> 0 && \text{ \Comment{} (\ref{eq:proof:adv:inter:i-greater-than-j})} 
		\end{align*}
		This proves that $m_{i,k}^2$ exit $S$ after $m_{i,k}^1$ for any $i \in \llbracket 1,q \rrbracket$ and any $k\in\mathbb{N}$. Also,
		\begin{align*}
			A_{i,k+1}^1 - A_{i,k}^2 &= \tau + D-d - I && \text{ \Comment{} (\ref{eq:proof:adv:A2}) and (\ref{eq:proof:adv:A1})}\\
			& = \tau - (I-J)  && \\
			& = (q-1)(I-J) + q\epsilon && \text{ \Comment{} (\ref{eq:proof:adv:inter:tau}) and (\ref{eq:proof:adv:inter:phi})}\\ 
			&> 0 && \text{ \Comment{} } q \ge 3 
		\end{align*}
		And this proves that for any $i\in\llbracket 1,q\rrbracket$ and any $k\in\mathbb{N}$, $m_{i,k}^2$ exits $S$ before $m_{i,k+1}^1$.
		Therefore, $S$ is \ac{FIFO} for $f_i$, for any $i\in\llbracket 1,q\rrbracket$.
		Furthermore, each data unit is transported through exactly one path (either $P_1$ or $P_2$), thus $S$ is also lossless. This proves that Property~\ref{item:thm:ir:lossless-fifo-per-flow} of Theorem~\ref{thm:regulators:ir-instable} holds.

		Last, we note that
		\begin{equation}\label{eq:proof:adv:inversion}
		\begin{aligned}
			A_{i+1,k}^1 - A_{i,k}^2 &= x_{i+1} - x_i + D -d -I && \text{ \Comment{} (\ref{eq:proof:adv:A2}) and (\ref{eq:proof:adv:A1})}\\
			&= \phi + J - I && \text{ \Comment{} (\ref{eq:proof:adv:inter:j}) and (\ref{eq:proof:adv:inter:xi})}\\
			&= \epsilon > 0 && \text{ \Comment{} (\ref{eq:proof:adv:inter:phi})}\\
		\end{aligned}
		\end{equation}
		Therefore, $m_{i+1,k}^1$, the first packet of the flow $f_{i+1}$ in the $k$-th period exits the \acp{PEF} $\epsilon$ seconds after the second packet of the flow $f_i$ in the $k$-th period, as described in Figure~\ref{fig:proof:adv:timeline-input-ir}.

		\paragraph{Output of the \ac{IR} in $\mathcal{M}_{D_1 < d_2}$}
		For $i\in\llbracket1,q\rrbracket$, $n\in\mathbb{N}$ and $w \in \{1,2\}$, we denote by $D_{i,k}^w$ the absolute time at which data unit $m_{i,k}^w$ leaves the \ac{IR}.

		The bottom line of Figure~\ref{fig:proof:adv:timeline-input-ir} shows the release time of the data units out of the \ac{IR}.
		Assume for example that the source has been idle for a while, then the regulator is empty and data unit $m_{i,k}^1$ can be released immediately without violating the shaping curve for $f_i$, thus $D_{i,k}^1 = A_{i,k}^1$.
		
		However, data unit $m_{i,k}^2$ arrives at the \ac{IR} too soon with respect to the shaping curve $\sigma_{f_i}$. 
		By applying the equations of the \ac{IR}~\cite{leboudecTheoryTrafficRegulators2018}, we note that the \ac{IR} must delay  $m_{i,k}^2$ and
		\begin{equation}\label{eq:adv:twoDSameFlow}
			\forall i \in \llbracket 1,q\rrbracket, \forall k \in\mathbb{N}, \quad D_{i,k}^2 \ge D_{i,k}^1 + b/r
		\end{equation}
		By~(\ref{eq:proof:adv:inversion}), data unit $m_{i+1,k}^1$ arrives after the data unit $m_{i,k}^2$.
		As the \ac{IR} looks only at the head-of-line packet and is itself a \ac{FIFO} system, we obtain
		\begin{equation}\label{eq:adv:twoDDiffFlow}
			\forall i \in \llbracket 1,q\rrbracket, \forall k \in \mathbb{N}, \quad D_{i+1,k}^1 \ge D_{i,k}^2
		\end{equation}

		Combining Equations (\ref{eq:adv:twoDSameFlow}) and (\ref{eq:adv:twoDDiffFlow}) gives, by induction,
		\begin{equation}\label{eq:firstDvsLastD} \forall k \in \mathbb{N}, \quad D^2_{q,k}\ge D^1_{1,k} + q\frac{b}{r}\end{equation}
		Now we note that
		\begin{align*}
			A^1_{1,k+1}&=x_1 + (k+1)\tau + D\quad &&\text{\Comment{} (\ref{eq:proof:adv:A1})}\\
			&=x_1 + k\tau + q\phi + D &&\text{\Comment{} (\ref{eq:proof:adv:inter:tau})}\\
			&=x_q+k\tau+\phi+D &&\text{\Comment{} (\ref{eq:proof:adv:inter:xi})}\\
			&=x_q + k\tau + I - J + \epsilon + D &&\text{\Comment{} (\ref{eq:proof:adv:inter:phi})}\\
			&=x_q+k\tau+I+d+\epsilon &&\text{\Comment{} (\ref{eq:proof:adv:inter:j})}\\
			&=A_{q,k}^2 + \epsilon &&\text{\Comment{} (\ref{eq:proof:adv:A2})}
		\end{align*}
		Therefore, the first data unit of the $(k+1)$-th period of $f_1$ arrives $\epsilon$ seconds after the second data unit of the $k$-th period of the last flow $f_q$. The \ac{IR} being \ac{FIFO}, we have
		\begin{equation}\label{eq:adv:aroundAPeriod}
		\forall k \in \mathbb{N},\quad D^1_{1,k+1} \ge D^2_{q,k}
		\end{equation}
		which, combined with (\ref{eq:adv:twoDDiffFlow}), gives
		\begin{equation}\label{eq:adv:loop}
			\forall k \in \mathbb{N},\quad D^1_{1,k+1} \ge D^1_{1,k} + q\frac{b}{r}
		\end{equation}
		At period $k=0$, the network is empty and $D^1_{1,0} = A^1_{1,0} = x_1$. The induction of (\ref{eq:adv:loop}) thus gives
		\begin{equation}\label{eq:adv:loopb}
			\forall k \in \mathbb{N},\quad D^1_{1,k}\ge x_1 + k q\frac{b}{r}
		\end{equation}
		And the delay, through the \ac{IR}, suffered by the first data unit of the $k$-th period of the first flow $f_1$ is
		\begin{align*}
			&D^1_{1,k} - A^1_{1,k}&&\\
			&\ge x_1 + k q\frac{b}{r} - x_1-k\tau-D&&\text{\Comment{} (\ref{eq:adv:loopb}) and (\ref{eq:proof:adv:A1})}\\
			&\ge -D + k q\left(\frac{b}{r} - \phi\right)&&\text{\Comment{} (\ref{eq:proof:adv:inter:tau})}\\
		\end{align*}
		By~(\ref{eq:proof:adv:IvsBR}), $\frac{b}{r} - \phi >0$. Thus the above delay lower-bound diverges as $k$ increases and Property \ref{item:thm:ir:undbounded} of the Theorem holds.
	
		\subsubsection{Adversarial model for the case $d_2 < D_1$}
		
		The adversarial model $\mathcal{M}_{d_2<D_1}$ follows the same principle as the adversarial model $\mathcal{M}_{D_1<d_2}$ described above.
		In the following, we detail only the differences.
		
		\paragraph{Constants of $\mathcal{M}_{d_2<D_1}$}\label{par:adv1:constants}
		By assumption, $D_1 - d_2 > 0$.
		Furthermore, $q_{\min}$ now equals $3$ and $q \ge q_{\min}$, $\left(\frac{q-2}{2}\right)\frac{b}{r} > 0$.

		We hence select $J$ such that
		\begin{equation}\label{eq:adv1:J}
			0 < J < \min\left(\frac{q-2}{2}\frac{b}{r},D_1 - d_2\right)
		\end{equation}
		And we re-define
		\begin{equation}\label{eq:adv1:D-d}
			D \triangleq D_1 \qquad d \triangleq D - J
		\end{equation}
		As $J < D_1 -d_2$=, $J >0$, and $D_2 \ge D_1$ by on the indexes, we obtain $D_2\ge D_1>d>d_2$ thus 
		\begin{equation}\label{eq:adv1:dInInterval}
			d \in [d_2,D_2]	
		\end{equation}
		
		By definition, $ \frac{2}{q-2} J < \frac{b}{r}$, thus we define $\epsilon$, $I$, $\phi$, $\tau$ and $x_i$ as in Appendix~\ref{par:adv:constants}, \emph{i.e.}, per Equations (\ref{eq:proof:adv:inter:epsilon}), (\ref{eq:proof:adv:inter:I}), (\ref{eq:proof:adv:inter:phi}), (\ref{eq:proof:adv:inter:tau}) and (\ref{eq:proof:adv:inter:xi}).

		\paragraph{Properties on the constants of $\mathcal{M}_{d_2<D_1}$}
		None of the properties established in Appendix~\ref{par:adv:constants-prop} depends on the definition of $J$, $d$ or $D$. They are all obtained thanks to the definitions of the other constants. As we re-use the same definitions, all the properties obtained in Appendix~\ref{par:adv:constants-prop} are also valid for $\mathcal{M}_{d_ 2 < D_1}$.

		\paragraph{Adversarial traffic generation at the source in $\mathcal{M}_{d_2 < D_1}$}
		The adversarial model $\mathcal{M}_{d_2 < D_1}$ uses the same traffic generation as $\mathcal{M}_{D_1 < d_2}$. It is described in Appendix~\ref{par:adv:source} and summarized in Figure~\ref{fig:proof:adv:timeline-local}.
	
		\paragraph{Properties on the traffic generation at the source in $\mathcal{M}_{d_2 < D_1}$}
		The traffic generation is not modified, thus the properties established in Appendix~\ref{par:adv:source-prop} also hold for $\mathcal{M}_{d_2<D_1}$. In particular, Property 1/ of Theorem~\ref{thm:regulators:ir-instable} holds.

		\paragraph{Adversarial paths in  $\mathcal{M}_{d_2 < D_1}$} With respect to the model $\mathcal{M}_{D_1 < d_2}$, the adversarial model $\mathcal{M}_{d_2 < D_1}$ simply flips the the roles of each paths. Specifically, 
		\begin{itemize}[-]
			\item For any $k\in\mathbb{N}$ and any $i\in\llbracket 1,q \rrbracket$, path $P_1$ forwards the packet containing the data unit $m^1_{i,k}$ with a delay $D$ and drops the packet containing the data unit $m^2_{i,k}$.
			\item For any $k\in\mathbb{N}$ and any $i\in\llbracket 1,q \rrbracket$, path $P_2$ drops the packet containing the data unit $m^1_{i,k}$ and forwards the packet containing the data unit $m^2_{i,k}$ with a delay $d$.
			\item Any other path $P_j$ with $j \ge 3$ drops all packets.
		\end{itemize}
		The output of both paths is shown in Figure~\ref{fig:proof:adv:timeline-paths-flip}.
		We can note the symmetry with Figure~\ref{fig:proof:adv:timeline-paths}.

		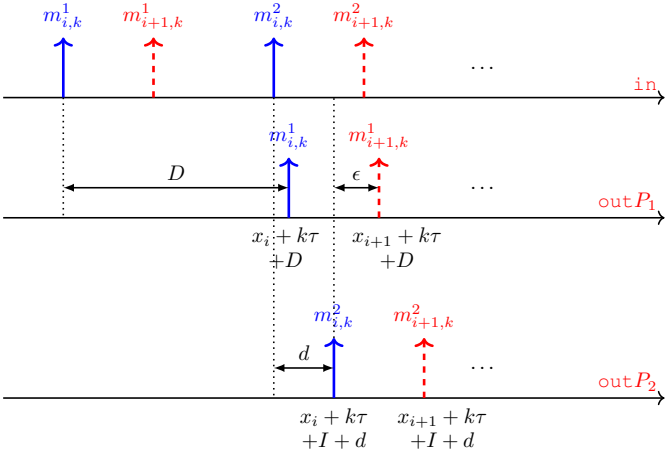
\begin{figure}\centering 
			\resizebox{\linewidth}{!}{

\begin{tikzpicture}

	\tikzstyle{pkn} = [pos=1, anchor=south]
	\tikzstyle{pki} = [pos=0, anchor=north, black]
	\tikzset{
		f/.cd,
		1/.style={blue},
		2/.style={red, dashed},
		3/.style={green, dotted},
		4/.style={black, dashdotted},
		5/.style={black},
		6/.style={black},
		7/.style={black},
		8/.style={black},
		9/.style={black},
		10/.style={black},
		11/.style={black},
		12/.style={black},
		13/.style={black},
		14/.style={black},
	}
	\tikzset{
		role/.cd,
		1/.style={->, line width=1.25pt},
		2/.style={->>, line width=1.25pt}
	}

	\def\mxmax{11}
	\draw[->] (0,0) -- (\mxmax,0) node[pos=1, anchor=east, black] {};
	\node[red, anchor=south east] at (\mxmax,0) {\texttt{in}};
	\draw[f/1, role/1] (1,0) -- (1,1)       node[pkn] {$m_{i,k}^1$}     ;
	\draw[f/1, role/1] (4.5,0) -- (4.5,1)   node[pkn] {$m_{i,k}^2$}     ;

    \draw[f/2, role/1] (2.5,0) -- (2.5,1)     node[pkn] {$m_{i+1,k}^1$}  ;
	\draw[f/2, role/1] (6,0) -- (6,1)         node[pkn] {$m_{i+1,k}^2$}  ;
	\node at (8,0.5) {\ldots};

	\draw[->] (0,-5) -- (\mxmax,-5) node[pos=1, anchor=east, black] {};
	\node[red, anchor=south east] at (\mxmax,-5) {\texttt{out}$P_2$};
	\draw[f/1, role/1] (5.5,-5) -- (5.5, -4)  				node[pkn] {$m_{i,k}^2$};
	\draw[f/2, role/1] (7,		-5) -- 	(7,		-4) 	  	node[pkn] {$m_{i+1,k}^2$};
	\node at (8,-4.5) {\ldots};

	\node[black, anchor=north] at (5.5,-5) {\makecell{$x_{i}+k\tau$\\$+I+d$}};
	\node[black, anchor=north] at (7.3,-5) {\makecell{$x_{i+1}+k\tau$\\$+I+d$}};  

	\draw[->] (0,-2) -- (\mxmax,-2) node[pos=1, anchor=east, black] {};
	\node[red, anchor=south east] at (\mxmax,-2) {\texttt{out}$P_1$};
	\draw[f/1, role/1] (4.75,-2) -- (4.75, -1)  				node[pkn] {$m_{i,k}^1$};
	\draw[f/2, role/1] (6.25,		-2) -- 	(6.25,		-1) 	  	node[pkn] {$m_{i+1,k}^1$};
	\node at (8,-1.5) {\ldots};

	\node[black, anchor=north] at (4.7,-2) {\makecell{$x_{i}+k\tau$\\$+D$}};
	\node[black, anchor=north] at (6.55,-2) {\makecell{$x_{i+1}+k\tau$\\$+D$}};  

	\draw[line width=0.75pt, dotted] (1,0) -- (1,-2);
	\draw[line width=0.75pt, dotted] (4.5,0) -- (4.5,-5);
	\draw[latex-latex, line width=0.75pt] (1,-1.5) -- (4.75,-1.5) node[pos=0.5,above] {$D$};
	\draw[latex-latex, line width=0.75pt] (4.5,-4.5) -- (5.5,-4.5) node[pos=0.5,above] {$d$};

	\draw[line width=0.75pt, dotted] (5.5,-4) -- (5.5,-0);
	\draw[latex-latex, line width=0.75pt] (5.5,-1.5) -- (6.25,-1.5) node[pos=0.5,above]{$\epsilon$};
\end{tikzpicture}} %
			\caption{\label{fig:proof:adv:timeline-paths-flip} Traffic profile for the flows $f_i$ and $f_{i+1}$ at the output of the two adversarial paths $P_1$ and $P_2$, in the model $\mathcal{M}_{d_2 < D_1}$. $P_1$ drops all $m^2$ data units whereas $P_2$ drops all $m^2$ data units. With respect to Figure~\ref{fig:proof:adv:timeline-paths}, the roles of $P_1$ and $P_2$ have been exchanged.}
		\end{figure}

		\paragraph{Properties of the paths in $\mathcal{M}_{d_2 < D_1}$}
		The packets not lost in $P_1$ have the same delay through $P_1$ equal to $D$.
		Similarly, the packets not lost in $P_2$ have the same delay through $P_2$ equal to $d$.
		Thus both $P_1$ and $P_2$ are \ac{FIFO} and Property~\ref{item:thm:ir:both-fifo} of Theorem~\ref{thm:regulators:ir-instable} hold.
		
		Furthermore, by (\ref{eq:adv1:D-d}),  $D = D_1$ and by (\ref{eq:adv1:dInInterval}), $d\in[d_2,D_2]$.
		Thus Property~\ref{item:thm:ir:path-delay} holds.

		\paragraph{Effect of the \acsp{PEF} in $M_{d_2<D_1}$}
		As for the $\mathcal{M}_{D_1 < d_2}$ model, each data unit arrives in a unique packet at the \acp{PEF} ($m^1$ data units arrive only through $P_1$ and $m^2$ data units arrive only through $P^2$).
		Thus the \acsp{PEF} are transparent and forward the sum of both output traffic, \texttt{out}$P_1$ and \texttt{out}$P_2$ from Figure~\ref{fig:proof:adv:timeline-paths-flip}.
		We observe that the sum of them gives the same output as on the first line of Figure~\ref{fig:proof:adv:timeline-input-ir}.

		Therefore, all the remaining steps of the proof (properties on the system $S$, output of the \ac{IR} with diverging delays) can be followed as in the model $\mathcal{M}_{D_1<d_2}$.

		\subsubsection{Adversarial model for the case $d_2 = D_1$}
		If $d_2 = D_1$, then $q_{\min} = 3$.
		By Condition~(e) of Theorem~\ref{thm:regulators:ir-instable}, one of the two intervals $[d_1, D_1]$ or $[d_2,D_2]$ has a strictly positive length. 
		Assume for example that $d_2 < D_2$.
		Then we simply select $d_2'$ such that 
		\begin{equation}\label{eq:adv2:d2prime}
			d_2 < d_2' < \min\left(D_2,D_1+\frac{b}{2r}\right)
		\end{equation}
		We obtain    
		\begin{equation*}
			\left\lfloor \frac{2r\left|d_2'-D_1\right|^+}{b} + 2\right\rfloor + 1 = 3
		\end{equation*}
		because 
		\begin{equation*}
			\frac{2r\left|d_2'-D_1\right|^+}{b} < 1
		\end{equation*}
		by choice of $d_2'$.
		This means that we can apply model $\mathcal{M}_{d_2' < D_1}$ with parameters $q,r,b,d_1,d_2',D_1,D_2$ and the same number of flows ($q \ge 3$).
		This model will provide Properties~\ref{item:thm:ir:first} to \ref{item:thm:ir:last} of Theorem~\ref{thm:regulators:ir-instable} for the choice of parameters $q,r,b,d_1,d_2',D_1,D_2$, thus providing Properties~\ref{item:thm:ir:first} to \ref{item:thm:ir:last} of Theorem~\ref{thm:regulators:ir-instable} for the choice parameters $r,b,d_1,d_2,D_1,D_2$.
		
	\end{proof}

\subsection{Proof of Corollary~\ref{cor:regulators:ir-after-non-fifo}}
\begin{proof}[Proof of Corollary~\ref{cor:regulators:ir-after-non-fifo}]\label{proof:cor:regulators:ir-after-non-fifo}
    We simply construct $\mathcal{S}$ as a system containing a \acf{PRF}, two alternative paths $P_1$, $P_2$ and a set of \acp{PEF}, as in Figure~\ref{fig:regulators:ir-overview}. We then apply Theorem~\ref{thm:regulators:ir-instable} with lossless and \ac{FIFO} paths $P_1$ and $P_2$ that have both the same delay interval $[d_1,D_1] = [d_2,D_2] = [0,D_{\max}]$.
\end{proof}

\subsection{Proof of Theorem~\ref{thm:regulators:preof-for-free}}
\begin{proof}[Proof of Theorem~\ref{thm:regulators:preof-for-free}]\label{proof:thm:regulators:preof-for-free}
    We denote by $[d^\dagger, D^\dagger]$ the lower and upper delay bounds of the non-lost data units through the system $\mathcal{S}^\dagger$ between the output of $a$ and the output of the \ac{POF} (Figure~\ref{fig:regulators:preof-for-free}).

    The output of a vertex corresponds to the output of the packetizer, thus the flow aggregate $\mathcal{F}$ is packetized at the output of vertex $a$.
    We can hence define the packet sequence $(A,L,F)$ for the aggregate $\mathcal{F}$ as in \cite[\S II.A]{leboudecTheoryTrafficRegulators2018}: 
    
    \begin{itemize}
        \item  $A$ is the sequence of the arrival times at the observation point $a^*$ for the data units that belong to the flow aggregate $\mathcal{F}$. 
        $A$ is a wide-sense increasing sequence.
        \emph{I.e.,} $A_n$ is the arrival time at $a^*$ of the $n$-th data unit of the aggregate $\mathcal{F}$.
        \item $L$ is the sequence of packet length for the above data units.
        \emph{I.e.,} $L_n$ is the length of the packet that transports the $n$-th data unit that arrives at $a^*$ and belongs to $\mathcal{F}$.
        \item $F$ is the sequence of flow identifiers for the above data units.
        \emph{I.e.,} $F_n=f$ means that the $n$-th data unit of the aggregate $\mathcal{F}$ at $a^*$ belongs to flow $f$.
    \end{itemize}

    We also define the $\Pi^f$ regulator for each flow $f$ of $\mathcal{F}$ that corresponds to the shaping curve $\sigma_{f,n}$ of the \ac{IR} $\texttt{REG}_n(\mathcal{F},a)$ \cite[IV.A]{leboudecTheoryTrafficRegulators2018}.
    By configuration of $\texttt{REG}_n(\mathcal{F},a)$, each flow $f$ of the aggregate is $\sigma_{f,n}$-constrained thus $\Pi^f$-regular at $a^*$, the input of the systems $\mathcal{S}$, $\mathcal{S}^{\dagger}$ and $S'$.
    
    $-$ 
    If $\mathcal{S}$ is \textbf{lossless} for $\mathcal{F}$, we apply \cite[Theorem~4]{mohammadpourPacketReorderingTimeSensitive2020} and obtain $d^\dagger = d$ and $D^\dagger = D$. 
    
    Then, by definition of the \ac{POF} and considering its configuration $\texttt{POF}_n(\{f\},a)$, system $\mathcal{S}^\dagger$ is \ac{FIFO} and lossless for the aggregate $\mathcal{F}$ processed by the regulator.
    Therefore, applying~\cite[Theorem~5]{leboudecTheoryTrafficRegulators2018} gives $d' = d^\dagger$ and $D' = D^\dagger$.

    $-$ 
    If $\mathcal{S}$ is \textbf{not lossless} for $\mathcal{F}$, then the application of \cite[Theorem~4]{mohammadpourPacketReorderingTimeSensitive2020} gives $d^\dagger = d$ and $D^\dagger = D+T$. 

    Then, by definition of the \ac{POF} and considering its configuration $\texttt{POF}_n(\{f\},a)$, system $\mathcal{S}^\dagger$ is \ac{FIFO} but not lossless for the aggregate $\mathcal{F}$ processed by the regulator.

    Like in the proof of Proposition~\ref{prop:appendix:ac-after-lossy}, we decompose the packet sequence $(A,L,F)$ at the input of $\mathcal{S}^\dagger$ into the sub-sequences $(A_1,L_1,F_1)$ and $(A_2,L_2,F_2)$ that correspond respectively to the data units that are not lost inside  $\mathcal{S}^\dagger$ and to the data units that are lost inside  $\mathcal{S}^\dagger$.

    We consider the cumulative functions $R_{f,1}$ and $R_{f,2}$ of each flow $f$ that correspond to the sequence $(A_1,L_1,F_1)$ and $(A_2,L_2,F_2)$, respectively.
    Then, $R_f\triangleq R_{f,1}+R_{f,2}$, the overall cumulative function, corresponds to the sequence $(A,L,F)$ for flow $f$ thus $R_{f}$ is $\sigma_{f,n}$-constrained.
    Re-using an argument from the proof of Proposition~\ref{prop:appendix:ac-after-lossy}, the cumulative sub-function $R_{f,1}$ remains $\sigma_{f,n}$-constrained.
    Thus flow $f$ in sub-sequence $(A_1,L_1,F_1)$ remains $\Pi^f$-regular.

    Therefore, we apply \cite[Theorem~5]{leboudecTheoryTrafficRegulators2018} on the sub-sequence $(A_1,L_1,F_1)$ and we obtain that for the corresponding \ac{IR} output sequence $(D_1,L_1,F_1)$, the delay through $\mathcal{S'}$ verifies  $d' = d^\dagger$ and $D' = D^\dagger$.
    The delay of the non-lost data units through $\mathcal{S'}$ is hence within $[d',D'] = [d,D+T]$.
\end{proof}



	\ifCLASSOPTIONcaptionsoff
	\newpage
	\fi

	
	
	\bibliographystyle{ieeetr}
	%

\end{document}